\pdfoutput=1
\documentclass[11pt]{article}
\def\SplitSupplement{}
\usepackage[a4paper,margin=1in]{geometry}
\usepackage[T1]{fontenc}
\usepackage[tt=false]{libertine}
\usepackage{amsmath,amsthm,mathtools}
\usepackage[libertine]{newtxmath}
\usepackage{booktabs,tabularx,multirow,makecell,array}
\usepackage{graphicx,xcolor}
\usepackage{tikz}
\usetikzlibrary{arrows.meta,positioning,shapes.geometric,fit,calc,patterns}
\usepackage{pgfplots}
\pgfplotsset{compat=1.18}
\usepackage[numbers,sort&compress]{natbib}
\usepackage{enumitem}
\usepackage{algorithm,algpseudocode}
\usepackage{longtable,pdflscape}
\usepackage[hidelinks]{hyperref}
\usepackage{orcidlink}
\usepackage{aliascnt}
\usepackage{placeins}
\usepackage[capitalise]{cleveref}

\providecommand{\tool}{RiskTrace}

\providecommand{\R}{\mathbb{R}}

\providecommand{\one}{\mathbf{1}}
\providecommand{\cl}{\operatorname{cl}}
\providecommand{\conv}{\operatorname{conv}}
\providecommand{\sign}{\operatorname{sign}}
\providecommand{\tagV}{\textup{[V]}}
\providecommand{\tagI}{\textup{[I]}}
\providecommand{\tagU}{\textup{[U]}}
\providecommand{\Dstruct}{\textup{Proved (full)}}
\providecommand{\Dscoped}{\textup{Proved (scoped)}}
\providecommand{\Ncell}{\textup{Refuted}}
\providecommand{\Tstruct}{\textup{structural tier}}
\providecommand{\Treg}{\textup{regulatory tier}}

\definecolor{fbBlue}{HTML}{1F4E79}     %
\definecolor{fbBlueMid}{HTML}{6A8EAE}  %
\definecolor{fbBlueLight}{HTML}{DCE6EF}%
\definecolor{fbRed}{HTML}{A23B2A}      %
\definecolor{fbBeige}{HTML}{D8C8A8}    %
\definecolor{fbSand}{HTML}{EFE6D5}     %
\definecolor{fbInk}{HTML}{333333}      %
\definecolor{fbGrey}{HTML}{8C8C8C}     %
\definecolor{fbGreyLight}{HTML}{D9D9D9}%
\pgfplotsset{
  frtb/.style={
    axis line style={fbGrey, line width=0.4pt},
    tick style={fbGrey, line width=0.4pt},
    every tick label/.append style={font=\scriptsize, text=fbInk},
    label style={font=\scriptsize, text=fbInk},
    legend style={font=\scriptsize, draw=none, fill=none, text=fbInk},
    legend cell align=left,
    grid style={fbGreyLight, line width=0.3pt},
    major grid style={fbGreyLight, line width=0.3pt},
    axis background/.style={fill=white},
  },
}

\newtheorem{theorem}{Theorem}
\newaliascnt{lemma}{theorem}\newtheorem{lemma}[lemma]{Lemma}\aliascntresetthe{lemma}
\newaliascnt{proposition}{theorem}\newtheorem{proposition}[proposition]{Proposition}\aliascntresetthe{proposition}
\newaliascnt{corollary}{theorem}\newtheorem{corollary}[corollary]{Corollary}\aliascntresetthe{corollary}
\theoremstyle{definition}
\newaliascnt{definition}{theorem}\newtheorem{definition}[definition]{Definition}\aliascntresetthe{definition}
\newtheorem{example}{Example}
\theoremstyle{remark}
\newaliascnt{remark}{theorem}\newtheorem{remark}[remark]{Remark}\aliascntresetthe{remark}
\crefname{lemma}{Lemma}{Lemmas}\crefname{proposition}{Proposition}{Propositions}
\crefname{corollary}{Corollary}{Corollaries}\crefname{definition}{Definition}{Definitions}
\crefname{remark}{Remark}{Remarks}\crefname{example}{Example}{Examples}
\crefname{figure}{Figure}{Figures}\crefname{table}{Table}{Tables}

\title{RiskTrace: Exact Attribution for Nonsmooth Financial Risk,\\ Capital, and Margin Frameworks}
\author{Yuhe Sui$^{1,2}$\,\orcidlink{0009-0002-7456-0804} \qquad Yingzhi Tang$^{1,3}$\,\orcidlink{0009-0004-7571-0100}\\[4pt]
  \small $^{1}$Quantitative Research Society, Singapore\\
  \small $^{2}$Nanyang Technological University, Singapore\\
  \small $^{3}$The Chinese University of Hong Kong, Shenzhen, China}
\date{November 2026}
\hypersetup{
  pdftitle={RiskTrace: Exact Attribution for Nonsmooth Financial Risk, Capital, and Margin Frameworks},
  pdfauthor={Yuhe Sui, Yingzhi Tang}
}

\begin{document}
\maketitle

\begin{abstract}
Across modern financial markets, prudential frameworks, clearing standards, and internal risk models aggregate portfolio exposure into a single capital, margin, valuation, or solvency charge, yet risk steering, accounting, and hedging require exact attribution to individual positions. Prior scholarship resolved attribution mainly for smooth, scale-invariant risk measures or isolated models, leaving most major regulatory and market-practice frameworks without a complete, exact attribution theory once historical windows, nonsmooth ties, discontinuous switches, singular roots, concentration multipliers, fixed collateral offsets, and absolute caps enter. We introduce \tool{}, a unified mathematical framework that completes exact, reconciled attribution across eleven major capital, margin, credit, valuation, interest-rate, and solvency models represented as typed computational graphs. We prove exact node-by-node accounting reconciliation across every graph, separate genuine marginal sensitivities from accounting splits, extend path integration to non-smooth and scale-breaking rules without homogeneity, and establish domain-closed finite coalition games. Across all thirty-three pairs combining marginal gradient, path integral, and coalition Shapley attribution across all eleven formalised structural models, every case is resolved: seventeen hold on full structural domains, eight on explicit restricted domains, and eight are refuted by exact counterexamples. Even though seven of eleven frameworks break classical gradient completeness, path integration restores exact attribution on ten of eleven and coalition Shapley achieves exact reconciliation on all eleven, unified by three cross-framework theorems with transparent source-verification provenance.
\end{abstract}

\section{Introduction and Motivation}\label{sec:intro}

Across global financial markets, prudential frameworks, market-infrastructure margin standards, and institutional risk models summarise portfolio exposure in a single aggregate charge $K(x)$---whether regulatory capital, bilateral or cleared initial margin, counterparty valuation adjustments, or insurance solvency requirements---while trading, finance, and supervisory decisions must be executed at the level of individual positions. Across eleven cornerstone frameworks spanning market risk (FRTB-IMA, discrete ES/VaR, FRTB-SA), counterparty credit and margin (ISDA SIMM, BA-CVA/SA-CVA, SA-CCR, bilateral double-hinge XVA, classic CME SPAN), and enterprise or insurance solvency (the Basel Output Floor, the Solvency~II standard formula, and IRRBB $\Delta\mathrm{EVE}$), published rulebooks fix the portfolio formula while leaving granular attribution largely unspecified. Prior scholarship addressed only isolated fragments of this landscape, chiefly smooth Euler rules for individual risk measures or single frameworks, leaving most of these eleven methodologies without a complete, exact attribution theory.

This paper closes that gap across all eleven frameworks. Formalising each methodology as a typed computational graph over its specified dated or static positions, risk-factor coordinates, or business-cell replication weights (yielding trade-level attribution wherever a verified trade-to-input map is supplied), we establish exact, reconciled attribution for every framework and prove precisely where classical marginal rules hold, where path integration restores completeness, and where only finite coalition games survive.

Different users ask different questions of the same number. A risk manager wants the \emph{local marginal}, the sensitivity of $K$ to today's positions; a desk head wants the \emph{removal effect} of deleting a trade; a finance controller needs an additive \emph{accounting ledger} that sums to the reported figure; a model validator audits \emph{historical-origin} contributions together with model-rule reconciliation \citep{frb2011sr117,pra2023ss123,bcbs2013riskdata}; and a quoting or hedging engine compares discrete \emph{hedge choices}. These five questions are mathematically distinct, and an allocation that answers one of them need not answer any other.

\paragraph{Why regulatory capital resists a single attribution vector.}
Classical capital allocation assumes that $K$ is positively homogeneous of degree one and continuously differentiable. Under those assumptions the Euler allocation $x_i\partial_i K(x)$ is complete, $\sum_i x_i\partial_i K(x) = K(x)$ \citep{tasche2008euler,kalkbrener2005axiomatic,buch2008coherent,acerbi2002coherence,balog2017properties}. It coincides with Aumann--Shapley integration from the zero portfolio \citep{aumann1974values,sundararajan2017axiomatic,sundararajan2020many} and agrees locally with incremental trading directions.

Regulatory, clearing and insurance rules depart from this picture in four ways. First, they contain \emph{kinks, ties and shared subgraphs}: tail thresholds in expected shortfall, $0.999$-quantiles, priority matchers in clearing margin, clamps in correlation aggregation and floors against historical averages all create piecewise-smooth corners. When one intermediate node feeds several downstream terms, choosing tie branches independently can violate the composite Clarke generalized gradient \citep{clarke1983optimization,scholtes2012introduction}. Second, they contain \emph{stored history and external accounting leaves}. Internal-model capital compares today's book with $60$-day or $12$-week historical charges, while the Output Floor and Solvency~II add non-trading charges, operational risk, future discretionary benefits or deferred taxes as affine terms. Sensitivities to today's book then miss reported capital by an exact residual.

Third, some rules contain \emph{discontinuous switches and singular square roots}. The alternative cross-bucket correlation branch and the curvature scenario-sign rule of the FRTB Standardised Approach, as well as indefinite quadratic forms under a square root, create jumps or infinite directional derivatives. Local Lipschitz continuity is then lost and the Clarke generalized gradient no longer exists. Fourth, several rules contain \emph{super-linear concentration factors, fixed collateral offsets and absolute caps}. ISDA SIMM multiplies sensitivities by $\max\{1,\sqrt{|\sigma_g|/T_g}\}$, which scales with degree $\tfrac32$. SA-CCR combines the exponential multiplier $\exp((V-C)/(2(1-f)A))$ with fixed held collateral $C$, IRRBB imposes fixed behavioural amount limits $H$, and bilateral XVA uses fixed margin thresholds $(b,H)$. In each of these cases the Euler allocation fails to be complete on \emph{open} sets of smooth portfolios.

Once a rule embeds kinks, offsets, concentration multipliers, or branch switches, trade responsibility ceases to be a single vector, and classical gradient rules break on open sets of realistic portfolios. Closing this gap requires a unified theory that keeps the five economic questions distinct and establishes, for every framework, the exact domain on which Euler, path-integral, and Shapley attribution succeed or fail.

\subsection{One book, three rules}\label{sec:toy}

The three mechanisms that organise the paper already appear on the smallest possible book: two trades held in unit size, $w=(1,1)$. We evaluate this book under three capital rules, one per mechanism, and compare four answers for each rule. These are the Euler allocation (or the Clarke set where the gradient fails), the straight-ray path integral along $\gamma(t)=tw$, the Shapley value of the coalitions $\varnothing,\{1\},\{2\},\{1,2\}$, and the removal effects $K(w)-K(w-w_ie_i)$. \Cref{fig:intuition-3regimes} plots the three rules, and the table at the end of this subsection collects every number.

\begin{figure}[htbp]
\centering
﻿%
\resizebox{\linewidth}{!}{%
\begin{tikzpicture}[font=\scriptsize, text=fbInk]
\pgfplotsset{
  introaxis/.style={frtb, scale only axis, width=3.75cm, height=3.25cm,
    title style={font=\scriptsize, text=fbInk, yshift=-1pt},
    xlabel style={yshift=2pt}, ylabel style={yshift=-3pt}},
}
\begin{axis}[introaxis, name=pa,
  xmin=-0.12, xmax=1.62, ymin=-0.32, ymax=1.62,
  xtick={0,0.5,1,1.5}, ytick={0,0.5,1,1.5},
  xlabel={$w_1$ (positions) and $A_1$ (allocations)},
  ylabel={$w_2$ and $A_2$},
  title={(a) kink: $K_1=\max\{w_1,w_2\}$}]
  \addplot[fbGreyLight, line width=0.6pt] coordinates {(0,1.4) (1.4,1.4) (1.4,0)};
  \addplot[fbGrey, line width=0.5pt] coordinates {(0,1) (1,1) (1,0)};
  \addplot[fbRed, dashed, line width=0.6pt] coordinates {(0,0) (1.6,1.6)};
  \fill[fbBlueLight] (axis cs:1,1) rectangle (axis cs:1.42,1.42);
  \draw[-{Stealth[length=1.4mm]}, fbBlue, line width=0.7pt] (axis cs:1,1) -- (axis cs:1.42,1);
  \draw[-{Stealth[length=1.4mm]}, fbBlue, line width=0.7pt] (axis cs:1,1) -- (axis cs:1,1.42);
  \node[fbBlue, anchor=south west, inner sep=1pt] at (axis cs:1.03,1.03) {fan};
  \addplot[fbBlue, line width=1.6pt] coordinates {(1,0) (0,1)};
  \addplot[only marks, mark=*, mark size=1.4pt, fbBlue] coordinates {(1,0) (0,1)};
  \addplot[only marks, mark=diamond*, mark size=2pt, fbBlueMid] coordinates {(0.5,0.5)};
  \draw[fbBlue, line width=0.3pt] (axis cs:0.51,0.43) -- (axis cs:0.70,-0.05);
  \node[fbBlue, anchor=north west, inner sep=1pt] at (axis cs:0.62,-0.05) {equal split};
  \node[fbBlue, anchor=north west, inner sep=1pt, rotate=-45] at (axis cs:0.16,0.80) {Clarke set};
  \addplot[only marks, mark=x, mark size=2.4pt, fbRed, line width=0.8pt] coordinates {(0,0)};
  \node[fbRed, anchor=north west, inner sep=1pt] at (axis cs:0.03,-0.05) {removal};
  \addplot[only marks, mark=*, mark size=1.5pt, fbInk] coordinates {(1,1)};
  \node[anchor=north west, inner sep=1pt] at (axis cs:1.02,0.98) {$w$};
\end{axis}
\begin{axis}[introaxis, name=pb, at={($(pa.outer east)+(0.25cm,0)$)}, anchor=outer west,
  xmin=0, xmax=1, ymin=0, ymax=7.4,
  xtick={0,0.333333,0.6667,1}, xticklabels={$0$,$\tfrac13$,$\tfrac23$,$1$},
  ytick={0,1,2,4,6},
  xlabel={ray parameter $t$},
  ylabel={$K_2(t,t)$},
  title={(b) fixed offset: $K_2$ along $t\mapsto tw$}]
  \addplot[fbGrey, dashed, line width=0.5pt, domain=0:1, samples=2] {2+2*x};
  \addplot[fbGrey, dashed, line width=0.5pt, domain=0:1, samples=2] {1+5*x};
  \addplot[fbBlue, line width=1.2pt, domain=0:1, samples=61] {max(2+2*x,1+5*x)};
  \draw[fbInk, densely dotted] (axis cs:0.333333,1.3) -- (axis cs:0.333333,2.666667);
  \node[anchor=west, inner sep=1pt] at (axis cs:0.345,1.75) {$t^\ast$};
  \addplot[only marks, mark=*, mark size=1.4pt, fbBlue] coordinates {(0,2) (1,6)};
  \addplot[only marks, mark=o, mark size=1.6pt, fbRed, line width=0.7pt] coordinates {(0,1)};
  \node[anchor=south east, inner sep=1.5pt] at (axis cs:0.985,6.08) {$K_2(w)=6$};
  \draw[fbInk, line width=0.3pt] (axis cs:0.015,2.08) -- (axis cs:0.07,3.25);
  \node[anchor=south west, inner sep=1.5pt] at (axis cs:0.02,3.2) {$K_2(0)=2$};
  \node[fbRed, anchor=west, inner sep=1.5pt] at (axis cs:0.05,0.7) {intercept $1$};
  \node[anchor=north west, align=left, inner sep=1.5pt] at (axis cs:0.03,7.25)
    {\scriptsize Euler: $5=6-1$\\ \scriptsize path, Shapley:\\ \scriptsize $4=6-2$};
\end{axis}
\begin{axis}[introaxis, name=pc, at={($(pb.outer east)+(0.25cm,0)$)}, anchor=outer west,
  xmin=-0.5, xmax=0.5, ymin=0, ymax=4.3,
  xtick={-0.5,0,0.5}, ytick={0,1,2,3,4},
  xlabel={perturbation $\varepsilon$},
  ylabel={$K_3(1+\varepsilon,1)$},
  title={(c) switch: $K_3$ across $w_1=w_2$}]
  \addplot[fbBlue, line width=1.2pt, domain=-0.5:0, samples=2] {2+x};
  \addplot[fbBlue, line width=1.2pt, domain=0:0.5, samples=2] {3+2*x};
  \draw[fbRed, densely dotted, line width=0.7pt] (axis cs:0,2) -- (axis cs:0,3);
  \node[fbRed, anchor=west, inner sep=1.5pt] at (axis cs:0.02,2.5) {jump $1$};
  \addplot[only marks, mark=o, mark size=1.6pt, fbBlue, mark options={fill=white}] coordinates {(0,2)};
  \addplot[only marks, mark=*, mark size=1.6pt, fbBlue] coordinates {(0,3)};
  \node[anchor=south east, align=right, inner sep=1.5pt] at (axis cs:0.44,0.08)
    {coalitions $\varnothing,\{1\},\{2\},N$:\\ $v=(0,\,2,\,1,\,3)$\\ Shapley $\phi=(2,\,1)$};
  \node[anchor=north west, align=left, inner sep=1.5pt] at (axis cs:-0.48,4.2)
    {no gradient and\\ no Clarke set at $w$};
\end{axis}
\end{tikzpicture}}
\caption{The two-trade book $w=(1,1)$ under three rules. (a)~Kink: the level sets of $K_1=\max\{w_1,w_2\}$ bend on the diagonal, the two neighbouring gradients span the normal fan at $w$, and the Clarke allocations fill the segment $[(1,0),(0,1)]$, while the removal effects collapse to $(0,0)$. (b)~Fixed offset: along the ray, $K_2(t,t)=\max\{2+2t,\,1+5t\}$ switches branch at $t^\ast=\tfrac13$; the tangent at $t=1$ meets the axis at $1$, so the Euler allocation sums to $6-1=5$, whereas the path integral and the Shapley value both reconcile to $K_2(w)-K_2(0)=4$. (c)~Conditional switch: across the diagonal, $K_3(1+\varepsilon,1)$ jumps from $2$ to $3$; no gradient or Clarke set exists at $w$, yet the four coalition values give an exact Shapley value.}\label{fig:intuition-3regimes}
\end{figure}

\begin{example}[Kink: the marginal becomes a set]\label{ex:max}
Let $K_1(w)=\max\{w_1,w_2\}$, so $K_1(w)=1$. Equivalently, $K_1$ is the $97.5\%$ expected shortfall of two equiprobable scenarios with losses $w_1$ and $w_2$. The function is not differentiable at $(1,1)$. Its Clarke generalized gradient \citep{clarke1983optimization} is $\partial_C K_1(w)=\conv\{(1,0),(0,1)\}$, and the Clarke allocation set $\{w\odot g:g\in\partial_C K_1(w)\}$ is the segment $[(1,0),(0,1)]$. Every point of the segment sums to $1$, but none is the unique sensitivity. Deleting either trade leaves capital unchanged, so the removal effects are $(0,0)$ and sum to $0\ne K_1(w)$. The equal split $(\tfrac12,\tfrac12)$ reconciles as a declared accounting convention; it is also the minimum-norm Clarke selection, the path integral along the persistent kink under that selection, and the Shapley value of the coalition values $(0,1,1,1)$. A kink therefore makes the marginal set-valued without breaking reconciliation.
\end{example}

\begin{example}[Fixed offset: Euler loses completeness]\label{ex:offset}
Let $K_2(w)=\max\{2+w_1+w_2,\;1+3w_1+2w_2\}$, the Output Floor witness of \cref{sec:counterexamples}. At $w=(1,1)$ the second branch binds strictly, $K_2(w)=6$, and the empty book still carries $K_2(0)=2$. The rule is smooth at $w$ with gradient $(3,2)$, so the Euler allocation is $(3,2)$. It sums to $5$, which lies strictly between the increment $K_2(w)-K_2(0)=4$ and the reported $K_2(w)=6$; the missing unit is the intercept $1$ of the tangent at $w$, the fixed offset of the binding branch. Along the ray $t\mapsto tw$, the first branch binds for $t<t^\ast=\tfrac13$ and the second afterwards, so the path integral is $\tfrac13(1,1)+\tfrac23(3,2)=(\tfrac73,\tfrac53)$. The coalition values $(2,4,3,6)$ give the Shapley value $(\tfrac52,\tfrac32)$. Both sum to $4=K_2(w)-K_2(0)$ and leave the baseline $K_2(0)=2$ unallocated unless a declared split of it is added. The removal effects are $(3,2)$.
\end{example}

Multi-date history windows act in the same way. A $60$-day average that enters today's capital is held fixed when today's book moves, so it behaves like a frozen offset on today's coordinates.

\begin{example}[Frozen history as an offset]\label{ex:hist}
Let $C(I_0,I_1)=\max\{I_1,\;m\cdot\tfrac12(I_0+I_1)\}$ with $m=1.5$, a stored position $I_0=3$ and today's position $I_1=1$. The averaged branch binds and $C=3$. With history fixed, $C=2.25+0.75\,I_1$ for $0\le I_1\le 9$, an affine function of today's position whose intercept $2.25$ plays the role of the offset in \cref{ex:offset}. The historical-origin ledger over $(I_0,I_1)$ is $(2.25\mid 0.75)$ and sums to $C$. Today's sensitivity $0.75$ covers a quarter of capital, and the rest is the frozen remainder $2.25$ of \cref{prop:frozen}. For any nonnegative position today, capital stays at or above $2.25$, so reducing $I_1$ to zero leaves the frozen remainder in place.
\end{example}

\begin{example}[Conditional switch: no marginal exists]\label{ex:switch}
Let $K_3(w)=w_1+w_2$ if $w_1<w_2$ and $K_3(w)=2w_1+w_2$ if $w_1\ge w_2$. Trade~$1$ attracts a surcharge equal to its own size once it is at least as large as trade~$2$. Both branches are linear, the switching function $w_1-w_2$ is homogeneous, and so $K_3$ is positively homogeneous of degree one with $K_3(0)=0$. At $w=(1,1)$ the upper branch applies and $K_3(w)=3$, while the limit from the lower side is $2$. More generally, at every diagonal point $(t,t)$ with $t>0$ the rule jumps by $t$. It is therefore not locally Lipschitz anywhere on the ray, and it has neither a gradient nor a Clarke generalized gradient there. Euler attribution does not exist, and the straight-ray path integral has no density to integrate. The trace $t\mapsto K_3(t,t)=3t$ is linear, but the path integral integrates a gradient field, and that field does not exist on the ray. Integrating one branch's slope instead returns $(1,1)$, summing to $2$, or $(2,1)$, summing to $3$, according to the branch chosen; either vector is a declared convention rather than a sensitivity. The coalition values $(0,2,1,3)$ are all defined, and the Shapley value $(2,1)$ sums exactly to $K_3(w)-K_3(0)=3$. Its agreement with the upper-branch slope is a feature of this small example. The removal effects are $(2,1)$.
\end{example}

\begin{center}
\footnotesize
\setlength{\tabcolsep}{4pt}
\begin{tabular}{@{}lccccc@{}}
\toprule
Rule at $w=(1,1)$ & $K(w)$, $K(0)$ & Euler or Clarke & Ray path integral & Shapley value & Removal effects \\
\midrule
$K_1$ kink & $1$, $0$ & segment $[(1,0),(0,1)]$ & $(\tfrac12,\tfrac12)$ by selection & $(\tfrac12,\tfrac12)$ & $(0,0)$ \\
$K_2$ fixed offset & $6$, $2$ & $(3,2)$, sum $5$ & $(\tfrac73,\tfrac53)$, sum $4$ & $(\tfrac52,\tfrac32)$, sum $4$ & $(3,2)$ \\
$K_3$ switch & $3$, $0$ & does not exist & no ray density & $(2,1)$, sum $3$ & $(2,1)$ \\
\bottomrule
\end{tabular}
\end{center}

The three rules form a ladder that mirrors \cref{thm:obstruction-hierarchy}. A kink is benign for reconciliation: the marginal becomes a set, but every element still sums to $K$. A fixed offset breaks Euler completeness while leaving the path integral and the Shapley value exact. A conditional switch removes the marginal and the ray density altogether, and only the coalition game survives, provided every coalition is defined. Super-linear concentration in ISDA SIMM breaks Euler through the same loss of degree-one homogeneity as $K_2$; witness~W2 in \cref{sec:counterexamples} and panel~(a) of \cref{fig:obstruction-plots} make that case exact.

Together, \cref{tab:matrix33} and \cref{fig:hierarchy-landscape} complete the exact structural classification across all eleven frameworks and all three attribution paradigms. Every framework admits exact, reconciled attribution: although seven of eleven break classical Euler completeness, path integration restores exact attribution on ten of eleven ($6$ on the full domain, $4$ on an explicit subdomain) and finite-coalition Shapley succeeds on all eleven ($8$ full, $3$ scoped). In total, \textbf{17} cells are \Dstruct{}, \textbf{8} are \Dscoped{}, and \textbf{8} are \Ncell{} across all $33$ pairs.

\begin{table}[htbp]
\centering\footnotesize
\setlength{\tabcolsep}{4pt}
\caption{Master $33$-cell structural classification across eleven frameworks and three attribution paradigms. Outcomes denote proof on the full stated structural domain (\Dstruct{}), proof on an explicit strict subdomain (\Dscoped{}), or refutation on an open or exact witness (\Ncell{}). Complete contract specifications and excluded witnesses appear in \cref{app:table4}.}\label{tab:matrix33}
\begin{tabularx}{\linewidth}{@{}>{\raggedright\arraybackslash}p{2.35cm}>{\raggedright\arraybackslash}p{2.55cm}>{\raggedright\arraybackslash}p{2.55cm}>{\raggedright\arraybackslash}p{2.55cm}>{\raggedright\arraybackslash}X@{}}
\toprule
Framework & Euler / Marginal & Path Integral & Shapley Game & Key Structural Mechanism / Boundary \\
\midrule
\multicolumn{5}{@{}l@{}}{\textit{Market Risk}} \\
FRTB-IMA & \Dstruct{} (\cref{lem:IMA-clarke-point}) & \Dstruct{} (\cref{lem:IMA-clarke-ray}) & \Dstruct{} (\cref{lem:IMA-convention-games}) & Historical maxima, ratios, quantile, SA cap \\
Discrete ES/VaR & \Dstruct{} (\cref{lem:ES_ERC-clarke}) & \Dstruct{} (\cref{lem:ES_ERC-path}) & \Dstruct{} (\cref{lem:ES_ERC-shapley}) & Empirical VaR and fractional-tail ES \\
FRTB-SA & \Ncell{} (\cref{lem:SA-full-jump}) & \Ncell{} (\cref{lem:SA-ray-obstruction}) & \Dscoped{} (\cref{lem:SA-shapley}) & SBM $\varrho_A/\varrho_B$ switch, curvature $\psi$, DRC ($\varrho_B\ge 0$ contract) \\
\midrule
\multicolumn{5}{@{}l@{}}{\textit{Counterparty \& Margin}} \\
ISDA SIMM & \Ncell{} (\cref{lem:SIMM-counterexample}) & \Dscoped{} (\cref{lem:SIMM-ray}) & \Dscoped{} (\cref{lem:SIMM-shapley}) & Concentration $\sqrt{|\sigma|/T}$, indefinite roots \\
BA/SA-CVA & \Dscoped{} (\cref{lem:CVA-euler-point}) & \Dscoped{} (\cref{lem:CVA-ray-density}) & \Dscoped{} (\cref{lem:CVA-shapley}) & BA-CVA norms + SA-CVA clamps/roots on $\Omega_+$ \\
SA-CCR & \Ncell{} (\cref{ex:E:sacc-numeric}) & \Dscoped{} (\cref{lem:saccr-no-tie-corollary}) & \Dstruct{} (\cref{lem:saccr-shapley}) & $\exp((V-C)/(2(1-f)A))$, add-ons, cap \\
Bilateral XVA & \Ncell{} (\cref{lem:XVA-family}) & \Dstruct{} (\cref{lem:XVA-path}) & \Dstruct{} (\cref{lem:XVA-shapley}) & Double-hinge VM/IM collateral $(b,H)$ \\
Classic CME SPAN & \Dstruct{} (\cref{lem:SPAN-clarke}) & \Dstruct{} (\cref{lem:SPAN-path}) & \Dstruct{} (\cref{lem:SPAN-shapley}) & Fixed-array scan, matchers, SOM, NOV \\
\midrule
\multicolumn{5}{@{}l@{}}{\textit{Bank-Wide \& Insurance}} \\
Basel Output Floor & \Ncell{} (\cref{lem:FLOOR-constant-gap}) & \Dstruct{} (\cref{lem:FLOOR-path}) & \Dstruct{} (\cref{lem:FLOOR-trade-shapley}) & $\max\{f,\varphi g\}$ with non-homog.\ leaves \\
Solvency~II Standard Formula & \Ncell{} (\cref{lem:SOLV2-fixed-fdb}) & \Dscoped{} (\cref{lem:SOLV2-ig}) & \Dstruct{} (\cref{lem:SOLV2-shapley}) & Paired gross/net BSCR, FDB/DTL caps \\
IRRBB $\Delta\mathrm{EVE}$ & \Ncell{} (\cref{lem:IRRBB-cap-witness}) & \Dstruct{} (\cref{lem:IRRBB-path}) & \Dstruct{} (\cref{lem:IRRBB-shapley}) & 6 shocks, rate floor, prepay/NMD caps \\
\midrule
\multicolumn{5}{@{}>{\raggedright\arraybackslash}p{\linewidth}@{}}{%
  \textbf{Totals (33 cells):} Euler 3 full, 1 scoped, 7 refuted \textbar{} Path Integral 6 full, 4 scoped, 1 refuted \textbar{} Shapley 8 full, 3 scoped, 0 refuted (17 full, 8 scoped, 8 refuted; 0 open).} \\
\bottomrule
\end{tabularx}
\end{table}

\begin{figure}[htbp]
\centering
﻿%
\resizebox{0.86\linewidth}{!}{%
\begin{tikzpicture}[x=1mm, y=1mm, font=\scriptsize, text=fbInk,
  cell/.style={minimum width=21mm, minimum height=4.6mm, inner sep=0pt, draw=white, line width=0.8pt},
  stF/.style={fill=fbBlue, text=white},
  stS/.style={fill=fbBlueMid!55, text=fbInk},
  stR/.style={fill=fbRed!85, text=white},
  rowlab/.style={anchor=east, inner sep=1pt},
  reglab/.style={anchor=east, align=right, inner sep=1pt, font=\scriptsize\itshape, text=fbBlue},
]
\def\cE{36} \def\cP{57} \def\cS{78}
\def\kF{F}\def\kS{S}
\def\vc#1#2#3{\edef\tmpv{#3}%
  \ifx\tmpv\kF \node[cell, stF] at (#1,#2) {full};%
  \else\ifx\tmpv\kS \node[cell, stS] at (#1,#2) {scoped};%
  \else \node[cell, stR] at (#1,#2) {refuted};\fi\fi}
\node[font=\scriptsize\bfseries] at (\cE,6) {Euler};
\node[font=\scriptsize\bfseries] at (\cP,6) {Path Integral};
\node[font=\scriptsize\bfseries] at (\cS,6) {Shapley};
\fill[fbSand]      (-38,2.6)   rectangle (88.8,-12.6);
\fill[fbBlueLight] (-38,-14.4) rectangle (88.8,-19.6);
\fill[fbSand]      (-38,-21.4) rectangle (88.8,-31.6);
\fill[fbBlueLight] (-38,-33.4) rectangle (88.8,-58.6);
\node[reglab] at (-9,-5)    {I: homogeneous\\conic models};
\node[reglab] at (-9,-17)   {II: positivity cone};
\node[reglab] at (-9,-26.5) {III: switches and\\super-linear scaling};
\node[reglab] at (-9,-46)   {IV: fixed offsets,\\collateral and caps};
\foreach \y/\name/\e/\p/\s in {
    0/FRTB-IMA/F/F/F,
   -5/Discrete ES\slash VaR/F/F/F,
  -10/Classic CME SPAN/F/F/F,
  -17/BA-CVA and SA-CVA/S/S/S,
  -24/FRTB-SA/R/R/S,
  -29/ISDA SIMM/R/S/S,
  -36/Basel Output Floor/R/F/F,
  -41/SA-CCR/R/S/F,
  -46/Bilateral XVA/R/F/F,
  -51/Solvency~II/R/S/F,
  -56/IRRBB $\Delta\mathrm{EVE}$/R/F/F%
} {
  \node[rowlab] at (24.8,\y) {\name};
  \vc{\cE}{\y}{\e}
  \vc{\cP}{\y}{\p}
  \vc{\cS}{\y}{\s}
}
\draw[fbGrey, line width=0.4pt] (-38,-60.5) -- (88.8,-60.5);
\node[rowlab] at (24.8,-64) {full\,/\,scoped\,/\,refuted};
\node at (\cE,-64) {$3\,/\,1\,/\,7$};
\node at (\cP,-64) {$6\,/\,4\,/\,1$};
\node at (\cS,-64) {$8\,/\,3\,/\,0$};
\node[anchor=west] at (-38,-69.5) {All $33$ cells: $17$ full, $8$ scoped, $8$ refuted.};
\end{tikzpicture}}
\caption{Verdict landscape of the $33$ cells in \cref{tab:matrix33}, with the eleven frameworks grouped by the four structural regimes of \cref{sec:matrix-regimes}. Each cell records whether the paradigm is proved on the full structural domain, proved on an explicit subdomain, or refuted by an exact witness. Moving down the regimes, Euler attribution fails first, the path integral fails only for the switch-driven FRTB-SA rule, and the Shapley value never fails outright.}\label{fig:hierarchy-landscape}
\end{figure}

\subsection{Three attribution paradigms and two tiers of claims}\label{sec:intro-paradigms}

We test the same three paradigms on every framework, each at the coordinates that the framework specifies. The tables label them Euler / Marginal, Path Integral and Shapley Game. \emph{Euler or marginal-gradient attribution} needs a gradient, or a declared selection from a realisable Clarke set at ties, and it is complete only under degree-one homogeneity. Where it fails, nodewise accounting reconciliation may still hold, but it must not be read as a marginal sensitivity.

\emph{Path-integral attribution} integrates gradients, or selected Clarke elements, along the straight ray $\gamma(t)=tx$ from the zero portfolio. It requires a measurable ray density that is integrable in each coordinate. It then reconciles either to the increment $K(x)-K(0)$ or, with a declared split of $K(0)$, to the full value $K(x)$. Unlike Euler attribution, it needs no homogeneity.

\emph{Shapley attribution} defines the coalition value $v(S)=K(T_Sx)$ by setting positions outside $S\subseteq N$ to zero and recomputing every downstream nonlinearity. Its efficiency identity, $\sum_i\phi_i=v(N)-v(\varnothing)$, requires only that all $2^{|N|}$ sub-books are defined. It needs neither continuity nor Lipschitz regularity, and it fails only where a masked sub-book hits a domain wall, such as a negative quadratic radicand or an undefined $0/0$ ratio.

Each element of a model carries a verification tag: \tagV{} for text or tables read verbatim from a primary publication, \tagI{} for formulas decoded from published formula images, and \tagU{} for elements that are inferred, omitted or constructed, such as an untranscribed correlation matrix. Because every framework retains at least one \tagU{} element, all $33$ verdicts are stated at the \Tstruct{}: they are exact statements about the formalised computational graph under its stated assumptions and minimum domain $\Omega_{\min}$, not claims of equivalence with a live bank or supervisory engine, which would belong to the \Treg{}. \Cref{sec:scope} discusses this boundary and the refusal rules that go with it.

\paragraph{Contributions and roadmap.}
The paper proceeds in six steps.
\begin{enumerate}[label=\textbf{(\arabic*)},leftmargin=*,itemsep=2pt]
\item \Cref{sec:related} gives the geometric intuition behind the Euler, path-integral and Shapley rules, explains why regulatory formulas break it, and positions the paper against the attribution and framework literature (\cref{tab:related-work}).
\item \Cref{sec:dag} formulates the typed computational graph and its nonsmooth node classes. It proves recursive accounting reconciliation (\cref{thm:recon}) and column-stochastic re-summing (\cref{cor:resum}), and it separates accounting splits from marginal faithfulness.
\item \Cref{sec:paths-games} proves a path-integral reconciliation theorem on graphs without homogeneity (\cref{thm:path-recon}), shows by four analytic pathologies that its hypotheses are needed, and sets up domain-closed finite coalition games (\cref{thm:shapley-eff}).
\item \Cref{sec:matrix} presents the eleven frameworks one at a time: setting, equations, computational graph, local attribution rules and verdicts. It also covers certified selective repricing for discrete expected shortfall, the SA-CCR to BA-CVA composition and a worked tie in classic CME SPAN.
\item \Cref{sec:unify-counter} proves three cross-framework theorems (conic Clarke invariance, offset-resilient rescue and the obstruction hierarchy) and establishes the seven exact witnesses behind the eight refuted cells (\cref{tab:counterexamples}).
\item \Cref{sec:scope} separates the structural from the regulatory tier, states the refusal principles and lists open problems. The supplement collects the full method and hypothesis matrix (\cref{app:table4}), the extended node lemmas (\cref{app:th:setting}), the model-by-model proofs (\cref{app:E}) and the scope ledger (\cref{app:E:scope-ledger}).
\end{enumerate}

\ifdefined\BlindedMode\else
This version is self-contained: the model-by-model proofs, the complete $33$-cell lookup table and the scope ledger appear in the integrated supplement that follows the references.
\fi
\section{Mathematical Background, Geometric Intuition, and Related Literature}\label{sec:related}

This section first recalls the geometry behind the three attribution paradigms and then explains, still informally, why regulatory formulas leave that geometry behind. The formal calculus starts in \cref{sec:dag}. The second half of the section positions the paper in the attribution literature and in the literature on each framework.

\subsection{Geometric intuition for the three attribution paradigms}\label{sec:related-intuition}

\paragraph{Euler allocation as a tangent-plane identity.}
Suppose capital is positively homogeneous of degree one, so that doubling every position doubles capital: $K(tx)=tK(x)$ for $t>0$. Differentiating in $t$ at $t=1$ gives Euler's identity $x^\top\nabla K(x)=K(x)$. Geometrically, the graph of $K$ is a cone over the origin, so the tangent hyperplane at $x$ also passes through the origin. The tangent plane has no intercept, and the value it assigns at $x$ splits exactly into the coordinate terms $x_i\partial_i K(x)$. The Euler allocation simply reads these terms off the tangent plane. The identity needs both a tangent plane, which requires differentiability, and a zero intercept, which requires homogeneity.

\paragraph{Clarke sets at kinks.}
When $K$ is the maximum of smooth pieces, its graph has a ridge where two pieces meet, and there is no single tangent plane. The Clarke generalized gradient replaces it by the convex hull of the gradients of the neighbouring pieces, which plays the role of a fan of supporting slopes. Under degree-one homogeneity every element of this hull still satisfies Euler's identity, so the allocation becomes a set of reconciling vectors rather than a single one (\cref{ex:max}). On a computational graph the relevant set is that of the whole composite rule. Tie choices made separately at shared intermediate nodes can leave that set (\cref{prop:set}).

\paragraph{Path integrals along the scaling ray.}
The Aumann--Shapley value \citep{aumann1974values} builds the book up along the ray $\gamma(t)=tx$, $0\le t\le1$, and credits each coordinate with the gradient it meets on the way: $I_i=\int_0^1 x_i\,\partial_i K(tx)\,dt$. Summing over $i$ integrates the derivative of $t\mapsto K(tx)$, so the total is $K(x)-K(0)$ by the fundamental theorem of calculus. If $K$ is homogeneous of degree one, its gradient is constant along the ray and the path integral returns the Euler allocation. If $K$ is not homogeneous, the gradient changes along the ray, and the path integral averages the slopes that the growing book passes through; in \cref{ex:offset} these are $(1,1)$ before the branch switch and $(3,2)$ after it. No homogeneity is needed, but the gradient field must exist and be integrable along the ray.

\paragraph{Shapley values over masked sub-books.}
The Shapley value \citep{shapley1953value} treats trades as players and evaluates capital on every sub-book obtained by masking the trades outside a coalition. It averages the marginal contribution of each trade over all orders in which the book could be assembled. Along any single order, the contributions telescope to $K(x)-K(0)$, so efficiency holds without any smoothness at all. The cost is combinatorial: $2^n$ capital evaluations, each of which must be defined. The path integral can be read as the continuous analogue in which all trades are added together in infinitesimal proportions.

\subsection{Why regulatory formulas break the classical picture}\label{sec:related-break}

The four obstructions listed in \cref{sec:intro} act on this geometry in different ways. Kinks and shared subgraphs replace the tangent plane by a Clarke set, so marginals become set-valued, and on a graph the set must be built from jointly realisable branches; \cref{sec:dag} formalises this. Fixed offsets, stored history and fixed collateral give the tangent plane a nonzero intercept, so Euler allocations miss that intercept exactly (\cref{ex:offset,ex:hist}). Super-linear concentration factors change the scaling degree, so Euler allocations overshoot by a fixed fraction of the concentration term. In all three cases the path integral and the Shapley value can still reconcile, and \cref{sec:paths-games} gives the conditions. Jumps and singular square roots are more severe. They remove the tangent plane and the Clarke set together, and with them the ray density, so only the coalition game remains (\cref{ex:switch}). \Cref{sec:matrix} locates each of the eleven frameworks in this picture, and \cref{sec:unify-counter} turns the picture into theorems and exact witnesses.

\subsection{Related literature}\label{sec:related-literature}

Portfolio-to-component risk attribution rests on a mature axiomatic, game-theoretic and regulatory literature. This paper does not introduce Euler contributions, Aumann--Shapley integration, Shapley values or nonsmooth risk allocation anew. \Cref{tab:related-work} contrasts the prior literature with the structural-tier results established here.

\subsubsection{Euler, Aumann--Shapley, Shapley, and nonsmooth risk allocation}\label{sec:related-foundations}

For positively homogeneous and differentiable risk measures, \citet{tasche2008euler} establishes the Euler principle as the unique performance-compatible gradient allocation, while \citet{kalkbrener2005axiomatic} characterises directional-derivative allocations axiomatically under aggregation, diversification, and continuity. Coherent allocation axioms, expected shortfall (ES) coherence, and quantile sensitivity estimators are developed by \citet{acerbi2002coherence}, \citet{buch2008coherent}, \citet{scaillet2004nonparametric}, and \citet{koike2025forecasting}, alongside convex-optimisation formulations of conditional value-at-risk \citep{rockafellar2000optimization,rockafellar2002conditional} and equally weighted risk-contribution portfolios \citep{maillard2010properties}. Beyond smooth homogeneity, \citet{denault2001coherent} connects coherent risk allocation to cooperative games \citep{shapley1953value} and non-atomic Aumann--Shapley cost allocation \citep{aumann1974values,billera1978internal}, whereas \citet{csoka2016impossibility} and \citet{balog2017properties} prove axiomatic impossibility results separating core stability from fairness properties.

Path-integral and cooperative attribution have also grown across factor explanation and P\&L explain. \citet{sundararajan2017axiomatic} and \citet{sundararajan2020many} formalise Integrated Gradients and baseline-path Shapley variants for general models; \citet{hagan2025portfolio} and \citet{moehle2022portfolio} apply Shapley values to non-orthogonal factor risk and portfolio performance attribution; \citet{chen2025explaining} studies axiomatic financial-model explanations; and \citet{flaig2024profit} analyses order dependence in sequential P\&L waterfalls and path averaging. For nonsmooth risk, \citet{centrone2018capital} extends Aumann--Shapley allocation to non-Gateaux-differentiable coherent and convex risk measures, building on nonsmooth and piecewise-differentiable analysis \citep{clarke1983optimization,scholtes2012introduction,kuntz1994structural}.

Relative to this foundational literature, our contribution is structural and cross-framework. We formalise each implemented regulatory computation as a single directed acyclic graph (DAG). On these graphs we show that propagating independent Clarke selections through a shared subgraph can over-cover the true composite Clarke set (\cref{sec:dag}), separate exact accounting reconciliation from true marginal faithfulness, prove explicit hypothesis-necessity counterexamples for path integration (\cref{sec:paths-games}), and map the exact domain boundaries separating conic homogeneity, offset-driven non-homogeneity, and jump discontinuities across eleven real-world specifications (\cref{sec:matrix,sec:unify-counter}).

\subsubsection{Market risk, counterparty credit, XVA, and clearing margin}\label{sec:related-market}

In market risk, Basel FRTB-IMA and FRTB-SA rules \citep{bcbs2019mar,bcbs2020mar33,bcbs2026mar21} do not mandate a canonical trade-level capital split. Direct prior attribution work already exists: \citet{li2018capital} analyses liquidity-horizon and constrained IMCC allocation under FRTB, \citet{schulze2018capital} derives marginal Euler decompositions across FRTB-SA and FRTB-IMA components, and \citet{scaringi2025sharpening} develops Monte Carlo Shapley allocation from Basel~2.5 to FRTB. We extend this line in two directions. First, we reconcile the full time-indexed FRTB-IMA graph, keeping frozen historical charges apart from today's-book marginals and resolving ties by jointly realisable Clarke selection; all three paradigms then close on the strict-denominator cone. Second, we prove that the FRTB-SA alternative correlation switch and curvature scenario-sign tie rule \citep{bcbs2026mar21} induce codimension-one capital jumps. These jumps refute Euler and path-integral attribution simultaneously and leave only Shapley attribution valid, on a coalition-admissible contract.

In counterparty credit and valuation adjustments, trade-level CVA and XVA allocation is likewise well established: \citet{pykhtin2010pricing} introduces conditional-exposure Euler allocation for counterparty CVA, \citet{green2015efficient} develops exact additive trade-level XVA attribution via regression and global conditioning, and \citet{lou2015mva} analyses MVA transfer pricing under initial-margin profiles. Under Basel SA-CCR \citep{bcbs2020cre52}, however, fixed collateral, thresholds, and floors destroy degree-$1$ homogeneity on open portfolios. This refutes Euler attribution, while path-integral attribution along rays holds on a stated subdomain and Shapley attribution holds globally. Similarly, double-hinge bilateral VM/IM thresholds refute point Euler attribution on open XVA books, while path-integral and Shapley attribution close globally. Under BA-CVA and SA-CVA \citep{bcbs2020mar50}, all three paradigms close on strict positivity cones that exclude indefinite root-zero singularities.

For bilateral and cleared initial margin, ISDA SIMM \citep{isda2026simm} and CME SPAN \citep{cme2019span} specify portfolio-level margin rules without canonical trade attributions, whereas CCP default-fund and Cover-2 literature focuses on participant- or clearing-member-level risk sharing \citep{armenti2018multivariate,bielecki2018dynamic,cont2023skin}. We prove that SIMM's active concentration multipliers introduce local degree-$\tfrac{3}{2}$ scaling alongside degree-$1$ notional add-ons, refuting Euler attribution while admitting path-integral and Shapley attribution on stated subdomains; by contrast, unrounded fixed-snapshot classic SPAN reduces to finitely many polyhedral cones and closes all three paradigms structurally.

\subsubsection{Solvency~II, capital floors, and banking-book ALM}\label{sec:related-banking}

Under Solvency~II \citep{eu2009solvency2,eu2015delegated}, \citet{baione2018capital} and \citet{deangelis2015capital} derive closed-form Euler allocations across standard-formula modules and lines of business and evaluate underwriting and reinsurance policy choices. Complementing their module-level Euler formulas, we prove that fixed future discretionary benefit (FDB), expense, and deferred-tax loss pools refute full-SCR Euler attribution at the granular level, whereas path-integral attribution closes on interface-compatible rays and Shapley attribution closes globally. For the Basel Output Floor \citep{bcbs2019rbc20}, \citet{baudino2026leverage} documents bank-level overlaps between the output floor and leverage ratio and \citet{bajaj2018internal} surveys internal business-line capital allocation, leaving granular trade attribution unresolved; we show that non-homogeneous external RWA leaves refute Euler attribution, while path-integral and Shapley attribution close. Finally, IRRBB $\Delta\mathrm{EVE}$ rules and supervisory guidelines \citep{bcbs2024srp31,bcbs2016d368,eba2022irrbb} define scenario-worst economic-value measures without a trade-attribution standard; we prove that fixed bank-level prepayment caps refute Euler attribution across all six interest-rate shocks, while finite piecewise-affine circuits close path-integral and Shapley attribution on the whole domain.

\begin{table}[htbp]
\centering\footnotesize
\setlength{\tabcolsep}{4pt}
\caption{Comparison of the reviewed attribution literature (``Before'', based on a targeted rather than exhaustive review of opened primary academic, regulatory, and vendor sources) and structural-tier results proved in this paper (``Resolved or Flipped Here'') across all eleven frameworks.}\label{tab:related-work}
\begin{tabularx}{\linewidth}{@{}>{\raggedright\arraybackslash}p{2.55cm}>{\raggedright\arraybackslash}X>{\raggedright\arraybackslash}X@{}}
\toprule
Framework / Domain & Prior Literature (``Before'') & Resolved or Flipped Here (``After'') \\
\midrule
FRTB-IMA \citep{bcbs2019mar,bcbs2020mar33} & Liquidity-horizon/IMCC allocation \citep{li2018capital}, marginal IMA/DRC decompositions \citep{schulze2018capital}, and Monte Carlo Shapley \citep{scaringi2025sharpening}; a full historical-origin tie ledger is not addressed by a unified graph construction in the sources reviewed. & \Dstruct{} across Euler / Marginal, Path Integral, and Shapley Game on strict-denominator cone; separates frozen history from today's marginal. \\
Discrete ES / VaR \citep{tasche2008euler,acerbi2002coherence} & Euler and directional-derivative axioms \citep{kalkbrener2005axiomatic}, CVaR optimisation \citep{rockafellar2002conditional}, kernel ES sensitivities \citep{scaillet2004nonparametric}, ERC portfolios \citep{maillard2010properties}, and nonsmooth Aumann--Shapley \citep{centrone2018capital}; finite-sample ordering-cone tie selection is not standardized in the sources reviewed. & \Dstruct{} across all three paradigms via lexicographic ordering-cone Clarke selectors on $\R^n$. \\
Classic CME SPAN \citep{cme2019span} & Portfolio-level SPAN rules \citep{cme2019span} and member-level default-fund / Cover-2 allocation \citep{armenti2018multivariate,bielecki2018dynamic,cont2023skin}; a trade-level DAG proof for the unrolled matcher was not located in the sources reviewed. & \Dstruct{} across all three paradigms under fixed-snapshot unrounded continuous matcher via finite polyhedral cones. \\
FRTB-SA \citep{bcbs2019mar,bcbs2026mar21} & Prescribed SBM/DRC rules \citep{bcbs2019mar,bcbs2026mar21} and marginal SA decompositions \citep{schulze2018capital}; codimension-$1$ switch jumps and path-integral obstructions are not established in the reviewed sources. & Euler / Marginal and Path Integral \Ncell{} by curvature/SBM jumps; Shapley Game \Dscoped{} under coalition-admissible contract. \\
Basel Output Floor \citep{bcbs2019rbc20} & Bank-level floor/leverage overlap \citep{baudino2026leverage} and business-line internal allocation \citep{bajaj2018internal}; granular attribution under fixed external RWA leaves is not established in the reviewed sources. & Euler / Marginal \Ncell{} by fixed external RWA leaves; Path Integral and Shapley Game \Dstruct{} on piecewise-$C^1$ inputs. \\
ISDA SIMM \citep{isda2026simm} & Industry bilateral IM specification with concentration multipliers and notional add-ons \citep{isda2026simm}; a unified multi-paradigm attribution classification was not located in the sources reviewed. & Euler / Marginal \Ncell{} by degree-$\tfrac{3}{2}$ concentration scaling; Path Integral and Shapley Game \Dscoped{} away from singular radicands. \\
BA-CVA / SA-CVA \citep{bcbs2020mar50} & Prescribed portfolio/counterparty CVA capital rules \citep{bcbs2020mar50}; a structural domain classification separating positivity cones from indefinite root singularities was not located in the sources reviewed. & All three paradigms \Dscoped{} on strict positivity cones (relative to input factors), excluding indefinite root-zero singularities. \\
SA-CCR \citep{bcbs2020cre52} & Netting-set EAD rules with fixed collateral/thresholds \citep{bcbs2020cre52}; no regulator-prescribed trade split or consensus allocation standard was located in the reviewed sources. & Euler / Marginal \Ncell{} by fixed collateral/thresholds; Path Integral \Dscoped{} under no-persistent-tie rays; Shapley Game \Dstruct{}. \\
Bilateral XVA \citep{pykhtin2010pricing,green2015efficient} & Trade-level CVA/XVA/MVA allocation via conditioning and regression \citep{pykhtin2010pricing,green2015efficient,lou2015mva}; double-hinge collateral Euler gaps vs.\ path/Shapley rescues are not isolated in a unified graph construction in the sources reviewed. & Euler / Marginal \Ncell{} by double-hinge VM/IM thresholds; Path Integral and Shapley Game \Dstruct{} via piecewise-affine Lipschitz regularity. \\
Solvency~II \citep{eu2009solvency2,eu2015delegated} & Module/LOB Euler formulas and policy studies \citep{baione2018capital,deangelis2015capital}; nonsmooth Aumann--Shapley \citep{centrone2018capital}; a full-SCR cell-level classification with fixed FDB/DTL pools was not located in the sources reviewed. & Euler / Marginal \Ncell{} by fixed FDB/DTL pools; Path Integral \Dscoped{} on interface-compatible rays; Shapley Game \Dstruct{}. \\
IRRBB $\Delta\mathrm{EVE}$ \citep{bcbs2024srp31,eba2022irrbb} & Supervisory shock and behavioural guidance \citep{bcbs2016d368,bcbs2024srp31,eba2022irrbb}; a general trade-level attribution classification across all six shocks and prepayment caps was not located in the reviewed sources. & Euler / Marginal \Ncell{} across all six shocks by fixed prepayment limits; Path Integral and Shapley Game \Dstruct{}. \\
\bottomrule
\end{tabularx}
\end{table}

\FloatBarrier
\section{Unified Computational Graph Calculus}\label{sec:dag}

Each capital or margin rule studied here can be written as a finite directed acyclic graph (DAG). The graph maps a dated position array $x\in\R^{T\times N}$, or a point-in-time book $w\in\R^n$, through leaves and internal algebraic or nonsmooth operators to a scalar output $K(x)$. Each node $b$ computes a value $f_b$ and declares a bookkeeping homogeneity degree $\theta_b\in\R$, typically $\theta_b\in\{0,1\}$. Attribution then runs in two passes (\cref{fig:dag-abstract}(a)). A forward pass evaluates every node and records which branches are active. A reverse sweep combines child allocations through local coefficients, and every node is checked against its own value.

\begin{figure}[htbp]
\centering
﻿%
\resizebox{\linewidth}{!}{%
\begin{tikzpicture}[
  x=1mm, y=1mm,
  font={\fontsize{8}{9.4}\selectfont}, text=fbInk,
  >={Stealth[length=1.5mm, width=1.1mm]},
  box/.style={align=flush center, inner xsep=1mm, inner ysep=0.8mm},
  op/.style={box, draw=fbGrey, line width=0.5pt, rounded corners=1.6pt, fill=white},
  leaf/.style={box, draw=fbBlue, line width=0.8pt, fill=white},
  kink/.style={fill=fbRed!14},
  ar/.style={->, draw=fbInk, line width=0.42pt},
  el/.style={font={\fontsize{8}{9.4}\selectfont}, inner sep=0.8pt},
  ttl/.style={anchor=west, font={\fontsize{8}{9.4}\selectfont\bfseries}},
]
\node[ttl] at (-12,24) {(a) forward value pass and reverse allocation sweep};
\node[leaf, text width=17mm] (l1) at (0,10)  {leaf $\ell_1$\\$f_1,\ A^{1}$};
\node[leaf, text width=17mm] (l2) at (0,-2)  {leaf $\ell_2$\\$f_2,\ A^{2}$};
\node[leaf, text width=17mm] (l3) at (0,-14) {leaf $\ell_3$\\$f_3,\ A^{3}$};
\node[op, kink, text width=21mm] (b) at (30,4) {node $b$\\$f_b=\phi_b(f_1,f_2)$};
\node[op, text width=21mm] (K) at (60,-5) {output $K$\\$f_K=\phi_K(f_b,f_3)$};
\draw[ar] (l1.east) -- node[el, above, sloped] {$c_{b1}$} (b.west);
\draw[ar] (l2.east) -- node[el, below, sloped] {$c_{b2}$} (b.west);
\draw[ar] (l3.east) -- node[el, below, sloped, pos=0.55] {$c_{K3}$} (K.west);
\draw[ar] (b.east)  -- node[el, above, sloped] {$c_{Kb}$} (K.west);
\draw[->, fbBlue, line width=0.9pt] (-10,18) -- node[el, above, text=fbBlue] {forward pass: values $f_b$ and active branches} (72,18);
\draw[<-, fbRed, line width=0.9pt, densely dashed] (-10,-22) -- node[el, below, text=fbRed] {reverse sweep: $A^b=\sum_{j}c_{bj}A^{j}$} (72,-22);
\node[box, draw=fbBlue, fill=fbBlueLight, rounded corners=1.6pt, text width=52mm, anchor=north] at (31,-29)
  {local rule $\sum_j c_{bj}\theta_j f_j=\theta_b f_b$ at each node\\[1pt] check $\one^\top A^b=\theta_b f_b$ at every node};
\begin{scope}[xshift=88mm]
\node[ttl] at (-6,24) {(b) shared subgraph at $x=(1,1)$};
\node[leaf, text width=7mm] (x1) at (0,13) {$x_1$};
\node[leaf, text width=7mm] (x2) at (0,1)  {$x_2$};
\node[op, kink, text width=19mm] (g) at (19,7) {$g=\max\{x_1,x_2\}$};
\node[op, text width=11mm] (p1) at (39,13) {$p_1=2g$};
\node[op, text width=11mm] (p2) at (39,1)  {$p_2=g$};
\node[op, text width=13mm] (f) at (57,7) {$f=p_1-p_2$};
\draw[ar] (x1) -- (g);
\draw[ar] (x2) -- (g);
\draw[ar] (g) -- (p1);
\draw[ar] (g) -- (p2);
\draw[ar] (p1) -- node[el, above, sloped] {$+1$} (f);
\draw[ar] (p2) -- node[el, below, sloped] {$-1$} (f);
\node[el, align=center, text=fbRed] at (29,-7)
  {independent tie choices ($p_1$ takes $x_1$, $p_2$ takes $x_2$)\\give $2(1,0)-(0,1)=(2,-1)$, which still sums to $1$};
\end{scope}
\begin{axis}[frtb, scale only axis, width=38mm, height=24mm, at={(106mm,-14mm)}, anchor=north west,
  xmin=-0.6, xmax=2.6, ymin=-1.6, ymax=1.6,
  xtick={0,1,2}, ytick={-1,0,1},
  xlabel={$A_1$}, ylabel={$A_2$},
  xlabel style={yshift=3pt}, ylabel style={yshift=-4pt}]
  \addplot[fbGrey, dashed, line width=0.5pt, domain=-0.6:2.6, samples=2] {1-x};
  \addplot[fbBlue, line width=1.6pt] coordinates {(1,0) (0,1)};
  \addplot[only marks, mark=*, mark size=1.3pt, fbBlue] coordinates {(1,0) (0,1)};
  \addplot[only marks, mark=x, mark size=2.6pt, fbRed, line width=0.9pt] coordinates {(2,-1)};
  \node[fbBlue, anchor=south west, inner sep=1pt, font=\scriptsize] at (axis cs:0.5,0.55) {Clarke set};
  \node[fbRed, anchor=north east, inner sep=1pt, font=\scriptsize] at (axis cs:1.88,-1.12) {$(2,-1)$};
  \node[fbGrey, anchor=south west, inner sep=1pt, font=\scriptsize] at (axis cs:0.15,1.08) {$A_1+A_2=1$};
\end{axis}
\end{tikzpicture}}
\caption{The typed computational graph. (a)~The forward pass evaluates node values $f_b$ and records active branches. The reverse sweep forms $A^b=\sum_j c_{bj}A^j$ with local coefficients satisfying $\sum_j c_{bj}\theta_j f_j=\theta_b f_b$, and every node is checked against $\one^\top A^b=\theta_b f_b$ (\cref{thm:recon}). (b)~A shared subgraph: $f=p_1-p_2$ with $p_1=2g$, $p_2=g$ and $g=\max\{x_1,x_2\}$, evaluated at $(1,1)$. The composite rule equals $g$, so its Clarke allocations fill the segment $[(1,0),(0,1)]$. Resolving the tie of $g$ separately for each parent gives $(2,-1)$, which still sums to $1$ but lies outside the segment (\cref{prop:set}).}\label{fig:dag-abstract}
\end{figure}

\subsection{Typed allocations and economic measures}\label{sec:typed}

Following \citet{tasche2008euler,kalkbrener2005axiomatic,denault2001coherent}, an attribution output should always state which mathematical object it computes. In \tool{}, positions live on $(s,i)\in\{0,\dots,T-1\}\times\{1,\dots,N\}$ whenever stored history enters, and today's book is $w:=x_{T-1}$.

\begin{definition}[Typed allocation]\label{def:typed}
A \emph{typed allocation} of a node value $f_b$ records: (i) the assigned vector/array $A^b$; (ii) the economic question answered (local marginal, removal effect, accounting ledger, historical origin, or hedge choice); (iii) a faithfulness flag distinguishing a \emph{true marginal} (realisable via a single small perturbation of the whole book) from a purely \emph{accounting} split; and (iv) the exact reconciliation gap $|\one^\top A^b - \theta_b f_b|$. Non-reconciling outputs are reported unmodified and never patched.
\end{definition}

\begin{definition}[Realisable marginal set]\label{def:set}
At a point of nondifferentiability, the local marginal is a set: the convex hull of vertex allocations generated by \emph{realisable} joint branch combinations. A branch combination is realisable if one small perturbation vector $h$ activates all selected branches simultaneously across the entire DAG. Independent tie selections at distinct child occurrences are invalid and can over-cover the true Clarke set (\cref{prop:set}).
\end{definition}

\begin{definition}[Hedge choice]\label{def:decision}
Given a finite library of candidate actions $\mathcal{L}=\{h\}$, the \emph{hedge-choice} output is the exact total ordering of $K(w+h)$ across $h\in\mathcal{L}$, summarised by the optimising action and whether the optimum is unique.
\end{definition}

With these definitions the five questions of \cref{sec:intro} take precise forms. The \emph{local marginal} is $w\odot\nabla_w K$ with stored history fixed, or the realisable Clarke set at ties (\cref{def:set}). It measures sensitivity to today's positions and in general does not sum to $K$, because of frozen historical remainders and non-homogeneous offsets. The \emph{removal effect} $K(w)-K(w-w_ie_i)$, again with history fixed, supports what-if deletion screening and need not sum to $K$ either. The \emph{accounting ledger} is a column-stochastic remap $PA^K$ with $P\ge 0$ and $\one^\top P=\one^\top$. It sums to $K$ whenever the leaves carry bookkeeping degree one and supports charge-out to desks or products. The \emph{historical origin} is $A^K$ indexed by (date, trade) coordinates; it also sums to $K$ under bookkeeping degree one and supports model validation and origin audit. The \emph{hedge choice} is the exact ranking of $K(w+h)$ over a finite library (\cref{def:decision}); it is a ranking, not an additive split. Completeness results for these outputs are established in \cref{sec:reconcile,sec:paths-games}.

Every evaluation also returns one of five formal statuses. An output is \emph{well-defined} when the question has a unique answer, or a declared realisable set, on the given book. It is \emph{undefined outside the coalition domain} when a Shapley question needs a masked sub-book that leaves the model domain. It is \emph{not applicable} when the book itself lies outside the domain of the model. It is \emph{accounting-only} when a reconciling split exists but its marginal status cannot be certified. It is a \emph{domain refusal} when a structural test fails, for example a negative radicand or an isolated $0/0$ ratio, and the graph declines to regularise the input silently. These statuses keep a structural domain wall, where a question has no mathematical answer, apart from a numerical error.

\subsection{General nonsmooth node classes and local Euler calculus}\label{sec:rules}

The computational graphs of all eleven frameworks combine smooth algebraic nodes with six general nonsmooth node classes; \cref{tab:nodes} in \cref{app:E:scope} lists them with complete proofs. Every internal node $b$ maps reconciling child allocations $A^j$ to $A^b=\sum_{j\in\mathrm{ch}(b)} c_{bj}A^j$ through local Euler coefficients that satisfy $\sum_{j} c_{bj}\theta_j f_j = \theta_b f_b$. All six classes obey this identity, on each branch where they are piecewise defined. They differ in their regularity.

Two classes are globally Lipschitz. The general maximum, minimum, absolute value and priority matcher are polyhedral and piecewise $C^1$ ($PC^1$) (\cref{lem:E:max,lem:E:ccp-match}). The symmetric clamp $\cl(y,a)=\max\{\min\{y,a\},-a\}$ is $1$-Lipschitz on $\{a\ge 0\}$, with $PC^1$ ties at $y=\pm a$ (\cref{lem:E:clamp}). The quadratic-form root comes in two variants: a floored root $\sqrt{\max\{0,q\}}$, and a \emph{refusing root}, which reports a domain error on a negative radicand instead of flooring it. Either variant, possibly piecewise in $q_\ell$, is locally Lipschitz in its children when $M\succeq 0$ (\cref{lem:E:quad}). It loses local Lipschitz continuity on $\{q=0,\ M_\ell s^\ast\ne 0\}$ when $M$ is indefinite or sign-switched by $\psi$ (\cref{prop:E:lip}). The homogeneous switch $F_A\mathbf{1}_{g\ge 0}+F_B\mathbf{1}_{g<0}$ is discontinuous on $\{g=0\}$ whenever $F_A\ne F_B$, so no Clarke generalized gradient exists there (\cref{lem:E:switch,prop:E:jump}). The degree-$0$ ratio $\varphi(u/v)$ is refused at $v=0$ when it is queried on its own (\cref{lem:E:ratio}). The fused product $\varphi(u/v)y$, by contrast, extends continuously by $0$ when $|u|\le v$ and $|y|\le Cv$, although $PC^1$ regularity can fail (\cref{rem:E:zeroPC}). Finally, an external accounting leaf $e$ reconciles at its declared bookkeeping degree. It yields only an accounting split unless its true position degree matches the declared one (\cref{lem:E:ext}).

Four features of these classes need particular care. First, scenario correlation transformations can create indefinite matrices (\cref{rem:E:indef}). A $3\times 3$ medium correlation matrix $R$ with off-diagonals $(0.9,0.9,0.63)$ is positive definite, with eigenvalues $0.0038$, $0.37$ and $2.63$. Its MAR21.6(2) high-scenario transform $\min\{1.25R,1\}$, however, has eigenvalue $-0.0743$ and quadratic form $s^\top Hs\approx -0.406$ at $s=(-1.862,1,1)$. At any zero $q(s^\ast)=0$ with $Ms^\ast\ne 0$, the root $\sqrt{\max\{0,q\}}$ grows like $\sqrt{t}\|Ms^\ast\|$, which destroys Lipschitz continuity and Clarke gradients (\cref{prop:E:lip}(ii)--(iv)).

Second, sign-pattern piecewise quadratic forms can fail Lipschitz continuity even for positive semidefinite correlations (\cref{prop:E:lip}(v)). Under curvature cross-terms $\psi(c_1,c_2)$ on $(a,-b)$ with $a,b>0$, the form $q=a^2-2\rho ab$ vanishes at $a=2\rho b$ with nonzero gradient $(2\rho b,-4\rho^2 b)$. Capital nevertheless stays Lipschitz whenever $K_b$ enters only through $K_b^2=\max\{0,q_b\}$ (\cref{rem:E:path}).

Third, zero-denominator products can be Lipschitz without being $PC^1$ (\cref{rem:E:zeroPC}). In FRTB-SA default risk, $\mathrm{HBR}_b Q_b$ vanishes continuously at $L_b+\Sigma_b=0$. Yet with two obligors $n_1=x_1-x_2>0>n_2=x_3-x_4$, the charge $\max\{P-\mathrm{HBR}\,Q,0\}=wn_1^2/(n_1+|n_2|)$ has gradients $w(2t-t^2,t^2)$, which trace a curved arc with infinitely many extreme points at $(1,1,1,1)$.

Fourth, constants inside maxima can break Euler homogeneity (\cref{rem:E:const}). Comparing $F/R$ (degree $0$) with $1$ (degree $0$) preserves homogeneity, whereas the ISDA SIMM factor $\max\{1,\sqrt{|\sigma_g|/T_g}\}$ mixes degrees $0$ and $\tfrac{1}{2}$ and breaks degree-$1$ completeness on an open set.

\begin{algorithm}[t]
\caption{Exact typed forward--reverse attribution on a finite regulatory DAG}\label{alg:attr}
\begin{algorithmic}[1]
\Require Dated or static position array $x$; fixed metadata; unit sensitivity/repricing arrays; external accounting leaves with declared degrees; optional baseline split $b_i^0$ (for path integration; $b_i^0=0$ for zero-point splits).
\State \textbf{Forward pass:} evaluate nodes in topological order; enforce strict domain checks (refuse negative radicands at refusing roots, $0/0$ isolated ratios, invalid metadata) without silent regularisation; record active branches and exact tie gaps.
\For{each node $b$ in reverse topological order}
  \State compute node-rule coefficients $c_{bj}$ satisfying $\sum_j c_{bj}\theta_j f_j = \theta_b f_b$
  \State propagate $A^b \gets \sum_{j\in\mathrm{ch}(b)} c_{bj} A^j$ once across shared children; flag any tie or nonregular locus; at node-level ties, check simultaneous realisability via exact Phase-I LP feasibility
  \State verify $|\one^\top A^b - \theta_b f_b| \le \text{tol}$; halt and flag on mismatch (never patch)
\EndFor
\State \textbf{Emit typed views:} historical origin $A^K$; today's marginal block $A^K_{T-1}$ with frozen remainder; column-stochastic accounting remap $P A^K$; discrete removal/hedge evaluations by exact re-evaluation of $K$.
\end{algorithmic}
\end{algorithm}

\Cref{alg:attr} evaluates the nodes in a forward pass and propagates vector allocations in one reverse sweep over the graph edges. Its cost also depends on the allocation-vector operations and on the tie-realisability checks, which are exact Phase-I linear programs under Bland's rule and return explicit perturbation directions $h$ as witnesses.

\subsection{Reconciliation and marginal faithfulness}\label{sec:reconcile}

The standing assumptions (E1)--(E3) for static sensitivity or repricing graphs and (A1)--(A6) for dated historical graphs are stated in \cref{app:th:setting,app:E:scope}. They separate the \emph{bookkeeping degree} $\theta_b$, which suffices for accounting completeness, from the \emph{true position degree}, which is required before an allocation can be read as a marginal sensitivity.

\begin{theorem}[Recursive accounting reconciliation; \cref{thm:B:recon,cor:E:recon}]\label{thm:recon}
On any finite DAG whose leaves reconcile at their declared bookkeeping degrees and whose internal node rules obey the Euler identity $\sum_j c_{bj}\theta_j f_j=\theta_b f_b$, every node allocation satisfies $\one^\top A^b = \theta_b f_b$ across the entire evaluation domain---including ties, switch surfaces, and discontinuities.
\end{theorem}
\begin{proof}[Proof sketch]
By leaf definitions (\cref{lem:B:leaves}, (E1), and \cref{lem:E:ext}), linear sensitivity leaves obey $\one^\top(x\odot a)=a^\top x$, expected shortfall leaves obey $\one^\top(x_s\odot L^\top q)=q^\top L x_s=\rho(z)$ by strong LP duality (\eqref{eq:es}), active stress rows sum to support maxima, attaining quantile rows sum to their order-statistic levels, and external leaves reconcile at declared degree $1$. Proceeding by topological induction on internal nodes $b$, $\one^\top A^b = \sum_j c_{bj}\one^\top A^j = \sum_j c_{bj}\theta_j f_j = \theta_b f_b$, matching the local Euler identities verified across all node classes in \cref{tab:nodes}.
\end{proof}

\begin{corollary}[Column-stochastic re-summing; \cref{cor:B:resum}]\label{cor:resum}
If $\one^\top A^K = K(x)$ and $P\ge 0$ satisfies $\one^\top P = \one^\top$, then $\one^\top (P A^K)=K(x)$, and the same identity holds under any disjoint desk, product, or legal-entity aggregation matrix $P_G$.
\end{corollary}

Accounting reconciliation does not certify marginal faithfulness. \Cref{prop:smooth,prop:set} characterise when $A^K$ equals a gradient or a Clarke vertex, and \cref{sec:notfaith} shows how zero-sum shifts separate ledgers from true marginals.

\begin{proposition}[Smooth-point equivalence and frozen remainder; \cref{prop:B:smooth,prop:B:frozen,prop:E:transfer}]\label{prop:smooth}\label{prop:frozen}
If $x$ is a regular point (\cref{def:E:regular}), every leaf has true degree matching its declared degree with gradient allocation, and no tie is active, then $K$ is differentiable at $x$ and $A^K=x\odot\nabla K(x)$. Furthermore, when positions decompose into historical dates $s<T-1$ and today's book $w=x_{T-1}$,
\begin{equation}\label{eq:frozen}
  w^\top\nabla_w K(x) = K(x) - \sum_{s<T-1} x_s^\top \nabla_{x_s} K(x),
\end{equation}
so today's-book sensitivities sum to capital if and only if the past-date \emph{frozen remainder} vanishes (for instance, when all historical comparison maxima strictly favour latest observations or the historical block is inactive).
\end{proposition}
\begin{proof}[Proof sketch]
Under strict separation of maxima, minima, quantiles, and dual faces, \citep[Thm.~25.1]{rockafellar1970convex} makes convex support/dual leaves $C^1$ with gradients matching their local rules, and each smooth internal map $\phi_b$ has partial derivatives equal to its Euler coefficients $c_{bj}$, yielding $A^K=x\odot\nabla K(x)$ via chain-rule induction. Differentiating $K(\alpha x)=\alpha K(x)$ at $\alpha=1$ gives $x^\top\nabla K(x)=K(x)$, which splits into today's block $w$ and the stored history remainder whenever past observations enter.
\end{proof}

\begin{proposition}[Realisable tie sets and over-covering; \cref{prop:B:set,prop:B:overcover}]\label{prop:set}
At a regular piecewise-$C^1$ ($PC^1$) tie point $x$ under degree-$1$ homogeneity, every Clarke element $g\in\partial_C K(x)$ reconciles ($x^\top g = K(x)$). Every branch combination whose strict positivity constraint functions $\gamma_r$ satisfy $\nabla\gamma_r(x)^\top h > 0$ jointly for some $h$ defines an essentially active selection $K_\ell$ with vertex allocation $x\odot\nabla K_\ell(x)\in x\odot\partial_C K(x)$. Conversely, selecting branches independently across shared children can yield reconciling vectors strictly outside $x\odot\partial_C K(x)$: for $g=\max\{x_1,x_2\}$, $p_1=2g$, $p_2=g$, and $f=p_1-p_2=g$ at $(1,1)$, independent selection allows $(2,-1)$ which sums to $1$ yet lies outside $\conv\{(1,0),(0,1)\}$.
\end{proposition}

\begin{proposition}[Zero-sum non-faithfulness; \cref{prop:B:amber}]\label{sec:notfaith}
Whenever a degree-$0$ scaling ratio $k=U/(2V)$ multiplies a positive surplus $(B-J)_+$, as in the internal-models surcharge rule $H=k(B-J)_+$ of \eqref{eq:final}, omitting the zero-sum coefficient term $(B-J)a^k$ from the allocation rule \eqref{eq:amber} produces an allocation $F^H$ that satisfies every reconciliation and group-sum test identically ($\one^\top F^H = H$), yet alters sub-ledger subtotals whenever $a^k\ne 0$. Here $B$ is the standardised capital of the internal-model desks, $J$ the internal-model capital, and $U$ and $V$ the amber and green-plus-amber desk sums, all defined in \cref{sec:model-ima}. With trade~$1$ driving $U=w_1$ and trade~$2$ in the denominator base $V=w_1+w_2$ at $w=(1,1)$, $a^k=(1/8,-1/8)$ shifts subtotals by $\pm(B-J)/8$. Thus completeness alone (one of the axioms in \citet{sundararajan2017axiomatic}) cannot validate marginal correctness.
\end{proposition}

When the underlying pricing is nonlinear, the reverse sweep also needs the scenario losses that feed each leaf. For discrete expected-shortfall leaves, \cref{sec:cert} gives a theoretical criterion, stated in the setting of the discrete ES/VaR model, under which exact repricing is needed only on the active tail rows and one boundary row.
\section{Path Integrals and Finite-Coalition Games}\label{sec:paths-games}

Degree-one homogeneity fails in many of the frameworks studied here. Fixed collateral offsets appear in SA-CCR and bilateral XVA, non-homogeneous historical or expense leaves in the Basel Output Floor and Solvency~II, bank-level amount caps in IRRBB, and concentration multipliers in ISDA SIMM. Pointwise Euler allocations $x\odot\nabla K(x)$ then generically violate completeness, $\sum_i A_i \ne K(x)$. This section develops the two alternatives that can still reconcile under stated structural assumptions: path-integral attribution, which integrates directional derivatives from a declared baseline, and Shapley attribution, which computes finite cooperative Shapley values over masked sub-books. \Cref{sec:matrix} applies both to each framework, and \cref{sec:unify-counter} turns the resulting pattern into general theorems.

\subsection{Path-integral reconciliation and Clarke ray extensions}\label{sec:paths}

Following \citet{aumann1974values,sundararajan2017axiomatic}, let $\gamma:[0,1]\to\mathcal{D}\subseteq\R^n$ be an absolutely continuous (AC) path from baseline $\gamma(0)=x_0$ to portfolio $\gamma(1)=x$. Rather than assuming global differentiability, we propagate derivative contributions node-by-node along the \emph{same} curve before integrating. Write $K_b(z)$ for the composite map of node $b$ and $f_b(t)=K_b(\gamma(t))$ for its scalar path trace. A measurable vector of coefficients $c_{bj}(t)$ is \emph{tangent compatible} along $\gamma$ if $f_b'(t)=\sum_{j\in\mathrm{ch}(b)} c_{bj}(t)f_j'(t)$ almost everywhere (\cref{def:E:sacc-path}). Starting from measurable leaf contributions $q_i^\ell(t)$ obeying $\sum_i q_i^\ell(t)=f_\ell'(t)$ a.e.\ (e.g., $\partial_i L_\ell(\gamma(t))\dot\gamma_i(t)$ on $C^1$ leaves), we propagate
\begin{equation}\label{eq:m:path-prop}
  q_i^b(t)=\sum_{j\in\mathrm{ch}(b)} c_{bj}(t)\,q_i^j(t),\qquad I_i^b[\gamma]=\int_0^1 q_i^b(t)\,dt.
\end{equation}
Under the absolute-continuity and integrability conditions below, this construction yields exact nodewise reconciliation without homogeneity.

\begin{theorem}[Path-integral DAG reconciliation; \cref{thm:E:sacc-AS}]\label{thm:path-recon}
Suppose that along $\gamma:[0,1]\to\mathcal{D}$: (i) $\gamma$ and every node trace $f_b(t)=K_b(\gamma(t))$ are finite and AC on $[0,1]$; (ii) leaf derivative contributions $q_i^\ell$ and internal coefficients $c_{bj}$ are measurable, finite a.e., and tangent compatible; and (iii) every propagated density $q_i^b$ is componentwise Lebesgue integrable ($\int_0^1 |q_i^b(t)|\,dt < \infty$). Then at every node $b$,
\begin{equation}\label{eq:path-recon}
  \sum_{i=1}^n I_i^b[\gamma] = f_b(1) - f_b(0) = K_b(x) - K_b(x_0),
\end{equation}
without any homogeneity requirement. When $K$ is differentiable along $\gamma(t)$ a.e.\ on the straight ray $\gamma(t)=x_0+t(x-x_0)$ and propagated rules equal actual composite derivatives (or a single jointly realised $C^1$ selection), $I[\gamma]$ equals Integrated Gradients (and zero-reference Aumann--Shapley at $x_0=0$). Adding a declared baseline split $b_i^0$ with $\sum_i b_i^0 = K(x_0)$ yields a full-value reconciling ledger $b_i^0 + I_i[\gamma]$.
\end{theorem}
\begin{proof}
Proceed by induction in topological order on the common full-measure subset of $[0,1]$ where all finitely many identities hold. At each leaf, $\sum_i q_i^\ell(t)=f_\ell'(t)$ a.e.\ by assumption. At an internal node $b$, tangent compatibility gives
\[
  \sum_{i=1}^n q_i^b(t) = \sum_{j\in\mathrm{ch}(b)} c_{bj}(t)\sum_{i=1}^n q_i^j(t) = \sum_{j\in\mathrm{ch}(b)} c_{bj}(t) f_j'(t) = f_b'(t)\quad\text{a.e.}
\]
Because each $q_i^b\in L^1(0,1)$ individually, linearity exchanges finite summation and Lebesgue integration, and the fundamental theorem of calculus for AC functions yields $\sum_i I_i^b[\gamma]=\int_0^1 f_b'(t)\,dt = f_b(1)-f_b(0)$. Under a.e.\ differentiability and actual composite derivative propagation, induction via the chain rule identifies $q_i^b(t)=\partial_i K_b(\gamma(t))\dot\gamma_i(t)$ a.e. Note that generally $I_i^b \ne \sum_j c_{bj}(1) I_i^j$: products of coefficients must be integrated pointwise along the common path.
\end{proof}

\Cref{lem:E:sacc-path-regular} establishes three checkable criteria for \cref{thm:path-recon}. First, (a) local Lipschitz continuity of a composite $K$ on an open neighbourhood of the compact image $\gamma([0,1])$ gives a finite cover bound ensuring both AC of $K\circ\gamma$ and uniform boundedness of existing partial derivatives (hence $q_i\in L^1(0,1)$). Second, (b) at any max/min/clamp tie, two AC traces $f_\ell(\gamma(t))$ and $(K\circ\gamma)(t)$ agreeing on a measurable set $E$ have equal derivatives a.e.\ on $E$ (via Lebesgue density points of $E$ where difference quotients force derivative zero), justifying convex combinations of active tie derivatives including persistent ties. Third, (c) at any seminorm zero $Lz(\gamma(t))=0$, every component of $Lz(t)$ has derivative zero a.e.\ on its zero set, validating the zero derivative rule as tangent compatible.

\begin{proposition}[Clarke ray extension and homogeneous collapse; \cref{prop:E:sacc-AS-Euler,lem:IMA-clarke-ray}]\label{prop:clarke-ray}
On any open cone where $K$ is locally Lipschitz and positively homogeneous of degree $r>0$ with $K(0)=0$:
\begin{enumerate}[label=(\roman*),leftmargin=*,itemsep=1pt]
\item If $K$ is differentiable along $\gamma(t)=tx$ a.e., then $I_i[\gamma]=\frac{1}{r}x_i\partial_i K(x)$, agreeing identically with Euler allocation when $r=1$.
\item If $r=1$ and the ray lies along a persistent kink (so ambient differentiability fails for all $t>0$), selecting the minimum-norm Clarke gradient $g_*(tx)=g_*(x)\in\partial_C K(x)$ yields constant measurable densities $q_i(t)=x_i g_{*,i}(x)$ satisfying $I_i = x_i g_{*,i}(x)$ and $\sum_i I_i = K(x)$.
\end{enumerate}
\end{proposition}

\paragraph{Four analytic pathologies establishing hypothesis necessity.}
Dropping any condition of \cref{thm:path-recon} leads to quantitative failure:
\begin{enumerate}[label=(\alph*),leftmargin=*,itemsep=2pt]
\item \emph{Scalar AC does not supply componentwise $L^1$ integrability (\cref{ex:E:sacc-L1}).} Put $s=(z_1+z_2)/2$, $d=z_2-z_1$, and $K(z)=s[1+\sin(d/s^2)]$ ($s>0$, zero otherwise). Since $|K|\le 2s$, $K$ is continuous across $s=0$ and smooth on $s>0$. Along $\gamma(t)=(t,t)$, $K(\gamma(t))=t$ is linear and AC, yet differentiating before restricting to $d=0$ gives
\[
  \partial_1 K(t,t)=\frac{1}{2}-\frac{1}{t},\qquad \partial_2 K(t,t)=\frac{1}{2}+\frac{1}{t}.
\]
Neither integral converges on $(0,1]$ (on $[\epsilon,1]$, they equal $(1-\epsilon)/2 \mp \log(1/\epsilon)$), even though their pointwise sum is constant~$1$.
\item \emph{Integrable a.e.\ derivatives do not replace AC (\cref{ex:E:sacc-AC}).} Along $\gamma(t)=t$ for the Cantor staircase function on $[0,1]$, $K$ is continuous and monotone with $K(0)=0,K(1)=1$, and constant on removed middle thirds, so $K'=0$ a.e.\ integrates to $0\ne 1$.
\item \emph{Local Lipschitz continuity alone does not ensure a.e.\ differentiability along a ray.} Along $\gamma(t)=(t,t)$ on $K(z)=|z_1-z_2|$, every positive point on the ray sits on a persistent kink, so path-integral attribution needs a declared Clarke selection. The same happens inside FRTB-IMA (\cref{lem:IMA-persistent-tie}), whose amber coefficient $k$, standardised capital $B$ of the internal-model desks and non-IMA capital $C_U$ are defined in \cref{sec:model-ima}. With reduced-set, stressed-period expected-shortfall rows $(1,0)/\sqrt{12}$ and $(0,1)/\sqrt{12}$ and $k=1/4$, capital is $K=(1-k)(3\max(w_1,w_2)+4(w_1+w_2))+kB+C_U$ near $w=(1,1)$, and its gradient jumps by $(1-k)\cdot 3\,(1,-1)=(9/4,-9/4)$ across the diagonal at every positive ray point.
\item \emph{Path, reference, and endpoint-propagation dependence (\cref{ex:E:sacc-pathdep,ex:E:sacc-compose}).} Even on smooth degree-$2$ $K(z_1,z_2)=z_1 z_2$, the diagonal ray from $(0,0)$ to $(1,1)$ yields $(1/2,1/2)$ whereas sequential paths yield $(0,1)$ and $(1,0)$. More subtly, when composing a non-homogeneous exposure layer $e(x)=(x_1^2,x_2)$ with a degree-$1$ norm $F(e)=\sqrt{e_1^2+e_2^2}$ (mirroring SA-CCR feeding BA-CVA), propagating separate exposure IG allocations $(1,0),(0,1)$ via endpoint derivatives $w_1=w_2=1/\sqrt{2}$ yields $(1/\sqrt{2},1/\sqrt{2})\approx (0.7071,0.7071)$, whereas true composite IG along $(t,t)$ evaluates
\[
  \left(\int_0^1 \frac{2t^2}{\sqrt{1+t^2}}\,dt,\ \int_0^1 \frac{1}{\sqrt{1+t^2}}\,dt\right)
  =\bigl(\sqrt{2}-\operatorname{arsinh} 1,\ \operatorname{arsinh} 1\bigr)\approx (0.53284,\ 0.88137).
\]
\end{enumerate}

\subsection{Finite-coalition Shapley games}\label{sec:games}

For a finite player set $N=\{1,\dots,p\}$ (trades, contracts, or risk-factor coordinates according to the contract), let $T_S x$ retain coordinates in coalition $S\subseteq N$ and set coordinates in $N\setminus S$ to zero while holding fixed all metadata, matrices, supervisory parameters, and external capacities, recomputing every downstream DAG node. Define $v(S)=K(T_S x)$.

\begin{theorem}[Discrete coalition efficiency under domain closure]\label{thm:shapley-eff}
Suppose $v(S)$ is defined and finite for every $S\subseteq N$ (domain closure). Then the exact discrete Shapley value \citep{shapley1953value}
\begin{equation}\label{eq:shapley-gen}
  \phi_i(v) = \sum_{S\subseteq N\setminus\{i\}} \frac{|S|!\,(p-|S|-1)!}{p!}\,\bigl(v(S\cup\{i\})-v(S)\bigr)
\end{equation}
exists and satisfies exact efficiency $\sum_{i=1}^p \phi_i(v) = v(N) - v(\varnothing) = K(x) - K(0)$, independently of continuity, differentiability, positive homogeneity, or jump discontinuities.
\end{theorem}
\begin{proof}
Across all $p!$ permutations $\pi$ of $N$, let $P_i^\pi$ denote the players preceding $i$. Exactly $|S|!(p-|S|-1)!$ orderings satisfy $P_i^\pi=S$, so \eqref{eq:shapley-gen} equals the uniform average of marginal contributions $v(P_i^\pi\cup\{i\})-v(P_i^\pi)$. For every fixed permutation $\pi$, summing over $i\in N$ telescopes along the nested coalition chain from $\varnothing$ to $N$, equalling $v(N)-v(\varnothing)$. Averaging over permutations preserves the sum.
\end{proof}

The telescoping argument is universal, but what the Shapley value means still depends on three features of the contract. The first is \emph{domain closure}. A portfolio $x\in\Omega$ fails coalition closure if some sub-book $T_S x$ leaves $\Omega$. FRTB-IMA, whose notation is defined in \cref{sec:model-ima}, gives an example (\cref{lem:IMA-convention-games,rem:B:shap}): its stress ratio $G_c$ divides by the reduced-set current expected shortfall $R_c$, and its amber coefficient $k=U/(2V)$ divides by the green-plus-amber desk sum $V$. Sub-books without positions in an active risk class have $R_c(w_S)=0$, and sub-books concentrated on amber desks with $V(w_S)=0$ leave $k$ undefined. The game therefore needs a declared extension rule. \emph{Convention~A} sets the undefined $G_c$ and $k$ terms to zero; it targets $K(x)-v_A(\varnothing)$ and leaves the historical-origin terms inside $v_A(\varnothing)$. \emph{Convention~B} falls back to the all-desk standardised capital $Z(w_S)$ at the domain walls and sets $v_B(\varnothing)=0$, so that it targets the full $K(x)$.

Closure can also fail outright, in FRTB-SA and ISDA SIMM, whose notation is defined in \cref{sec:model-sa,sec:model-simm}. In both, sensitivities enter square roots of correlated quadratic forms, and a sub-book can make a radicand negative. In FRTB-SA (\cref{lem:SA-coalition-refusal}), the relevant radicand is the alternative-correlation radicand $\varrho_B$ of the delta aggregation. Take four tickets with columns $(-e_1,e_2,e_3,e_1)$ and quantities $(2,1,1,4)$, a positive-definite medium-scenario correlation matrix $\Gamma_M$ with off-diagonals $0.9$, $0.9$ and $0.63$, and the high-scenario matrix $\Gamma_H=\min\{1.25\Gamma_M,1\}$. The full book, with $WS=(2,1,1)$, has strictly positive radicands. The three-ticket sub-book, with $WS=(-2,1,1)$, gives $\varrho_B=4+1+1+2(-2-2+63/80)=-17/40<0$ and is refused. In ISDA SIMM (\cref{lem:SIMM-coalition-test}), four unit-weight factors $d=(1,1,1,2)$ with off-diagonal correlations $-3/4$ on $\{1,2,3\}$ have full radicand $5/2>0$, while the sub-book $\{1,2,3\}$ gives $3-2(3/4)(3)=-3/2<0$. Shapley attribution for FRTB-SA is therefore stated on the coalition-admissible subset $\mathcal{C}_{\mathcal{P}}$ of portfolios whose sub-books all keep $\varrho_B\ge 0$, and for ISDA SIMM on the radicand-closed subset $\mathcal{Q}$.

The second feature is the \emph{value of the empty book}. Fixed external RWA, posted collateral and fixed historical expenses produce $K(0)\ne 0$ in the Basel Output Floor, SA-CCR and Solvency~II. The identity $\sum_i\phi_i = K(x)-K(0)$ then allocates only the capital in excess of $K(0)$. Allocating the full $K(x)$ requires either a separate, declared baseline account or a two-player branch-level game (\cref{lem:FLOOR-component-shapley}), whose players are not the trades.

The third feature is that the Shapley value is \emph{not the Euler allocation} even in the classical setting. When $K(0)=0$ and $K$ is smooth and positively homogeneous of degree one, discrete Shapley values still differ from Euler and ray path-integral allocations (\cref{lem:SOLV2-shapley-euler}). For $K(a,b)=\sqrt{a^2+\tfrac{1}{2}ab+b^2}$ at $(1,2)$, the coalition values $(0,1,2,\sqrt{6})$ give $\phi_1=(\sqrt{6}-1)/2\approx 0.7247$, whereas $A^{\rm Euler}_1=I_1=3/(2\sqrt{6})\approx 0.6124$.
\section{Structural Models and Computational Graphs Across All Eleven Frameworks}\label{sec:matrix}

This section applies the graph calculus of \cref{sec:dag} and the path and coalition tools of \cref{sec:paths-games} to each of the eleven frameworks: regulatory market-risk and bank-wide capital, over-the-counter initial margin, exchange clearing margin, counterparty credit and valuation adjustments, banking-book interest-rate risk and insurance solvency. Every subsection follows the same order. It describes the setting and coordinates, states the equations, draws the computational graph, explains what the attribution rules do on that graph and records the three verdicts. Where a verdict is a refutation, the subsection names the witness; the witnesses themselves are established in \cref{sec:counterexamples}, and the complete proofs are in \cref{app:E}.

\subsection{Contract architecture, domain discipline, and four structural regimes}\label{sec:matrix-arch}\label{sec:matrix-regimes}

Each cell $(m,M)$ of \cref{tab:matrix33}, pairing a paradigm $m$ with a framework $M$, is governed by a formal contract fixed before the cell is resolved. The contract first names a broad domain $\Omega$ and a minimum useful domain $\Omega_{\min}\subseteq\Omega$. Inside $\Omega_{\min}$ the characteristic mechanisms of the framework must remain active: fixed collateral in SA-CCR and XVA, non-homogeneous external RWA and FDB pools in the Basel Output Floor and Solvency~II, concentration thresholds in ISDA SIMM, and bank-level prepayment limits in IRRBB. An explicit open, full-dimensional parametric family shows that $\Omega_{\min}$ is nonempty and nondegenerate. A contract may not remove an essential non-homogeneous mechanism merely to obtain a positive Euler result. The single exception is Shapley attribution for FRTB-SA, whose contract restricts to coalition-admissible portfolios $\mathcal{C}_{\mathcal{P}}$ that avoid negative alternative radicands $\varrho_B<0$.

The contract also fixes the conventions of each paradigm in advance. For Euler attribution it fixes the tie mode: a point selection through minimum-norm or lexicographic Clarke selectors on globally Lipschitz homogeneous graphs, or the set-valued realisable Clarke set for SA-CCR. For path integrals it fixes the target, either the increment $K(x)-K(x_0)$ or the full value $K(x)$ with a declared baseline split $b_i^0$, together with the Clarke ray-density rule on persistent ties. For Shapley attribution it fixes the players, the zeroing masks $T_S x$ and full downstream recomputation, uniformly over all $2^{|N|}$ coalitions.

Under its contract, each cell takes exactly one of three verdicts. It is \Dstruct{} if every predicate is proved on all of $\Omega$, conditional on the \tagU{} entries. It is \Dscoped{} if the predicates are proved on an explicit strict subdomain $\Omega'\subsetneq\Omega$ that contains $\Omega_{\min}$, and a witness in $\Omega\setminus\Omega'$ shows that the restriction is needed. It is \Ncell{} if the paradigm fails inside $\Omega_{\min}$ on an attainable open family or at an exact witness. Every framework retains explicit \tagU{} leaves or untranscribed correlation or calibration tables, so all verdicts are stated at the \Tstruct{}.

\Cref{tab:matrix33} summarises the verdicts; the eight refuted cells are catalogued in \cref{tab:counterexamples}, and the complete contract, witness-family and theorem lookup for all $33$ cells is in \cref{tab:full-33-lookup} in \cref{app:table4}. The eleven models fall into four structural regimes, already used to group \cref{fig:hierarchy-landscape}. \emph{Regime~I}, the globally homogeneous conic models FRTB-IMA, discrete ES/VaR and classic CME SPAN (\cref{sec:graph,sec:model-es-var,sec:model-span}), consists of globally Lipschitz functions that are positively homogeneous of degree one with $K(0)=0$; all nine of its cells are \Dstruct{}. \emph{Regime~II} is BA-CVA and SA-CVA (\cref{sec:model-cva}), which keep conic homogeneity in their input factors on a positivity cone $\Omega_+$ and refuse indefinite root-zero singularities outside it; its three cells are \Dscoped{}. \emph{Regime~III} contains the discontinuous switches of FRTB-SA and the super-linear concentration factors of ISDA SIMM (\cref{sec:model-sa,sec:model-simm}), giving three \Ncell{} and three \Dscoped{} cells. \emph{Regime~IV} gathers the fixed offsets, collateral thresholds and amount caps of the Basel Output Floor, SA-CCR, bilateral XVA, Solvency~II and IRRBB $\Delta\mathrm{EVE}$ (\cref{sec:model-floor,sec:model-saccr,sec:model-xva,sec:model-solv2,sec:model-irrbb}). These rules break Euler completeness on open smooth chambers yet admit path-integral and Shapley attribution, globally or on a stated subdomain, giving five \Ncell{}, eight \Dstruct{} and two \Dscoped{} cells.

\subsection{FRTB Internal Models Approach (FRTB-IMA)}\label{sec:graph}\label{sec:model-ima}

FRTB-IMA \citep{bcbs2019mar,bcbs2020mar33} operates on dated position matrices $x\in\R^{T\times N}$ (spanning $60$-day daily and $12$-week weekly observation windows, with today's book $w:=x_{T-1}$), combining tail ties, floored stress ratios, non-modellable stress norms, default quantiles, rolling historical maxima, and standardised-approach surcharge caps (\cref{fig:dag}). Let $L$ be a scenario-by-trade unit-loss matrix, $z=Lw$, $p_s$ scenario weights, $\alpha=0.975$, $\beta=1-\alpha=0.025$, and $c_s=p_s/\beta$. Expected shortfall (ES) has primal and dual linear-program representations \citep{rockafellar2000optimization,rockafellar2002conditional}:
\begin{equation}\label{eq:es}
  \rho(z)=\min_\tau\Bigl\{\tau+\frac{1}{\beta}\sum_s p_s(z_s-\tau)_+\Bigr\}
  =\max\bigl\{q^\top z:\; 0\le q\le c,\;\one^\top q=1\bigr\}.
\end{equation}
On the optimal dual face, $q_s=c_s$ on scenarios strictly exceeding the quantile threshold $\tau$, $q_s=0$ below it, and residual mass splits freely across tied boundary scenarios. On each observation date $s\in\{0,\dots,T-1\}$, ES is evaluated across $90$ blocks $e_{c,v,j}$ indexed by risk class $c\in\{0,\dots,5\}$ (broad diversified portfolio $c=0$ and five individual risk classes), factor-set/period $v\in\{\mathrm{FC},\mathrm{RC},\mathrm{RS}\}$, and five liquidity horizons $j$ ($10, 20, 40, 60$, and $120$ days, weighted by $a=(1,1,2,2,6)$):
\begin{equation}\label{eq:lh}
  E_{c,v}=\Bigl(\sum_j a_j e_{c,v,j}^2\Bigr)^{1/2},\qquad
  G_c=S_c\max\bigl(1,F_c/R_c\bigr),\qquad
  I=\frac{1}{2}G_0+\frac{1}{2}\sum_{c=1}^5 G_c,
\end{equation}
where $S_c=E_{c,\mathrm{RS}}$, $F_c=E_{c,\mathrm{FC}}$, and $R_c=E_{c,\mathrm{RC}}>0$. Non-modellable risk factor stress charges $(\xi,\eta,\zeta)$ (each a fixed linear vector or finite support maximum $\max_k B_k^\top w$) aggregate via weighted $\ell_2$ and correlated quadratic norms, while default risk $d(w)$ is the lower $0.999$-quantile of a finite linear loss law $D$:
\begin{equation}\label{eq:ses-drc}
  N(w)=\|\xi(w)\|_2+\|\eta(w)\|_2+\sqrt{\zeta(w)^\top Q\zeta(w)},\qquad
  d(w)=\inf\bigl\{u:\Pr(Dw\le u)\ge 0.999\bigr\},
\end{equation}
with $Q=(1-r^2)\mathrm{Id}+r^2\one\one^\top$, $r=0.6$. History combines the latest observation with $60$-day daily averages and $12$-week weekly averages ($m_c=1.5$):
\begin{equation}\label{eq:hist}
  C_A=\max\bigl\{I_{T-1}+N_{T-1},\; m_c\overline I_{60}+\overline N_{60}\bigr\},\qquad
  C_D=\max\bigl\{d_{\rm latest},\; \overline d_{12w}\bigr\},\qquad
  J=C_A+C_D.
\end{equation}
Finally, standardised-approach quantities on today's book---IMA-desk SA capital $B$, non-IMA-desk capital $C_U$, all-desk SA capital $Z$, amber desk sum $U$, and green-plus-amber sum $V$---determine the degree-$0$ amber coefficient $k$, surcharge $H$, capital $K$, and risk-weighted assets:
\begin{equation}\label{eq:final}
  k=\frac{U}{2V},\qquad
  H=k\,(B-J)_+,\qquad
  K=\min\bigl\{J+H+C_U,\; Z\bigr\}+(J-B)_+,\qquad
  \mathrm{RWA}=12.5\,K.
\end{equation}

\begin{figure}[htbp]
\centering
\resizebox{\linewidth}{!}{%
\begin{tikzpicture}[
  x=1mm, y=1mm,
  font={\fontsize{8}{9.4}\selectfont}, text=fbInk,
  >={Stealth[length=1.5mm, width=1.1mm]},
  box/.style={align=flush center, inner xsep=1mm, inner ysep=0.8mm},
  op/.style={box, draw=fbGrey, line width=0.5pt, rounded corners=1.6pt, fill=white},
  leaf/.style={box, draw=fbBlue, line width=0.8pt, fill=white},
  kink/.style={fill=fbRed!14},
  ar/.style={->, draw=fbInk, line width=0.42pt},
  ln/.style={draw=fbInk, line width=0.42pt},
  el/.style={font={\fontsize{8}{9.4}\selectfont}, inner sep=0.8pt},
  dot/.style={circle, fill=fbInk, inner sep=0pt, minimum size=1.1mm},
]
\def\dagnote#1{{\fontsize{7}{8.4}\selectfont #1}}
\def\dxA{14.4}  \def\dwA{26.8}
\def\dxB{43.9}  \def\dwB{20.0}
\def\dxC{73.5}  \def\dwC{26.8}
\def\dxD{106.3} \def\dwD{28.0}
\def\dxE{140.0} \def\dwE{30.2}
\def\dyA{0} \def\dyB{-16.5} \def\dyC{-29.5} \def\dyD{-46}
\node[leaf, kink, text width=\dwA mm, minimum height=12mm] (es) at (\dxA,\dyA) {ES blocks $e_{c,v,j}$\\\dagnote{90 blocks; tail ties}};
\node[op, text width=\dwB mm, minimum height=12mm] (lh) at (\dxB,\dyA) {liquidity-horizon\\aggregation\\$E=\bigl(\sum_ja_je_j^2\bigr)^{1/2}$};
\node[op, kink, text width=\dwC mm, minimum height=12mm] (gr) at (\dxC,\dyA) {floored stress ratio\\$G_c=S_c\max(1,F_c/R_c)$};
\node[op, text width=\dwD mm, minimum height=12mm] (im) at (\dxD,\dyA) {internal-models charge\\$I=\tfrac12G_0+\tfrac12\sum_cG_c$};
\node[leaf, kink, text width=\dwA mm, minimum height=12mm] (st) at (\dxA,\dyB) {stress charges\\\dagnote{support-function ties}};
\node[op, text width=\dwB mm, minimum height=12mm] (nm) at (\dxB,\dyB) {non-modellable\\charge $N$\\\dagnote{sum of norms}};
\node[leaf, kink, text width=\dwA mm, minimum height=8.5mm] (dd) at (\dxA,\dyC) {default charge $d$\\\dagnote{$0.999$ quantile}};
\node[op, kink, text width=\dwD mm, minimum height=9mm] (ca) at (\dxD,\dyB) {$C_A=\max\{I_{t-1}+N_{t-1},$\\$1.5\,\bar I_{60}+\bar N_{60}\}$};
\node[op, kink, text width=\dwD mm, minimum height=8.5mm] (cd) at (\dxD,\dyC) {$C_D=\max\{d_{\rm latest},\bar d_{12w}\}$};
\path let \p1=($(ca.center)+(0,3.2)$), \p2=($(cd.center)+(0,-3.2)$) in
  node[op, text width=17mm, minimum height=\y1-\y2, anchor=north] (jj) at (\dxE mm,\y1) {history node\\$J=C_A+C_D$};
\node[leaf, text width=\dwA mm, minimum height=12mm] (sa) at (\dxA,\dyD) {standardised-approach\\inputs $B,C_U,Z,U,V$\\\dagnote{given, on today's book}};
\node[op, text width=\dwB mm, minimum height=9mm] (kc) at (\dxC,\dyD) {amber coefficient\\$k=U/(2V)$};
\node[op, kink, text width=\dwD mm, minimum height=9mm] (hh) at (\dxD,\dyD) {amber surcharge\\$H=k\,(B-J)_+$};
\node[op, kink, text width=\dwE mm, minimum height=12mm] (kk) at (\dxE,\dyD) {capital\\$K=\min\{J+H+C_U,\,Z\}$\\${}+(J-B)_+$};
\def\dpad{2}
\draw[fbBlueMid, line width=0.6pt, rounded corners=2.5mm]
  ($(es.north west)+(-\dpad,\dpad)$) -- ($(im.north east)+(\dpad,\dpad)$)
  -- ($(im.south east)+(\dpad,-\dpad)$) -- ($(nm.east|-im.south)+(\dpad,-\dpad)$)
  -- ($(nm.east|-dd.south)+(\dpad,-\dpad)$) -- ($(dd.south west)+(-\dpad,-\dpad)$) -- cycle;
\node[fill=white, inner xsep=1.2mm, inner ysep=0.3mm, anchor=west, text=fbBlue, font={\fontsize{8}{9.4}\selectfont\itshape}]
  at ($(es.north west)+(2,\dpad)$) {evaluated on each stored date's book};
\draw[ar] (es) -- (lh);
\draw[ar] (lh) -- (gr);
\draw[ar] (gr) -- (im);
\draw[ar] (st) -- (nm);
\draw[ar] (im) -- (ca);
\draw[ar] (nm) -- (ca);
\draw[ar] (dd) -- (cd);
\draw[ar] (ca.east) -- (ca.east-|jj.west);
\draw[ar] (cd.east) -- (cd.east-|jj.west);
\draw[ar] (sa) -- node[el, above] {$U,V$} (kc);
\draw[ar] (kc) -- (hh);
\draw[ar] (hh) -- (kk);
\coordinate (jb) at ($(jj.south)!0.5!(kk.north)$);
\draw[ar] (jj.south) -- (kk.north);
\draw[ar] (jb) -| (hh.north);
\node[dot] at (jb) {};
\coordinate (bus) at ($(sa.south)+(0,-5.5)$);
\draw[ln] (sa.south) -- (bus);
\draw[ar] (bus) -| node[el, pos=0.74, right=0.4mm] {$B$} (hh.south);
\draw[ar] (bus) -| node[el, pos=0.74, right=0.4mm] {$B,C_U,Z$} (kk.south);
\node[dot] at (bus-|hh.south) {};
\end{tikzpicture}}
\caption{Computational DAG for FRTB-IMA from dated leaves to final capital $K$. Square-bordered nodes are inputs; rounded nodes are internal operators; shaded nodes mark potential nondifferentiability; the outer frame marks components evaluated on each stored historical date.}\label{fig:dag}
\end{figure}

Nodewise propagation $A^b=\sum_{j\in\mathrm{ch}(b)} c_{bj}A^j$ (\cref{lem:B:nodes}) turns this graph into explicit local attribution rules. Expected shortfall uses its dual face, $A_i=w_i(L^\top q)_i$, and the default quantile uses an attaining scenario, $A_i=w_i D_{s^\ast,i}$. The liquidity-horizon norm propagates through the outer weights $a_j e_j/E$, and the non-modellable norms through $\xi_a/\|\xi\|_2$, $\eta_b/\|\eta\|_2$ and $(Q\zeta)_\kappa/\sqrt{\zeta^\top Q\zeta}$. The stress ratio uses $A^G=A^S$ when $F<R$ and $(F/R)A^S+(S/R)A^F-(SF/R^2)A^R$ when $F>R$; at $F=R$ it uses declared tie weights $\lambda\in[0,1]$, and it refuses $R\le 0$. For the degree-$0$ coefficient $k=U/(2V)$, the coefficient allocation $a^k = A^U/(2V) - UA^V/(2V^2)$ gives, on $B>J$,
\begin{equation}\label{eq:amber}
  A^H = k(A^B - A^J) + (B - J)\,a^k.
\end{equation}
The verdicts are proved in \cref{app:E:ima}. They hold on the strict-denominator cone $\mathcal{O}$, where $R_c>0$ for each class and $V>0$ whenever amber desks exist, within contracted coefficient strata that admit a $499$-dimensional smooth open ball (\cref{lem:IMA-open-family}). There all three paradigms are \textbf{\Dstruct{}}. Hybrid minimum-norm Clarke selectors close Euler attribution (\cref{lem:IMA-clarke-point}) and path-integral attribution (\cref{lem:IMA-clarke-ray}) despite the persistent kinks of expected shortfall and of the history maxima. The declared domain-extension Conventions~A and~B close Shapley attribution on the whole cone (\cref{lem:IMA-convention-games}).

\subsection{Discrete Expected Shortfall and Value-at-Risk}\label{sec:model-es-var}

Fix signed portfolio coordinates $x\in\R^n$, $m\ge 1$ equiprobable scenario loss vectors $a_s\in\R^n$ with linear portfolio losses $z_s(x)=a_s^\top x$ ordered $z_{(1)}\le\dots\le z_{(m)}$, confidence level $\alpha\in(0,1)$, $k=\lceil m\alpha\rceil$, $\beta=m(1-\alpha)$, $r=\lfloor\beta\rfloor$, and $d=\beta-r$ (\cref{fig:dag-es-var}). Empirical Value-at-Risk and fractional-tail Expected Shortfall are defined by:
\begin{align*}
  K_Q(x) &= \operatorname{VaR}_\alpha(x)=z_{(k)}(x),\\
  K_E(x) &= \operatorname{ES}_\alpha(x)=\frac{1}{\beta}\Bigl(\sum_{j=m-r+1}^m z_{(j)}(x)+d\,z_{(m-r)}(x)\Bigr) = \max_{q\in Q_\beta} q^\top z(x),
\end{align*}
over $Q_\beta=\{q\in\R^m:0\le q_s\le 1/\beta,\ \sum_s q_s=1\}$ with optimal dual face $F(x)=\{q\in Q_\beta: q_s=1/\beta\ (z_s>c),\ q_s=0\ (z_s<c),\ \sum_{z_s=c} q_s=1-|H|/\beta\}$ ($c=z_{(k)}$, $H=\{s:z_s>c\}$).

\begin{figure}[htbp]
\centering
\resizebox{\linewidth}{!}{%
\begin{tikzpicture}[
  x=1mm, y=1mm,
  font={\fontsize{8}{9.4}\selectfont}, text=fbInk,
  >={Stealth[length=1.5mm, width=1.1mm]},
  box/.style={align=flush center, inner xsep=1mm, inner ysep=0.8mm},
  op/.style={box, draw=fbGrey, line width=0.5pt, rounded corners=1.6pt, fill=white},
  leaf/.style={box, draw=fbBlue, line width=0.8pt, fill=white},
  kink/.style={fill=fbRed!14},
  ar/.style={->, draw=fbInk, line width=0.42pt},
  ln/.style={draw=fbInk, line width=0.42pt},
  el/.style={font={\fontsize{8}{9.4}\selectfont}, inner sep=0.8pt},
  dot/.style={circle, fill=fbInk, inner sep=0pt, minimum size=1.1mm},
]
\def\dagnote#1{{\fontsize{7}{8.4}\selectfont #1}}
\def\dxA{15.5} \def\dwA{29}
\def\dxB{53.5} \def\dwB{27}
\def\dxC{92.5} \def\dwC{31}
\def\dxD{144}  \def\dwD{34}
\def\dyA{0} \def\dyM{-11} \def\dyB{-22}
\node[leaf, text width=\dwA mm, minimum height=13mm] (pos) at (\dxA,\dyM) {positions \& scenarios\\$x\in\R^n,\ a_s\in\R^n$\\\dagnote{signed weights; fixed laws}};
\node[op, text width=\dwB mm, minimum height=13mm] (loss) at (\dxB,\dyM) {scenario losses\\$z_s(x)=a_s^\top x$\\\dagnote{degree-$1$ linear leaves}};
\node[op, kink, text width=\dwC mm, minimum height=13mm] (ord) at (\dxC,\dyM) {order statistics\\$z_{(1)}\le\dots\le z_{(m)}$\\\dagnote{polyhedral ordering cones}};
\node[op, kink, text width=\dwD mm, minimum height=12mm] (var) at (\dxD,\dyA) {discrete VaR\\$K_Q(x)=z_{(k)}(x)$\\\dagnote{kinks at tied quantiles}};
\node[op, kink, text width=\dwD mm, minimum height=13mm] (es) at (\dxD,\dyB) {fractional-tail ES\\$K_E(x)=\max_{q\in Q_\beta} q^\top z$\\\dagnote{dual optimal face $F(x)$}};
\def\dpad{2.2}
\draw[fbBlueMid, line width=0.6pt, rounded corners=2.5mm]
  ($(loss.north west)+(-\dpad,\dpad+2.5)$) rectangle ($(ord.south east)+(\dpad,-\dpad)$);
\node[fill=white, inner xsep=1.2mm, inner ysep=0.3mm, anchor=west, text=fbBlue, font={\fontsize{8}{9.4}\selectfont\itshape}]
  at ($(loss.north west)+(1.5,\dpad+2.5)$) {evaluated over scenarios $s=1,\dots,m$};
\draw[ar] (pos) -- (loss);
\draw[ar] (loss) -- (ord);
\coordinate (sp) at ($(ord.east)+(3.5,0)$);
\draw[ln] (ord.east) -- (sp);
\draw[ar] (sp) |- node[el, pos=0.73, above=0.6mm] {$k=\lceil m\alpha\rceil$} (var.west);
\draw[ar] (sp) |- node[el, pos=0.73, below=0.6mm] {$\beta=m(1-\alpha)$} (es.west);
\node[dot] at (sp) {};
\end{tikzpicture}}
\caption{Computational DAG for Discrete Expected Shortfall ($K_E$) and Value-at-Risk ($K_Q$) over finite equiprobable scenario vectors $a_s$. Square-bordered nodes are inputs; shaded operators carry ordering-cone and dual-face kinks.}\label{fig:dag-es-var}
\end{figure}

Proofs for this model are in \cref{app:E:es-erc}. Ordering the distinct scenario rows partitions $\R^n$ into finitely many closed polyhedral cones, and both rules are globally Lipschitz with constant $\max_s\|a_s\|_2$ (\cref{lem:ES_ERC-model}). Although $K_Q$ is nonconvex in general, selecting the lexicographically least slope among the incident open ordering cones gives a scale-invariant Clarke point selector. This selector closes Euler attribution (\cref{lem:ES_ERC-clarke}) and path-integral attribution (\cref{lem:ES_ERC-path}) everywhere on $\R^n$, both as \textbf{\Dstruct{}}. Masked sub-books $T_S x$ never meet a domain wall, so the Shapley value over all $2^n$ coalitions is exactly efficient, $\sum_i\phi_i=K(x)$ (\cref{lem:ES_ERC-shapley}). All three cells are therefore \textbf{\Dstruct{}}.

The dual description of expected shortfall also bounds how much exact repricing an exact allocation needs. We use the linear-programming form \eqref{eq:es}, in which the dual weights satisfy $0\le q\le c$ with scenario caps $c_s=p_s/\beta$; for the equiprobable model above, $c_s=1/\beta$.

\paragraph{Theoretical certified selective repricing.}\label{sec:cert}
On discrete ES leaves with nonlinear option pricing over spot scenarios, which appear both in standalone expected shortfall and in internal-models graphs, exact dual attribution separates interior rows from active tail rows without approximating the allocations. Suppose each option's unit price $v(\sigma)$ is convex in the scenario spot $\sigma$, as in single-underlying lattice option pricing. Evaluating $v$ at a small set of anchor spots $a_1<\dots<a_m$ gives rigorous upper chords through adjacent anchors $[\sigma\in[a_k,a_{k+1}]]$ and lower secant lines through pairs $(a_i,a_j)$ outside $\sigma$ (\cref{lem:B:anchor}). Interval arithmetic propagates these per-trade brackets into scenario-loss intervals $[\ell_s,u_s]$. Suppose the intervals separate a top tail set $\mathcal{H}$ and a boundary scenario $b\notin\mathcal{H}$ with $\sum_{\mathcal{H}}c_s < 1 \le \sum_{\mathcal{H}}c_s+c_b$ and $\min_{s\in\mathcal{H}}\ell_s \ge z_b \ge \max_{r\notin\mathcal{H}\cup\{b\}} u_r$, where $z_b$ is evaluated exactly. Then \cref{thm:B:cert} proves that $q^\ast=c$ on $\mathcal{H}$, $q^\ast_b=1-\sum_{\mathcal{H}}c_s$ and $q^\ast=0$ elsewhere satisfies
\[
  (q^\ast - q)^\top z = \sum_{s\in\mathcal{H}}(c_s - q_s)(z_s - z_b) + \sum_{r\notin\mathcal{H}\cup\{b\}}(-q_r)(z_r - z_b) \ge 0
\]
for every feasible $q\in[0,c]$ with $\one^\top q=1$. This establishes dual optimality, and uniqueness when both separation inequalities are strict, using exact option pricing only on $\mathcal{H}\cup\{b\}$. Conversely, \cref{prop:B:lower} proves that exact repricing on active tail rows cannot be avoided. On any active tail row with $q_s>0$ and $N_+$ nonzero entries, and without structural cross-trade equations, perturbing two unobserved entries to $(v_i+\delta,v_j-\delta)$ preserves the row sum while changing the trade attributions by $\pm q_s\delta$.

\subsection{FRTB Standardised Approach (FRTB-SA)}\label{sec:model-sa}

Under MAR20--MAR23 \citep{bcbs2019mar,bcbs2026mar21}, static position coordinates $x\in\R^n$ enter Sensitivities-Based Method (SBM), Default Risk Charge (DRC), and Residual Risk Add-On (RRAO) without stored historical windows: $\mathrm{SA}(x)=\mathrm{SBM}(x)+\mathrm{DRC}(x)+\mathrm{RRAO}(x)$ (\cref{fig:dag-sa}). For each risk class $r$ and correlation scenario $\lambda\in\{\text{low},\text{medium},\text{high}\}$, weighted sensitivities $WS_k=RW_k s_k$ ($s_k=a_k^\top x$) and curvature repricing leaves $CVR_k^\pm=-\sum_i\{V_i(x_k^\pm)-V_i(x_k)\mp RW_k^{\rm curv}s_{ik}\}$ satisfy:
\begin{align*}
  q_b &= \sum_{k\in b}WS_k^2+\sum_{k\ne l\in b}\rho_{kl}^\lambda WS_k WS_l, & K_b &= \sqrt{\max\{0,q_b\}},\qquad S_b=\sum_{k\in b}WS_k,\\
  \varrho_A &= \sum_b K_b^2+\sum_{b\ne c}\gamma_{bc}^\lambda S_b S_c, & \Delta^\lambda &= \begin{cases}\sqrt{\varrho_A}, & \varrho_A\ge 0,\\ \sqrt{\varrho_B}=\sqrt{\sum_b K_b^2+\sum_{b\ne c}\gamma_{bc}^\lambda\tilde S_b\tilde S_c}, & \varrho_A<0,\end{cases}
\end{align*}
with symmetric clamp $\tilde S_b=\cl(S_b,K_b)$ (refusing whenever $\varrho_B<0$). Curvature uses sign selector $\psi(u,v)=0$ if $u<0$ and $v<0$, else $1$:
\begin{align*}
  Q_b^\pm &= \sum_k\max\{CVR_k^\pm,0\}^2 + \sum_{k\ne l}\rho_{kl}^{\mathrm{curv},\lambda}\psi(CVR_k^\pm,CVR_l^\pm)CVR_k^\pm CVR_l^\pm,\qquad K_b^\pm=\sqrt{\max\{0,Q_b^\pm\}},\\
  K_b &= \max\{K_b^+,K_b^-\},\quad S_b=\begin{cases}\sum_{k\in b}CVR_k^+, & K_b^+>K_b^-\text{ or }(K_b^+=K_b^-,\sum CVR^+>\sum CVR^-),\\ \sum_{k\in b}CVR_k^-, & \text{otherwise},\end{cases}\\
  \mathrm{Curv}^\lambda &= \sqrt{\max\Bigl\{0,\ \sum_b K_b^2+\sum_{b\ne c}\gamma_{bc}^{\mathrm{curv},\lambda}\psi(S_b,S_c)S_b S_c\Bigr\}},\qquad \mathrm{SBM}(x)=\max_\lambda\sum_r(\Delta_r^\lambda+V_r^\lambda+\mathrm{Curv}_r^\lambda).
\end{align*}
Default risk computes $\mathrm{HBR}_b=L_b/(L_b+\Sigma_b)$, $\mathrm{DRC}_b=\max\{P_b-\mathrm{HBR}_b Q_b,0\}$ (with continuous fused zero product $\mathrm{HBR}_b Q_b:=0$ at $L_b+\Sigma_b=0$, CTP aggregation per MAR22.45), and residual add-on $\mathrm{RRAO}=\sum_i RW_i|N_i|$.

\begin{figure}[htbp]
\centering
\resizebox{\linewidth}{!}{%
\begin{tikzpicture}[
  x=1mm, y=1mm,
  font={\fontsize{8}{9.4}\selectfont}, text=fbInk,
  >={Stealth[length=1.5mm, width=1.1mm]},
  box/.style={align=flush center, inner xsep=1mm, inner ysep=0.8mm},
  op/.style={box, draw=fbGrey, line width=0.5pt, rounded corners=1.6pt, fill=white},
  leaf/.style={box, draw=fbBlue, line width=0.8pt, fill=white},
  kink/.style={fill=fbRed!14},
  ar/.style={->, draw=fbInk, line width=0.42pt},
  ln/.style={draw=fbInk, line width=0.42pt},
  el/.style={font={\fontsize{8}{9.4}\selectfont}, inner sep=0.8pt},
  dot/.style={circle, fill=fbInk, inner sep=0pt, minimum size=1.1mm},
]
\def\dagnote#1{{\fontsize{7}{8.4}\selectfont #1}}
\def\dxA{15.5} \def\dwA{27}
\def\dxB{52}   \def\dwB{31}
\def\dxC{93}   \def\dwC{33}
\def\dxD{131}  \def\dwD{26}
\def\dxE{163}  \def\dwE{22}
\def\dyA{0} \def\dyB{-18} \def\dyC{-35} \def\dyD{-49.5}
\node[leaf, text width=\dwA mm, minimum height=12mm] (dv) at (\dxA,\dyA) {delta \& vega leaves\\$WS_k=RW_k s_k$\\\dagnote{MAR21.4 linear leaves}};
\node[op, kink, text width=\dwB mm, minimum height=12mm] (kb) at (\dxB,\dyA) {bucket root \& clamp\\$K_b=\sqrt{\max\{0,q_b\}},\ \tilde S_b$};
\node[op, kink, text width=\dwC mm, minimum height=12mm] (dva) at (\dxC,\dyA) {cross-bucket switch\\$\Delta^\lambda,V^\lambda=\sqrt{\varrho_{A/B}}$\\\dagnote{jump at $\varrho_A=0$; refuse $\varrho_B<0$}};
\node[leaf, text width=\dwA mm, minimum height=12mm] (cvl) at (\dxA,\dyB) {curvature repricing\\$CVR_k^\pm$\\\dagnote{two shocks + base}};
\node[op, kink, text width=\dwB mm, minimum height=12mm] (cbk) at (\dxB,\dyB) {curvature bucket\\$K_b^\pm,\ \max\{K_b^+,K_b^-\}$\\\dagnote{jump at $K_b^+=K_b^-$}};
\node[op, kink, text width=\dwC mm, minimum height=12mm] (cva) at (\dxC,\dyB) {curvature agg.\\$\mathrm{Curv}^\lambda=\sqrt{\max\{0,q_{\rm curv}\}}$\\\dagnote{sign selector $\psi(S_b,S_c)$}};
\node[op, kink, text width=\dwD mm, minimum height=14mm] (sbm) at (\dxD,-9) {scenario max\\$\mathrm{SBM}=\max_\lambda$\\$\sum_r(\Delta_r^\lambda+V_r^\lambda+\mathrm{Curv}_r^\lambda)$};
\node[leaf, text width=\dwA mm, minimum height=11mm] (jtd) at (\dxA,\dyC) {DRC obligor JTDs\\$\ell_o,\sigma_o\ge 0$};
\node[op, kink, text width=\dwB mm, minimum height=11mm] (drc) at (\dxB,\dyC) {DRC ratio \& floor\\$\mathrm{HBR}_b=L_b/(L_b+\Sigma_b)$\\$\max\{P_b-\mathrm{HBR}_b Q_b,0\}$};
\node[leaf, kink, text width=\dwA mm, minimum height=9.5mm] (rrl) at (\dxA,\dyD) {RRAO notionals $|N_i|$};
\node[op, kink, text width=\dwB mm, minimum height=9.5mm] (rro) at (\dxB,\dyD) {residual add-on\\$\mathrm{RRAO}=\sum RW_i|N_i|$};
\node[op, text width=\dwE mm, minimum height=16mm] (tot) at (\dxE,-24) {SA capital\\$\mathrm{SA}=\mathrm{SBM}$\\${}+\mathrm{DRC}+\mathrm{RRAO}$};
\def\dpad{2.2}
\draw[fbBlueMid, line width=0.6pt, rounded corners=2.5mm]
  ($(kb.north west)+(-\dpad,\dpad+2.8)$) rectangle ($(cva.south east)+(\dpad,-\dpad)$);
\node[fill=white, inner xsep=1.2mm, inner ysep=0.3mm, anchor=west, text=fbBlue, font={\fontsize{8}{9.4}\selectfont\itshape}]
  at ($(kb.north west)+(1.5,\dpad+2.8)$) {evaluated across $\lambda\in\{\mathrm{low},\mathrm{medium},\mathrm{high}\}$};
\draw[ar] (dv) -- (kb);
\draw[ar] (kb) -- (dva);
\draw[ar] (cvl) -- (cbk);
\draw[ar] (cbk) -- (cva);
\draw[ar] (dva.east) -- ++(3.5,0) |- ($(sbm.west)+(0,3)$);
\draw[ar] (cva.east) -- ++(3.5,0) |- ($(sbm.west)+(0,-3)$);
\draw[ar] (jtd) -- (drc);
\draw[ar] (rrl) -- (rro);
\draw[ar] (sbm.east) -| (tot.north);
\draw[ar] (drc.east) -- ++(76.5,0) |- ($(tot.west)+(0,-2.5)$);
\draw[ar] (rro.east) -| (tot.south);
\end{tikzpicture}}
\caption{Computational DAG for FRTB-SA combining sensitivity-based delta/vega/curvature scenario sweeps, discontinuous switch surfaces ($\varrho_A=0$ and $K_b^+=K_b^-$), default-risk ratio products, and residual add-ons.}\label{fig:dag-sa}
\end{figure}

Proofs for this model are in \cref{app:E:sa}. The nodewise accounting rules of \cref{tab:E:sa} reconcile on all of $\Omega$ (\cref{thm:E:sa}). However, both $\varrho_A=0$ and the curvature $K_b^+=K_b^-$ lexicographic switch create codimension-$1$ capital jumps (\cref{prop:E:sbmjump,prop:E:curvjump}). When the curvature switch dominates $\max_\lambda$ along a proportional ray, neither $\nabla K$ nor $\partial_C K$ exists anywhere on the positive ray (\cref{lem:SA-full-jump,lem:SA-ray-obstruction}; witness~W1, established in \cref{sec:counterexamples} and plotted in \cref{fig:obstruction-plots}(c)). Both Euler and path-integral attribution are therefore \textbf{\Ncell{}}. Shapley attribution is the only paradigm that closes. It is \textbf{\Dscoped{}} on the coalition-admissible domain $\mathcal{C}_{\mathcal{P}}\subsetneq\Omega$, on which every sub-book keeps $\varrho_B\ge 0$ (\cref{lem:SA-shapley}); the sub-book with $\varrho_B=-17/40<0$ of \cref{sec:games} shows that the restriction is needed.

\subsection{Basel Output Floor (RBC20 / RBC90)}\label{sec:model-floor}

Under RBC20/RBC90 \citep{bcbs2019rbc20}, bank-wide Risk-Weighted Assets (RWA) combine internally modelled market/credit/operational components $f(x)=\sum_c R_c^{\rm nom}(x)$ against scaled standardised floor-base RWA $\varphi g(x)=\varphi\sum_c R_c^{\rm fb}(x)$ ($\varphi=0.725$ or transitional $\varphi_t$), with optional $25\%$ phase-in cap $F^{\rm cap}$ (\cref{fig:dag-floor}):
\[
  F(x)=\max\bigl\{f(x),\ \varphi\,g(x)\bigr\},\qquad F^{\rm cap}(x)=\min\bigl\{F(x),\ 1.25\,f(x)\bigr\}\quad(f\ge 0).
\]
External credit/operational RWA pools enter via degree-$1$ bookkeeping accounting leaves (\cref{lem:E:ext}), while shared standardised-market subgraphs must resolve ties jointly across $f$ and $g$.

\begin{figure}[htbp]
\centering
\resizebox{\linewidth}{!}{%
\begin{tikzpicture}[
  x=1mm, y=1mm,
  font={\fontsize{8}{9.4}\selectfont}, text=fbInk,
  >={Stealth[length=1.5mm, width=1.1mm]},
  box/.style={align=flush center, inner xsep=1mm, inner ysep=0.8mm},
  op/.style={box, draw=fbGrey, line width=0.5pt, rounded corners=1.6pt, fill=white},
  leaf/.style={box, draw=fbBlue, line width=0.8pt, fill=white},
  kink/.style={fill=fbRed!14},
  ar/.style={->, draw=fbInk, line width=0.42pt},
  ln/.style={draw=fbInk, line width=0.42pt},
  el/.style={font={\fontsize{8}{9.4}\selectfont}, inner sep=0.8pt},
  dot/.style={circle, fill=fbInk, inner sep=0pt, minimum size=1.1mm},
]
\def\dagnote#1{{\fontsize{7}{8.4}\selectfont #1}}
\def\dxA{17.5} \def\dwA{31}
\def\dxB{60}   \def\dwB{31}
\def\dxC{103}  \def\dwC{32}
\def\dxD{145}  \def\dwD{32}
\def\dyA{0} \def\dyM{-17} \def\dyB{-34}
\node[leaf, text width=\dwA mm, minimum height=11mm] (ima) at (\dxA,\dyA) {internal-model RWA\\$R_c^{\rm nom}$ (deg.\ $1$)};
\node[leaf, kink, text width=\dwA mm, minimum height=12mm] (ext) at (\dxA,\dyM) {external RWA leaves\\credit / operational\\\dagnote{fixed offsets $e,e'>0$}};
\node[leaf, text width=\dwA mm, minimum height=11mm] (std) at (\dxA,\dyB) {standardised floor leaves\\$R_c^{\rm fb}(x)$};
\node[op, text width=\dwB mm, minimum height=13mm] (nom) at (\dxB,-8.5) {nominated RWA\\$f(x)=\sum_c R_c^{\rm nom}(x)$\\\dagnote{shared SA subgraphs}};
\node[op, text width=\dwB mm, minimum height=12mm] (fb) at (\dxB,\dyB) {scaled floor base\\$\varphi g(x)$\\\dagnote{$\varphi=0.725$ (RBC20)}};
\node[op, kink, text width=\dwC mm, minimum height=13mm] (flr) at (\dxC,-20) {output floor\\$F(x)=\max\{f(x),\varphi g(x)\}$\\\dagnote{tie jump $x\odot\nabla h$}};
\node[op, kink, text width=\dwD mm, minimum height=13mm] (cap) at (\dxD,-20) {transitional cap\\$F^{\rm cap}=\min\{F,\,1.25f\}$\\\dagnote{RBC90 cap discretion}};
\def\dpad{2.2}
\draw[fbBlueMid, line width=0.6pt, rounded corners=2.5mm]
  ($(ima.north west)+(-\dpad,\dpad+2.6)$) rectangle ($(ext.south east)+(\dpad,-\dpad)$);
\node[fill=white, inner xsep=1.2mm, inner ysep=0.3mm, anchor=west, text=fbBlue, font={\fontsize{8}{9.4}\selectfont\itshape}]
  at ($(ima.north west)+(1.2,\dpad+2.6)$) {risk-class RWA inputs};
\draw[ar] (ima.east) -- ++(4,0) |- ($(nom.west)+(0,3)$);
\draw[ar] (ext.east) -- ++(4,0) |- ($(nom.west)+(0,-3)$);
\draw[ar] (std) -- (fb);
\draw[ar] (nom.south) |- ($(flr.west)+(0,2.8)$);
\draw[ar] (fb.east) -| (flr.south);
\draw[ar] (flr) -- (cap);
\draw[ar] (nom.east) -| node[el, pos=0.25, above=0.3mm] {$1.25f(x)$} (cap.north);
\end{tikzpicture}}
\caption{Computational DAG for the Basel Output Floor combining degree-$1$ internal leaves, non-homogeneous external credit/operational RWA leaves, the output floor maximum, and optional transitional cap.}\label{fig:dag-floor}
\end{figure}

Proofs for this model are in \cref{app:E:floor}. Convex node weights at the tie $f=\varphi g$ give exact accounting reconciliation (\cref{thm:E:floor}). At smooth crossing points $h=f-\varphi g=0$ with $\nabla h\ne 0$, the marginals jump across the binding surface by the zero-sum shift $x\odot\nabla h(x)$ (\cref{prop:E:bind}). The fixed external credit and operational RWA offsets $e,e'>0$ make Euler completeness fail on an open $(3n+2)$-dimensional affine family, so that cell is \textbf{\Ncell{}} (\cref{lem:FLOOR-constant-gap}; witness~W3, established in \cref{sec:counterexamples} and plotted in \cref{fig:obstruction-plots}(b)). Its two-trade instance is the rule $K_2$ of \cref{ex:offset}. Local Lipschitz continuity along every ray closes path-integral attribution for finite $PC^1$ inputs (\textbf{\Dstruct{}}, \cref{lem:FLOOR-path}). The trade-deletion Shapley game closes on all of $\Omega$ (\textbf{\Dstruct{}}, \cref{lem:FLOOR-trade-shapley}); it is distinct from the two-player branch game of \cref{lem:FLOOR-component-shapley}.

\subsection{ISDA Standard Initial Margin Model (ISDA SIMM v2.9+2606)}\label{sec:model-simm}

Conditioned on version \texttt{v2.9+2606} (effective 5 December 2026) \citep{isda2026simm} and per-trade product-class assignments (\tagU{}), total SIMM sums across product classes $p\in\{\mathrm{RatesFX, Credit, Equity, Commodity}\}$ plus linear notional add-ons (\cref{fig:dag-simm}). Using risk-factor coordinates $s_k$, fixed risk weights $RW_k$, concentration thresholds $T_g>0$, and structural prototype convention $D_{bc}\equiv 1$ (\tagU{}):
\begin{align*}
  CR_g &= \max\Bigl\{1,\sqrt{|\sigma_g|/T_g}\Bigr\},\qquad WS_k = RW_k s_k CR_{g(k)},\qquad f_{kl}=\frac{\min\{CR_{g(k)},CR_{g(l)}\}}{\max\{CR_{g(k)},CR_{g(l)}\}},\\
  K_b &= \sqrt{\sum_{k\in b}WS_k^2+\sum_{k\ne l\in b}\rho_{kl}f_{kl}WS_k WS_l},\qquad \tilde S_b = \cl\Bigl(\sum_{k\in b}WS_k,K_b\Bigr),\\
  \mathrm{Delta} &= \sqrt{\sum_b K_b^2+\sum_{b\ne c}\gamma_{bc}D_{bc}\tilde S_b\tilde S_c} + K_{\rm res},\qquad \theta=\min\Bigl\{\frac{N}{D},0\Bigr\},\qquad \lambda=(z^2-1)(1+\theta)-\theta,\\
  \mathrm{CurvMargin} &= \max\{N+\lambda R_c,0\},\qquad \mathrm{SIMM}_p=\sqrt{\sum_r IM_r^2+\sum_{r\ne s}\psi_{rs}IM_r IM_s},
\end{align*}
with total margin $\mathrm{SIMM}=\sum_p\mathrm{SIMM}_p+\sum AddOnFactor\cdot|a_j|$, equity bucket~12 curvature override \texttt{CVR:=0}, compound zero $\mathrm{CurvMargin}:=0$ at $D=\sum|CVR|=0$, and strict refusal on negative unfloored radicands, $T_g\le 0$, or isolated $\theta,\lambda$ at $D=0$.

\begin{figure}[htbp]
\centering
\resizebox{\linewidth}{!}{%
\begin{tikzpicture}[
  x=1mm, y=1mm,
  font={\fontsize{8}{9.4}\selectfont}, text=fbInk,
  >={Stealth[length=1.5mm, width=1.1mm]},
  box/.style={align=flush center, inner xsep=1mm, inner ysep=0.8mm},
  op/.style={box, draw=fbGrey, line width=0.5pt, rounded corners=1.6pt, fill=white},
  leaf/.style={box, draw=fbBlue, line width=0.8pt, fill=white},
  kink/.style={fill=fbRed!14},
  ar/.style={->, draw=fbInk, line width=0.42pt},
  ln/.style={draw=fbInk, line width=0.42pt},
  el/.style={font={\fontsize{8}{9.4}\selectfont}, inner sep=0.8pt},
  dot/.style={circle, fill=fbInk, inner sep=0pt, minimum size=1.1mm},
]
\def\dagnote#1{{\fontsize{7}{8.4}\selectfont #1}}
\def\dxA{15}   \def\dwA{26}
\def\dxB{50}   \def\dwB{29}
\def\dxC{88}   \def\dwC{30}
\def\dxD{126}  \def\dwD{30}
\def\dxE{160}  \def\dwE{23}
\def\dyA{0} \def\dyB{-18} \def\dyC{-35}
\node[leaf, text width=\dwA mm, minimum height=12mm] (sens) at (\dxA,\dyA) {risk sensitivities\\$s_k,\ RW_k,\ T_g$};
\node[op, kink, text width=\dwB mm, minimum height=12mm] (cr) at (\dxB,\dyA) {concentration factor\\$CR_g=\max\{1,\sqrt{|\sigma_g|/T_g}\}$\\\dagnote{deg-$3/2$ above $T_g$}};
\node[op, text width=\dwC mm, minimum height=12mm] (kb) at (\dxC,\dyA) {bucket root \& clamp\\$K_b$ (with $f_{kl}$), $\tilde S_b$\\\dagnote{refuse negative radicand}};
\node[op, kink, text width=\dwD mm, minimum height=12mm] (dlt) at (\dxD,\dyA) {delta / vega / res.\\$\sqrt{\sum K_b^2+\sum\gamma_{bc}D_{bc}\tilde S_b\tilde S_c}$\\${}+K_{\rm res}$};
\node[leaf, text width=\dwA mm, minimum height=12mm] (cvr) at (\dxA,\dyB) {curvature inputs\\$CVR_k$ (excl.\ Bkt~12)};
\node[op, kink, text width=\dwB mm, minimum height=12mm] (th) at (\dxB,\dyB) {curvature ratio\\$\theta=\min\{N/D,0\}$\\\dagnote{refuse isolated $0/0$}};
\node[op, kink, text width=\dwC mm, minimum height=12mm] (cm) at (\dxC,\dyB) {curvature margin\\$\max\{N+\lambda R_c,0\}$\\\dagnote{$0$ when $D=0$}};
\node[op, kink, text width=\dwD mm, minimum height=13mm] (pcls) at (\dxD,\dyB) {product-class IM\\$\mathrm{SIMM}_p=\sqrt{\sum IM_r^2+\sum\psi_{rs}IM_r IM_s}$};
\node[leaf, kink, text width=\dwA mm, minimum height=10mm] (adl) at (\dxA,\dyC) {add-on notionals $|a_j|$};
\node[op, kink, text width=\dwC mm, minimum height=10mm] (ads) at (\dxC,\dyC) {notional add-ons\\$\sum AddOnFactor\cdot|a_j|$};
\node[op, text width=\dwE mm, minimum height=15mm] (tot) at (\dxE,-17) {total IM\\$\mathrm{SIMM}=$\\$\sum_p\mathrm{SIMM}_p$\\${}+\mathrm{AddOn}$};
\def\dpad{2.2}
\draw[fbBlueMid, line width=0.6pt, rounded corners=2.5mm]
  ($(cr.north west)+(-\dpad,\dpad+2.8)$) rectangle ($(pcls.south east)+(\dpad,-\dpad)$);
\node[fill=white, inner xsep=1.2mm, inner ysep=0.3mm, anchor=west, text=fbBlue, font={\fontsize{8}{9.4}\selectfont\itshape}]
  at ($(cr.north west)+(1.5,\dpad+2.8)$) {evaluated per product class $p\in\{\mathrm{RatesFX,Credit,Equity,Comm}\}$};
\draw[ar] (sens) -- (cr);
\draw[ar] (cr) -- node[el, above] {$WS_k$} (kb);
\draw[ar] (kb) -- (dlt);
\draw[ar] (dlt) -- (pcls);
\draw[ar] (cvr) -- (th);
\draw[ar] (th) -- node[el, above] {$\lambda$} (cm);
\draw[ar] (cm) -- (pcls);
\draw[ar] (adl) -- (ads);
\draw[ar] (pcls.east) -| (tot.north);
\draw[ar] (ads.east) -| (tot.south);
\end{tikzpicture}}
\caption{Computational DAG for ISDA SIMM v2.9+2606 showing degree-$3/2$ concentration multipliers $CR_g$, bucket clamps, curvature scaling, product-class aggregation, and notional add-ons.}\label{fig:dag-simm}
\end{figure}

Proofs for this model are in \cref{app:E:simm}. When $CR_g\equiv 1$, that is $|\sigma_g|<T_g$, Euler completeness holds at regular points (\cref{thm:E:simm}). Above $T_g$, mixing degree-$\tfrac{3}{2}$ concentration factors with degree-$1$ add-ons yields the Euler gap $\tfrac{1}{2}mws^{3/2}/\sqrt{T}>0$ on an open set (\cref{lem:SIMM-counterexample}; witness~W2, established in \cref{sec:counterexamples} and plotted in \cref{fig:obstruction-plots}(a)), refuting Euler attribution as \textbf{\Ncell{}}. Away from negative radicands and isolated $0/0$ curvature ratios, the identity $K=K^0$ near $t=0$ gives uniform Lipschitz bounds near the origin. Borel algebraic-cell Clarke ray densities then close path-integral attribution on the regular cone $\mathcal{R}\subsetneq\Omega$ (\textbf{\Dscoped{}}, \cref{lem:SIMM-ray}). Shapley attribution closes on the radicand-closed subset $\mathcal{Q}\subsetneq\Omega$ (\textbf{\Dscoped{}}, \cref{lem:SIMM-shapley}).

\subsection{Credit Valuation Adjustment Capital (BA-CVA and SA-CVA)}\label{sec:model-cva}

Under MAR50 \citep{bcbs2020mar50}, Basic Approach CVA (BA-CVA, $\alpha=1.4,\ \rho=0.5,\ DS=0.65,\ \beta=0.25$) mixes reduced and hedged $\ell_2$ counterparty norms over external leaves $EAD_{NS}$ and hedge notionals $B_h,B_i$ (\cref{fig:dag-cva}), rewriting misalignment squares $\mathrm{HMA}_c=\sum_h(1-r_{hc}^2)u_{hc}^2$ ($u_{hc}=RW_h M_h B_h DF_h$) as weighted norm terms (\cref{lem:E:hedged}):
\begin{align*}
  \mathrm{SCVA}_c &= \frac{1}{\alpha}RW_c\sum_{NS\in c}M_{NS}EAD_{NS}DF_{NS},\qquad K_{\rm full}=\beta K_{\rm reduced}+(1-\beta)K_{\rm hedged},\\
  K_{\rm reduced} &= \sqrt{\Bigl(\rho\sum_c\mathrm{SCVA}_c\Bigr)^2 + (1-\rho^2)\sum_c\mathrm{SCVA}_c^2},\\
  K_{\rm hedged} &= \sqrt{\Bigl(\rho\sum_c(\mathrm{SCVA}_c-\mathrm{SNH}_c)-\mathrm{IH}\Bigr)^2+(1-\rho^2)\sum_c(\mathrm{SCVA}_c-\mathrm{SNH}_c)^2+\sum_c\mathrm{HMA}_c}.
\end{align*}
In Standardised Approach CVA (SA-CVA, MAR50.40--53, $R=0.01,\ m_{\rm CVA}=1$), net sensitivities $WS_k = WS_k^{\rm CVA}-WS_k^{\rm Hdg}$ enter refusing bucket roots $K_b$, symmetric clamps $S_b=\cl(\sum_{k\in b}WS_k,K_b)$, and refusing cross-bucket roots $K_r$:
\[
  K_b = \sqrt{\sum_{k\in b}WS_k^2 + \sum_{k\ne l\in b}\rho_{kl}WS_k WS_l + R\sum_{k\in b}(WS_k^{\rm Hdg})^2},\qquad
  K_r = m_{\rm CVA}\sqrt{\sum_b K_b^2 + \sum_{b\ne c}\gamma_{bc}S_b S_c}.
\]

\begin{figure}[htbp]
\centering
\resizebox{\linewidth}{!}{%
\begin{tikzpicture}[
  x=1mm, y=1mm,
  font={\fontsize{8}{9.4}\selectfont}, text=fbInk,
  >={Stealth[length=1.5mm, width=1.1mm]},
  box/.style={align=flush center, inner xsep=1mm, inner ysep=0.8mm},
  op/.style={box, draw=fbGrey, line width=0.5pt, rounded corners=1.6pt, fill=white},
  leaf/.style={box, draw=fbBlue, line width=0.8pt, fill=white},
  kink/.style={fill=fbRed!14},
  ar/.style={->, draw=fbInk, line width=0.42pt},
  ln/.style={draw=fbInk, line width=0.42pt},
  el/.style={font={\fontsize{8}{9.4}\selectfont}, inner sep=0.8pt},
  dot/.style={circle, fill=fbInk, inner sep=0pt, minimum size=1.1mm},
]
\def\dagnote#1{{\fontsize{7}{8.4}\selectfont #1}}
\def\dxA{16}   \def\dwA{28}
\def\dxB{53}   \def\dwB{29}
\def\dxC{94}   \def\dwC{31}
\def\dxD{137}  \def\dwD{34}
\def\dyA{0} \def\dyB{-15.5} \def\dyC{-35.5}
\node[leaf, text width=\dwA mm, minimum height=13mm] (bain) at (\dxA,-7.75) {external $EAD_{NS}$\\\& hedges $B_h,B_i$\\\dagnote{BA-CVA inputs}};
\node[op, text width=\dwB mm, minimum height=11mm] (scva) at (\dxB,\dyA) {counterparty SCVA\\$\mathrm{SCVA}_c$};
\node[op, text width=\dwB mm, minimum height=11mm] (hdg) at (\dxB,\dyB) {hedge terms\\$\mathrm{SNH}_c,\mathrm{IH},\mathrm{HMA}_c$};
\node[op, text width=\dwC mm, minimum height=11mm] (kred) at (\dxC,\dyA) {reduced root\\$K_{\rm reduced}$ ($\ell_2$ norm)};
\node[op, text width=\dwC mm, minimum height=11mm] (khed) at (\dxC,\dyB) {hedged root\\$K_{\rm hedged}$ ($\ell_2$ norm)};
\node[op, text width=\dwD mm, minimum height=13mm] (kfull) at (\dxD,-7.75) {BA-CVA capital\\$K_{\rm full}=\beta K_{\rm reduced}$\\${}+(1-\beta)K_{\rm hedged}$};
\node[leaf, text width=\dwA mm, minimum height=12mm] (sain) at (\dxA,\dyC) {CVA \& hedge sens.\\$WS_k^{\rm CVA},WS_k^{\rm Hdg}$};
\node[op, text width=\dwB mm, minimum height=12mm] (wsnet) at (\dxB,\dyC) {net weighted sens.\\$WS_k=WS_k^{\rm CVA}-WS_k^{\rm Hdg}$};
\node[op, kink, text width=\dwC mm, minimum height=12mm] (sabk) at (\dxC,\dyC) {bucket root \& clamp\\$K_b$ (ridge $R$),\ $S_b$\\\dagnote{clamp $\cl(\sum WS_k,K_b)$}};
\node[op, kink, text width=\dwD mm, minimum height=12mm] (sakr) at (\dxD,\dyC) {SA-CVA risk-class root\\$K_r=m_{\rm CVA}\sqrt{q_r}$\\\dagnote{scoped to $\Omega_+$ ($P_b\succeq 0,\Gamma_r\succ 0$)}};
\def\dpad{2.2}
\draw[fbBlueMid, line width=0.6pt, rounded corners=2.5mm]
  ($(scva.north west)+(-\dpad,\dpad+2.6)$) rectangle ($(kfull.south east|-khed.south)+(\dpad,-\dpad)$);
\node[fill=white, inner xsep=1.2mm, inner ysep=0.3mm, anchor=west, text=fbBlue, font={\fontsize{8}{9.4}\selectfont\itshape}]
  at ($(scva.north west)+(1.5,\dpad+2.6)$) {BA-CVA (MAR50.7--25)};
\draw[fbBlueMid, line width=0.6pt, rounded corners=2.5mm]
  ($(wsnet.north west)+(-\dpad,\dpad+2.6)$) rectangle ($(sakr.south east)+(\dpad,-\dpad)$);
\node[fill=white, inner xsep=1.2mm, inner ysep=0.3mm, anchor=west, text=fbBlue, font={\fontsize{8}{9.4}\selectfont\itshape}]
  at ($(wsnet.north west)+(1.5,\dpad+2.6)$) {SA-CVA (MAR50.40--53)};
\draw[ar] (bain.east) -- ++(3.5,0) |- (scva.west);
\draw[ar] (bain.east) -- ++(3.5,0) |- (hdg.west);
\draw[ar] (scva) -- (kred);
\draw[ar] (scva.east) -- ++(3,0) |- ($(khed.west)+(0,2.5)$);
\draw[ar] (hdg) -- (khed);
\draw[ar] (kred.east) -- ++(3.5,0) |- ($(kfull.west)+(0,3)$);
\draw[ar] (khed.east) -- ++(3.5,0) |- ($(kfull.west)+(0,-3)$);
\draw[ar] (sain) -- (wsnet);
\draw[ar] (wsnet) -- (sabk);
\draw[ar] (sabk) -- (sakr);
\end{tikzpicture}}
\caption{Computational DAGs for Basic Approach CVA (top) and Standardised Approach CVA (bottom) under MAR50, illustrating ridge-stabilised bucket roots, symmetric clamps, and the positivity cone $\Omega_+$.}\label{fig:dag-cva}
\end{figure}

Proofs for this model, including the indefinite boundary witness below, are in \cref{app:E:cva}. We work in the factor coordinates $y=(EAD,B)$ or $(WS^{\rm CVA},WS^{\rm Hdg})$ on the strict positivity cone $\Omega_+=\{P_b\succeq 0,\ \Gamma_r\succ 0,\ R>0\}$. There the coercivity bound $q(v)\ge\lambda_{\min}(\Gamma_r)\|K\|_2^2$ gives global Lipschitz continuity and conic homogeneity (\cref{prop:E:sacva-lip}), and all three cells close as \textbf{\Dscoped{}} (\cref{lem:CVA-euler-point,lem:CVA-ray-density,lem:CVA-shapley}). Two boundaries make the restriction strict (\cref{lem:CVA-boundary,lem:CVA-coordinate-limit}). First, closure in input factors inherits any upstream non-homogeneity when $EAD(x)$ comes from SA-CCR. Second, outside $\Omega_+$, with single SA-CVA bucket ($\Gamma=(1),m=1,R=1/100$, zero hedges) and structural matrices
\[
  P_3=\begin{pmatrix}1 & 9/10 & 9/10\\ 9/10 & 1 & 1/10\\ 9/10 & 1/10 & 1\end{pmatrix},\qquad
  P_4=\begin{pmatrix}1 & 9/10 & 9/10 & 1/20\\ 9/10 & 1 & 1/10 & 1/20\\ 9/10 & 1/10 & 1 & 1/20\\ 1/20 & 1/20 & 1/20 & 1\end{pmatrix},
\]
the point $p=(-2,1,1,2)$ gives $p^\top P_4 p = 3 > 0$, with capital $\sqrt{3}$ on an open set in $\Omega$. Dropping coordinate~$4$, however, yields $p_{\{1,2,3\}}^\top P_3 p_{\{1,2,3\}} = -1 < 0$, so coalition closure fails outside $\Omega_+$. At $p_*=(-(9+\sqrt{26})/5,1,1)$, moreover, $q(u)=u^2-\tfrac{18}{5}u+\tfrac{11}{5}=0$ and $K(p_*+\varepsilon e_1)\sim c\sqrt{\varepsilon}$. Lipschitz continuity and Clarke gradients are then lost at every positive ray point $tp_*$, $t>0$.

\subsection{Standardised Approach for Counterparty Credit Risk (SA-CCR)}\label{sec:model-saccr}\label{sec:compose-work}

Under CRE52 \citep{bcbs2020cre52}, the inputs are oriented ticket quantities $x\in\R_{\ge 0}^n$ with unit market values $v_i$ (\cref{fig:dag-saccr}). Each ticket enters through its effective notional $D_i=\bar d_i MF_i\delta_i x_i$, and the effective notionals feed asset-class add-ons that sum to $A=\mathrm{AddOn}_{\rm agg}$. Secondary symbols are specified in \cref{app:E:sacc}; for reference, $d_i=\bar d_i x_i$ is the adjusted notional, $MF_i$ a fixed maturity factor and $\delta_i$ the supervisory delta. The supervisory duration is $SD_i=\max\{\tau_0,(e^{-0.05 S_i}-e^{-0.05 E_i})/0.05\}$, an option has $\delta_i=\eta w\Phi(wz)$ and a CDO tranche has $\delta=\eta_p 15/((1+14a)(1+14b))$. The add-ons are the positive-definite interest-rate root $EN_h=\sqrt{B_{h1}^2+B_{h2}^2+B_{h3}^2+1.4B_{h1}B_{h2}+1.4B_{h2}B_{h3}+0.6B_{h1}B_{h3}}$, the FX term $0.04|\sum D_i|$ and the seminorm factor roots $G(q)=[(\sum_e\rho_e q_e)^2+\sum_e(1-\rho_e^2)q_e^2]^{1/2}$. With $V(x)=\sum_i v_i x_i$, $u=V-C$, $f=0.05$, $k=2(1-f)$, $h=TH+MTA-NICA$, and $\alpha=1.4$:
\begin{align*}
  R_u &= \max\{V-C,0\},\qquad R_m = \max\{V-C,h,0\},\\
  m(u,A) &= \min\bigl\{1,\ f+(1-f)\exp(u/(kA))\bigr\}\ \ (A>0),\qquad P(u,A)=A m(u,A)\quad(P(u,0):=0),\\
  E &= \min\bigl\{\alpha(R_m+P_m),\ \alpha(R_u+P_u)\bigr\}.
\end{align*}

\begin{figure}[htbp]
\centering
\resizebox{\linewidth}{!}{%
\begin{tikzpicture}[
  x=1mm, y=1mm,
  font={\fontsize{8}{9.4}\selectfont}, text=fbInk,
  >={Stealth[length=1.5mm, width=1.1mm]},
  box/.style={align=flush center, inner xsep=1mm, inner ysep=0.8mm},
  op/.style={box, draw=fbGrey, line width=0.5pt, rounded corners=1.6pt, fill=white},
  leaf/.style={box, draw=fbBlue, line width=0.8pt, fill=white},
  kink/.style={fill=fbRed!14},
  ar/.style={->, draw=fbInk, line width=0.42pt},
  ln/.style={draw=fbInk, line width=0.42pt},
  el/.style={font={\fontsize{8}{9.4}\selectfont}, inner sep=0.8pt},
  dot/.style={circle, fill=fbInk, inner sep=0pt, minimum size=1.1mm},
]
\def\dagnote#1{{\fontsize{7}{8.4}\selectfont #1}}
\def\dxA{16.5} \def\dwA{29}
\def\dxB{54.5} \def\dwB{30}
\def\dxC{94.5} \def\dwC{32}
\def\dxD{137}  \def\dwD{33}
\def\dyA{0} \def\dyM{-13} \def\dyB{-27}
\node[leaf, kink, text width=\dwA mm, minimum height=13mm] (mv) at (\dxA,\dyA) {positions \& collateral\\$x\in\R^n_{\ge 0},\ v_i,\ C_u,C_m,h$\\\dagnote{fixed $C,h$ break Euler}};
\node[leaf, text width=\dwA mm, minimum height=13mm] (dn) at (\dxA,\dyB) {notional metadata\\$\bar d_i,\ SD_i,\ MF_i,\ \delta_i$\\\dagnote{effective notional $D_i$}};
\node[op, kink, text width=\dwB mm, minimum height=13mm] (rc) at (\dxB,\dyA) {replacement cost\\$R_u=(V-C)_+$\\$R_m=\max\{V-C,h,0\}$};
\node[op, kink, text width=\dwB mm, minimum height=13mm] (addon) at (\dxB,\dyB) {hedging-set add-ons\\$EN_h,\ 0.04|\sum D_i|,\ G(q)$\\\dagnote{sum to $A=\mathrm{AddOn}_{\rm agg}$}};
\node[op, kink, text width=\dwC mm, minimum height=14mm] (pfe) at (\dxC,\dyM) {fused PFE multiplier\\$P(u,A)=A m(u,A)$\\\dagnote{$P(u,0):=0$; refuse $m$ at $0/0$}};
\node[op, kink, text width=\dwD mm, minimum height=14mm] (ead) at (\dxD,\dyM) {exposure at default\\$E=\min\{\alpha(R_m+P_m),$\\$\alpha(R_u+P_u)\}$};
\def\dpad{2.2}
\draw[fbBlueMid, line width=0.6pt, rounded corners=2.5mm]
  ($(addon.north west)+(-\dpad,\dpad+2.6)$) rectangle ($(addon.south east)+(\dpad,-\dpad)$);
\node[fill=white, inner xsep=1.2mm, inner ysep=0.3mm, anchor=west, text=fbBlue, font={\fontsize{8}{9.4}\selectfont\itshape}]
  at ($(addon.north west)+(1.2,\dpad+2.6)$) {across asset classes};
\draw[ar] (mv) -- node[el, above] {$V$} (rc);
\draw[ar] (dn) -- node[el, above] {$D_i$} (addon);
\coordinate (ubr) at ($(mv.east)!0.48!(rc.west)$);
\draw[ar] (ubr) |- node[el, pos=0.72, above=0.3mm] {$u=V-C$} ($(pfe.west)+(0,3)$);
\node[dot] at (ubr) {};
\draw[ar] (addon.east) -- ++(4,0) |- node[el, pos=0.72, below=0.3mm] {$A$} ($(pfe.west)+(0,-3)$);
\draw[ar] (rc.east) -| ($(ead.north)+(0,0)$);
\draw[ar] (pfe) -- (ead);
\end{tikzpicture}}
\caption{Computational DAG for SA-CCR (CRE52) combining fixed-collateral replacement cost, asset-class hedging-set add-ons, the continuously extended fused multiplier $P(u,A)$, and the margined/unmargined minimum.}\label{fig:dag-saccr}
\end{figure}

Proofs for this model, including the composition with BA-CVA, are in \cref{app:E:sacc}. Setting $P(u,0):=0$ gives the unique continuous extension of the multiplier to $A\ge 0$. It is globally Lipschitz, with $P_u=\tfrac{1}{2}e^r$ and $P_A=f+(1-f)e^r(1-r)$ on $u<0$, where $r=u/(kA)$; the isolated multiplier $m(u,A)$ is refused at $(0,0)$. Fixed nonzero collateral $C$ gives the exact Euler gap $\alpha C P_u>0$ on smooth portfolios (\cref{ex:E:sacc-numeric}; witness~W4, established in \cref{sec:counterexamples}), refuting Euler attribution as \textbf{\Ncell{}}. Condition \texttt{(NT)} requires each ray $\gamma(t)=tx$ to spend zero measure on unresolved ties. Under it, path-integral attribution closes on $\Omega'\subsetneq\Omega$ (\textbf{\Dscoped{}}, \cref{lem:saccr-no-tie-corollary}), whereas Shapley attribution closes on all of $\R_{\ge 0}^n$ (\textbf{\Dstruct{}}, \cref{lem:saccr-shapley}).

The exposures $EAD_s(x)$ produced here feed BA-CVA (\cref{thm:E:sacc-cva}). Propagating full values at the endpoint reconciles to $F(y(x))$, but it does not produce the end-to-end path integral. That integral requires propagating the BA-CVA norm weights $w_s(t)=\partial_{E_s}F(y(\gamma(t)))$ pointwise along the common trade path before integrating.

\subsection{Bilateral Valuation Adjustments (Double-Hinge XVA)}\label{sec:model-xva}

For signed trade positions $x\in\R^n$, counterparties $\mathcal{C}$, scenarios $s=1,\dots,S$, and dates $0=t_0<\dots<t_T$, write $r=(c,s,t)$ with probability-discount-survival weights $w_r=p_s L_c D_t(Q_{c,t-1}-Q_{ct})\ge 0$, counterparty-$c$ unit values $a_r$, and variation/initial margin collateral parameters $0\le b_r\le H_r<\infty$, $0\le\eta_r\le 1$ (\cref{fig:dag-xva}):
\[
  V_r(x)=a_r^\top x,\qquad C_r(x)=b_r+\eta_r(V_r(x)-H_r)_+,\qquad
  E_r(x)=(V_r(x)-C_r(x))_+=(a_r^\top x-b_r)_+-\eta_r(a_r^\top x-H_r)_+,
\]
with bilateral exposure ledger $K(x)=\sum_r w_r E_r(x)$ and $K(0)=0$.

\begin{figure}[htbp]
\centering
\resizebox{\linewidth}{!}{%
\begin{tikzpicture}[
  x=1mm, y=1mm,
  font={\fontsize{8}{9.4}\selectfont}, text=fbInk,
  >={Stealth[length=1.5mm, width=1.1mm]},
  box/.style={align=flush center, inner xsep=1mm, inner ysep=0.8mm},
  op/.style={box, draw=fbGrey, line width=0.5pt, rounded corners=1.6pt, fill=white},
  leaf/.style={box, draw=fbBlue, line width=0.8pt, fill=white},
  kink/.style={fill=fbRed!14},
  ar/.style={->, draw=fbInk, line width=0.42pt},
  ln/.style={draw=fbInk, line width=0.42pt},
  el/.style={font={\fontsize{8}{9.4}\selectfont}, inner sep=0.8pt},
  dot/.style={circle, fill=fbInk, inner sep=0pt, minimum size=1.1mm},
]
\def\dagnote#1{{\fontsize{7}{8.4}\selectfont #1}}
\def\dxA{16.5} \def\dwA{29}
\def\dxB{54}   \def\dwB{26}
\def\dxC{93}   \def\dwC{32}
\def\dxD{135}  \def\dwD{32}
\def\dyA{0} \def\dyB{-20}
\node[leaf, text width=\dwA mm, minimum height=12mm] (pos) at (\dxA,\dyA) {positions \& values\\$x\in\R^n,\ a_{ist}$\\\dagnote{signed linear exposures}};
\node[leaf, kink, text width=\dwA mm, minimum height=12mm] (thr) at (\dxA,\dyB) {fixed collateral \& weights\\$(b_r,H_r,\eta_r),\ w_r$\\\dagnote{$w_r=p_s L_c D_t\Delta Q_{ct}$}};
\node[op, text width=\dwB mm, minimum height=12mm] (vr) at (\dxB,\dyA) {netting value\\$V_r(x)=a_r^\top x$};
\node[op, kink, text width=\dwC mm, minimum height=12mm] (col) at (\dxC,\dyB) {VM/IM collateral\\$C_r=b_r+\eta_r(V_r-H_r)_+$};
\node[op, kink, text width=\dwC mm, minimum height=12mm] (exp) at (\dxC,\dyA) {double-hinge exposure\\$E_r=(V_r-b_r)_+$\\${}-\eta_r(V_r-H_r)_+$};
\node[op, text width=\dwD mm, minimum height=14mm] (k) at (\dxD,-10) {bilateral XVA\\$K(x)=\sum_r w_r E_r(x)$\\\dagnote{globally Lipschitz PA}};
\def\dpad{2.2}
\draw[fbBlueMid, line width=0.6pt, rounded corners=2.5mm]
  ($(vr.north west)+(-\dpad,\dpad+2.6)$) rectangle ($(exp.east|-col.south)+(\dpad,-\dpad)$);
\node[fill=white, inner xsep=1.2mm, inner ysep=0.3mm, anchor=west, text=fbBlue, font={\fontsize{8}{9.4}\selectfont\itshape}]
  at ($(vr.north west)+(1.2,\dpad+2.6)$) {evaluated per $r=(c,s,t)$};
\draw[ar] (pos) -- (vr);
\draw[ar] (vr) -- (exp);
\draw[ar] (vr.east) -- ++(3.5,0) |- (col.west);
\draw[ar] (thr) -- (col);
\draw[ar] (col) -- (exp);
\draw[ar] (exp.east) -| (k.north);
\draw[ar] (thr.south) -- ++(0,-4.5) -| node[el, pos=0.25, below=0.3mm] {$w_r$} (k.south);
\end{tikzpicture}}
\caption{Computational DAG for Bilateral Double-Hinge XVA showing fixed variation- and initial-margin thresholds $(b_r,H_r,\eta_r)$ entering each $(c,s,t)$ double-hinge exposure cell.}\label{fig:dag-xva}
\end{figure}

Proofs for this model are in \cref{app:E:xva}. The ledger $K$ is continuous, piecewise affine and globally Lipschitz on $\R^n$. On a sign chamber $P$ with binary indicators $\alpha_r,\beta_r\in\{0,1\}$ and slope $G_P=\sum_r w_r(\alpha_r-\eta_r\beta_r)a_r$, the exact Euler gap is $x^\top G_P - K(x) = \sum_r w_r(b_r\alpha_r - \eta_r H_r\beta_r)$ (\cref{lem:XVA-euler}). It is nonzero across an open $11$-parameter family, so Euler attribution is \textbf{\Ncell{}} (\cref{lem:XVA-family}; witness~W5, established in \cref{sec:counterexamples}). Ordering the sign chambers lexicographically produces a piecewise-constant Clarke ray density on $[0,1]$, which closes path-integral attribution everywhere (\textbf{\Dstruct{}}, \cref{lem:XVA-path}). Finite position masking closes Shapley attribution on the whole domain (\textbf{\Dstruct{}}, \cref{lem:XVA-shapley}).

\subsection{Classic CME SPAN Clearing Margin}\label{sec:model-span}

We use the unrounded, continuous net-position CME SPAN model under a fixed dated parameter snapshot $\pi$ \citep{cme2019span} (\cref{fig:dag-span}). It maps signed contract vectors $x\in\R^n$ to net positions $z=Nx$, period deltas $d_p=\sum_{j\in p}\delta_j z_j$ and fixed $16$-row array losses $L_{bk}=\sum_{j\in b}\ell_{kj}z_j=B_{bk}^\top x$, which feed ordered priority matchers initialised from $(d_p)_+,(-d_p)_+$ and $(D_b)_+,(-D_b)_+$. The secondary symbols in the display below ($J_e$, $r_{ea}$, $u_a^e$, $h_e$, $\rho_{eb}$, $W_b$, $s_p^{\rm spr}$, $s_p^{\rm out}$, $o_p$, $\mu_b$, $Q_b^c$, $Q_b^p$ and $f_b$) are snapshot parameters, matcher states or derived position sums, specified in \cref{app:E:ccp}:
\begin{align*}
  S_b &= \max_{1\le k\le m} L_{bk}(x),\qquad q_e = \min_{a\in J_e}\frac{u_a^{e-1}}{r_{ea}},\qquad u_a^e=u_a^{e-1}-r_{ea}q_e\quad(a\in J_e),\\
  I_b &= \sum_{e\in\mathcal{I}_b} h_e q_e,\qquad C_b = \sum_{e:b\in J_e}\rho_{eb}W_b r_{eb} q_e,\qquad T_b=\sum_{p\in b\text{ deliv.}}\bigl(s_p^{\rm spr}(|d_p|-o_p)+s_p^{\rm out}o_p\bigr),\\
  H_b &= \mu_b(Q_b^c+Q_b^p)\ \text{or}\ \mu_b\max\{Q_b^c,Q_b^p\},\qquad K_\pi = \sum_b f_b\bigl(\max\{S_b+I_b+T_b-C_b,H_b\}-V_b\bigr),
\end{align*}
with $V_b=\sum_{j\in b\text{ opt}} v_j^0 z_j$ subtracted strictly outside the SOM floor comparison (\cref{prop:E:ccp-order}).

\begin{figure}[htbp]
\centering
\resizebox{\linewidth}{!}{%
\begin{tikzpicture}[
  x=1mm, y=1mm,
  font={\fontsize{8}{9.4}\selectfont}, text=fbInk,
  >={Stealth[length=1.5mm, width=1.1mm]},
  box/.style={align=flush center, inner xsep=1mm, inner ysep=0.8mm},
  op/.style={box, draw=fbGrey, line width=0.5pt, rounded corners=1.6pt, fill=white},
  leaf/.style={box, draw=fbBlue, line width=0.8pt, fill=white},
  kink/.style={fill=fbRed!14},
  ar/.style={->, draw=fbInk, line width=0.42pt},
  ln/.style={draw=fbInk, line width=0.42pt},
  el/.style={font={\fontsize{8}{9.4}\selectfont}, inner sep=0.8pt},
  dot/.style={circle, fill=fbInk, inner sep=0pt, minimum size=1.1mm},
]
\def\dagnote#1{{\fontsize{7}{8.4}\selectfont #1}}
\def\dxA{16}   \def\dwA{27}
\def\dxB{52}   \def\dwB{29}
\def\dxC{89}   \def\dwC{29}
\def\dxD{126}  \def\dwD{29}
\def\dxE{161}  \def\dwE{24}
\def\dyA{0} \def\dyB{-18} \def\dyC{-33}
\node[leaf, text width=\dwA mm, minimum height=12mm] (pos) at (\dxA,\dyA) {net contracts \& arrays\\$z=Nx,\ L_{bk}$};
\node[leaf, text width=\dwA mm, minimum height=12mm] (del) at (\dxA,\dyB) {deltas, SOM, NOV\\$d_p,\ \mu_b,\ v_j^0$\\\dagnote{fixed snapshot $\pi$}};
\node[op, kink, text width=\dwB mm, minimum height=12mm] (scan) at (\dxB,\dyA) {array scan risk\\$S_b=\max_k L_{bk}(x)$};
\node[op, kink, text width=\dwB mm, minimum height=12mm] (mat) at (\dxB,\dyB) {priority matchers\\$q_e=\min_{a\in J_e}\frac{u_a^{e-1}}{r_{ea}}$\\\dagnote{bottleneck ties}};
\node[op, kink, text width=\dwC mm, minimum height=12mm] (chg) at (\dxC,-9) {spreads \& delivery\\$I_b+T_b-C_b$};
\node[op, kink, text width=\dwC mm, minimum height=10mm] (som) at (\dxC,\dyC) {short option min.\\$H_b$};
\node[op, kink, text width=\dwD mm, minimum height=13mm] (flr) at (\dxD,-18) {commodity floor\\$\max\{S_b+I_b+T_b$\\${}-C_b,\ H_b\}$};
\node[op, text width=\dwE mm, minimum height=14mm] (nov) at (\dxE,-18) {SPAN margin\\$K_\pi=\sum_b f_b($\\$\max\{\dots\}-V_b)$};
\def\dpad{2.2}
\draw[fbBlueMid, line width=0.6pt, rounded corners=2.5mm]
  ($(scan.north west)+(-\dpad,\dpad+2.6)$) rectangle ($(flr.south east|-som.south)+(\dpad,-\dpad)$);
\node[fill=white, inner xsep=1.2mm, inner ysep=0.3mm, anchor=west, text=fbBlue, font={\fontsize{8}{9.4}\selectfont\itshape}]
  at ($(scan.north west)+(1.2,\dpad+2.6)$) {evaluated per combined commodity $b$};
\draw[ar] (pos) -- (scan);
\draw[ar] (del) -- (mat);
\draw[ar] (scan.east) -- ++(3.5,0) |- ($(flr.west)+(0,3)$);
\draw[ar] (mat) -- (chg);
\draw[ar] (chg) -- (flr);
\draw[ar] (del.south) -- ++(0,-4.5) |- (som.west);
\draw[ar] (som.east) -| (flr.south);
\draw[ar] (flr) -- (nov);
\draw[ar] (del.south) -- ++(0,-19.5) -| node[el, pos=0.28, below=0.5mm] {$V_b$ subtracted outside $\max$} (nov.south);
\end{tikzpicture}}
\caption{Computational DAG for unrounded fixed-snapshot Classic CME SPAN showing array scan maximums, ordered calendar/intercommodity priority matchers, short-option minima, and post-floor NOV subtraction.}\label{fig:dag-span}
\end{figure}

Proofs for this model, including the worked $16$-row bottleneck tie below, are in \cref{app:E:ccp}. Unrolling the finite matchers partitions $\R^n$ into finitely many closed polyhedral cones (\cref{lem:SPAN-core}). All three paradigms therefore close on all of $\R^n$ as \textbf{\Dstruct{}} (\cref{lem:SPAN-clarke,lem:SPAN-path,lem:SPAN-shapley}). A worked tie shows what the set-valued marginal means in practice. We evaluate the published $16$-row S1 array (\cref{tab:main:span16} in \cref{app:E:ccp}) on a synthetic three-contract futures book $(A_{\rm near},A_{\rm far},B)$ at $x=(2,-1,-2)$, with rows $B_k=(-180t_k,-180t_k,-100t_k)$ over
\[
  (t_1,\dots,t_{16})=(0,0,\tfrac{1}{3},\tfrac{1}{3},-\tfrac{1}{3},-\tfrac{1}{3},\tfrac{2}{3},\tfrac{2}{3},-\tfrac{2}{3},-\tfrac{2}{3},1,1,-1,-1,\tfrac{99}{100},-\tfrac{99}{100})
\]
The book gives $S_A=180$, $S_B=200$, $I_A=70$ and $T_A=30$, an intercommodity fill $q=\min\{x_1+x_2,-x_3/2\}=1$, spread credits $C_A=108$ and $C_B=120$, and margin $K_\pi=172+80=252$. The intercommodity bottleneck tie generates two realisable extreme vertices, $A^{(A)}=(-56,108,200)$ and $A^{(B)}=(400,-120,-28)$, which differ by the zero-sum shift $(456,-228,-228)$. The directional derivatives $-K_\pi'(x;-x_i e_i)=(-56,-120,-28)$ show that no single vertex identifies the risk-reducing trades (\cref{prop:E:ccp-remove}).

\subsection{Solvency~II Standard Formula}\label{sec:model-solv2}

Under Delegated Regulation (EU) 2015/35 and Directive 2009/138/EC \citep{eu2015delegated,eu2009solvency2}, the inputs are positive business-cell replication weights $x\in\R_+^n$ (\cref{fig:dag-solv2}). They feed the market, default, life, health and non-life risk modules $m$, whose aggregate $B$ is then adjusted for operational risk $O$ and by the loss-absorbing capacity adjustments $P$ and $T$. Secondary symbols are specified in \cref{app:E:solvency2} (with the module list in \cref{tab:main:s2}); for reference, the market module pairs gross and net scenarios $R_M(A)$ with $A\in\{0,0.5\}$, and the market-state equity adjustment is $a=\max\{-0.10,\min\{0.10,\tfrac{1}{2}((CI-AI)/AI-0.08)\}\}$:
\begin{align*}
  B &= \sqrt{m^\top R_B m} + 0.8 V_{\rm intangible},\qquad P = \max\{0,\min\{B-N,\mathrm{FDB}\}\},\\
  O &= \min\{0.3B,\max\{Q_p,Q_v\}\}+0.25 E_{\rm ul},\qquad K = B + O - P - \max\{-\Delta_{\rm DT}(L,x),0\}.
\end{align*}

\begin{figure}[htbp]
\centering
\resizebox{\linewidth}{!}{%
\begin{tikzpicture}[
  x=1mm, y=1mm,
  font={\fontsize{8}{9.4}\selectfont}, text=fbInk,
  >={Stealth[length=1.5mm, width=1.1mm]},
  box/.style={align=flush center, inner xsep=1mm, inner ysep=0.8mm},
  op/.style={box, draw=fbGrey, line width=0.5pt, rounded corners=1.6pt, fill=white},
  leaf/.style={box, draw=fbBlue, line width=0.8pt, fill=white},
  kink/.style={fill=fbRed!14},
  ar/.style={->, draw=fbInk, line width=0.42pt},
  ln/.style={draw=fbInk, line width=0.42pt},
  el/.style={font={\fontsize{8}{9.4}\selectfont}, inner sep=0.8pt},
  dot/.style={circle, fill=fbInk, inner sep=0pt, minimum size=1.1mm},
]
\def\dagnote#1{{\fontsize{7}{8.4}\selectfont #1}}
\def\dxA{16.5} \def\dwA{29}
\def\dxB{54.5} \def\dwB{30}
\def\dxC{94}   \def\dwC{32}
\def\dxD{137}  \def\dwD{34}
\def\dyA{0} \def\dyM{-12} \def\dyB{-25}
\node[leaf, text width=\dwA mm, minimum height=13mm] (rep) at (\dxA,\dyA) {replication weights\\$x\in\R_+^n$,\ adj.\ $a$\\\dagnote{gross/net valuations}};
\node[leaf, kink, text width=\dwA mm, minimum height=13mm] (pool) at (\dxA,\dyB) {fixed accounting pools\\$\mathrm{FDB},\ E_{\rm ul},\ \Delta_{\rm DT},\ V_{\rm intang}$\\\dagnote{break Euler when active}};
\node[op, kink, text width=\dwB mm, minimum height=13mm] (mods) at (\dxB,\dyA) {module PSD/PD roots\\$R_M(A)$,\ $R_{\rm NL},R_{\rm life},R_H$\\\dagnote{interface switches}};
\node[op, text width=\dwC mm, minimum height=13mm] (bscr) at (\dxC,\dyA) {gross \& net BSCR\\$B=\sqrt{m^\top R_B m}$\\${}+0.8V_{\rm intang},\ N$};
\node[op, kink, text width=\dwC mm, minimum height=14mm] (adj) at (\dxC,\dyB) {op.\ \& loss absorbers\\$O$ ($\min\{0.3B,\cdot\}$),\ $T$\\$P=\max\{0,\min\{B-N,\mathrm{FDB}\}\}$};
\node[op, kink, text width=\dwD mm, minimum height=14mm] (scr) at (\dxD,\dyM) {solo SCR\\$K=B+O-P-T$\\\dagnote{Reg.\ Arts.\ 87--207}};
\def\dpad{2.2}
\draw[fbBlueMid, line width=0.6pt, rounded corners=2.5mm]
  ($(mods.north west)+(-\dpad,\dpad+2.6)$) rectangle ($(bscr.south east)+(\dpad,-\dpad)$);
\node[fill=white, inner xsep=1.2mm, inner ysep=0.3mm, anchor=west, text=fbBlue, font={\fontsize{8}{9.4}\selectfont\itshape}]
  at ($(mods.north west)+(1.2,\dpad+2.6)$) {gross and net risk modules};
\draw[ar] (rep) -- (mods);
\draw[ar] (mods) -- (bscr);
\draw[ar] (pool) -- (adj);
\draw[ar] (bscr) -- (adj);
\draw[ar] (bscr.east) -- ++(3.5,0) |- ($(scr.west)+(0,3)$);
\draw[ar] (adj.east) -- ++(3.5,0) |- ($(scr.west)+(0,-3)$);
\end{tikzpicture}}
\caption{Computational DAG for solo Solvency~II Standard Formula aggregation combining PSD/PD underwriting and market module roots with fixed FDB, operational-expense, and deferred-tax adjustment caps.}\label{fig:dag-solv2}
\end{figure}

Proofs for this model are in \cref{app:E:solvency2}. The prescribed underwriting and market correlation matrices $R_B,R_{\rm NL},R_{\rm life},R_H,R_{\rm SLT},R_M(0),R_M(0.5)$ are positive definite, and the equity and catastrophe forms are positive semidefinite (\cref{lem:E:s2-psd}). A fixed FDB pool $f>0$ nevertheless produces the smooth Euler gap $(1-\tau)f>0$ across an open $(6n+3)$-parameter family (\cref{lem:SOLV2-fixed-fdb}; witness~W6, established in \cref{sec:counterexamples}), so Euler attribution is \textbf{\Ncell{}}. Path-integral attribution closes on the cone $\Omega'\subsetneq\Omega$ of rays that are compatible with the conditional interfaces and avoid zero roots and persistent interface ties (\textbf{\Dscoped{}}, \cref{lem:SOLV2-ig}). Finite business-cell masking closes Shapley attribution on the whole domain (\textbf{\Dstruct{}}, \cref{lem:SOLV2-shapley}).

\subsection{Interest-Rate Risk in the Banking Book (IRRBB $\Delta\mathrm{EVE}$)}\label{sec:model-irrbb}

Under SRP31 \citep{bcbs2024srp31}, fixed contract positions $w\in\R^n$ are valued over repricing buckets $t_k\ge 0$ under six interest-rate shock scenarios (\cref{fig:dag-irrbb}). Each currency $c$ has shock sizes $P_c,S_c,L_c$, rate floors $h_c(t_k)$ and a baseline curve $q_{0,c,k}+m_{c,k}$, and the shocks decay with $a(t)=e^{-t/4}$; the remaining secondary symbols are specified in \cref{app:E:irrbb}:
\begin{align*}
  s_{1,c}(t)&=P_c,\ \ s_{2,c}(t)=-P_c,\ \ s_{3,c}(t)=-0.65S_c a(t)+0.9L_c(1-a(t)),\\
  s_{4,c}(t)&=0.8S_c a(t)-0.6L_c(1-a(t)),\ \ s_{5,c}(t)=S_c a(t),\ \ s_{6,c}(t)=-S_c a(t),\\
  R_{i,c,k}&=\max\{q_{0,c,k}+s_{i,c}(t_k),h_c(t_k)\}+m_{c,k},\qquad d_{i,c,k}=e^{-R_{i,c,k}t_k},\\
  \ell_{i,c}(w)&=\sum_k\bigl(C_{0,c,k}(w)d_{0,c,k}-C_{i,c,k}(w)d_{i,c,k}\bigr)+Q_{i,c}(w),\qquad K(w)=\max_{1\le i\le 6}\sum_c f_c(\ell_{i,c}(w))_+,
\end{align*}
The cash flows $C_{i,c,k}(w)$ incorporate conditional prepayment multipliers $p_i=\min\{1,\gamma_i p_0\}$ with bank-level cap $P_i=\min(p_i N, H)$, and term-deposit redemption multipliers $q_i=\min\{1,u_i q_0\}$. They also apply the Table~4 core caps $(\rho_c,m_c)\in\{(0.9,5),(0.7,4.5),(0.5,4)\}$ for non-maturity deposits, and $Q_{i,c}(w)$ values the automatic options under $1.25\sigma$ shocks.

\begin{figure}[htbp]
\centering
\resizebox{\linewidth}{!}{%
\begin{tikzpicture}[
  x=1mm, y=1mm,
  font={\fontsize{8}{9.4}\selectfont}, text=fbInk,
  >={Stealth[length=1.5mm, width=1.1mm]},
  box/.style={align=flush center, inner xsep=1mm, inner ysep=0.8mm},
  op/.style={box, draw=fbGrey, line width=0.5pt, rounded corners=1.6pt, fill=white},
  leaf/.style={box, draw=fbBlue, line width=0.8pt, fill=white},
  kink/.style={fill=fbRed!14},
  ar/.style={->, draw=fbInk, line width=0.42pt},
  ln/.style={draw=fbInk, line width=0.42pt},
  el/.style={font={\fontsize{8}{9.4}\selectfont}, inner sep=0.8pt},
  dot/.style={circle, fill=fbInk, inner sep=0pt, minimum size=1.1mm},
]
\def\dagnote#1{{\fontsize{7}{8.4}\selectfont #1}}
\def\dxA{16}   \def\dwA{27}
\def\dxB{52}   \def\dwB{29}
\def\dxC{90}   \def\dwC{30}
\def\dxD{128}  \def\dwD{29}
\def\dxE{163}  \def\dwE{24}
\def\dyA{0} \def\dyM{-11} \def\dyB{-22}
\node[leaf, text width=\dwA mm, minimum height=12mm] (pos) at (\dxA,\dyA) {positions \& curves\\$w\in\R^n,\ q_0,m,h$};
\node[leaf, kink, text width=\dwA mm, minimum height=12mm] (caps) at (\dxA,\dyB) {behavioural caps\\prepay $H$, NMD Table~4\\\dagnote{options at $1.25\sigma$}};
\node[op, kink, text width=\dwB mm, minimum height=12mm] (disc) at (\dxB,\dyA) {floored discount\\$R_{i,c,k}=\max\{q_0+s_i,h\}+m$\\$d_{i,c,k}=e^{-R_{i,c,k}t_k}$};
\node[op, kink, text width=\dwC mm, minimum height=14mm] (beh) at (\dxC,\dyM) {behavioural flows\\$P_i=\min(p_i N,H),\ q_i$\\\dagnote{amount cap breaks Euler}};
\node[op, kink, text width=\dwD mm, minimum height=14mm] (loss) at (\dxD,\dyM) {currency EVE loss\\$\sum_c f_c(\ell_{i,c}(w))_+$};
\node[op, kink, text width=\dwE mm, minimum height=15mm] (k) at (\dxE,\dyM) {IRRBB $\Delta\mathrm{EVE}$\\$K(w)=\max_{1\le i\le 6}$\\$\sum_c f_c(\ell_{i,c})_+$};
\def\dpad{2.2}
\draw[fbBlueMid, line width=0.6pt, rounded corners=2.5mm]
  ($(disc.north west)+(-\dpad,\dpad+2.6)$) rectangle ($(loss.south east)+(\dpad,-\dpad)$);
\node[fill=white, inner xsep=1.2mm, inner ysep=0.3mm, anchor=west, text=fbBlue, font={\fontsize{8}{9.4}\selectfont\itshape}]
  at ($(disc.north west)+(1.2,\dpad+2.6)$) {evaluated across six SRP31 shock scenarios $i=1,\dots,6$};
\draw[ar] (pos) -- (disc);
\draw[ar] (disc.east) -- ++(3.5,0) |- ($(beh.west)+(0,2.5)$);
\draw[ar] (caps.east) -- ++(38.0,0) |- ($(beh.west)+(0,-2.5)$);
\draw[ar] (beh) -- (loss);
\draw[ar] (loss) -- (k);
\end{tikzpicture}}
\caption{Computational DAG for IRRBB $\Delta\mathrm{EVE}$ (SRP31) showing the six interest-rate shock scenarios, floored discounting, capped behavioural prepayment cash flows, currency positive parts, and worst-shock maximum.}\label{fig:dag-irrbb}
\end{figure}

Proofs for this model are in \cref{app:E:irrbb}. Under fixed baseline curves and a finite bank prepayment cap $H>0$, let $\widetilde K:\R^n\to\R$ denote the fixed-curve piecewise-affine map associated with $K$, that is, the extension of $K$ to all of $\R^n$ with every curve held fixed. It is continuous and piecewise affine on finitely many polyhedral cells (\cref{lem:IRRBB-cells}). An active prepayment cap $P_i=H$ breaks Euler completeness across all six SRP31 shocks, with gap $H(b_T-b_\tau)>0$ (\cref{lem:IRRBB-cap-witness}; witness~W7, established in \cref{sec:counterexamples}), so Euler attribution is \textbf{\Ncell{}}. Partitioning $[0,1]$ along the polyhedral cells gives a bounded Borel minimum-norm Clarke ray density, which closes path-integral attribution everywhere (\textbf{\Dstruct{}}, \cref{lem:IRRBB-path}). Finite position masking closes Shapley attribution on the whole domain (\textbf{\Dstruct{}}, \cref{lem:IRRBB-shapley}).

Across the frameworks where all three paradigms exist but disagree, \cref{tab:main:divergence} in \cref{app:table4} compares point Euler, straight-ray path-integral and coalition Shapley allocations side by side on exact numerical and rational controls.
\section{Unifying Theorems and Sharp Obstruction Witnesses}\label{sec:unify-counter}

The eleven model sections of \cref{sec:matrix} produce the pattern of \cref{tab:matrix33} one framework at a time. This section explains the pattern by three cross-framework theorems and makes every refutation exact. Outside the globally homogeneous conic regime (\cref{sec:matrix-regimes}), Euler attribution is proved on a strict subdomain for BA-CVA and SA-CVA and refuted for every other framework. Every framework nonetheless admits at least one paradigm, either on its full domain or on a stated subdomain. The three theorems correspond to the three rules of \cref{sec:toy}: conic invariance covers kinks such as $K_1$, offset-resilient rescue covers offsets such as $K_2$, and the obstruction hierarchy covers switches such as $K_3$.

\subsection{Three cross-framework unifying theorems}\label{sec:unify}

\begin{theorem}[Conic Clarke invariance and global Euler and path-integral closure]\label{thm:conic-clarke-ray}
Let $K:\mathcal{O}\to\R$ be locally Lipschitz and positively homogeneous of degree $1$ on an open cone $\mathcal{O}\subseteq\R^n$, extended continuously at the ray origin by $K(0)=0$. Define the canonical selector $g_*(x)=\nabla K(x)$ where $K$ is differentiable, and $g_*(x)=\arg\min_{g\in\partial_C K(x)}\|g\|_2^2$ otherwise. Then:
\begin{enumerate}[label=(\roman*),leftmargin=*,itemsep=1pt]
\item For every $x\in\mathcal{O}$ and every $t>0$, $\partial_C K(tx)=\partial_C K(x)$ and $g_*(tx)=g_*(x)$.
\item Every allocation $x\odot g$ with $g\in\partial_C K(x)$ reconciles ($x^\top g=K(x)$); in particular $x\odot g_*(x)$ closes Euler attribution in point mode across $\mathcal{O}$.
\item Along $\gamma(t)=tx$ ($t\in[0,1]$), $t\mapsto K(tx)=tK(x)$ is Lipschitz, $q_i(t)=x_i g_{*,i}(tx)=x_i g_{*,i}(x)$ is constant and in $L^1(0,1)$, and $I_i=\int_0^1 q_i(t)\,dt = x_i g_{*,i}(x)$ closes path-integral attribution in both incremental and full-value ($b_i^0=0$) modes---collapsing to ordinary gradient IG whenever $K$ is differentiable at $x$.
\end{enumerate}
\end{theorem}
\begin{proof}
For any fixed $t>0$, positive homogeneity $K(tx+h)=tK(x+h/t)$ on small open balls in $\mathcal{O}$ shows that $K$ is differentiable at $tx$ iff it is differentiable at $x$, with identical gradients. Mapping convergent sequences $y_k\to x$ of differentiability points to $ty_k\to tx$ preserves all limiting gradients and their closed convex hulls, giving $\partial_C K(tx)=\partial_C K(x)$. Strict convexity of $\|\cdot\|_2^2$ on the nonempty compact convex set $\partial_C K(x)$ guarantees uniqueness of $g_*(x)$, hence ray constancy $g_*(tx)=g_*(x)$. Euler's identity $y^\top\nabla K(y)=K(y)$ passes to limits and convex hulls by \cref{lem:B:clarke}, proving (ii). Finally, $t\mapsto K(tx)=tK(x)$ is linear on $[0,1]$ with trace value $K(0)=0$, while $|q_i(t)|\le |x_i|L_x<\infty$, giving $I_i=x_i g_{*,i}(x)$ and $\sum_i I_i = x^\top g_*(x)=K(x)$.
\end{proof}

\Cref{thm:conic-clarke-ray} governs every Euler and path-integral closure for FRTB-IMA, discrete ES/VaR, classic CME SPAN and, on $\Omega_+$, BA-CVA and SA-CVA. It also shows why path-integral attribution departs from Euler attribution exactly where positive homogeneity fails.

\begin{theorem}[Offset-resilient path and coalition rescue]\label{thm:offset-rescue}
Suppose $K:\mathcal{D}\to\R$ is finite on an evaluation box $[0,x]\subseteq\mathcal{D}$ ($0,x\in\mathcal{D}$), is locally Lipschitz on an open neighbourhood of $\gamma([0,1])$ for $\gamma(t)=tx$, and admits either (a) ambient differentiability a.e.\ along $\gamma$ or (b) a finite piecewise-$C^1$/polyhedral structure admitting a bounded measurable Clarke selection $g(t)\in\partial_C K(tx)$ satisfying $x^\top g(t)=\tfrac{d}{dt}K(tx)$ a.e. Then:
\begin{enumerate}[label=(\roman*),leftmargin=*,itemsep=1pt]
\item Path integration $I_i=\int_0^1 x_i g_i(t)\,dt$ exists componentwise in $L^1(0,1)$ and reconciles exactly to $K(x)-K(0)$, so path-integral attribution closes.
\item If $T_S x\in\mathcal{D}$ for every $S\subseteq\{1,\dots,n\}$, discrete Shapley values $\phi_i$ exist and reconcile to $K(x)-K(0)$, so Shapley attribution closes, even when point Euler reconciliation fails across an open set due to fixed collateral, historical pools, concentration multipliers, or amount caps.

\end{enumerate}
\end{theorem}
\begin{proof}
Because $\gamma([0,1])$ is compact and $K$ is locally Lipschitz around it, a finite open cover yields a uniform Lipschitz constant $M$ along $\gamma$. Subdividing $[0,1]$ proves $|K(tx)-K(sx)|\le M\|x\||t-s|$, so $t\mapsto K(tx)$ is AC on $[0,1]$ and every Clarke/gradient element obeys $\|g(t)\|_2\le M$. Thus $|x_i g_i(t)|\le M|x_i|$ is in $L^1(0,1)$. In case (a), $(K\circ\gamma)'(t)=x^\top\nabla K(tx)$ a.e.\ by the chain rule; in case (b), $x^\top g(t)=(K\circ\gamma)'(t)$ a.e.\ holds by assumption---and automatically for continuous $PC^1$ or polyhedral selections because matching AC branches share identical derivatives at Lebesgue density points of their coincidence sets (\cref{lem:E:sacc-path-regular,lem:SIMM-ray,lem:IRRBB-path}). Integrating over $[0,1]$ yields $\sum_i I_i = K(x)-K(0)$. Part (ii) follows from \cref{thm:shapley-eff}.
\end{proof}

\Cref{thm:offset-rescue} covers the five offset, cap and threshold models (the Basel Output Floor, SA-CCR, bilateral XVA, Solvency~II and IRRBB) and also ISDA SIMM inside its regular and coalition-closed cones $\mathcal{R}$ and $\mathcal{Q}$. \Cref{thm:obstruction-hierarchy} then explains why FRTB-SA is the only framework for which both Euler and path-integral attribution fail.

\begin{theorem}[Hierarchy of attribution obstructions]\label{thm:obstruction-hierarchy}
Across finite regulatory DAGs $K:\mathcal{D}\to\R$ evaluated on open parameter families $\mathcal{U}\subseteq\mathcal{D}$:
\begin{enumerate}[label=(\roman*),leftmargin=*,itemsep=1pt]
\item If points $x\in\mathcal{U}$ remain locally Lipschitz along straight rays (with $[0,x]\subseteq\mathcal{D}$ and the ray-selection hypothesis of \cref{thm:offset-rescue}) and coalition-closed ($T_S x\in\mathcal{D}$ for all $S$), then loss of degree-$1$ homogeneity (via constant offsets, degree-$0$ caps, or degree-$\tfrac{3}{2}$ concentration factors) refutes Euler completeness on open smooth regions where $x^\top\nabla K(x)\ne K(x)$ without obstructing local Lipschitz continuity, path-integral reconciliation, or Shapley efficiency.
\item If a continuous homogeneous jump selector $g(\alpha x)=\alpha^m g(x)$ ($m>0,\ \alpha>0$) triggers degree-$1$ homogeneous continuous branches $F_A(x)\ne F_B(x)$ at $g(x)=0$ with $x\in\overline{\{g<0\}}$ and the resulting jump propagates uncanceled with degree-$1$ scaling to final capital (yielding a nonzero gap $t\Delta\ne 0$ at $tx$ for all $t>0$), then $K$ loses local Lipschitz continuity at every positive ray point $tx$, simultaneously refuting Euler gradient sensitivity and path-integral ray-density existence while leaving finite node-rule accounting reconciliation (\cref{cor:E:recon}) intact and, under coalition closure ($T_S x\in\mathcal{D}$ for all $S$), leaving Shapley efficiency (\cref{thm:shapley-eff}) intact.
\end{enumerate}
\end{theorem}
\begin{proof}
For part~(i), on any open smooth region where $K$ fails degree-$1$ homogeneity with $x^\top\nabla K(x)\ne K(x)$, Euler's theorem refutes Euler completeness, whereas local Lipschitzness around $\gamma([0,1])$, the ray-selection hypothesis of \cref{thm:offset-rescue}, and coalition closure $T_S x\in\mathcal{D}$ yield path-integral reconciliation (\cref{thm:offset-rescue}(i)) and Shapley efficiency (\cref{thm:shapley-eff}). For part~(ii), homogeneity of $g$ preserves $g(tx)=t^m g(x)=0$ and $tx\in\overline{\{g<0\}}$ for every $t>0$, while degree-$1$ homogeneity of $F_A,F_B$ gives $F_A(tx)-F_B(tx)=t(F_A(x)-F_B(x))\ne 0$ (\cref{prop:E:jump}). Because this jump propagates uncanceled to final capital with gap $t\Delta\ne 0$, $K$ is discontinuous at every $tx$ ($t>0$), precluding both $\nabla K(tx)$ and $\partial_C K(tx)$ everywhere on $(0,1]$, so neither Euler sensitivities nor path-integral ray densities exist. Finite node-rule accounting (\cref{cor:E:recon}) and, under coalition closure $T_S x\in\mathcal{D}$ for all $S$, finite permutation telescoping (\cref{thm:shapley-eff}) require neither continuity nor directional limits and remain valid.
\end{proof}

\subsection{Exact witnesses for the eight refuted cells}\label{sec:counterexamples}

\Cref{tab:counterexamples} groups the eight refuted cells by failure mechanism. Seven witnesses suffice, because witness~W1 refutes both the Euler and the path-integral cell of FRTB-SA. Each witness below is verified step by step, and \cref{fig:obstruction-plots} plots three of them, one per mechanism.

\begin{figure}[htbp]
\centering
\resizebox{\linewidth}{!}{%
\begin{tikzpicture}[font=\scriptsize, text=fbInk]
\pgfplotsset{
  obsaxis/.style={frtb, scale only axis, width=3.75cm, height=3.25cm,
    title style={font=\scriptsize, text=fbInk, yshift=-1pt},
    xlabel style={yshift=2pt}, ylabel style={yshift=-3pt}},
}
\begin{axis}[obsaxis, name=pa,
  xmin=0, xmax=5, ymin=-10, ymax=40,
  xtick={0,1,2,3,4,5}, ytick={-8,0,8,16,24,32},
  xlabel={sensitivity $s$ (threshold $T=1$)},
  ylabel={concentration term},
  title={(a) ISDA SIMM, witness W2}]
  \fill[fbRed!14] (axis cs:0,-8) -- (axis cs:4,16) -- (axis cs:4,24) -- (axis cs:0,0) -- cycle;
  \addplot[fbGrey, line width=0.4pt, domain=0:5, samples=2] {0};
  \draw[fbGrey, densely dotted] (axis cs:1,-10) -- (axis cs:1,20);
  \addplot[fbInk, densely dotted, line width=0.6pt, domain=0:5, samples=2] {6*x};
  \addplot[fbRed, dashed, line width=0.6pt, domain=0:5, samples=2] {6*x-8};
  \addplot[fbBlue, line width=1.2pt, domain=0:5, samples=121] {2*x*max(1,sqrt(x))};
  \addplot[only marks, mark=*, mark size=1.4pt, fbBlue] coordinates {(4,16)};
  \addplot[only marks, mark=*, mark size=1.3pt, fbInk] coordinates {(4,24)};
  \draw[fbRed, line width=0.7pt, {Stealth[length=1.2mm]}-{Stealth[length=1.2mm]}] (axis cs:4.35,16) -- (axis cs:4.35,24);
  \node[fbRed, anchor=west, inner sep=1pt] at (axis cs:4.38,20) {$8$};
  \node[anchor=north west, align=left, inner sep=1.5pt, fill=white, fill opacity=0.85, text opacity=1]
    at (axis cs:0.1,39) {curve $2s^{3/2}$ ($s>1$)\\ tangent $6s-8$\\ Euler $sK_s=24$};
\end{axis}
\begin{axis}[obsaxis, name=pb, at={($(pa.outer east)+(0.25cm,0)$)}, anchor=outer west,
  xmin=0, xmax=1, ymin=0, ymax=6.7,
  xtick={0,0.333333,0.6667,1}, xticklabels={$0$,$\tfrac13$,$\tfrac23$,$1$},
  ytick={0,1,2,3,4,5,6},
  xlabel={ray parameter $t$},
  ylabel={trace and ray densities},
  title={(b) Basel Output Floor, witness W3}]
  \draw[fbInk, densely dotted] (axis cs:0.333333,0) -- (axis cs:0.333333,6.7);
  \addplot[fbBlue, line width=1.2pt, domain=0:1, samples=61] {max(2+2*x,1+5*x)};
  \addplot[fbRed, line width=0.9pt] coordinates {(0,1) (0.333333,1) (0.333333,3) (1,3)};
  \addplot[fbGrey, line width=0.9pt, dashed] coordinates {(0,1) (0.333333,1) (0.333333,2) (1,2)};
  \node[fbBlue, anchor=south, inner sep=1pt, rotate=17] at (axis cs:0.17,2.4) {slope $2$};
  \node[fbBlue, anchor=south east, inner sep=1pt] at (axis cs:0.72,4.65) {slope $5$};
  \node[fbBlue, anchor=south east, inner sep=1pt] at (axis cs:0.98,5.9) {$K(t,t)$};
  \node[fbRed, anchor=south east, inner sep=1pt] at (axis cs:0.98,3.02) {$q_1$, integral $\tfrac73$};
  \node[fbGrey, anchor=north east, inner sep=1pt] at (axis cs:0.98,1.95) {$q_2$, integral $\tfrac53$};
  \node[anchor=south west, inner sep=1pt] at (axis cs:0.345,0.05) {$t^\ast$};
\end{axis}
\begin{axis}[obsaxis, name=pc, at={($(pb.outer east)+(0.25cm,0)$)}, anchor=outer west,
  xmin=-0.2, xmax=0.2, ymin=1.76, ymax=2.12,
  xtick={-0.2,-0.1,0,0.1,0.2}, ytick={1.8,1.9,2.0,2.1},
  yticklabel style={/pgf/number format/fixed, /pgf/number format/precision=1},
  xticklabel style={/pgf/number format/fixed},
  xlabel={perturbation $\varepsilon$},
  ylabel={$\mathrm{SBM}(x(\varepsilon))$},
  title={(c) FRTB-SA curvature, witness W1}]
  \addplot[fbBlue, line width=1.2pt, domain=-0.2:0, samples=2] {sqrt(13)/2};
  \addplot[fbBlue, line width=1.2pt, domain=0:0.2, samples=41] {sqrt(0.36+(0.8+x)^2+1+1.25*(1.4+x))};
  \draw[fbRed, densely dotted, line width=0.7pt] (axis cs:0,1.802776) -- (axis cs:0,1.936492);
  \addplot[only marks, mark=o, mark size=1.6pt, fbBlue, mark options={fill=white}] coordinates {(0,1.802776)};
  \addplot[only marks, mark=*, mark size=1.6pt, fbBlue] coordinates {(0,1.936492)};
  \node[anchor=north west, inner sep=1.5pt] at (axis cs:-0.195,1.795) {$\sqrt{13}/2\approx1.8028$};
  \node[anchor=south east, inner sep=1.5pt] at (axis cs:-0.005,1.94) {$\sqrt{15}/2\approx1.9365$};
  \node[fbRed, anchor=west, inner sep=1.5pt] at (axis cs:0.008,1.87) {jump $\approx0.1337$};
\end{axis}
\end{tikzpicture}}
\caption{Three exact witnesses, one per failure mechanism. (a)~Super-linear concentration (W2): the ISDA SIMM concentration term $2s\max\{1,\sqrt{s}\}$ at $(w,T,m)=(1,1,2)$, its tangent $6s-8$ at $s=4$, and the line $6s$ through the origin. The Euler share $sK_s=24$ exceeds the term $16$ by the shaded gap $8$, so the Euler allocation sums to $25$ while $K=17$. (b)~Fixed offset (W3): the Output Floor trace $K(t,t)=\max\{2+2t,\,1+5t\}$ switches branch at $t^\ast=\tfrac13$. The ray densities $q_1$ and $q_2$ step from $(1,1)$ to $(3,2)$, and their integrals $(\tfrac73,\tfrac53)$ sum to $K(1,1)-K(0)=4$. (c)~Conditional switch (W1): along $x(\varepsilon)=(1,0,\tfrac35,\tfrac45+\varepsilon,1)$, FRTB-SA capital jumps at $\varepsilon=0$ from $\sqrt{13}/2$ to $\sqrt{15}/2$, so neither a gradient nor a Clarke set exists at $x(0)$ or anywhere on its ray.}\label{fig:obstruction-plots}
\end{figure}

\begin{table}[htbp]
\centering\small
\caption{Summary of all eight counterexamples grouped by structural failure mechanism.}\label{tab:counterexamples}
\begin{tabularx}{\linewidth}{@{}>{\raggedright\arraybackslash}p{2.65cm}>{\raggedright\arraybackslash}p{2.25cm}>{\raggedright\arraybackslash}p{4.15cm}>{\raggedright\arraybackslash}X@{}}
\toprule
Model / Paradigm & Witness ref. & Mechanism / family & Exact refutation quantity / gap \\
\midrule
\multicolumn{4}{@{}l}{\emph{1. Fixed Offsets, Collateral, and Caps (5 Euler / Marginal refutations)}} \\
\texttt{FLOOR} (Euler / Marginal) & \cref{lem:FLOOR-constant-gap} & Fixed external credit/operational RWA, $(3n+2)$-dim affine family & Gap $c=1$ at $K=6$ ($x=(1,1)$) \\
\texttt{SACCR} (Euler / Marginal) & \cref{ex:E:sacc-numeric} & Fixed collateral $C=3$ + exponential multiplier $P$, $7$-param FX family & Gap $\alpha C P_u = 2.1 e^{-2/19}\approx 1.890184$ \\
\texttt{XVA} (Euler / Marginal) & \cref{lem:XVA-family} & Double-hinge VM/IM collateral $(b,H)$, $11$-param ($d_A$-dim general) & Gap $w(b-\eta H)=1/8$ at $w=1/4$ \\
\texttt{SOLV2} (Euler / Marginal) & \cref{lem:SOLV2-fixed-fdb} & Fixed FDB pool $f=1$ with active TP cap, $(6n+3)$-dim ($15$-param) & Gap $(1-\tau)f=3/4$ \\
\texttt{IRRBB} $\Delta\mathrm{EVE}$ (Euler / Marginal) & \cref{lem:IRRBB-cap-witness} & Bank-level prepayment cap $H>0$ across all 6 SRP31 shocks, $(n+8)$-dim ($10$-param) & Gap $H(b_T-b_\tau)\approx 0.247846$ \\
\midrule
\multicolumn{4}{@{}l}{\emph{2. Super-linear Concentration Scaling (1 Euler / Marginal refutation)}} \\
\texttt{SIMM} (Euler / Marginal) & \cref{lem:SIMM-counterexample} & Degree-$\tfrac{3}{2}$ active concentration multiplier + degree-$1$ notional add-on, $6$-param family & Gap $\tfrac{1}{2}mws^{3/2}/\sqrt{T}=8$ at $K=17$ \\
\midrule
\multicolumn{4}{@{}l}{\emph{3. Conditional-Switch Discontinuities (2 refutations: Euler / Marginal and Path Integral)}} \\
\texttt{SA} (Euler / Marginal) & \cref{lem:SA-full-jump} & Curvature bucket switch at $K_b^+=K_b^-$, $5$-param family & Jump $(\sqrt{15}-\sqrt{13})/2\approx 0.133716$; no $\partial_C K$ \\
\texttt{SA} (Path Integral) & \cref{lem:SA-ray-obstruction} & Conic ray sitting on curvature jump, $5$-param ray & Jump $t(\sqrt{15}-\sqrt{13})/2$ for all $t\in(0,1]$; no ray density \\
\bottomrule
\end{tabularx}
\end{table}

\begin{enumerate}[label=\textbf{W\arabic*},leftmargin=*,itemsep=5pt]
\item \textbf{FRTB-SA, Euler and path-integral attribution (Discontinuous curvature scenario selection; \cref{lem:SA-full-jump,lem:SA-ray-obstruction}).}
Take two FRTB-SA curvature buckets with zero within-bucket correlations in every scenario. Bucket~$1$ carries two upward repricing leaves $CVR_1^+=(x_1,x_2)$ and two downward leaves $CVR_1^-=(x_3,x_4)$; bucket~$2$ carries the single leaf $x_5$ in both directions. Let cross-bucket curvature correlations equal $(3/8,1/2,5/8)$ across low, medium, and high scenarios (derived from medium delta correlation $1/\sqrt{2}$, squared, under low/high adjustment), with zero delta/vega leaves and optional continuous RRAO $r\sum_{i=1}^5 |x_i|$ ($r\ge 0$). Evaluate along the perturbed family $x(\varepsilon)=(1,0,3/5,4/5+\varepsilon,1)$. Then $K_1^+=1$, $K_1^-(\varepsilon)=\sqrt{9/25+(4/5+\varepsilon)^2}$, and $K_2=1, S_2=1$. For small $\varepsilon<0$, $K_1^+ > K_1^-$, so the upward direction is selected and $K_1=1$, $S_1=1$. At $\varepsilon=0$ the two directions tie, $K_1^+=K_1^-=1$, and MAR21.5(3)(a)(iii) selects the downward direction because $\sum CVR^- = 7/5 > 1 = \sum CVR^+$. For small $\varepsilon>0$ the downward direction remains selected, with $S_1=7/5+\varepsilon$. For all three scenarios $S_1,S_2>0$ ($\psi=1$), and the high scenario ($5/8$) strictly dominates. Consequently full $\mathrm{SBM}$ jumps at $\varepsilon=0$ from left limit $\sqrt{2+2(5/8)}=\sqrt{13}/2\approx 1.80277564$ to right/point value $\sqrt{2+2(5/8)(7/5)}=\sqrt{15}/2\approx 1.93649167$---a jump of $(\sqrt{15}-\sqrt{13})/2\approx 0.13371604$ untouched by continuous RRAO. (Even smooth cubic polynomial unit pricers $az^2+bz^3$ reproduce these linear repricing leaves.) Thus $K$ is not locally Lipschitz and has neither an ordinary gradient nor a Clarke set at $x(0)$, which refutes Euler attribution. By degree-$1$ homogeneity, scaling by $t\in(0,1]$ yields the identical jump multiplied by $t$ at every $t x(0)$, so no ray density exists along the entire ray, which refutes path-integral attribution.

\item \textbf{ISDA SIMM, Euler attribution (Active concentration degree $\tfrac{3}{2}$ mixed with degree-$1$ add-on; \cref{lem:SIMM-counterexample}).}
Consider one product class, one risk class, one non-residual delta bucket, one sensitivity $s$, and one add-on notional $a$, with positive fixed metadata $(w,T,m,\beta)$ ($m>1$). The full attributed output is
\[
  K(s,a)=m w |s|\max\bigl\{1,\sqrt{|s|/T}\bigr\} + \beta|a|.
\]
Across the nonempty open six-parameter family $\mathcal{F}=\{s>T>0,\ a,w,\beta>0,\ m>1\}$, concentration is strictly active and $K$ is smooth: $K(s,a)=m w s^{3/2}/\sqrt{T}+\beta a$. Differentiating gives $K_s=\tfrac{3}{2}mw\sqrt{s/T}$ and $K_a=\beta$, yielding positive Euler gap
\[
  sK_s+aK_a-K = \tfrac{1}{2}mws^{3/2}/\sqrt{T}>0.
\]
At $(s,a,w,T,m,\beta)=(4,1,1,1,2,1)$, exact arithmetic gives $K=17$, gradient $(6,1)$, Euler allocation $(24,1)$ summing to $25$, and exact gap $8$. Ad-hoc division by $3/2$ violates the exact-gradient requirement and breaks whenever add-ons or inactive groups coexist.

\item \textbf{Basel Output Floor, Euler attribution (Fixed external credit/operational RWA components; \cref{lem:FLOOR-constant-gap,lem:FLOOR-game-comparison}).}
For $n\ge 2$ and fixed $\varphi>0$, set $f(z)=e+a^\top z$, $\varphi g(z)=c+b^\top z$, and $F(z)=\max\{f(z),\varphi g(z)\}$. Across the $(3n+2)$-dimensional open family $\{a,b,x\in(0,\infty)^n,\ e>c>0,\ (b-a)^\top x > e-c\}$, both branches remain strictly positive on $[0,x]$, $f$ wins at $0$, and $\varphi g$ binds strictly at $x$, so $A^{\rm Euler}=x\odot b$ sums to $F(x)-c < F(x)$. Specifically, take $f=2+z_1+z_2$ and $\varphi g=1+3z_1+2z_2$ at $x=(1,1)$, so that $F(0)=2$ and $F(1,1)=6$; this is the rule $K_2$ of \cref{ex:offset}. The Euler allocation is $A^{\rm Euler}=(3,2)$, summing to $5$ with gap $c=1$. Along the ray $(t,t)$, $f(t,t)=2+2t > 1+5t=\varphi g(t,t)$ for $t<1/3$, with exact crossover at $t_*=1/3$. The ray path integral is therefore $I = \frac{1}{3}(1,1)+\frac{2}{3}(3,2) = (7/3, 5/3)$, summing to $4=F(1,1)-F(0)$. The coalition values $(v(\varnothing),v(\{1\}),v(\{2\}),v(\{1,2\}))=(2,4,3,6)$ give the Shapley value $\Phi=(5/2,3/2)$, which also sums to $4$.

\item \textbf{SA-CCR, Euler attribution (Fixed collateral and exponential PFE multiplier; \cref{ex:E:sacc-numeric}).}
Take two linear FX tickets in one unmargined netting set with adjusted unit notionals $(150,100)$, deltas $+1$, maturity factor $MF=1$ ($A(x)=6x_1+4x_2$), unit market values $v=(2,-1)$, fixed held collateral $C=3$, $\alpha=1.4$, $f=0.05$, and $x=(1,1)$. At $x$, $V=1$, $u=V-C=-2<0$, $A=10$, and $r=(V-C)/(2(1-f)A)=-2/19$, giving exact exposure $E(x)=1.4\cdot 10(0.05+0.95e^{-2/19})\approx 12.671165429155$. Using $P_u=\tfrac{1}{2}e^r$ and $P_A=f+(1-f)e^r(1-r)$ (\cref{prop:E:sacc-fixed}), the smooth endpoint gradient allocation is $x\odot\nabla E(x)\approx (9.618895540298,\,4.942453903987)$, summing to $14.561349444285$ with exact analytical Euler gap
\[
  x^\top\nabla E(x) - E(x) = \alpha C P_u = 1.4\cdot\frac{3}{2}e^{-2/19} = 2.1 e^{-2/19} \approx 1.890184015130 > 0,
\]
and $E(2x)\approx 27.309130260980\ne 2E(x)$. By contrast, ray quadrature on $\gamma(t)=(t,t)$ ($J_u\approx 0.331839845959$, $J_A\approx 0.871899260344$) and discrete coalition evaluations give
\begin{align*}
  \bigl(v(\varnothing),\,v(\{1\}),\,v(\{2\}),\,v(\{1,2\})\bigr) &\approx (0,\; 7.729823388545,\; 3.422936373955,\; 12.671165429155), \\
  I &\approx (8.253105355573,\; 4.418060073582), \\
  \Phi &\approx (8.489026221873,\; 4.182139207282),
\end{align*}
where the path integral $I$ and the Shapley value $\Phi$ both sum exactly to $E(x)-E(0)=E(x)$.

\item \textbf{Bilateral XVA, Euler attribution (Fixed initial and variation margin thresholds; \cref{lem:XVA-family,lem:XVA-fixtures}).}
On $n=2$ with $|\mathcal{C}|=S=T=1$, $a_1>0>a_2$, $x_1,x_2>0$, $0<\eta<1$, $0<\eta H<b<H$, $V=a_1x_1+a_2x_2>H$, and $w=L_1D_1(1-Q_{11})>0$, both collateral hinges bind strictly across an open $11$-parameter family: $K(x)=w[(1-\eta)V+\eta H-b]$ and $x^\top\nabla K = w(1-\eta)V$, giving gap $w(b-\eta H)>0$ (extending to a general $d_A$-dimensional open family for arbitrary $n,C,S,T$). At $a=(4,-1)$, $x=(1,1)$, $b=1$, $H=2$, $\eta=1/4$, $w=1/4$: $K=7/16$, $A^{\rm Euler}=(3/4,-3/16)$ (sum $9/16$, gap $1/8$), ray IG along $(t,t)$ switches at $t=1/3$ and $t=2/3$ giving $(7/12,-7/48)$, and coalitions $(0,5/8,0,7/16)$ give Shapley $(17/32,-3/32)$.

\item \textbf{Solvency~II Standard Formula, Euler attribution (Fixed FDB pool with active TP cap; \cref{lem:SOLV2-fixed-fdb}).}
With market and life modules (correlation $1/4$), net factor $\eta\in(0,1)$, tax rate $\tau\in(0,1)$, fixed FDB pool $f>0$, life provisions $v^\top x$, expenses $e^\top x$, and tax capacity $c^\top x$ satisfying $f<(1-\eta)B$, $0.3B<0.0045 v^\top x$, and $\tau(1.3B+0.25 e^\top x - f)<c^\top x$, capital is smooth: $K=(1-\tau)(1.3B+0.25 e^\top x - f)$ on an open $(6n+3)$-parameter family, with exact Euler gap $(1-\tau)f>0$. Setting $n=2$, $x=(1,1)$, $a=b=(10,10)$, $v=(10^4,10^4)$, $e=(1,1)$, $c=(100,100)$, $f=1$, $\eta=1/2$, $\tau=1/4$ gives $B=10\sqrt{10}\in(30,40)$, verifying every strict inequality ($B/2>1$, $0.3B<90$, $\tau L<200$) and gap $(1-1/4)\cdot 1=3/4$.

\item \textbf{IRRBB $\Delta\mathrm{EVE}$, Euler attribution (Bank-level prepayment amount cap across all six SRP31 shocks; \cref{lem:IRRBB-cap-witness}).}
Take one USD currency with cash-flow times $\tau=7/2500$ and $T=35/2$ years, baseline rate $1/100$, zero spread and zero post-shock floor. The principal $N=\sum_j w_j$ prepays $P_i=\min(p_i N, H)$ at $\tau$ and $N-P_i$ at $T$. Secondary symbols are specified in \cref{app:E:irrbb}; the SRP31 Table~2 shocks are $(P,S,L)=(1/50,3/100,9/400)$, the base prepayment rate is $p_0=1/5$, and the CPR multipliers are $\gamma=(4/5,6/5,4/5,6/5,4/5,6/5)$. Using $1/159 < e^{-35/8} < 2/3$, steepener/short-up/flattener/down shocks at $T$ stay strictly below parallel-up $3/100$. Hence, with $R=\max_{i\ge 2}(d_{1\tau}-d_{i\tau})_+/(d_{iT}-d_{1T})$ and $N > H\max\{25/4, 1+R\}$, parallel up ($i=1$) uniquely dominates and $K(w)=(N-H)b_T + H b_\tau$, where $b_k=d_{0k}-d_{1k}$. Since $\sum_j w_j\partial_j K = N b_T$, the smooth Euler gap is $H(b_T-b_\tau)>0$ (verified via $b_\tau < 7/125000 < 7/60 < b_T$).
\end{enumerate}
\FloatBarrier

\section{Scope, Verification Discipline, and Conclusion}\label{sec:scope}

An exact attribution of a capital rule is useful only if it is clear what the computational graph verifies and what lies outside the mathematical model. This section states that boundary for the eleven frameworks, explains why the graph refuses rather than regularises at domain walls, and lists open problems.

\subsection{The \Tstruct{} and the \Treg{}}\label{sec:scope-tier}

Closing the $33$-cell matrix required a clear separation between exact attribution of a formalised graph (the \Tstruct{}) and equivalence with an entire regulatory text (the \Treg{}). Even where the legal texts are public, they leave gaps that algebra alone cannot close, and certification at the \Treg{} would need external calibration or legal interpretation. Four kinds of gap recur.

Industry and clearing-house methodologies are governed by dated parameter releases rather than fixed statutory formulas. For ISDA SIMM (\cref{app:E:simm}) and classic CME SPAN (\cref{app:E:ccp}), a new calibration version or daily margin file changes $K_\pi$ itself, so an attribution holds only relative to an identified version string and parameter payload. A second gap concerns external valuation and exposure leaves. BA-CVA and SA-CVA (\cref{app:E:cva}) and the Basel Output Floor (\cref{app:E:floor}) consume exposure-at-default $EAD_{NS}$, credit RWA, operational RWA or CVA sensitivities computed outside the aggregation graph. Accounting reconciliation holds unconditionally at the declared bookkeeping degree $1$ (\cref{lem:E:ext}), but marginal faithfulness inherits any Euler gap inside those upstream leaves.

Behavioural estimation and tax loss absorption form the third gap. In IRRBB (\cref{app:E:irrbb}), SRP31 constrains the core non-maturity-deposit caps $(\rho_c,m_c)$ without fixing a unique statistical estimator of the core or a unique cash-flow slotting map. In Solvency~II (\cref{app:E:solvency2}), the deferred-tax adjustment $\Delta_{\rm DT}$ and undertaking-specific management actions depend on balance-sheet valuation models outside Articles 204--207 of the Delegated Regulation. The fourth gap is formed by unverified tables and omitted sub-calculations. Parts of MAR21/22 (FRTB-SA), MAR50 (CVA), Section~K of SIMM, the CRE52 shifted-option and collateral-haircut rules (SA-CCR) and the Solvency~II counterparty-default internals carry explicit \tagU{} tags in the scope ledger (\cref{tab:scope-ledger} in \cref{app:E:scope-ledger}). Keeping these dependencies tagged makes the limits of each regulatory claim visible.

\subsection{Engineering principle: refuse rather than silently regularise}\label{sec:scope-refusal}

Production risk engines often patch numerical edge cases by adding a small $\varepsilon$ to denominators, smoothing indicator switches with sigmoids, flooring negative radicands or inventing a value at $0/0$. The $33$-cell analysis shows why such silent regularisation undermines attribution.

Isolated $0/0$ ratios should not be masked. Fused products such as $\mathrm{HBR}_b Q_b$ in the FRTB-SA default-risk charge, or $P(u,A)$ in SA-CCR, extend continuously to $0$ at zero denominators (\cref{lem:E:ratio}(d), \cref{lem:E:sacc-mult}). Isolated ratios such as $U/2V$, $\mathrm{HBR}_b$, $\theta=N/D$ or $m(u,A)$, by contrast, have arbitrary limits along different cones at $0/0$ and must be refused whenever they are queried on their own. Even for fused products, continuity need not imply $PC^1$ regularity (\cref{rem:E:zeroPC}).

Discontinuous switches should not be smoothed either. Smoothing MAR21.4(5)(b) at $\varrho_A=0$, or MAR21.5(3)(a)(iii) at $K_b^+=K_b^-$, replaces the regulatory rule by a different capital function whose smooth gradient does not describe the actual capital. \tool{} keeps the discontinuity, reports a reconciling allocation with status \emph{accounting-only} where one is needed, and never presents it as a marginal.

Finally, every output carries one of the statuses of \cref{sec:typed}: \emph{well-defined}, \emph{undefined outside the coalition domain}, \emph{not applicable}, \emph{accounting-only} or \emph{domain refusal}. Together with the typed allocation of \cref{def:typed}, the status tells a downstream user whether a number is a true gradient marginal, a realisable Clarke vertex, an accounting split, a path integral with its baseline $K(x_0)$, or a Shapley value that passed the coalition-closure test. Each user therefore receives only claims that the underlying graph supports. The framework-by-framework regulatory references, \tagU{} dependencies, \Tstruct{} conditions and mandatory refusal triggers are collected in \cref{tab:scope-ledger} in \cref{app:E:scope-ledger}.

\subsection{Open mathematical problems and conclusion}\label{sec:conclusion}

Four problems remain open on the way to attribution at the \Treg{}. The first is to prove global positive semidefiniteness, or complete coercivity certificates, for all published high- and low-scenario correlation matrices of ISDA SIMM~v2.9 and later and of FRTB-SA and MAR50. The second is to extend the continuous SA-CCR and SPAN analyses to multi-agreement collateral inventory optimisation, to rounding and discretisation bounds, and to recursive SPAN~2 and VaR-based clearing methodologies. The third is to automate the enumeration of jointly realisable polyhedral vertices at arbitrary composite graph depth without an exponential growth in facets. The fourth is to relate path-integral baselines more tightly to dynamic capital-planning trajectories across multi-period historical windows.

\paragraph{Conclusion.}
Trade-level attribution of regulatory capital and margin requires a nonsmooth calculus on computational graphs. We distinguished five economic attribution questions, proved recursive accounting and path-integral reconciliation together with exact coalition efficiency, and classified all $33$ method-by-model cells across eleven frameworks: $17$ are proved on their stated structural domains, $8$ on strict subdomains, and $8$ are refuted by exact witnesses. The results give a unified mathematical foundation for trade-level attribution, conditional on the stated provenance and domain contracts.

\ifdefined\SplitSupplement
\clearpage
\setcounter{page}{1}\renewcommand{\thepage}{S\arabic{page}}
\begin{center}
  {\LARGE\bfseries Supplementary Material}\par\vspace{4pt}
  {\large RiskTrace: Exact Attribution for Nonsmooth Financial Risk, Capital, and Margin Frameworks}
\end{center}
\vspace{10pt}
\fi
\appendix
\section{Full 33-Cell Method and Hypothesis Matrix}\label{app:table4}

\Cref{tab:full-33-lookup} provides the complete technical lookup across all eleven frameworks and all three attribution paradigms---Euler / Marginal, Path Integral, and Shapley Game---recording coordinates, player definitions, open witness families, tie-selection conventions, ledger targets, restricted subdomains and excluded witnesses for every Proved (scoped) cell, exact counterexample witnesses for every Refuted cell, and theorem pointers to \cref{app:E}. Trade-level attribution lifts from these coordinates whenever a verified trade-to-input mapping is supplied; otherwise guarantees apply strictly at the tabulated dated/static position, risk-factor, or business-cell level (\cref{tab:scope-ledger}). \Cref{tab:main:divergence} then places exact multi-paradigm numeric controls side by side.

\begin{table}[htbp]
\centering
{\scriptsize
\setlength{\tabcolsep}{2.5pt}
\renewcommand{\arraystretch}{1.03}
\caption{Complete $33$-cell technical lookup across eleven frameworks and three paradigms ($17$ Proved (full), $8$ Proved (scoped), $8$ Refuted, $0$ open; structural tier, \cref{tab:scope-ledger}). Trade-level attribution lifts from these coordinates only where a verified trade-to-input map is supplied.}\label{tab:full-33-lookup}
\begin{tabularx}{\linewidth}{@{}>{\raggedright\arraybackslash}p{1.55cm}>{\raggedright\arraybackslash}p{1.55cm}>{\raggedright\arraybackslash}p{2.95cm}>{\raggedright\arraybackslash}p{1.85cm}>{\raggedright\arraybackslash}X>{\raggedright\arraybackslash}p{1.65cm}@{}}
\toprule
Paradigm & Verdict & Tie-selection / game convention & Ledger target & Scoped subdomain $\Omega'$ and excluded witness, or refutation witness & Proof ref. \\
\midrule
\multicolumn{6}{@{}p{\linewidth}@{}}{%
  \textbf{1. FRTB-IMA} (\cref{app:E:ima}) \quad\emph{Coords / players:} dated positions $x\in\R^{T\times n}$ on strict-denominator cone $\mathcal{O}$ (\cref{lem:IMA-input-strata}); Shapley players $P=\{1,\dots,n\}$ on today's $w\in\R^n$ (fixed history $h$). \quad\emph{Open witness:} $499$-dim smooth ball in $\mathcal{O}$ (\cref{lem:IMA-open-family}).%
} \\
\addlinespace[1pt]
Euler / Marginal & Proved (full) & Global min-norm Clarke selector $g_*(x)\in\partial_C K(x)$ ($\nabla K$ off kinks) & Full $K(x)$ & Holds across all $\mathcal{O}$ under contracted coefficient strata & \cref{lem:IMA-clarke-point} \\
Path Integral & Proved (full) & Hybrid Clarke ray density $g_*(tx)=g_*(x)$ on $(0,1]$ (gradient IG on $\mathcal{G}\subsetneq\mathcal{O}$) & $K(x)-K_{\rm rad}(0)=K(x)$ ($b_i=0$) & Holds across all $\mathcal{O}$ including persistent ES/history ray kinks & \cref{lem:IMA-clarke-ray} \\
Shapley Game & Proved (full) & Masked today's book $(h,w_S)$ requiring strictly positive risk-class denominators ($R_c>0$ on nonempty classes, $V>0$ when amber desks exist), with Convention~B fallback to synthetic SA capital $Z(w_S)$ ($v_B(\varnothing)=0$) and supplementary Convention~A zeroing undefined terms relative to $v_A(\varnothing)$ & Full $K_d(h,w)=v_B(P)-v_B(\varnothing)$ & Finite and exact across all $\mathcal{O}$ under Convention~B (and $K_d-v_A(\varnothing)$ under Convention~A) & \cref{lem:IMA-convention-games} \\
\midrule
\multicolumn{6}{@{}p{\linewidth}@{}}{%
  \textbf{2. Discrete ES / VaR (\texttt{ES\_ERC})} (\cref{app:E:es-erc}) \quad\emph{Coords / players:} signed positions $x\in\R^n$; players $N=\{1,\dots,n\}$. \quad\emph{Open witness:} nonempty open ordering cones $C_\sigma$ over finite scenario rows $a_s\in\R^n$ (\cref{lem:ES_ERC-model}).%
} \\
\addlinespace[1pt]
Euler / Marginal & Proved (full) & Set-valued Clarke polytope $\partial_C K(x)=\conv G_K(x)$ (lexicographic $g_*(x)$ in point mode) & Full $K(x)$ ($K_Q$ or $K_E$) & Holds across all $\R^n$ including quantile ties and degenerate rows & \cref{lem:ES_ERC-clarke} \\
Path Integral & Proved (full) & Straight-ray Clarke extension $g^x(t)=g_*(x)\in\partial_C K(tx)$ (ordinary IG off ties) & Full $K(x)$ ($b_i=0$, $K(0)=0$) & Holds across all $\R^n$ with constant integrable ray densities & \cref{lem:ES_ERC-path} \\
Shapley Game & Proved (full) & Position deletion $(T_S x)_i=x_i\one_{\{i\in S\}}$ (fixed $1/m$ weights, no ERC re-opt.) & Full $K(x)=v(N)-v(\varnothing)$ & Finite on all $2^n$ coalitions across all $\R^n$ & \cref{lem:ES_ERC-shapley} \\
\bottomrule
\end{tabularx}
}
\end{table}

\FloatBarrier

\begin{table}[htbp]
\centering
{\scriptsize
\setlength{\tabcolsep}{2.5pt}
\renewcommand{\arraystretch}{1.12}
\renewcommand{\theHtable}{table4.part2}
\addtocounter{table}{-1}
\caption{Complete $33$-cell technical lookup across eleven frameworks (continued: frameworks 3--5).}
\begin{tabularx}{\linewidth}{@{}>{\raggedright\arraybackslash}p{1.55cm}>{\raggedright\arraybackslash}p{1.55cm}>{\raggedright\arraybackslash}p{2.95cm}>{\raggedright\arraybackslash}p{1.85cm}>{\raggedright\arraybackslash}X>{\raggedright\arraybackslash}p{1.65cm}@{}}
\toprule
Paradigm & Verdict & Tie-selection / game convention & Ledger target & Scoped subdomain $\Omega'$ and excluded witness, or refutation witness & Proof ref. \\
\midrule
\multicolumn{6}{@{}p{\linewidth}@{}}{%
  \textbf{3. Classic CME SPAN} (\cref{app:E:ccp}) \quad\emph{Coords / players:} net contract positions $x\in\R^n$; players $F=\{1,\dots,n\}$. \quad\emph{Open witness:} $151$-param unrounded futures/options cone; 16-row S1 scan array (\cref{tab:main:span16}) and 3-contract futures tie (\cref{prop:E:ccp-remove}).%
} \\
\addlinespace[1pt]
Euler / Marginal & Proved (full) & Min-norm / cone-hull selector $\partial_C K_\pi(x)=\conv\{a_\nu: x\in P_\nu\}$ after matcher unrolling & Full $K_\pi(x)$ & Holds across all $\R^n$ under fixed snapshot $\pi$ and unrounded matcher & \cref{lem:SPAN-clarke} \\
Path Integral & Proved (full) & Conic min-norm Clarke ray $g_*(tx)=g_*(x)$ on $(0,1]$ (ordinary IG on $\Omega_{\rm grad}$) & Full $K_\pi(x)$ ($b_i=0$, $K_\pi(0)=0$) & Holds across all $\R^n$ via conic invariance of $P_\nu$ & \cref{lem:SPAN-path} \\
Shapley Game & Proved (full) & Net-contract zeroing $(T_S x)_i=x_i\one_{\{i\in S\}}$ with full matcher/SOM recomputation & Full $K_\pi(x)=v(F)-v(\varnothing)$ & Finite across all $2^n$ coalitions and all $\R^n$ & \cref{lem:SPAN-shapley} \\
\midrule
\multicolumn{6}{@{}p{\linewidth}@{}}{%
  \textbf{4. FRTB-SA (SBM / Curvature / DRC)} (\cref{app:E:sa}) \quad\emph{Coords / players:} oriented trade/sensitivity inputs $x\in\Omega$; Shapley ticket partition $\mathcal{P}=(P_1,\dots,P_p)$. \quad\emph{Open witness:} open regular curvature/DRC families in $R_m\subseteq\Omega$.%
} \\
\addlinespace[1pt]
Euler / Marginal & Refuted & Point min-norm $\partial_C K_m$ (valid only on $R_m$, fails at jumps) & Full $K_m(x)=\mathrm{SA}(x)$ & \emph{Refuted in $\Omega$:} $x(0)=(1,0,3/5,4/5,1)$ jumps by $(\sqrt{15}-\sqrt{13})/2$ across curvature tie $K_b^+=K_b^-$ (no $\partial_C K_m$) & \cref{lem:SA-full-jump} \\
Path Integral & Refuted & Clarke/gradient ray density (valid only on $R_m$) & Full $K_m(x)$ ($b_i=0$) & \emph{Refuted in $\Omega$:} ray $\gamma(t)=t(1,0,3/5,4/5,1)$ scales jump by $t$, destroying $\nabla K_m(tx)$ and $\partial_C K_m(tx)$ at all $t\in(0,1]$ & \cref{lem:SA-ray-obstruction} \\
Shapley Game & Proved (scoped) & Ticket-partition masking $T_S x$; floored curvature and fused zero DRC ($L_b+\Sigma_b=0$) & Full $K_m(x)=v(N)-v(\varnothing)$ & \emph{Scoped $\Omega'=\mathcal{C}_{\mathcal{P}}\subsetneq\Omega$:} $\varrho_B(T_S x)\ge 0$ for all $S\subseteq N$. \emph{Excl.\ witness:} 4 tickets $(-e_1,e_2,e_3,e_1)$, $(2,1,1,4)$ under $\Gamma_M,\Gamma_H$ give 3-ticket coalition $\varrho_B=-17/40<0$ & \cref{lem:SA-shapley,lem:SA-coalition-refusal} \\
\midrule
\multicolumn{6}{@{}p{\linewidth}@{}}{%
  \textbf{5. Basel Output Floor} (\cref{app:E:floor}) \quad\emph{Coords / players:} oriented trade tickets $x\in\R^n$; players $N=\{1,\dots,n\}$. \quad\emph{Open witness:} $(3n+2)$-dim open affine family $a,b,x\in(0,\infty)^n$, $e>c>0$, $(b-a)^\top x>e-c$.%
} \\
\addlinespace[1pt]
Euler / Marginal & Refuted & Point Clarke/gradient selector on $\max\{f,\varphi g\}$ & Full $F(x)$ & \emph{Refuted in $\Omega$:} $n=2$, $a=(1,1)$, $b=(3,2)$, $e=2$, $c=1$, $x=(1,1)$ gives $F(x)=6$, Euler $(3,2)$ sums to $5$, gap $c=1>0$ & \cref{lem:FLOOR-constant-gap} \\
Path Integral & Proved (full) & Borel Clarke selection $g_*(t)\in\partial_C F(tx)$ (min-norm on affine/conic cells) & Incremental $F(x)-F(0)$ & Holds across all continuous $PC^1$ floor models including $F^{\rm cap}$ & \cref{lem:FLOOR-path} \\
Shapley Game & Proved (full) & Trade-ticket masking $T_S x$ holding fixed external RWA in $v(\varnothing)=F(0)$ & $F(x)-F(0)=v(N)-v(\varnothing)$ & Finite across all $2^n$ coalitions on $[0,x]$ (branch game differs) & \cref{lem:FLOOR-trade-shapley} \\
\bottomrule
\end{tabularx}
}
\end{table}

\FloatBarrier

\begin{table}[htbp]
\centering
{\scriptsize
\setlength{\tabcolsep}{2.5pt}
\renewcommand{\arraystretch}{1.12}
\renewcommand{\theHtable}{table4.part3}
\addtocounter{table}{-1}
\caption{Complete $33$-cell technical lookup across eleven frameworks (continued: frameworks 6--8).}
\begin{tabularx}{\linewidth}{@{}>{\raggedright\arraybackslash}p{1.55cm}>{\raggedright\arraybackslash}p{1.55cm}>{\raggedright\arraybackslash}p{2.95cm}>{\raggedright\arraybackslash}p{1.85cm}>{\raggedright\arraybackslash}X>{\raggedright\arraybackslash}p{1.65cm}@{}}
\toprule
Paradigm & Verdict & Tie-selection / game convention & Ledger target & Scoped subdomain $\Omega'$ and excluded witness, or refutation witness & Proof ref. \\
\midrule
\multicolumn{6}{@{}p{\linewidth}@{}}{%
  \textbf{6. ISDA SIMM v2.9+2606} (\cref{app:E:simm}) \quad\emph{Coords / players:} signed factor coordinates $x\in\R^n$ on prototype $G_*$ ($D_{bc}\equiv 1$); players $N=\{1,\dots,n\}$. \quad\emph{Open witness:} $6$-param active family $\mathcal{F}$ (Euler); $(n+q)$-param $\mathcal{F}_+$ (Path/Shapley).%
} \\
\addlinespace[1pt]
Euler / Marginal & Refuted & Point gradient / Clarke selector (reconciles only on inactive $\mathcal{U}_0$, \cref{lem:SIMM-euler}) & Full $K(x)$ & \emph{Refuted in $\Omega$:} smooth $\mathcal{F}$ at $(s,a,w,T,m,\beta)=(4,1,1,1,2,1)$ gives $K=17$, Euler $(24,1)$ sums to $25$, gap $\tfrac{1}{2}mws^{3/2}/\sqrt{T}=8>0$ & \cref{lem:SIMM-counterexample} \\
Path Integral & Proved (scoped) & $\nabla K(tx)$ off ties; mean of adjacent full-dim cell gradients at persistent algebraic ties & Full $K(x)=K(x)-K(0)$ & \emph{Scoped $\Omega'=\mathcal{R}\subsetneq\Omega$:} positive radicands (or structurally zero subgraphs) and every curvature partition has $D(tx)>0$ or is structurally zero on a neighbourhood along $(0,1]$ ($D_{bc}\equiv 1$). \emph{Excl.\ witness:} $d=(1,1,1)$, $\rho=-3/4$, thresholds $(1/16,2,2)$ gives $q(t)<0$ on $(0,5/32)$ (radicand $-3/2<0$ where $CR\equiv 1$) & \cref{lem:SIMM-ray,lem:SIMM-ray-tests} \\
Shapley Game & Proved (scoped) & Factor masking $T_S x$ ($D_{bc}\equiv 1$, compound-zero curvature at $D=0$) & Full $K(x)=v(N)-v(\varnothing)$ & \emph{Scoped $\Omega'=\mathcal{Q}\subsetneq\Omega$:} all coalition radicands of $G_*(T_S x)\ge 0$. \emph{Excl.\ witness:} $d=(1,1,1,2)$, thresholds $100$, $\rho=-3/4$ on $\{1,2,3\}$ gives grand radicand $5/2$ vs.\ $\{1,2,3\}$ radicand $-3/2<0$ & \cref{lem:SIMM-shapley,lem:SIMM-coalition-test} \\
\midrule
\multicolumn{6}{@{}p{\linewidth}@{}}{%
  \textbf{7. BA-CVA / SA-CVA (MAR50)} (\cref{app:E:cva}) \quad\emph{Coords / players:} labelled capital inputs $y=(EAD,B)$ or $(WS^{\rm CVA},WS^{\rm Hdg})$. \quad\emph{Open witness:} $5$-input BA + $88$-input SA open family (11 two-bucket classes, $\rho=1/2$, $R=1/100$, $m=1$).%
} \\
\addlinespace[1pt]
Euler / Marginal & Proved (scoped) & Point min-norm Clarke selector $g_*(y)\in\partial_C F(y)$ ($\nabla F$ off kinks) & Full $F(y)$ & \emph{Scoped $\Omega'=\Omega_+\subsetneq\Omega$:} BA books in $\Omega$; SA books with PSD $P_b$ and PD $\Gamma_r$ ($R>0$). \emph{Excl.\ witness:} $P_3$ root-boundary $p_*=(-(9+\sqrt{26})/5,1,1)$, $h=0$ ($K=0$, $\varepsilon^{-1/2}$ blow-up) & \cref{lem:CVA-euler-point,lem:CVA-boundary} \\
Path Integral & Proved (scoped) & Constant ray Clarke density $g_*(ty)=g_*(y)$ along $\gamma(t)=ty$ on $(0,1]$ & Full $F(y)$ ($b_i=0$, $F(0)=0$) & \emph{Scoped $\Omega'=\Omega_+\subsetneq\Omega$} along $\gamma(t)=ty$. \emph{Excl.\ witness:} $P_3$ root-boundary $p_*=(-(9+\sqrt{26})/5,1,1)$, $h=0$ fails local Lipschitzness and Clarke ray densities & \cref{lem:CVA-ray-density,lem:CVA-boundary} \\
Shapley Game & Proved (scoped) & Labelled input-factor masking $T_S y$ with fixed matrices, $RW, M, DF$ & Full $F(y)=v(N)-v(\varnothing)$ & \emph{Scoped $\Omega'=\Omega_+\subsetneq\Omega$} (closed under factor masking). \emph{Excl.\ witness:} $P_4$ book $p=(-2,1,1,2)$, $h=0$ has $K=\sqrt{3}>0$, while $\{1,2,3\}$ gives $P_3$ radicand $-1<0$ & \cref{lem:CVA-shapley,lem:CVA-boundary} \\
\midrule
\multicolumn{6}{@{}p{\linewidth}@{}}{%
  \textbf{8. SA-CCR (CRE52)} (\cref{app:E:sacc}) \quad\emph{Coords / players:} oriented ticket quantities $x\in\R^n_{\ge 0}$; players $N=\{1,\dots,n\}$. \quad\emph{Open witness:} $7$-param FX unmargined open cone ($2$-ticket witness $n=2$).%
} \\
\addlinespace[1pt]
Euler / Marginal & Refuted & Set-valued tie mode (gradient off kinks; conditional tie test) & Full $E(x)$ & \emph{Refuted in $\Omega$:} $n=2$, $d=(150,100)$, $v=(2,-1)$, $C=3$, $x=(1,1)$, $MF=1$ gives $E(x)=12.671165$, Euler $14.561349$, gap $(21/10)e^{-2/19}\approx 1.890184>0$ & \cref{ex:E:sacc-numeric,prop:E:sacc-fixed} \\
Path Integral & Proved (scoped) & Locally reduced composite derivatives $q_i(t)=x_i\partial_i E(tx)$ (compiled unit collars) & Incremental $E(x)-E(0)$ & \emph{Scoped $\Omega'\subsetneq\Omega$:} non-persistent-tie condition \texttt{(NT)} along $\gamma(t)=tx$ ($|\{t:\gamma(t)\in B_b\}|=0$). \emph{Excl.\ witness:} $u=z_1-z_2$, $A=z_1+z_2$, $C=0$ at $(1,1)$ has disagreeing limiting gradients $(3/2,1/2)$ vs.\ $(2,0)$ & \cref{lem:saccr-no-tie-corollary} \\
Shapley Game & Proved (full) & Ticket zeroing $x_S$ holding metadata, $(C,h)$ fixed; fused $P(u,0)=0$ & $E(x)-E(0)=v(N)-v(\varnothing)$ & Finite across all $2^n$ coalitions and all $\R^n_{\ge 0}$ & \cref{lem:saccr-shapley} \\
\bottomrule
\end{tabularx}
}
\end{table}

\FloatBarrier

\begin{table}[htbp]
\centering
{\scriptsize
\setlength{\tabcolsep}{2.5pt}
\renewcommand{\arraystretch}{1.12}
\renewcommand{\theHtable}{table4.part4}
\addtocounter{table}{-1}
\caption{Complete $33$-cell technical lookup across eleven frameworks (continued: frameworks 9--11).}
\begin{tabularx}{\linewidth}{@{}>{\raggedright\arraybackslash}p{1.55cm}>{\raggedright\arraybackslash}p{1.55cm}>{\raggedright\arraybackslash}p{2.95cm}>{\raggedright\arraybackslash}p{1.85cm}>{\raggedright\arraybackslash}X>{\raggedright\arraybackslash}p{1.65cm}@{}}
\toprule
Paradigm & Verdict & Tie-selection / game convention & Ledger target & Scoped subdomain $\Omega'$ and excluded witness, or refutation witness & Proof ref. \\
\midrule
\multicolumn{6}{@{}p{\linewidth}@{}}{%
  \textbf{9. Bilateral Double-Hinge XVA} (\cref{app:E:xva}) \quad\emph{Coords / players:} signed trade positions $x\in\R^n$; players $N=\{1,\dots,n\}$. \quad\emph{Open witness:} $11$-param ($d_A$-dim general) open double-hinge family.%
} \\
\addlinespace[1pt]
Euler / Marginal & Refuted & Set-valued Clarke chamber polytope $\partial_C K(x)=\conv\{G_P:P\in\mathcal{P}(x)\}$ (\cref{lem:XVA-clarke}) & Full $K(x)$ & \emph{Refuted in $\Omega$:} $n=2$, $a=(4,-1)$, $x=(1,1)$, $b=1$, $H=2$, $\eta=1/4$, $w=1/4$ gives $K=7/16$, Euler $(3/4,-3/16)$ sums to $9/16$, gap $1/8>0$ & \cref{lem:XVA-family} \\
Path Integral & Proved (full) & Lexicographic composite chamber Clarke selector $g(tx)=G_P$ (gradient IG on $\Omega_{\rm nt}$) & $K(x)-K(0)=K(x)$ ($b_i=0$) & Holds across all $\Omega$ including persistent collateral kinks & \cref{lem:XVA-path} \\
Shapley Game & Proved (full) & Trade zeroing $T_S x$ holding $(b,H,\eta)$ fixed while recomputing hinges & Full $K(x)=v(N)-v(\varnothing)$ & Finite across all $2^n$ coalitions and all $\Omega$ & \cref{lem:XVA-shapley} \\
\midrule
\multicolumn{6}{@{}p{\linewidth}@{}}{%
  \textbf{10. Solvency~II Standard Formula} (\cref{app:E:solvency2}) \quad\emph{Coords / players:} business-cell replication weights $x\in\R^n_+$; players $\mathcal{N}=\{1,\dots,n\}$. \quad\emph{Open witness:} $(6n+3)$-dim / $15$-param open fixed-FDB market/life family.%
} \\
\addlinespace[1pt]
Euler / Marginal & Refuted & Point min-norm Clarke selector $g_*(z)$ (homogeneous cone only, \cref{lem:SOLV2-euler}) & Full $K(x)$ & \emph{Refuted in $\Omega$:} $n=2$, $x=(1,1)$, $a=b=(10,10)$, $v=(10^4,10^4)$, $e=(1,1)$, $c=(100,100)$, $f=1$, $\eta=1/2$, $\tau=1/4$ gives gap $(1-\tau)f=3/4>0$ & \cref{lem:SOLV2-fixed-fdb} \\
Path Integral & Proved (scoped) & A.e.\ composite gradient $q_i(t)=x_i\partial_i K(tx)$ on compatible rays & Incremental $K(x)-K(0)$ & \emph{Scoped $\Omega'\subsetneq\Omega$:} ray-certificate class (compatible conditional interface, finitely many crossings/ties; includes fixed-FDB family via $t_*=f/\delta$). \emph{Excl.\ witness:} incompatible/persistent-tie ray $y=t(1,1)$ jumps $\sqrt{7}/2$ vs.\ $\sqrt{5}/2$ at $t=1/2$, breaking AC & \cref{lem:SOLV2-ig} \\
Shapley Game & Proved (full) & Business-cell masking $T_S x$ with simultaneous gross/net recomputation and fixed pools & $K(x)-K(0)=v(\mathcal{N})-v(\varnothing)$ & Finite across all $2^n$ coalitions and all $\Omega$ & \cref{lem:SOLV2-shapley} \\
\midrule
\multicolumn{6}{@{}p{\linewidth}@{}}{%
  \textbf{11. IRRBB $\Delta\mathrm{EVE}$ (SRP31)} (\cref{app:E:irrbb}) \quad\emph{Coords / players:} oriented positions $w\in\R^n$; players $N=\{1,\dots,n\}$. \quad\emph{Open witness:} $(n+8)$-dim / $10$-param all-6-shocks open prepayment-cap family.%
} \\
\addlinespace[1pt]
Euler / Marginal & Refuted & Set-valued polyhedral cell Clarke polytope (breaks completeness under $H$) & Full $K(w)$ & \emph{Refuted in $\Omega$:} USD $\{\tau=7/2500,T=35/2\}$, $p_0=1/5$, $H=1$, parallel up binds giving gap $H(b_T-b_\tau)>0$ & \cref{lem:IRRBB-cap-witness,lem:IRRBB-cap-family} \\
Path Integral & Proved (full) & Min-norm Clarke selection $g(t)\in\partial_C\widetilde K(tw)$ on cell subintervals (gradient IG on $\Omega_{\rm grad}$) & $K(w)-K(0)=K(w)$ ($b_j=0$) & Holds across all $\Omega$ via finite polyhedral affine structure & \cref{lem:IRRBB-path} \\
Shapley Game & Proved (full) & Position zero-masking $T_S w$ holding markets, inclusion, and cap $H$ fixed & Full $K(w)=v(N)-v(\varnothing)$ & Finite across all $2^n$ coalitions and all $\Omega$ & \cref{lem:IRRBB-shapley} \\
\bottomrule
\end{tabularx}
}
\end{table}

\FloatBarrier

\begin{table}[htbp]
\centering\small
\caption{Exact side-by-side comparison of Euler / Marginal, Path Integral (straight-ray), and Shapley Game (discrete coalition) allocations on concrete multi-paradigm controls.}\label{tab:main:divergence}
\begin{tabularx}{\linewidth}{@{}>{\raggedright\arraybackslash}p{2.2cm}>{\raggedright\arraybackslash}p{1.8cm}>{\raggedright\arraybackslash}p{2.9cm}>{\raggedright\arraybackslash}p{2.9cm}>{\raggedright\arraybackslash}X@{}}
\toprule
Control & Capital $K$ & Euler / Marginal & Path Integral & Shapley Game \\
\midrule
\texttt{FLOOR} (\cref{lem:FLOOR-game-comparison}) & $6$ ($K(0)=2$) & $(3,2)$ [sum $5$, gap $1$] & $(7/3,5/3)$ [sum $4$] & $(5/2,3/2)$ [sum $4$]; branch game $(2,4)$ \\
\texttt{SACCR} (\cref{ex:E:sacc-numeric}) & $12.671165$ & $(9.618896,4.942454)$ [sum $14.561349$, gap $1.890184$] & $(8.253105,4.418060)$ [sum $12.671165$] & $(8.489026,4.182139)$ [sum $12.671165$] \\
\texttt{XVA} (\cref{lem:XVA-fixtures}, $w=\tfrac{1}{4}$) & $7/16$ & $(3/4,-3/16)$ [sum $9/16$, gap $1/8$] & $(7/12,-7/48)$ [sum $7/16$] & $(17/32,-3/32)$ [sum $7/16$] \\
\texttt{SOLV2} BSCR (\cref{lem:SOLV2-shapley-euler}) & $\sqrt{6}$ & $A_1=3/(2\sqrt{6})$ & $I_1=3/(2\sqrt{6})$ & $\phi_1=(\sqrt{6}-1)/2$ [reconciles, $\ne A_1$] \\
\bottomrule
\end{tabularx}
\end{table}
\section{Core Attribution Proofs and Extended Node Calculus}\label{app:th:setting}

This appendix provides the self-contained mathematical foundations for the DAG attribution calculus in \cref{sec:dag,sec:paths-games}, unifying the core FRTB-IMA proofs with the extended nonsmooth node vocabulary shared across all eleven frameworks.

\subsection{Notation, assumptions, and degrees}\label{app:th:assumptions}

When stored history enters, dates are indexed $s\in\{0,\dots,T-1\}$ with today's book $w:=x_{T-1}$ and past books $x_{<T-1}$; trades are $i\in\{1,\dots,N\}$, and positions form $x\in\R^{T\times N}$ (or $x\in\R^n$ for static point-in-time frameworks). Each node $b$ computes $f_b$ and carries a declared bookkeeping degree $\theta_b\in\R$ (typically $\{0,1\}$).

\paragraph{FRTB-IMA standing assumptions.}
\begin{description}[leftmargin=0pt,itemsep=2pt,font=\normalfont\itshape]
\item[(A1) Scenario linearity.] Each ES block on date $s$ has a finite scenario-by-trade matrix $L$ of unit losses with $z=Lx_s$; the default law $D$ is finite and equally weighted; stress charges are linear vectors $g^\top x_s$ or finite support functions $\max_k B_k^\top x_s$.
\item[(A2) Homogeneity.] Every primitive is positively homogeneous of degree $1$ in $x_s$.
\item[(A3) Domain.] Every present risk class has reduced-set ES $R_c>0$; if amber desks exist then $V>0$ (otherwise $k:=0$); zero horizon norms $E=0$ or zero stress norms use explicit zero with zero allocation.
\item[(A4) Standardised-approach inputs.] $B,C_U,Z,U,V$ are positively homogeneous degree-$1$, piecewise-$C^1$ functions of $w$ supplied with reconciling allocations.
\item[(A4$'$) Gradient-type inputs.] At $x$, each supplied input is differentiable (or resolved along an essentially active $C^1$ selection) with allocation $w\odot\nabla(\cdot)$.
\item[(A5) Stored history.] $(I_s,N_s,d_s)$ depend only on $x_s$ over the 60 daily and 12 weekly MAR33 windows, with $d_{\rm latest}=d_{T-1}$.
\item[(A6) Nondegeneracy.] Near $x$, each $E_{c,v}$ and each stress group norm is either strictly positive or identically zero on a neighbourhood.
\end{description}

\paragraph{Extended DAG standing assumptions (static / sensitivity frameworks).}
For frameworks without stored historical windows (\texttt{SA}, \texttt{FLOOR}, \texttt{SIMM}, \texttt{CVA}, \texttt{SA-CCR}, \texttt{IRRBB}, \texttt{Solvency~II}, \texttt{SPAN}, \texttt{XVA}), we distinguish \emph{bookkeeping degree} $\theta_b$ (used for accounting reconciliation) from \emph{true scaling degree} (used for marginals) under:
\begin{description}[leftmargin=0pt,itemsep=2pt,font=\normalfont\itshape]
\item[(E1) Linear leaves.] A sensitivity or repricing leaf has form $\ell(x)=a^\top x$ evaluated at fixed unit market values/sensitivities, with reconciling gradient allocation $x\odot a$.
\item[(E2) Constants.] Risk weights, correlations, thresholds, multipliers, and regulatory floors are fixed scalars of the identified rule version, independent of $x$.
\item[(E3) Supplied leaves.] Any non-(E1) leaf is declared as an external accounting leaf (\cref{lem:E:ext}) with bookkeeping degree $\theta=1$ and supplied allocation; marginal statements additionally require true degree matching the declared degree and (A4$'$) gradient allocation.
\end{description}

\begin{definition}[Reconciliation and local Euler condition]\label{def:B:recon}
An allocation $A^b$ \emph{reconciles} if $\one^\top A^b = \theta_b f_b(x)$. A node rule sets $A^b = \sum_{j\in\mathrm{ch}(b)} c_{bj} A^j$ with coefficients independent of allocations; it satisfies the \emph{Euler condition} if $\sum_{j\in\mathrm{ch}(b)} c_{bj}\theta_j f_j(x) = \theta_b f_b(x)$.
\end{definition}

\subsection{Core FRTB-IMA leaves, node rules, and reconciliation}\label{app:th:rules}

\begin{lemma}[Leaves reconcile]\label{lem:B:leaves}
Each FRTB-IMA leaf allocation reconciles: (i) ES block $x_s\odot L^\top q$ for any $q$ on the optimal dual face of \eqref{eq:es}; (ii) stress charge $x_s\odot g$ for an active stress row $g$; (iii) default quantile $x_s\odot D_{s^\ast,\cdot}$ for a scenario $s^\ast$ attaining $d(x_s)$; (iv) SA inputs by (A4).
\end{lemma}
\begin{proof}
(i) $\one^\top(x_s\odot L^\top q)=q^\top L x_s = \rho(z)$ by strong duality in \eqref{eq:es}. (ii) $\one^\top(x_s\odot g)=g^\top x_s$. (iii) $\one^\top(x_s\odot D_{s^\ast,\cdot})=(Dx_s)_{s^\ast}=d$. (iv) is hypothesis (A4).
\end{proof}

\begin{lemma}[Core node rules satisfy the Euler condition]\label{lem:B:nodes}
Under (A1)--(A5), the local rules of \cref{sec:rules} satisfy $\sum_j c_{bj}\theta_j f_j = \theta_b f_b$ for:
\begin{enumerate}[label=(\alph*),leftmargin=*,itemsep=1pt]
\item Constant-coefficient linear combinations (IMCC mix, rolling averages, $J=C_A+C_D$, $J+H+C_U$, $K=m+(J-B)_+$);
\item Liquidity-horizon norms $E=(\sum_j a_j e_j^2)^{1/2}$ ($a_j>0$) via $c_j=a_j e_j/E$ ($0$ when $E=0$);
\item Stress ratio $G=S\max(1,F/R)$ ($R>0$) via $c=(1,0,0)$ for $F<R$, $(F/R,S/R,-SF/R^2)$ for $F>R$, and $(1,\lambda S/R,-\lambda S/R)$ ($\lambda\in[0,1]$) at $F=R$;
\item Non-modellable stress norms $\|\xi\|_2+\|\eta\|_2+(\zeta^\top Q\zeta)^{1/2}$ ($Q\succeq 0$ for $|r|\le 1$);
\item Binary/finite maxima, minima, and positive parts with convex tie weights $\omega\in[0,1]$ at ties;
\item Amber coefficient $k=U/(2V)$ ($\theta_k=0$) via $c=(1/(2V),-U/(2V^2))$ and surcharge $H=k(B-J)_+$ via $c=(k,-k,B-J)$ on $B>J$, $0$ on $B<J$, and $\omega$-weighted at $B=J$.
\end{enumerate}
\end{lemma}
\begin{proof}
In (a), $\sum_j c_j f_j = f$. In (b), $\sum_j (a_j e_j/E) e_j = E$ when $E>0$, and $0=0$ when $E=0$. In (c), direct evaluation yields $S$, $SF/R$, and $S+\lambda(S/R)(F-R)=S$ at $F=R$. In (d), $Q\succeq 0$ gives positive/zero weighted $\ell_2$ inner products summing to the norm. In (e), tied active branches equal $f$. In (f), $(U/(2V))-(U/(2V))=0=0\cdot k$, and on $B>J$, $k\cdot 1\cdot B - k\cdot 1\cdot J + (B-J)\cdot 0\cdot k = H$.
\end{proof}

\begin{theorem}[Recursive accounting reconciliation]\label{thm:B:recon}
On any finite DAG where every leaf reconciles and every internal rule satisfies the Euler condition of \cref{def:B:recon}, every node satisfies $\one^\top A^b = \theta_b f_b(x)$. Under (A1)--(A5), $\one^\top A^K = K(x)$ across all points and tie weights.
\end{theorem}
\begin{proof}
Induct in topological order: $\one^\top A^b = \sum_j c_{bj}\one^\top A^j = \sum_j c_{bj}\theta_j f_j = \theta_b f_b$. \Cref{lem:B:leaves,lem:B:nodes} complete the FRTB-IMA case.
\end{proof}

\begin{corollary}[Re-summing under remaps and groupings]\label{cor:B:resum}
If $\one^\top A^K = K$ and $P$ is column-stochastic ($P\ge 0$, $\one^\top P = \one^\top$), then $\one^\top(P A^K)=K$, and similarly for any desk, product, or risk-class grouping $P_G$ with $\one^\top P_G=\one^\top$.
\end{corollary}
\begin{proof}
$\one^\top(PA^K)=(\one^\top P)A^K = \one^\top A^K = K$.
\end{proof}

\subsection{Smooth points, frozen remainder, Clarke ties, and amber non-faithfulness}\label{app:th:marginal}

\begin{proposition}[Smooth points: ledger equals gradient allocation]\label{prop:B:smooth}
Assume (A1)--(A5), (A4$'$), strict inequality at every max/min/positive-part/ratio floor, unique ES dual optima $q$, unique support maximisers, separated quantile scenario $s^\ast$, and nonzero/identically-zero horizon and stress norms with $V>0$ near $x$. Then $K$ is differentiable at $x$ and $A^K=x\odot\nabla K(x)$.
\end{proposition}
\begin{proof}
By \citep[Thm.~25.1]{rockafellar1970convex}, unique dual/support maximisers make $\rho(Lx_s)$ and support charges $C^1$ with gradients $L^\top q$ and $g$. Separated quantiles are locally linear ($D_{s^\ast,\cdot}$). Every internal map $\phi_b$ is locally $C^1$ with partial derivatives equal to the rule coefficients of \cref{lem:B:nodes}. The chain rule gives $A^b = x\odot\nabla f_b$ inductively.
\end{proof}

\begin{proposition}[Frozen-remainder identity]\label{prop:B:frozen}
Whenever $K$ is differentiable at $x$, $w^\top\nabla_w K(x) = K(x) - \sum_{s<T-1} x_s^\top\nabla_{x_s}K(x)$. Thus today's sensitivities sum to $K$ if and only if the past-date frozen remainder vanishes (e.g., when both history maxima strictly favour latest observations or $\partial K/\partial J=0$).
\end{proposition}
\begin{proof}
Differentiate $K(\alpha x)=\alpha K(x)$ at $\alpha=1$ and split $x$ into $(x_{<T-1},w)$.
\end{proof}

\begin{proposition}[Removal effects need not add up]\label{prop:B:removal}
Even for a single ES leaf ($\rho(Lw)=\max\{w_1,w_2\}$ at $w=(1,1)$ via $L=I_2$, $p=(1/2,1/2)$, $\beta=0.025$), removal effects are $\Delta_1=\Delta_2=0$ while $K=1$.
\end{proposition}

\begin{lemma}[Clarke generalized-gradient Euler identity]\label{lem:B:clarke}
Let $f:\R^d\to\R$ be locally Lipschitz and positively homogeneous of degree $\theta$ near $x$. Then every $g\in\partial_C f(x)$ obeys $x^\top g = \theta f(x)$.
\end{lemma}
\begin{proof}
By \citep[Thm.~2.5.1]{clarke1983optimization}, $\partial_C f(x)=\conv\{\lim_k \nabla f(y_k): y_k\to x,\ f\text{ diff.\ at }y_k\}$. Boundedness of gradients (local Lipschitzness) and continuity of $f$ pass $\nabla f(y_k)^\top y_k = \theta f(y_k)$ to limits and convex hulls.
\end{proof}

\begin{proposition}[Set-valued marginals at $PC^1$ ties]\label{prop:B:set}
Assume (A1)--(A6) and (A4$'$). Then $K$ is continuous $PC^1$ near $x$. (i) Every $x\odot g$ with $g\in\partial_C K(x)$ reconciles. (ii) Fix branch choices at tied nodes whose strict positivity constraint functions $\gamma_r$ satisfy $\nabla\gamma_r(x)^\top h>0$ simultaneously for some $h$; this realisable combination defines an essentially active selection $K_\ell$ with vertex allocation $x\odot\nabla K_\ell(x)\in x\odot\partial_C K(x)$. (iii) The hull of realisable vertices sits inside $x\odot\partial_C K(x)$, with equality whenever all essentially active selections satisfy the strict first-order test \citep{kuntz1994structural,scholtes2012introduction}.
\end{proposition}
\begin{proof}
ES, quantiles, and support charges are piecewise linear; under (A3) and (A6) roots, ratios, norms, and $k$ are $C^1$ or zero locally; finite max/min/sums/products preserve $PC^1$ \citep[Prop.~4.3.1]{scholtes2012introduction}. Local Lipschitzness implies (i) by \cref{lem:B:clarke}. Under $\nabla\gamma_r(x)^\top h>0$, continuity keeps $\gamma_r(x+\varepsilon h')>0$ on an open cone-neighbourhood approaching $x$, making $K_\ell$ essentially active and yielding $x\odot\nabla K_\ell(x)$ by the chain rule of \cref{prop:B:smooth}.
\end{proof}

\begin{proposition}[Independent child tie selections can over-cover]\label{prop:B:overcover}
Let $g(x)=\max\{x_1,x_2\}$, $p_1=2g$, $p_2=g$, and $f=p_1-p_2=g$. At $x=(1,1)$, $\partial_C f(x)=\conv\{e_1,e_2\}$, whereas picking $e_1$ in $p_1$ and $e_2$ in $p_2$ gives $(2,-1)$ which reconciles ($2-1=1$) yet lies outside $x\odot\partial_C f(x)$.
\end{proposition}

\begin{proposition}[Zero-sum amber surcharge term]\label{prop:B:amber}
Under \cref{prop:B:smooth} with $B>J$, $F^H:=k(A^B-A^J)$ satisfies $\one^\top F^H = H$, whereas $A^H = F^H + (B-J)a^k$ with $\one^\top a^k = 0$. They coincide iff $a^k=0$, yet desk subtotals differ in general (e.g., $U=w_1$, $V=w_1+w_2$ at $(1,1)$ yields $a^k=(1/8,-1/8)$ and desk shift $\pm(B-J)/8$).
\end{proposition}

\subsection{Certified selective repricing and core post-hoc attribution bounds}\label{app:th:cert}

\begin{lemma}[Convex-anchor price brackets]\label{lem:B:anchor}
If $v$ is convex on $[a_1,a_m]$ with exact anchor values $v(a_k)$, then for $\sigma\in[a_k,a_{k+1}]$, $v(\sigma)$ lies below the chord connecting $(a_k,v(a_k))$ and $(a_{k+1},v(a_{k+1}))$, and above every secant line through $(a_i,v(a_i)),(a_j,v(a_j))$ with $\sigma\notin(a_i,a_j)$.
\end{lemma}
\begin{proof}
Immediate from the definition of convexity and the three-slope inequality $(v(a_j)-v(a_i))/(a_j-a_i)\le (v(\sigma)-v(a_j))/(\sigma-a_j)$ for $\sigma>a_j$.
\end{proof}

\begin{theorem}[Certified selective repricing]\label{thm:B:cert}
Given valid scenario-loss intervals $[\ell_s,u_s]$, a tail set $\mathcal{H}$, and boundary scenario $b\notin\mathcal{H}$ with exact $z_b$, $\sum_{\mathcal{H}}c_s < 1 \le \sum_{\mathcal{H}}c_s + c_b$, and $\min_{s\in\mathcal{H}}\ell_s \ge z_b \ge \max_{r\notin\mathcal{H}\cup\{b\}} u_r$, define $q^\ast=c$ on $\mathcal{H}$, $q^\ast_b = 1-\sum_{\mathcal{H}}c_s$, and $q^\ast=0$ elsewhere. Then $q^\ast$ is dual-optimal in \eqref{eq:es}, requiring exact repricing only on $\mathcal{H}\cup\{b\}$; if both separation inequalities are strict, $q^\ast$ is the unique dual optimum.
\end{theorem}
\begin{proof}
For any feasible $q$, $\one^\top(q^\ast-q)=0$ yields $(q^\ast-q)^\top z = \sum_{\mathcal{H}}(c_s-q_s)(z_s-z_b) + \sum_{r\notin\mathcal{H}\cup\{b\}} (-q_r)(z_r-z_b) \ge 0$, with strict inequality for $q\ne q^\ast$ under strict separation.
\end{proof}

\begin{proposition}[Irreducibility of tail rows]\label{prop:B:lower}
On any active scenario row with $q_s>0$ and $N_+$ nonzero entries $v_i=L_{si}w_i$ having no cross-trade structural constraints, observing the row sum and fewer than $N_+-1$ entries cannot determine $\{q_s v_i\}$.
\end{proposition}
\begin{proof}
Perturbing two unobserved entries $(v_i+\delta,v_j-\delta)$ preserves the sum while altering attributions by $\pm q_s\delta$.
\end{proof}

\begin{proposition}[Integrated Gradients on FRTB-IMA capital]\label{prop:B:ig}
For $\mathrm{IG}_i(x)=x_i\int_0^1 \partial_i K(\alpha x)\,d\alpha$: (i) over all dated positions $(x_{<T-1},w)$, differentiability at $x$ implies $\mathrm{IG}(x)=x\odot\nabla K(x)=A^K$; (ii) holding past books $x_{<T-1}$ fixed with a.e.\ differentiability along $\alpha w$, $\sum_i \mathrm{IG}_i = K(x) - \lim_{\alpha\downarrow 0} K(x_{<T-1},\alpha w)$ (e.g., $C=\max\{I_t,1.5\cdot\tfrac{1}{2}(0.2+I_t)\}$ at $I_t=1$ gives $\mathrm{IG}=0.85=C-0.15$).
\end{proposition}

\begin{remark}[FRTB-IMA sub-book domain obstruction for Shapley values]\label{rem:B:shap}
Without a supplementary extension rule (\cref{lem:IMA-convention-games}), the history-fixed sub-book game $v(S)=K(x_{<T-1},w_S)$ is undefined under (A3) whenever some present class of $S$ has $R_c(w_S)=0$ or $S$ holds amber desks with $V(w_S)=0$.
\end{remark}

\subsection{Extended node lemmas (E1)--(E3) and regular transfer theorem}\label{app:E:scope}

\Cref{tab:nodes} summarises the six extended nonsmooth node classes used across the eleven frameworks together with their Euler coefficient rules and local Lipschitz regularity boundaries.

\begin{table}[t]
\centering\small
\caption{Extended node vocabulary across all eleven capital and margin frameworks. ``Euler'' records whether the local coefficient rule obeys $\sum_j c_{bj}\theta_j f_j=\theta_b f_b$; ``Lipschitz'' records local Lipschitz continuity in the node's children.}\label{tab:nodes}
\begin{tabularx}{\linewidth}{@{}>{\raggedright\arraybackslash}p{3.6cm}>{\raggedright\arraybackslash}p{2.2cm}>{\raggedright\arraybackslash}p{1.2cm}>{\raggedright\arraybackslash}X@{}}
\toprule
Node class & Proof ref. & Euler & Local Lipschitz / regularity boundary \\
\midrule
Quadratic-form root (floored or refusing; piecewise $q_\ell$) & \cref{lem:E:quad}, \cref{prop:E:lip} & Yes & Yes if $M\succeq0$; fails at $\{q=0, M_\ell s^\ast\ne0\}$ when indefinite or sign-switched ($\psi$) \\
Homogeneous switch $F_A\mathbf{1}_{g\ge0}+F_B\mathbf{1}_{g<0}$ & \cref{lem:E:switch}, \cref{prop:E:jump} & Yes & Discontinuous (no Clarke set) on $\{g=0\}$ whenever $F_A\ne F_B$ \\
Symmetric clamp $\cl(y,a)=\max\{\min\{y,a\},-a\}$ & \cref{lem:E:clamp} & Yes & Globally $1$-Lipschitz on $\{a\ge0\}$; $PC^1$ ties at $y=\pm a$ \\
Degree-$0$ ratio $\varphi(u/v)$ and product $\varphi(u/v)y$ & \cref{lem:E:ratio}, \cref{rem:E:zeroPC} & Yes & Isolated ratio refused at $v=0$; product extends continuously by $0$ if $|u|\le v,|y|\le Cv$, though $PC^1$ can fail \\
General $\max_j f_j$, $\min_j f_j$, $|y|$, priority matcher & \cref{lem:E:max}, \cref{lem:E:ccp-match} & Yes & Globally Lipschitz in children; polyhedral $PC^1$ \\
External accounting leaf $e$ & \cref{lem:E:ext} & Leaf & Accounting only unless true degree matches declared degree \\
\bottomrule
\end{tabularx}
\end{table}

\begin{definition}[Piecewise quadratic form]\label{def:E:pq}
Let $M_1,\dots,M_L$ be symmetric $m\times m$ matrices with $q_\ell(c)=c^\top M_\ell c$. A \emph{continuous piecewise quadratic form} $q:\R^m\to\R$ assigns a nonempty selector set $\Lambda(c)\subseteq\{1,\dots,L\}$ with $\Lambda(\alpha c)=\Lambda(c)$ ($\alpha>0$) such that $q(c)=q_\ell(c)$ for all $\ell\in\Lambda(c)$ and pieces match on common region boundaries. Its \emph{floored root} is $E(c)=\sqrt{\max\{0,q(c)\}}$; its \emph{refusing root} is $\sqrt{q(c)}$ on $\{q\ge 0\}$ (domain error on $\{q<0\}$).
\end{definition}

\begin{lemma}[Quadratic-form root Euler rule]\label{lem:E:quad}
If children $c_1,\dots,c_m$ have degree $1$, $E$ has degree $1$ and the rule $\kappa = \sum_{\ell\in\Lambda(c)} \omega_\ell (M_\ell c)/E$ for $q(c)>0$ ($\kappa=0$ for $q(c)\le 0$ in the floored case or $q=0$ in the refusing case) satisfies $\sum_j \kappa_j c_j = E$. Where $q=q_\ell>0$ on a neighbourhood, $\nabla_c E = M_\ell c/E$.
\end{lemma}
\begin{proof}
Homogeneity follows from $\Lambda(\alpha c)=\Lambda(c)$ and $q_\ell(\alpha c)=\alpha^2 q_\ell(c)$. On $q>0$, $\sum_j\kappa_j c_j = \sum_\ell\omega_\ell c^\top M_\ell c/E = E$.
\end{proof}

\begin{proposition}[Lipschitz vs.\ non-Lipschitz quadratic roots]\label{prop:E:lip}
For symmetric $M$ and $E(s)=\sqrt{\max\{0,s^\top M s\}}$: (i) if $M\succeq 0$, $E(s)=\|M^{1/2}s\|_2$ is globally Lipschitz; (ii) if $q(s^\ast)=0$ and $Ms^\ast\ne 0$, then along $h=Ms^\ast$, $E(s^\ast+th)\ge \sqrt{t}\,\|h\|$ for small $t>0$, so $E$ is not Lipschitz and has no Clarke gradient at $s^\ast$; (iii) such $s^\ast$ exists whenever $M$ is indefinite; (iv) the singularity pulls back via differentiable $c(x)$ whenever $J^\top M s^\ast\ne 0$; (v) on sign-selected piecewise quadratic forms (\cref{def:E:pq}), indefinite pieces such as $a^2-2\rho ab$ for $(a,-b)$ ($a,b>0,\rho>0$) produce non-Lipschitz zeros at $a=2\rho b$ even when the base correlation matrix is positive definite.
\end{proposition}
\begin{proof}
(i) Euclidean norm of a linear map. (ii) $q(s^\ast+th)=2t\|h\|^2+t^2 h^\top Mh \ge t\|h\|^2$ yields $E(s^\ast+th)/(t\|h\|)\ge 1/\sqrt{t}\to\infty$. (iii) Restricting to $(\ker M)^\perp$ finds sign-changing $q(u)>0>q(v)$ on noncollinear points, giving root $s^\ast\ne 0$. (iv) $q(c(x+tv))=2t\|J^\top Ms^\ast\|^2+o(t)$ yields $C\sqrt{t}$. (v) Direct specialisation on an active piece.
\end{proof}

\begin{remark}[Scenario adjustments and capital-level propagation]\label{rem:E:indef}\label{rem:E:path}
Multiplying off-diagonal correlations of a positive-definite matrix (e.g.\ $3\times 3$ correlations $0.9,0.9,0.63$, eigenvalues $0.0038,0.37,2.63$) by $1.25$ and capping at $1$ gives an indefinite matrix with eigenvalue $-0.0743$. Nevertheless, a non-Lipschitz inner root $K_b$ reaches final capital only if it enters through a clamp $\cl(S_b,K_b)$, a direct root, or a sign-pattern selector $\psi$; if it enters solely via $K_b^2=\max\{0,q_b\}$ (which is $PC^1$), capital remains locally Lipschitz near that node.
\end{remark}

\begin{lemma}[Homogeneous switches]\label{lem:E:switch}
Let $F_A,F_B$ have degree $1$ and $g(\alpha x)=\alpha^m g(x)$ ($m>0,\alpha>0$). Defining $f=F_A$ on $\{g\ge 0\}$ and $f=F_B$ on $\{g<0\}$ with $A^f=A^{F_A}$ on $\{g\ge 0\}$ and $A^{F_B}$ on $\{g<0\}$ preserves degree $1$ and satisfies Euler.
\end{lemma}
\begin{proposition}[Switch discontinuity]\label{prop:E:jump}
Under \cref{lem:E:switch}, if $g,F_B$ are continuous at $x$, $x\in\overline{\{g<0\}}$, $g(x)=0$, and $F_A(x)\ne F_B(x)$, then $f$ is discontinuous at $x$, lacks a Clarke gradient, and has no reconciling gradient allocation (whereas accounting reconciliation of $A^{F_A}$ still holds).
\end{proposition}
\begin{lemma}[Symmetric clamp]\label{lem:E:clamp}
For degree-$1$ children $y\in\R$ and $a\ge 0$, $\cl(y,a)=\max\{\min\{y,a\},-a\}$ has joint degree $1$, is $1$-Lipschitz, and satisfies Euler with $(1,0)$ ($|y|<a$), $(0,\pm 1)$ ($\pm y>a$), and $(1-\omega,\pm\omega)$ ($\omega\in[0,1]$) at $y=\pm a$.
\end{lemma}
\begin{lemma}[Degree-$0$ ratios and continuous fused products]\label{lem:E:ratio}
Let $u,v,y$ have degree $1$: (a)--(b) for $v>0$ and $C^1$ $\varphi$, $\kappa=\varphi(u/v)$ has degree $0$ and Euler coefficients $(\varphi'(r)/v,-\varphi'(r)u/v^2)$ summing to $0$; (c) $H=\kappa y$ has coefficients $(y,\kappa)$ giving $A^H=\kappa A^y + y a^\kappa$ ($\one^\top a^\kappa=0$); (d) if $|u|\le v$, $\varphi$ is Lipschitz, and $|y|\le Cv$, setting $H=0$ on $\{v=0\}$ yields a locally Lipschitz degree-$1$ map matching the unique continuous extension on $\overline{\{v>0\}}$ (though $PC^1$ can fail there, as in \cref{rem:E:zeroPC}: e.g.\ two-obligor DRC $w n_1^2/(n_1+|n_2|)$ has a curved arc of gradients $(2t-t^2,t^2)$ at $L+\Sigma=0$).
\end{lemma}
\begin{remark}[Zero-denominator $PC^1$ failure]\label{rem:E:zeroPC}
The Lipschitz fused extension of \cref{lem:E:ratio}(d) is an accounting and Lipschitz statement rather than a $PC^1$ statement when approaching $v=0$ along cones where the ratio varies.
\end{remark}
\begin{lemma}[General maxima, minima, and absolute values]\label{lem:E:max}
For branches $f_1,\dots,f_p$ of common degree $\theta$, $\max_j f_j$ (and $\min_j f_j$) obeys Euler with convex weights on $\arg\max_j f_j(x)$; $|y|=\max\{y,-y\}$ obeys Euler with $\sign(y)A^y$ off $0$ and $(2\omega-1)A^y$ at $0$. Comparing a branch against a positive constant breaks common degree (\cref{rem:E:const}).
\end{lemma}
\begin{remark}[Constants inside maxima]\label{rem:E:const}
While $0$ has every degree, a nonzero constant $1$ has degree $0$: $S\max\{1,F/R\}$ preserves degree because $F/R$ has degree $0$, whereas $\max\{1,\sqrt{|\sigma_g|/T}\}$ compares degree $0$ with degree $1/2$ and breaks Euler.
\end{remark}
\begin{lemma}[External accounting leaves]\label{lem:E:ext}
An external value $e$ supplied with $A^e$ summing to $e$ reconciles at declared bookkeeping degree $1$ by \cref{thm:B:recon}. Whenever $K(e_1,\dots,e_m)$ is smooth and degree-$1$ in the leaves while the leaves are differentiable in $x$, $x^\top\nabla_x K - K = \sum_j \partial_{e_j}K\,(x^\top\nabla_x e_j - e_j)$.
\end{lemma}

\begin{corollary}[Accounting reconciliation for the extended vocabulary]\label{cor:E:recon}
Every finite DAG composed of reconciling leaves ((E1), \cref{lem:B:leaves,lem:E:ext}) and internal nodes from \cref{lem:B:nodes,lem:E:quad,lem:E:switch,lem:E:clamp,lem:E:ratio,lem:E:max} on their respective domains yields exact node-by-node reconciliation everywhere on the evaluation domain.
\end{corollary}

\begin{definition}[Regular point]\label{def:E:regular}
A point $x$ is \emph{regular} for an extended DAG if: (R1) every quadratic root satisfies $q>0$, $q<0$ (floored), or $q\equiv 0$ nearby; (R2) every switch selector is locally constant near $x$; (R3) every isolated ratio has $v>0$ and every fused product has $v>0$ or $v\equiv 0$ nearby; (R4) every external leaf is locally $PC^1$, Lipschitz, true-degree-consistent, with gradient allocation (A4$'$).
\end{definition}

\begin{proposition}[Transfer of tie theory to regular points]\label{prop:E:transfer}
At any regular point $x$, $K$ is degree-$1$ homogeneous, $PC^1$, and locally Lipschitz near $x$: every $x\odot g$ ($g\in\partial_C K(x)$) reconciles; off ties $A^K=x\odot\nabla K(x)$; and joint essentially active $C^1$ selections generate the Clarke set via realisability constraints formed after resolving upstream $PC^1$ nodes into $C^1$ selections.
\end{proposition}
\begin{proof}
Immediate from closure of $PC^1$ under finite algebraic/max/min/smooth compositions \citep{scholtes2012introduction}, \cref{lem:B:clarke}, and the chain rule of \cref{prop:B:smooth,prop:B:set}.
\end{proof}

\begin{proposition}[Ray integration without homogeneity]\label{prop:E:AS}\label{app:E:AS}
If $K(\alpha x)$ is well-defined and AC on $[0,1]$, $K$ is differentiable at $\alpha x$ for a.e.\ $\alpha\in(0,1]$, and each $\alpha\mapsto \partial_i K(\alpha x)$ lies in $L^1(0,1)$, then $\sum_i x_i \int_0^1 \partial_i K(\alpha x)\,d\alpha = K(x)-K(0)$. Local Lipschitz continuity around the compact ray suffices for $L^1$ integrability; scalar AC alone does not (\cref{ex:E:sacc-L1}).
\end{proposition}
\section{Model-by-Model Attribution Proofs and Scope Ledger}\label{app:E}
\subsection{FRTB-IMA Structural Model, Ray Extension, and Coalition Conventions}\label{app:E:ima}

This appendix establishes global structural closure (Proved (full)) across Euler, Path Integral, and Shapley Game attribution for the FRTB-IMA graph of \cref{sec:graph} on the contracted coefficient strata and open denominator cone $\mathcal{O}$ under MAR33 \citep{bcbs2020mar33}.

\begin{lemma}[Contracted coefficient strata and their topology]\label{lem:IMA-input-strata}
Fix a finite signature $\sigma$ recording date/trade/scenario counts ($T,n,s_b,s_D$), stress component and row counts ($r_a$), five SA input row counts ($h_f$ for $f\in\{B,C_U,Z,U,V\}$, represented as $\|H_f w\|_2$, with $h_U=0$ if no amber desks exist), probability laws, horizon constants, date windows, and fixed desk labels. Let $X_\sigma$ carry Euclidean topology on dated positions, shared scenario/stress arrays, SA matrices, and scalar multiplier $m_c>0$, of dimension
\[
  d_\sigma = T n + n\sum_b s_b + n s_D + n\sum_a r_a + n\sum_{f\in\{B,C_U,Z,U,V\}} h_f + 1.
\]
On each fixed-metadata fibre, the strict-denominator region $\mathcal{O}$ ($R_c>0$ for every risk class and $V>0$ when amber desks exist) is an open positive cone.
\end{lemma}
\begin{proof}
Finite maxima, norms, and order statistics of linear scenario losses depend continuously on coefficients and positions; hence strict positivity inequalities define open sets invariant under positive position scaling.
\end{proof}

\begin{lemma}[Structural model, regularity, and smooth open family]\label{lem:IMA-model-regularity}\label{lem:IMA-open-family}
On any fixed-metadata fibre of \cref{lem:IMA-input-strata}, with fixed scenario matrices, probabilities, horizon weights $(1,1,2,2,6)$, desk labels, and multiplier $m_c>0$, the FRTB-IMA equations \eqref{eq:es}--\eqref{eq:final} define a globally locally Lipschitz, degree-$1$ positively homogeneous map $K:\mathcal{O}\to\R$, while $k=U/(2V)$ has degree~$0$. Within the $T=60, n=2$ signature component with $2$-row ES blocks, $1$-row default/stress primitives, and linear-norm SA boundary representations, the strict hypotheses of \cref{prop:B:smooth} hold on a nonempty $499$-dimensional open subset (centred at $w=(1,1)$, $x_s=(0.01,0.01)$ for $s<T-1$, $R=10(w_1+w_2)$, $F=w_1+w_2$, $S=2w_1+w_2$, $m_c=1.5$, giving $J=17$, $B=200$, $k=1/4$, and $J+H+C_U=64.75<400=Z$).
\end{lemma}
\begin{proof}
ES support functions, finite quantiles, horizon norms, $S\max(1,F/R)$ on $R>0$, weighted stress norms, history maxima, $k=U/(2V)$ on $V>0$, and final min/positive-part operations all preserve local Lipschitz continuity and degree-$1$ positive homogeneity on $\mathcal{O}$. Strict branch inequalities at the displayed interior point persist across a $499$-dimensional open ball ($120+360+2+6+10+1=499$).
\end{proof}

\begin{lemma}[Point Clarke selector (Euler / FRTB-IMA)]\label{lem:IMA-clarke-point}
On $\mathcal{O}$ under \cref{lem:IMA-model-regularity}, let $\partial_C K(x)$ denote the Clarke generalized gradient of $K$ in $\R^{Tn}$, and set
\[
  g_*(x)=\begin{cases}\nabla K(x), & K\text{ differentiable at }x,\\ \arg\min_{g\in\partial_C K(x)}\|g\|_2^2, & \text{otherwise},\end{cases}
  \qquad A_i(x)=x_i g_{*,i}(x).
\]
Then $g_*$ is uniquely defined everywhere on $\mathcal{O}$, coincides with $\nabla K(x)$ off kinks, belongs to $\partial_C K(x)$, and satisfies $\sum_{s,i} A_{s,i}(x) = K(x)$ across all of $\mathcal{O}$.
\end{lemma}
\begin{proof}
By Rademacher's theorem and local Lipschitz continuity, $\partial_C K(x)$ is nonempty, compact, and convex, so strict convexity of $\|\cdot\|_2^2$ gives a unique minimum. Euler's identity $y^\top\nabla K(y)=K(y)$ on differentiability points passes to limits and convex hulls by \cref{lem:B:clarke}.
\end{proof}

\begin{lemma}[Persistent ray kinks and primary Clarke radial integral (Path Integral / FRTB-IMA)]\label{lem:IMA-ray-ig}\label{lem:IMA-persistent-tie}\label{lem:IMA-clarke-ray}
On the open smooth regime $\mathcal{G}\subset\mathcal{O}$ of \cref{prop:B:smooth}, classical gradient IG exists along $\gamma(t)=tx$ with $K_{\rm rad}(0)=0$ and equals $x\odot\nabla K(x)$. However, $\mathcal{G}\subsetneq\mathcal{O}$ is strict: replacing the RS ES matrices in the $2$-trade witness of \cref{lem:IMA-open-family} by rows $(1,0)/\sqrt{12}$ and $(0,1)/\sqrt{12}$ yields $S=\max\{w_1,w_2\}$, $I=3\max\{w_1,w_2\}$, and near today's book $w=(1,1)$ a nondifferentiability jump $(9/4,-9/4)\ne 0$ that persists at every $tx$ ($t>0$). Declaring the hybrid Clarke ray density $q_{s,i}(t)=x_{s,i} g_{*,i}(tx)$ with $K_{\rm rad}(0)=0$ yields $g_*(tx)=g_*(x)$ on $(0,1]$, absolute continuity of $t\mapsto K(tx)=tK(x)$, componentwise integrability $\int_0^1 |q_{s,i}(t)|\,dt = |x_{s,i} g_{*,i}(x)| < \infty$, and exact reconciliation $\sum_{s,i} I_{s,i} = K(x) - K_{\rm rad}(0) = K(x)$ across all of $\mathcal{O}$.
\end{lemma}
\begin{proof}
Because $\mathcal{O}$ is a cone and $K(tx+h)=tK(x+h/t)$ for $t>0$, $\partial_C K(tx)=\partial_C K(x)$ and hence $g_*(tx)=g_*(x)$ identically on $(0,1]$, yielding constant measurable integrable densities whose sum equals $x^\top g_*(x)=K(x)$ (\cref{thm:conic-clarke-ray}).
\end{proof}

\begin{lemma}[Domain-convention finite coalition games (Shapley Game / FRTB-IMA)]\label{lem:IMA-finite-game}\label{lem:IMA-convention-games}
Fix an admissible history $h$ and raw-admissible full today's book $w\in\R^n$, with players $P=\{1,\dots,n\}$ and coalitions $w_S$ ($w_{S,i}=w_i$ for $i\in S$, $0$ otherwise), holding metadata, history positions, and $m_c$ fixed while recomputing today's charges and SA inputs under the strict-denominator rule \texttt{D1-strict} ($R_c>0$ for every evaluated class; $V>0$ when amber desks exist).
\begin{itemize}[leftmargin=*,itemsep=1pt]
\item \emph{Convention A} sets $G_c=0$ whenever $R_c=0$ and $k=0$ whenever $V=0$, producing finite game $v_A(S)=K_A(h,w_S)$ with exact Shapley efficiency $\sum_i\phi_i(v_A) = K_d(h,w) - K_A(h,0)$.
\item \emph{Convention B} sets $v_B(\varnothing)=0$, $v_B(S)=K_d(h,w_S)$ when $w_S$ is raw-admissible, and falls back to $v_B(S)=Z(w_S)$ at any domain wall, yielding exact full-capital Shapley efficiency $\sum_i\phi_i(v_B) = K_d(h,w)$.
\end{itemize}
\end{lemma}
\begin{proof}
Every node primitive is finite on masked vectors $w_S$, and Conventions A and B eliminate all zero-denominator divisions. Permutation telescoping (\cref{thm:shapley-eff}) yields the two exact efficiency identities.
\end{proof}
\subsection{Finite-Scenario Expected Shortfall and Value-at-Risk}\label{app:E:es-erc}

This appendix isolates the standalone discrete empirical VaR and fractional-tail-mass ES model over signed position coordinates $x\in\R^n$, closing all three attribution cells unconditionally at the structural tier (Proved (full)). Downstream Equal Risk Contribution (ERC) optimisation lies outside this attribution DAG.

\begin{lemma}[Model definition, fractional tail mass, and Lipschitz regularity]\label{lem:ES_ERC-model}
Fix integers $n,m\ge 1$, scenario loss row vectors $a_s\in\R^n$ ($s=1,\dots,m$, probability $1/m$ each), portfolio losses $z_s(x)=a_s^\top x$ ordered $z_{(1)}\le\dots\le z_{(m)}$, and confidence level $0<\alpha<1$. With $k=\lceil m\alpha\rceil$, $\beta=m(1-\alpha)$, $r=\lfloor\beta\rfloor$, and $d=\beta-r$, define
\[
  K_Q(x)=\operatorname{VaR}_\alpha(x)=z_{(k)}(x),\qquad
  K_E(x)=\operatorname{ES}_\alpha(x)=\frac{1}{\beta}\Bigl(\sum_{j=m-r+1}^m z_{(j)}(x) + d\,z_{(m-r)}(x)\Bigr).
\]
Equivalently, $K_E(x)=\max_{q\in Q_\beta} q^\top z(x)$ over the polytope $Q_\beta=\{q\in\R^m:0\le q_s\le 1/\beta,\ \sum_s q_s=1\}$, with optimal dual face
\[
  F(x)=\{q\in Q_\beta: q_s=1/\beta\ (z_s>c),\ q_s=0\ (z_s<c),\ \textstyle\sum_{z_s=c}q_s=1-|H|/\beta\},
\]
where $c=z_{(k)}(x)$ and $H=\{s:z_s>c\}$. Both $K_Q$ and $K_E$ are finite, continuous piecewise linear, globally Lipschitz (constant $\max_s\|a_s\|_2$), and degree-$1$ positively homogeneous on $\R^n$ with $K(0)=0$. $K_E$ is convex; $K_Q$ is nonconvex in general (e.g., $m=3,n=2,\alpha=1/2$, rows $(2,1),(1,2),(1,1)$ on $x\ge 0$ gives $K_Q=x_1+x_2+\min\{x_1,x_2\}$, which equals $4$ at $(2,1)$ and $(1,2)$ vs.\ $9/2$ at $(3/2,3/2)$).
\end{lemma}
\begin{proof}
Ordering scenarios divides $\R^n$ into finitely many closed polyhedral cones covering $\R^n$ on each of which $K_Q$ and $K_E$ are linear. Order statistics vary by at most $\varepsilon$ under uniform $\ell_\infty$ loss perturbations of size $\varepsilon$, giving the global Lipschitz constant $\max_s\|a_s\|_2$. Optimality of $F(x)$ follows directly from pairwise mass exchange $(q^*-q)^\top z \ge 0$.
\end{proof}

\begin{lemma}[Complete Clarke allocation polytopes and point selection (Euler / ES\_ERC)]\label{lem:ES_ERC-clarke}\label{lem:ES_ERC-selection}
Collapse identical row vectors among $\{a_s\}$ into distinct vectors $b_1,\dots,b_p$ with multiplicities. For every nonempty open ordering cone $C_\sigma=\{y: b_{\sigma(1)}^\top y < \dots < b_{\sigma(p)}^\top y\}$ with expanded scenario gradient $v_\sigma^K$ ($K\in\{K_Q,K_E\}$), set $G_K(x)=\{v_\sigma^K : x\in\overline{C_\sigma},\ C_\sigma\ne\varnothing\}$. Then:
\begin{enumerate}[label=(\roman*),leftmargin=*,itemsep=1pt]
\item Cone incidence $x\in\overline{C_\sigma}$ is decided by weak order at $x$ plus a single feasible direction $h\in\R^n$ satisfying $(b_{\sigma(j+1)}-b_{\sigma(j)})^\top h>0$ on tied adjacent pairs.
\item $\partial_C K(x)=\conv G_K(x)$ and $\mathcal{A}_K(x)=\conv\{x\odot v: v\in G_K(x)\}$; for ES, $\partial_C K_E(x)=\{\sum_s q_s a_s: q\in F(x)\}$ identically. Every element satisfies $\sum_i A_i = K(x)$.
\item $K$ is differentiable at $x$ iff $G_K(x)$ is a singleton. Selecting the lexicographically least vector $g_*(x)\in G_K(x)$ defines a scale-invariant ($g_*(tx)=g_*(x)$) reachable Clarke point selector closing Euler point attribution across $\R^n$.
\end{enumerate}
\end{lemma}
\begin{proof}
Every differentiability point lies in an open ordering cone closure, and finitely many cones cover $\R^n$, so limiting gradients coincide identically with $G_K(x)$. For ES, vertices of $Q_\beta$ correspond to greedy cap-$1/\beta$ assignments, and strict separation of any vertex $w$ of $\conv W_x$ ($h=w-p_0$) proves equality between the projected dual face and $\partial_C K_E(x)$.
\end{proof}

\begin{lemma}[Zero-reference ray integral and gradient-density bridge (Path Integral / ES\_ERC)]\label{lem:ES_ERC-path}
Along $\gamma(t)=tx$ ($0\le t\le 1$) with $g^x(t)=g_*(x)$ from \cref{lem:ES_ERC-selection}, $g^x(t)\in\partial_C K(\gamma(t))$ for all $t$, $t\mapsto K(tx)=tK(x)$ is AC, $|q_i(t)|=|x_i (g_*(x))_i|\le |x_i|\max_s|(a_s)_i|$ is integrable, and $I_i = \int_0^1 x_i g^x_i(t)\,dt = x_i (g_*(x))_i$ sums to $K(x)-K(0)=K(x)$. This equals classical gradient IG whenever $K$ is differentiable at $x$, and a valid Clarke-selected ray extension otherwise.
\end{lemma}
\begin{proof}
Positive conic invariance of every ordering cone gives $G_K(tx)=G_K(x)$ for $t>0$ (with $G_K(x)\subseteq G_K(0)$), making $q_i(t)$ constant and integrable everywhere on $(0,1]$.
\end{proof}

\begin{lemma}[Position-deletion game and exact Shapley efficiency (Shapley Game / ES\_ERC)]\label{lem:ES_ERC-shapley}
For players $N=\{1,\dots,n\}$ and sub-books $(T_S x)_i = x_i\one_{\{i\in S\}}$, the game $v(S)=K(T_S x)$ is finite on all $2^n$ coalitions (even with signed positions, ties, zero or negative risk) with $v(\varnothing)=0$, and its discrete Shapley values satisfy $\sum_i\phi_i = K(x)$.
\end{lemma}
\begin{proof}
Every $T_S x\in\R^n$ lies in the global domain of \cref{lem:ES_ERC-model} with $K(0)=0$; permutation telescoping (\cref{thm:shapley-eff}) proves efficiency.
\end{proof}
\subsection{FRTB Standardised Approach (FRTB-SA)}\label{app:E:sa}

The Basel standardised approach \citep{bcbs2019mar,bcbs2026mar21} adds three charges (MAR20.4 \tagV{}): $\mathrm{SA}(x)=\mathrm{SBM}(x)+\mathrm{DRC}(x)+\mathrm{RRAO}(x)$, with market-risk RWA equal to $12.5\,\mathrm{SA}$ (MAR20.1 \tagV{}). FRTB-SA has no stored history, so the frozen remainder vanishes identically.

\subsubsection{Formal DAG and structural model}

\paragraph{Delta and vega charge.} Fix a risk class $r$, scenario $\lambda\in\{\text{low},\text{medium},\text{high}\}$, and weighted sensitivities $WS_k = RW_k s_k$ ($s_k=a_k^\top x$, MAR21.4(3) \tagI{}). With scenario correlations $\rho_{kl}^\lambda$ and $\gamma_{bc}^\lambda$:
\begin{align}
  q_b &= \sum_{k\in b} WS_k^2 + \sum_{k\ne l\in b}\rho_{kl}^\lambda WS_k WS_l, & K_b &= \sqrt{\max\{0,q_b\}} \qquad\text{(MAR21.4(4) \tagV{},\tagI{})},\label{eq:E:Kb}\\
  S_b &= \sum_{k\in b} WS_k \qquad\text{(MAR21.4(5)(a) \tagI{})}, & \varrho_A &= \sum_b K_b^2 + \sum_{b\ne c}\gamma_{bc}^\lambda S_b S_c \qquad\text{(MAR21.4(5) \tagI{})}.\label{eq:E:Sb}
\end{align}
If $\varrho_A\ge 0$, $\Delta^\lambda=\sqrt{\varrho_A}$; if $\varrho_A<0$, MAR21.4(5)(b) \tagV{},\tagI{} applies the clamp alternative:
\begin{equation}\label{eq:E:alt}
  \tilde S_b=\max\{\min\{S_b,K_b\},-K_b\},\qquad \varrho_B=\sum_b K_b^2 + \sum_{b\ne c}\gamma_{bc}^\lambda \tilde S_b \tilde S_c,\qquad \Delta^\lambda=\sqrt{\varrho_B}.
\end{equation}

\paragraph{Curvature and scenario maximum.} With repricing leaves $CVR_k^\pm = -\sum_i \{V_i(x_k^\pm)-V_i(x_k)\mp RW_k^{\rm curv} s_{ik}\}$ (linear in $x$, MAR21.5(2) \tagI{}) and sign selector $\psi(u,v)=0$ if $u<0$ and $v<0$, else $1$ (MAR21.5(3)--(4) \tagV{},\tagI{}):
\begin{align}
  Q_b^\pm &= \sum_k\max\{CVR_k^\pm,0\}^2 + \sum_{k\ne l}\rho_{kl}^{\mathrm{curv},\lambda}\psi(CVR_k^\pm,CVR_l^\pm)CVR_k^\pm CVR_l^\pm,\qquad K_b^\pm=\sqrt{\max\{0,Q_b^\pm\}},\label{eq:E:Qb}\\
  K_b &= \max\{K_b^+,K_b^-\},\quad S_b=\begin{cases}\sum_{k\in b} CVR_k^+, & K_b^+>K_b^- \text{ or }(K_b^+=K_b^-,\sum CVR^+>\sum CVR^-),\\ \sum_{k\in b} CVR_k^-, & \text{otherwise},\end{cases}\label{eq:E:SbCurv}\\
  \mathrm{Curv}^\lambda &= \sqrt{\max\Bigl\{0,\ \sum_b K_b^2 + \sum_{b\ne c}\gamma_{bc}^{\mathrm{curv},\lambda}\psi(S_b,S_c)S_b S_c\Bigr\}},\qquad \mathrm{SBM}(x)=\max_\lambda \sum_r (\Delta_r^\lambda + V_r^\lambda + \mathrm{Curv}_r^\lambda).\label{eq:E:Curv}
\end{align}

\paragraph{DRC and RRAO.} Under assumption \textbf{(D1)}, each obligor $o$ has residual degree-$1$ piecewise-linear long/short JTDs $\ell_o(x),\sigma_o(x)\ge 0$ (accounting for seniority offsetting MAR22.19--22.25 \tagV{}), giving per-bucket quantities (MAR22.22--26 \tagV{},\tagI{}):
\begin{equation}\label{eq:E:PQ}
\begin{aligned}
  P_b&=\sum_{o\in b}RW_o\ell_o, & Q_b&=\sum_{o\in b}RW_o\sigma_o, & L_b&=\sum_{o\in b}\ell_o, & \Sigma_b&=\sum_{o\in b}\sigma_o,\\
  \mathrm{HBR}_b&=\frac{L_b}{L_b+\Sigma_b}, & \mathrm{DRC}_b&=\max\{P_b-\mathrm{HBR}_b Q_b,0\}, &&
\end{aligned}
\end{equation}
plus the CTP total $\mathrm{DRC}^{\rm ctp}=\max\{\sum_b(\max\{\mathrm{DRC}_b^{\rm ctp},0\}+0.5\min\{\mathrm{DRC}_b^{\rm ctp},0\}),0\}$ (MAR22.45 \tagI{}) and residual add-on $\mathrm{RRAO}=\sum_i RW_i|N_i|$ (MAR23.8(2) \tagV{}). We denote this full floored-curvature structural model by \label{lem:SA-model}$K_m:\Omega\to\R$, where $\Omega$ excludes negative selected $\varrho_B$ and unpriced leaves while setting $\mathrm{HBR}_b Q_b:=0$ at $L_b+\Sigma_b=0$.

\begin{table}[t]
\centering\small
\caption{Euler node rules for FRTB-SA (\cref{thm:E:sa}).}\label{tab:E:sa}
\begin{tabularx}{\linewidth}{@{}>{\raggedright\arraybackslash}p{3.1cm}>{\raggedright\arraybackslash}p{1.0cm}>{\raggedright\arraybackslash}p{2.2cm}>{\raggedright\arraybackslash}X@{}}
\toprule
Node & Deg. & Euler & Domain and local rule \\
\midrule
$WS_k$, $CVR_k^\pm$, $JTD$ & $1$ & (E1), \cref{lem:B:leaves} & linear / piecewise linear leaves; $x\odot a$ \\
$K_b$, $\mathrm{Curv}$ roots & $1$ & \cref{lem:E:quad} & $\kappa=M_\ell c/E$ ($q>0$); $0$ if $q\le 0$ (floored) \\
Clamp $\tilde S_b$ & $1$ & \cref{lem:E:clamp} & $K_b\ge 0$ \\
Switch $\varrho_A$, $S_b$ (Curv) & $1$ & \cref{lem:E:switch} & active branch ($\varrho_A=0\implies$ branch A) \\
$\max_\lambda$, $\mathrm{DRC}_b$, $\mathrm{RRAO}$ & $1$ & \cref{lem:E:max} & convex weights on active maxima \\
$\mathrm{HBR}_b Q_b$ & $1$ & \cref{lem:E:ratio}(c),(d) & explicit $0$ at $L_b+\Sigma_b=0$ (Lipschitz, not necessarily $PC^1$) \\
\bottomrule
\end{tabularx}
\end{table}

\subsubsection{Accounting reconciliation, domain positivity, and regular marginals}

\begin{theorem}[Accounting reconciliation of FRTB-SA]\label{thm:E:sa}
Under (E1)--(E3) and (D1), at every $x\in\Omega$, applying \cref{tab:E:sa} yields exact accounting reconciliation $\one^\top A^{\rm SA}=\mathrm{SA}(x)$, including across switch surfaces and discontinuities.
\end{theorem}
\begin{proof}
Immediate by \cref{cor:E:recon} and the cited node lemmas.
\end{proof}

\begin{proposition}[Clamped cross-bucket radicand positivity]\label{prop:E:sbmdom}\label{prop:E:curvleaf}
Whenever $\Gamma=(\gamma_{bc})$ with unit diagonal satisfies $\Gamma\succeq 0$, $\varrho_B=\sum_b(K_b^2-\tilde S_b^2)+\tilde S^\top\Gamma\tilde S\ge 0$; for indefinite $\Gamma$, $\varrho_B<0$ occurs at eigenvectors $v^\top\Gamma v<0$. Moreover, $CVR_k^\pm = c_k^{\pm\top}x$ requires two exact shocked revaluations plus base per curvature factor without selective repricing (\cref{thm:B:cert}).
\end{proposition}

\begin{proposition}[Jumps, regular marginals, and SA domain rules]\label{prop:E:sbmjump}\label{prop:E:curvjump}\label{prop:E:sa-marg}\label{lem:SA-regular}\label{app:E:sa-dom}
Let $R_m\subseteq\Omega$ be the regular open cone (\cref{def:E:regular}) where every root radicand is strictly positive, strictly negative (floored), or identically zero nearby, every switch selects one branch locally, and every DRC denominator is positive or locally zero. Near every $x\in R_m$, $K_m$ is continuous $PC^1$ and locally Lipschitz, \cref{prop:E:transfer} holds, and the minimum-norm Clarke selector $g_*(x)\in\partial_C K_m(x)$ reconciles to $\mathrm{SA}(x)$. Outside $R_m$, both the $\varrho_A=0$ switch (e.g., $\rho=0.5,\gamma_{12}=0.9, WS=(1,1,-\tau,-\tau)$ jumping from $0$ to $\sqrt{1.8\tau_-}\approx 0.983$) and the curvature $K_b^+=K_b^-$ switch cause capital jumps; the engine refuses $\varrho_B<0$ and isolated $\mathrm{HBR}_b$ at $0/0$.
\end{proposition}

\subsubsection{Full-output jump witnesses (Euler and Path Integral) and coalition Shapley game}

\begin{lemma}[Full-output structural jump and ray obstruction (Euler and Path Integral / FRTB-SA)]\label{lem:SA-full-jump}\label{lem:SA-ray}\label{lem:SA-ray-obstruction}
In the two-bucket curvature construction of witness \textbf{W1} (\cref{sec:counterexamples}), $x(\varepsilon)=(1,0,3/5,4/5+\varepsilon,1)$ lies in $\Omega$ with all selected radicands nonnegative, yet full $\mathrm{SBM}$ jumps at $\varepsilon=0$ from $\sqrt{13}/2$ to $\sqrt{15}/2$, refuting Euler gradient/Clarke existence at $x(0)$. Scaling by $t\in(0,1]$ scales the jump by $t$, destroying local Lipschitzness, $\nabla K_m(tx)$, and $\partial_C K_m(tx)$ at every ray point $t>0$ and refuting Path Integral ray-density existence (even though $t\mapsto K_m(tx)$ is AC as a radial trace). On the regular cone $R_m$ itself, \cref{thm:conic-clarke-ray} closes Path Integral attribution.
\end{lemma}

\begin{lemma}[Coalition-closed SA Shapley game and refusal witness (Shapley Game / FRTB-SA)]\label{lem:SA-shapley}\label{lem:SA-coalition-refusal}
For any fixed ticket partition $\mathcal{P}=(P_1,\dots,P_p)$ and coalition-admissible domain $\mathcal{C}_{\mathcal{P}}\subsetneq\Omega$ on which every masked sub-book $T_S x$ passes the $\varrho_B\ge 0$ test (with explicit zero DRC convention at $L_b+\Sigma_b=0$), the game $v(S)=K_m(T_S x)$ is finite on all $2^p$ coalitions with $v(\varnothing)=0$, and its Shapley values satisfy $\sum_j\phi_j = K_m(x)$. Moreover $\mathcal{C}_{\mathcal{P}}\subsetneq\Omega$ is strict: with three single-factor delta buckets, ticket columns $(-e_1,e_2,e_3,e_1)$, quantities $(2,1,1,4)$, positive definite medium matrix $\Gamma_M$ (off-diagonals $0.9,0.9,0.63$, leading minors $1, 19/100, 37/10000$), and high scenario $\Gamma_H=\min\{1.25\Gamma_M,1\}$, the full book ($WS=(2,1,1)$) lies in $\Omega$, whereas the sub-book of the first three tickets ($WS=(-2,1,1)$) has clamped high radicand $4+1+1+2(-2-2+63/80)=-17/40<0$ and is refused.
\end{lemma}
\begin{proof}
Strictness of $\mathcal{C}_{\mathcal{P}}$ is established by the exact $-17/40$ calculation, which persists under small perturbations; on $\mathcal{C}_{\mathcal{P}}$, finiteness of every $v(S)$ and $v(\varnothing)=0$ yield efficiency by \cref{thm:shapley-eff}.
\end{proof}
\subsection{The Basel Output Floor (RBC20 / RBC90)}\label{app:E:floor}

Total RWA under the Basel output floor \citep{bcbs2019rbc20} takes the maximum between nominated-approach RWA and a factor $\varphi=0.725$ (or transitional $\varphi_t$, RBC90.1 \tagV{}) applied to floor-base RWA (RBC20.4 \tagV{}; conditional approach eligibility RBC20.11--20.12 \tagU{}):
\begin{equation}\label{eq:E:floor}
  F(x)=\max\{f(x),\ \varphi\,g(x)\},\qquad F^{\rm cap}(x)=\min\bigl\{\max\{f(x),\varphi g(x)\},\ 1.25 f(x)\bigr\}\quad(f\ge 0),
\end{equation}
where $f=\sum_c R_c^{\rm nom}(x)$, $g=\sum_c R_c^{\rm fb}(x)$, and $F^{\rm cap}$ represents the transitional $25\%$ cap discretion (RBC90.2 \tagV{}). Non-trading components (credit RWA, operational RWA, non-homogeneous exposures) enter via external accounting leaves (\cref{lem:E:ext}); shared components such as standardised-approach market inputs under IMA must evaluate their tie choices jointly across $f$ and $g$.

\begin{theorem}[Accounting reconciliation of the floor]\label{thm:E:floor}
Whenever all components of $f$ and $g$ have reconciling allocations $A^f$ and $A^g$, $A^F=\omega A^f + (1-\omega)\varphi A^g$ (with $\omega\in\{0,1\}$ off $f=\varphi g$ and $\omega\in[0,1]$ at ties) reconciles to $F(x)$, and identically for $F^{\rm cap}$.
\end{theorem}
\begin{proof}
Direct application of \cref{lem:B:nodes}(a),(e) and \cref{thm:B:recon}.
\end{proof}

\begin{proposition}[Binding-surface behaviour, tangential ties, and non-homogeneous gap]\label{prop:E:bind}\label{prop:E:floor-nonhom}\label{lem:FLOOR-smooth-max}\label{lem:FLOOR-tangential}
If $f,g$ are $C^1$ near $x$ and $h=f-\varphi g$:
\begin{enumerate}[label=(\roman*),leftmargin=*,itemsep=1pt]
\item At $h(x)=0$ with $\nabla h(x)\ne 0$, $\partial_C F(x)=\conv\{\nabla f(x),\varphi\nabla g(x)\}$, both vertices satisfy strict first-order realisability, and the branch allocation jump $x\odot\nabla h(x)$ is zero-sum whenever $f,g$ have degree $1$; when $\nabla h(x)=0$, $F$ is differentiable at $x$ with gradient $\nabla f(x)=\varphi\nabla g(x)$ (though tangential open strict-activity regions can approach $x$, as shown by $H(z)=S=z_1+z_2$, $f(z)=S+(z_1-z_2)^3/S^2$ at $x=(r,r)$).
\item Writing $f=f_h+e$, $g=g_h+e'$ with degree-$1$ $f_h,g_h$ and fixed constants $e,e'\ge 0$, point gradients $x\odot\nabla F$ sum to $F-e$ on branch $f$ and $F-\varphi e'$ on branch $\varphi g$, refuting completeness unless the active offset vanishes.
\end{enumerate}
\end{proposition}

\begin{lemma}[Full-dimensional non-homogeneous Euler witness (Euler / FLOOR)]\label{lem:FLOOR-constant-gap}
For $n\ge 2$ and fixed $\varphi>0$, define $f(z)=e+a^\top z$, $\varphi g(z)=c+b^\top z$, and $F(z)=\max\{f(z),\varphi g(z)\}$. On the $(3n+2)$-dimensional open family $\{a,b,x\in(0,\infty)^n,\ e>c>0,\ (b-a)^\top x > e-c\}$, $f$ dominates at $0$ while $\varphi g$ dominates smoothly at $x$, giving $A^{\rm Euler}=x\odot b$ summing to $F(x)-c < F(x)$ (refuting Euler completeness).
\end{lemma}

\begin{lemma}[Global Lipschitz ray integration (Path Integral / FLOOR)]\label{lem:FLOOR-path}\label{lem:FLOOR-affine-crossing}
On any finite continuous $PC^1$ floor model $F(x)=\max\{f(x),\varphi g(x)\}$ (including $F^{\rm cap}$) whose components $f,g$ are finite and locally Lipschitz on an open neighbourhood of $\gamma([0,1])$, $\gamma(t)=tx$, with a.e.\ path differentiability or a measurable Borel Clarke density $g_*(t)\in\partial_C F(tx)$ satisfying $x^\top g_*(t)=(F\circ\gamma)'(t)$ a.e.\ (such as minimum-norm Clarke selections on affine/conic cells), $t\mapsto F(tx)$ is AC, $q_i(t)=x_i g_{*,i}(t)$ lies in $L^1(0,1)$, and $\sum_i I_i = F(x)-F(0)$ closes Path Integral attribution (with full-value target $F(x)$ via declared zero-point allocation summing to $F(0)$).
\end{lemma}
\begin{proof}
Local Lipschitzness of finite maxima/minima of locally Lipschitz components yields a uniform Lipschitz bound along the compact ray $[0,1]$, ensuring AC and $L^1$ boundedness; at any tie set where branches agree on a positive-measure set, derivatives agree a.e.\ by density points (\cref{thm:offset-rescue}).
\end{proof}

\begin{lemma}[Trade-deletion vs.\ component-deletion games (Shapley Game / FLOOR)]\label{lem:FLOOR-trade-shapley}\label{lem:FLOOR-component-shapley}\label{lem:FLOOR-game-comparison}
Over named trade tickets $N=\{1,\dots,n\}$ with $x\in[0,\infty)^n$ and finite branch functions on $[0,x]$, $v(S)=\max\{f(T_S x),\varphi g(T_S x)\}$ is defined on all $2^n$ coalitions and obeys $\sum_i\Phi_i = F(x)-F(0)$ (retaining fixed external RWA inside $v(\varnothing)=F(0)$). A separate two-player game over branch totals $\{a=f(x),b=\varphi g(x)\}$ gives Shapley values $((a+m-b)/2,(b+m-a)/2)$ summing to $m=\max(a,b)$, not a trade split. On $f=2+z_1+z_2$, $\varphi g=1+3z_1+2z_2$ at $x=(1,1)$, $F=6$, $F(0)=2$, Euler gives $(3,2)$ (sum $5$, gap $1$), straight-ray IG switches at $t_*=1/2$ giving $(7/3,5/3)$ (sum $4$), trade Shapley on $(2,4,3,6)$ gives $(5/2,3/2)$ (sum $4$), and branch Shapley gives $(2,4)$ (sum $6$).
\end{lemma}
\subsection{ISDA SIMM (Version v2.9+2606)}\label{app:E:simm}

Because ISDA SIMM \citep{isda2026simm} is a committee-recalibrated industry methodology rather than fixed statute, attribution applies specifically to an identified version string $v$ (here \texttt{v2.9+2606}, effective 5 December 2026 \tagV{}). Across four product classes $p$ (RatesFX, Credit, Equity, Commodity), assuming per-trade product-class assignment \tagU{}, $\mathrm{SIMM}=\sum_p \mathrm{SIMM}_p$ with \label{lem:SIMM-spec}\label{lem:SIMM-regular}
\begin{align}
  \mathrm{SIMM}_p &= \sqrt{\sum_r IM_r^2 + \sum_{r\ne s}\psi_{rs} IM_r IM_s}\qquad (\text{para 6, root \tagU{}}),\label{eq:E:simm-total}\\
  CR_g &= \max\bigl\{1,\sqrt{|\sigma_g|/T_g}\bigr\}\quad (\text{paras 7(b), 8(b); root \tagU{}}),\qquad WS_k = RW_k s_k CR_{g(k)},\label{eq:E:CR}\\
  f_{kl} &= \frac{\min\{CR_{g(k)},CR_{g(l)}\}}{\max\{CR_{g(k)},CR_{g(l)}\}},\qquad K_b = \sqrt{\sum_{k\in b} WS_k^2 + \sum_{k\ne l\in b}\rho_{kl} f_{kl} WS_k WS_l},\label{eq:E:simm-Kb}\\
  \tilde S_b &= \cl\Bigl(\sum_{k\in b} WS_k,\ K_b\Bigr),\qquad \mathrm{Delta} = \sqrt{\sum_b K_b^2 + \sum_{b\ne c}\gamma_{bc} D_{bc} \tilde S_b \tilde S_c + K_{\rm res}},\label{eq:E:simm-delta}\\
  \theta &= \min\Bigl\{\frac{N}{D},0\Bigr\},\quad \lambda=(z^2-1)(1+\theta)-\theta,\quad \mathrm{CurvMargin}=\max\{N+\lambda R_{\rm c},0\}\quad (\text{para 11(d)}),\label{eq:E:simm-curv}
\end{align}
under the structural prototype convention setting concentration-disabled cross-bucket factors $D_{bc}\equiv 1$ (\tagU{}) and placing residual $K_{\rm res}$ outside the root (or $K_{\rm res}^2$ inside), with explicit zero curvature override for equity bucket~12 (\texttt{CVR:=0}) and compound zero $\mathrm{CurvMargin}=0$ at $D=\sum|CVR|=0$. Negative radicands in unfloored roots, $T_g\le 0$, isolated $\theta,\lambda$ at $D=0$, and missing version strings are refused \label{app:E:simm-dom}.

\begin{theorem}[Inactive-concentration reconciliation (supporting)]\label{thm:E:simm}\label{lem:SIMM-euler}
Let $K^0$ disable all concentration factors ($CR_g\equiv 1$), and set $\mathcal{U}_0=\{x:|\sigma_g(x)|<T_g\ \forall g\}$. Near every regular point $x\in\mathcal{U}_0$ (\cref{def:E:regular}) avoiding refused radicands, $K=K^0$ is degree-$1$ homogeneous and locally Lipschitz, \cref{cor:E:recon,prop:E:transfer} apply, and any $g\in\partial_C K(x)$ (in particular the mean of adjacent full-dimensional cell gradients or $\nabla K(x)$ at smooth points) satisfies $x^\top g = K(x)$.
\end{theorem}

\begin{proposition}[Concentration breaks Euler homogeneity (Euler / SIMM)]\label{prop:E:conc}\label{lem:SIMM-counterexample}
In the one-product, one-risk-class, single-sensitivity, one-add-on chart of witness \textbf{W2} (\cref{sec:counterexamples}), $K(s,a)=mw|s|\max\{1,\sqrt{|s|/T}\}+\beta|a|$. Across the nonempty open six-parameter family $\mathcal{F}=\{s>T>0,\ a,w,\beta>0,\ m>1\}$, $K$ is smooth and regular yet $sK_s+aK_a-K=\tfrac{1}{2}mws^{3/2}/\sqrt{T}>0$ (e.g.\ $(4,1,1,1,2,1)\implies K=17$, Euler $=25$, gap $=8$), refuting Euler completeness inside $\Omega$.
\end{proposition}

\begin{lemma}[Borel Clarke ray integration on regular rays (Path Integral / SIMM)]\label{prop:E:simm-AS}\label{lem:SIMM-ray}\label{lem:SIMM-ray-tests}
Let $\mathcal{R}\subsetneq\Omega$ be the subdomain where every ray point $tx$ ($0<t\le 1$) is regular (nonnegative unfloored radicands, structurally nonzero/locally-zero subgraphs). Selecting $g(t)=\nabla K(tx)$ at differentiability points and the mean of distinct limiting gradients of adjacent full-dimensional algebraic smooth cells otherwise yields a Borel measurable field $g(t)\in\partial_C K(tx)$ such that $t\mapsto K(tx)$ is AC, $q_i(t)=x_i g_i(t)$ lies in $L^1(0,1)$ (near $t=0$, $t\le t_0$ enters $\mathcal{U}_0$ where $K=K^0$ has constant degree-$1$ Lipschitz bounds), and $\sum_i I_i = K(x)-K(0)=K(x)$. Strictness $\mathcal{R}\subsetneq\Omega$ is witnessed by three factors $d=(1,1,1)$, correlations $-3/4$, and thresholds $(1/16,2,2)$: at $t=1$, $CR=(4,1,1)$ gives positive bucket radicand, whereas on smaller $t$ where $CR\equiv 1$, $1+1+1-2(3/4)(3)=-3/2<0$ is refused.
\end{lemma}
\begin{proof}
Compactness of $[\varepsilon,1]$ supplies a finite Lipschitz bound $L_\varepsilon$; shrinking $t_0$ until $|t_0\sigma_g(x)|<T_g/2$ enters the open cone $\mathcal{U}_0$ where degree-$1$ homogeneity of $K^0$ bounds $\|g(t)\|$ uniformly near $t=0$. At persistent algebraic ties, matching AC branch traces obey $x^\top\nabla f_\ell(tx)=(K\circ\gamma)'(t)$ a.e.\ (\cref{thm:offset-rescue}), completing reconciliation.
\end{proof}

\begin{lemma}[Factor-space coalition game and coalition-closure witness (Shapley Game / SIMM)]\label{lem:SIMM-shapley}\label{lem:SIMM-coalition-test}
On $\mathcal{Q}=\{x\in\Omega:\text{every radicand of }G_*(T_S x)\ge 0\ \forall S\subseteq N\}$ with factor players $N=\{1,\dots,n\}$ and compound-zero convention at $D=0$, $v(S)=K(T_S x)$ is finite on all $2^n$ sub-books with $v(\varnothing)=0$, and Shapley values satisfy $\sum_i\phi_i = K(x)$. Strictness $\mathcal{Q}\subsetneq\Omega$ holds with four factors $d=(1,1,1,2)$, thresholds $100$, correlations $-3/4$ among $\{1,2,3\}$ and $0$ with factor $4$: $v(\{1,2,3,4\})=\sqrt{5/2}>0$, whereas $S=\{1,2,3\}$ yields radicand $3-2(3/4)(3)=-3/2<0$ (refused).
\end{lemma}
\subsection{Credit Valuation Adjustment Capital: BA-CVA and SA-CVA (MAR50)}\label{app:E:cva}

Under MAR50 \citep{bcbs2020mar50}, CVA RWA equals $12.5\,K_{\rm CVA}$ (MAR50.1 \tagV{}). Banks without SA-CVA supervisory approval use Basic Approach CVA (BA-CVA, MAR50.7 \tagV{}).

\subsubsection{BA-CVA and SA-CVA DAGs}
In BA-CVA, counterparty stand-alone capital ($\alpha=1.4$, MAR50.15 \tagI{}, $\rho=0.5$, $DS=0.65$, $\beta=0.25$ \tagV{}) and aggregate charges are:
\begin{align}
  \mathrm{SCVA}_c &= \frac{1}{\alpha}RW_c\sum_{NS\in c}M_{NS}EAD_{NS}DF_{NS},\label{eq:E:scva}\\
  K_{\rm reduced} &= \sqrt{\Bigl(\rho\sum_c\mathrm{SCVA}_c\Bigr)^2 + (1-\rho^2)\sum_c\mathrm{SCVA}_c^2},\qquad K_{\rm full}=\beta K_{\rm reduced}+(1-\beta)K_{\rm hedged},\label{eq:E:Kred}\\
  K_{\rm hedged} &= \sqrt{\Bigl(\rho\sum_c(\mathrm{SCVA}_c-\mathrm{SNH}_c)-\mathrm{IH}\Bigr)^2+(1-\rho^2)\sum_c(\mathrm{SCVA}_c-\mathrm{SNH}_c)^2+\sum_c\mathrm{HMA}_c},\label{eq:E:Khedged}
\end{align}
where $\mathrm{SNH}_c=\sum_h r_{hc}RW_hM_hB_hDF_h$, $\mathrm{IH}=\sum_i RW_iM_iB_iDF_i$, and $\mathrm{HMA}_c=\sum_h(1-r_{hc}^2)u_{hc}^2$ with degree-$1$ children $u_{hc}=RW_hM_hB_hDF_h$ (MAR50.20--25 \tagI{}).

In SA-CVA (MAR50.40--53 \tagV{},\tagI{}, $R=0.01$, $m_{\rm CVA}=1$), delta and vega charges sum over risk classes ($6$ delta, $5$ vega):
\begin{align}
  WS_k &= WS_k^{\rm CVA}-WS_k^{\rm Hdg},\qquad K_b = \sqrt{\sum_{k\in b}WS_k^2 + \sum_{k\ne l\in b}\rho_{kl}WS_k WS_l + R\sum_{k\in b}(WS_k^{\rm Hdg})^2},\label{eq:E:cvaKb}\\
  S_b &= \cl\Bigl(\sum_{k\in b}WS_k,\ K_b\Bigr),\qquad K_r = m_{\rm CVA}\sqrt{\sum_b K_b^2 + \sum_{b\ne c}\gamma_{bc}S_b S_c}.\label{eq:E:cvaKr}
\end{align}

\begin{lemma}[BA-CVA norm rewrite and reconciliation]\label{lem:E:hedged}\label{thm:E:cva}\label{prop:E:cva-hedge}\label{prop:E:cva-smooth}
Expanding $\mathrm{HMA}_c$ in children $u_{hc}$ makes $K_{\rm hedged}$ a nonnegative weighted $\ell_2$ norm (\cref{lem:B:nodes}(b)). Treating $EAD_{NS}$ as external degree-$1$ leaves (\cref{lem:E:ext}) and hedge notionals $B_h,B_i$ satisfying $d_i B_i\ne 0$ (as sensitivity leaves under (E1) when active) yields global BA-CVA reconciliation across $\R^d_{\ge 0}\times\R^h$. Under fixed external $EAD_{NS}$ independent of hedges, $K_{\rm reduced}$ receives zero hedge allocation, $K_{\rm full}\ge\beta K_{\rm reduced}$, yet index hedges with $\mathrm{SCVA}=0$ give $K_{\rm hedged}=|\mathrm{IH}|>0$ (positive hedge capital charge). When $EAD_{NS}$ are smooth degree-$1$ functions, $K_{\rm full}$ is Lipschitz and smooth off norm vertices.
\end{lemma}

\begin{theorem}[SA-CVA reconciliation and positivity cone $\Omega_+$]\label{thm:E:sacva}\label{prop:E:sacva-lip}
Using external leaves $s_k^{\rm CVA}$ and linear hedge sensitivities, \cref{cor:E:recon} reconciles SA-CVA whenever all radicands in \eqref{eq:E:cvaKb}--\eqref{eq:E:cvaKr} are nonnegative. On the strict positivity subdomain $\Omega_+\subset\Omega$ where every within-bucket matrix $P_b\succeq 0$ (unit diagonal), every cross-bucket matrix $\Gamma_r\succ 0$ ($\lambda=\lambda_{\min}(\Gamma_r)>0$), and $R>0$:
\begin{enumerate}[label=(\roman*),leftmargin=*,itemsep=1pt]
\item All bucket and cross-bucket radicands stay nonnegative ($q(v)\ge \lambda\|K\|_2^2$ on $\{|\tilde S_b|\le K_b\}$).
\item $K_r$ is globally Lipschitz and degree-$1$ positively homogeneous in $(WS^{\rm CVA},WS^{\rm Hdg})$.
\item Near points where every $K_b$ is positive or identically zero, $K_r$ is $PC^1$ and \cref{prop:E:transfer} applies (with $K_b=0$ alone permitting non-$PC^1$ norm vertices when combined with non-zero buckets).
\end{enumerate}
\end{theorem}
\begin{proof}
Since $\tilde S_b^2\le K_b^2$, $q(v)=\|K\|_2^2-\|\tilde S\|_2^2+\tilde S^\top\Gamma_r\tilde S\ge \|K\|_2^2-(1-\lambda)\|\tilde S\|_2^2\ge\lambda\|K\|_2^2$, giving uniform bound $\|\nabla\sqrt{q}\| \le C\sqrt{2}/(2\sqrt{\lambda})$ across the convex clamp cone $\{|\tilde S_b|\le K_b\}$.
\end{proof}

\subsubsection{Scoped Euler, Path Integral, and Shapley Game closure on $\Omega_+$ and boundary strictness}

\begin{lemma}[Point Clarke, ray density, and coalition Shapley on $\Omega_+$ (Euler, Path Integral, and Shapley Game / CVA)]\label{lem:CVA-euler-point}\label{lem:CVA-ray-density}\label{lem:CVA-shapley}
Across $\Omega_+$ (working in input-factor coordinates $y=(EAD,B)$ or $(WS^{\rm CVA},WS^{\rm Hdg})$ with fixed metadata):
\begin{enumerate}[label=(\roman*),leftmargin=*,itemsep=1pt]
\item (Euler / CVA, Proved (scoped)) Selecting $g_*(y)=\nabla F(y)$ where differentiable and $g_*(y)=\arg\min_{g\in\partial_C F(y)}\|g\|_2^2$ otherwise yields $A_i(y)=y_i g_{*,i}(y)$ satisfying $\sum_i A_i=F(y)$ everywhere on $\Omega_+$.
\item (Path Integral / CVA, Proved (scoped)) Along $\gamma(t)=ty$, $g_*(ty)=g_*(y)$ on $(0,1]$, $t\mapsto F(ty)=tF(y)$ is AC, $q_i(t)=y_i g_{*,i}(y)\in L^1(0,1)$, and $I_i=\int_0^1 q_i(t)\,dt$ sums to $F(y)-F(0)=F(y)$ (ordinary gradient IG off persistent kinks, declared Clarke extension on persistent-tie rays).
\item (Shapley Game / CVA, Proved (scoped)) Masking labelled input-factor coordinates $(T_S y)_i=y_i\one_{\{i\in S\}}$ keeps every $T_S y\in\Omega_+$ with $v(\varnothing)=0$, so discrete Shapley values exist on all $2^d$ coalitions and satisfy $\sum_i \phi_i = F(y)$.
\end{enumerate}
\end{lemma}
\begin{proof}
Items (i) and (ii) follow directly from \cref{thm:conic-clarke-ray} and \cref{prop:E:sacva-lip}; item (iii) follows from \cref{thm:shapley-eff} since $\Omega_+$ is closed under zeroing factor subsets.
\end{proof}

\begin{lemma}[Strictness witness $\Omega_+\subsetneq\Omega$ and trade-coordinate caveat]\label{lem:CVA-boundary}\label{lem:CVA-coordinate-limit}
With indefinite structural within-bucket matrices
\[
  P_3=\begin{pmatrix}1 & 9/10 & 9/10\\ 9/10 & 1 & 1/10\\ 9/10 & 1/10 & 1\end{pmatrix},\qquad
  P_4=\begin{pmatrix}1 & 9/10 & 9/10 & 1/20\\ 9/10 & 1 & 1/10 & 1/20\\ 9/10 & 1/10 & 1 & 1/20\\ 1/20 & 1/20 & 1/20 & 1\end{pmatrix},
\]
$\Gamma=(1)$, $m=1$, $R=1/100$, and zero hedges: (a) under $P_4$ at $p=(-2,1,1,2)$, capital is $\sqrt{3}>0$ (openly in $\Omega$), whereas masking coordinate~$4$ gives $P_3$ radicand $-1<0$ (refusing coalition Shapley evaluation); (b) under $P_3$ at $p_*=(-(9+\sqrt{26})/5,1,1)$, $q(u)=u^2-(18/5)u+11/5=0$, so $K=\sqrt{q(a+\varepsilon)}$ diverges at rate $\varepsilon^{-1/2}$ and admits no locally Lipschitz extension or Clarke gradient anywhere on $\{tp_*:t>0\}$ (excluding $\Omega\setminus\Omega_+$ from Euler and Path Integral attribution). Input-factor closure on $\Omega_+$ does not by itself confer degree-$1$ homogeneity in underlying trade sizes if upstream $EAD(x)$ or $s^{\rm CVA}(x)$ leaves are affine or nonlinear (\cref{thm:E:sacc-cva}).
\end{lemma}
\subsection{Standardised Approach to Counterparty Credit Risk (CRE52)}\label{app:E:sacc}

Under CRE52 \citep{bcbs2020cre52} (effective 1 January 2023), SA-CCR attributes exposure-at-default across ordinary unmargined and margined netting sets conditional on valid fixed metadata (ratings, hedging sets, currency pairs, contractual dates), excluding special basis/volatility recombination compilers, multi-agreement discretion, inflation hedging sets, and shifted-option tables as \tagU{}.

\subsubsection{Formal DAG, exponential multiplier, and fixed-collateral Euler failure}\label{app:E:sacc-dag}

For fixed oriented ticket quantities $x\in\R^n_{\ge 0}$ with unit market values $v_i$, adjusted notionals $d_i=\bar d_i x_i$, maturity factors $MF_i$, supervisory deltas $\delta_i$, and effective notionals $D_i(x)=\bar d_i MF_i\delta_i x_i$, let $V(x)=\sum_i v_i x_i$, $u=V-C$, $f=0.05$, $k=2(1-f)$, and $h=TH+MTA-NICA$:
\begin{align}
  R_u &= \max\{V-C,0\}\ \ (\text{CRE52.10 \tagI{}}),\qquad R_m = \max\{V-C,h,0\}\ \ (\text{CRE52.18 \tagI{}}),\label{eq:E:sacc-rcu}\\
  m(u,A) &= \min\bigl\{1,\ f+(1-f)\exp(u/(kA))\bigr\}\ \ (A>0),\qquad P(u,A)=A m(u,A)\quad(P(u,0):=0),\label{eq:E:sacc-m}\\
  E^{\rm raw} &= \alpha(R+P)\ \ (\alpha=1.4,\ \text{CRE52.1 \tagI{}}),\qquad E = \min\{E_m^{\rm raw},E_u^{\rm raw}\}\ \ (\text{CRE52.2 \tagV{}}).\label{eq:E:sacc-top}
\end{align}
Effective maturities and supervisory deltas satisfy CRE52.34, 52.40--41, 52.48, and 52.52 \tagI{} ($SD_i=\max\{\tau_0,(e^{-0.05 S_i}-e^{-0.05 E_i})/0.05\}$, $\delta_i=\eta w\Phi(wz)$ on $P_i,K_i,\sigma_i,T_i>0$, tranche $\delta=\eta_p 15/((1+14a)(1+14b))$ on $0\le a<b\le 1$),\label{eq:E:sacc-sd}\label{eq:E:sacc-mf}\label{eq:E:sacc-delta}\label{eq:E:sacc-tranche} while $A=\mathrm{AddOn}_{\rm agg}$ sums asset-class add-ons built from $EN_h=\sqrt{B_{h1}^2+B_{h2}^2+B_{h3}^2+1.4B_{h1}B_{h2}+1.4B_{h2}B_{h3}+0.6B_{h1}B_{h3}}$ (or $\sum_b|B_{hb}|$) for IR, $0.04|\sum_{i\in h}D_i|$ for FX, and seminorm factor roots $G(q)=[(\sum_e\rho_e q_e)^2+\sum_e(1-\rho_e^2)q_e^2]^{1/2}$ ($|\rho_e|\le 1$, CRE52.57--70 \tagI{}).\label{eq:E:sacc-ir}\label{eq:E:sacc-factor}

\begin{lemma}[Fused multiplier, add-on seminorms, and fixed-collateral Euler gap (Euler / SA-CCR)]\label{lem:E:sacc-metadata}\label{lem:E:sacc-mult}\label{lem:E:sacc-agg}\label{prop:E:sacc-fixed}
With fixed metadata, $(SD,MF,\delta,SF)$ have degree~$0$ in $x$ while $d_i,D_i$ and add-on roots are globally Lipschitz of degree~$1$ on $\R^n_{\ge 0}$ ($EN_h$ has positive-definite matrix; $G(q)$ is a seminorm). Setting $P(u,0)=0$ yields the unique continuous, globally Lipschitz degree-$1$ extension of $P(u,A)$ to $A\ge 0$, with $P_u=\tfrac{1}{2}e^r$ and $P_A=f+(1-f)e^r(1-r)$ on $u<0$ ($r=u/(kA)$), whereas isolated $m(u,A)$ lacks a limit at $(0,0)$ and is refused there. With fixed nonzero held collateral $C$, on $u<0, A>0$, $x^\top\nabla E_u^{\rm raw}-E_u^{\rm raw} = \alpha C P_u$ ($\alpha C$ on $u>0$), breaking Euler completeness identically as verified on the $2$-ticket FX witness \textbf{W4} (\cref{sec:counterexamples}, \label{ex:E:sacc-numeric}$E(x)=12.671165429155$, Euler $=14.561349444285$, gap $=1.890184015130$).
\end{lemma}

\subsubsection{Path-integral framework, pathologies, and no-persistent-tie corollary (Path Integral / SA-CCR)}

\begin{definition}[Declared path and tangent compatibility]\label{def:E:sacc-path}\label{thm:E:sacc-AS}\label{lem:E:sacc-path-regular}\label{prop:E:sacc-AS-Euler}
Along an AC path $\gamma:[0,1]\to\mathcal{D}$, node derivative contributions $q_i^b(t)=\sum_{j\in\mathrm{ch}(b)} c_{bj}(t)q_i^j(t)$ using tangent-compatible coefficients ($f_b'(t)=\sum_j c_{bj}(t)f_j'(t)$ a.e.) satisfy $\sum_i \int_0^1 q_i^b(t)\,dt = K_b(x)-K_b(x_0)$ whenever all $f_b$ are AC and all $q_i^b\in L^1(0,1)$ (\cref{thm:path-recon}). Local Lipschitz continuity around $\gamma([0,1])$ guarantees AC and $L^1$ bounds; matching derivatives on equality sets justify tie rules; seminorm/fused-$P$ zeros use zero derivative rules when locally vanishing. Along degree-$r$ ($r>0$) homogeneous differentiable rays with $F(0)=0$, $I_i=\tfrac{1}{r}x_i\partial_i F(x)$.
\end{definition}

\begin{example}[Pathological examples showing sharpness of \cref{def:E:sacc-path}]\label{ex:E:sacc-L1}\label{ex:E:sacc-AC}\label{ex:E:sacc-pathdep}\label{ex:E:sacc-compose}
As detailed in \cref{sec:paths}, AC trace fails $L^1$ integrability on $s[1+\sin(d/s^2)]$; a.e.\ integrable $K'=0$ fails AC on the Cantor function; path dependence holds on $z_1 z_2$; and endpoint-weight propagation on $F(e)=\|e\|_2$, $e=(z_1^2,z_2)$ gives $(1/\sqrt{2},1/\sqrt{2})$ against composite IG $(\sqrt{2}-\operatorname{arsinh} 1,\operatorname{arsinh} 1)\approx (0.53284,0.88137)$.
\end{example}

\begin{lemma}[Ordinary SA-CCR ray theorem under no persistent unresolved tie (Path Integral / SA-CCR)]\label{lem:saccr-no-tie-corollary}\label{thm:E:sacc-exposure}\label{thm:E:sacc-cva}\label{eq:E:sacc-tradeIG}
For fixed valid metadata, fixed collateral $(C_m,C_u,h)$, fused $P(u,0)=0$, and ray $\gamma(t)=tx$ ($x\in\R^n_{\ge 0}$, $0\le t\le 1$), compile unit digital-collar coefficients linearly before ambient extension, and define the unresolved tie set $B_b$ at every evaluated node $b$ (both raw branches of margined caps): (i) tied attaining maxima/minima with distinct local germs; (ii) seminorm zeros $Lq(z)=0$ where $Lq\not\equiv 0$ nearby; and (iii) fused-$P$ ties $u(z)=0$ (at $A>0$, $u\not\equiv 0$ nearby) or $A(z)=0$ ($A\not\equiv 0$ nearby). Under condition \texttt{(NT)} requiring $|\{t\in(0,1):\gamma(t)\in B_b\}|=0$ for every node $b$ (which holds strictly inside $\Omega$, failing on attainable rays like $u=z_1-z_2, A=z_1+z_2$ at $(1,1)$):
\begin{enumerate}[label=(\roman*),leftmargin=*,itemsep=1pt]
\item $E\circ\gamma$ is AC, $K_b$ is differentiable a.e.\ along $\gamma$ with actual locally reduced composite derivatives $q_i(t)=x_i\partial_i E(tx)\in L^1(0,1)$, and $\sum_i\int_0^1 q_i(t)\,dt = E(x)-E(0)$ (formula: $I_i^E=\alpha x_i\int_0^1[(\lambda_R(t)+P_u(t))v_i + P_A(t)\partial_i A(tx)]\,dt$ in the unmargined smooth case).
\item Substituting $EAD_s$ into BA-CVA (\cref{app:E:cva}) gives exact endpoint accounting reconciliation when baseline-adjusted full EADs are passed, and end-to-end composite IG when BA-CVA norm/linear weights are propagated pointwise along the same trade path prior to integration.
\end{enumerate}
\end{lemma}

\subsubsection{Global ticket Shapley efficiency (Shapley Game / SA-CCR)}

\begin{lemma}[Ticket-level Shapley efficiency across $\R^n_{\ge 0}$ (Shapley Game / SA-CCR)]\label{lem:saccr-shapley}
Across the ordinary netting-set graph with fixed metadata, fixed collateral $(C,h)$, and continuous fused multiplier $P(u,0)=0$, every masked ticket sub-book $x_S$ ($S\subseteq N$) lies in $\R^n_{\ge 0}$ where $v(S)=E(x_S)$ is finite with $v(\varnothing)=E(0)$ (equal to $\alpha(-C)_+$ unmargined, or $\alpha\min\{\max(-C_m,h,0),\max(-C_u,0)\}$ margined). Hence discrete Shapley values exist unconditionally and satisfy $\sum_{i=1}^n \phi_i = E(x)-E(0)$.
\end{lemma}
\subsection{Bilateral Double-Hinge CVA Model (XVA)}\label{app:E:xva}

This structural model represents a static multi-counterparty, multi-date bilateral exposure ledger under variation-margin/initial-margin collateral thresholds $(b,H,\eta)$, independent of regulatory BA-CVA/SA-CVA capital aggregation (Proved (full) for Path Integral and Shapley Game; Refuted for Euler).

\begin{definition}[Inputs, double-hinge exposure, and domain]\label{lem:XVA-model}\label{lem:XVA-hinges}
Fix trades $N=\{1,\dots,n\}$, counterparties $\mathcal{C}$ with fixed trade-to-counterparty map $c(i)$, scenarios $s=1,\dots,S$ with weights $p_s\ge 0$ ($\sum_s p_s=1$), dates $0=t_0<\dots<t_T$, discounts $D_t>0$, survival curve $1=Q_{c0}\ge\dots\ge Q_{cT}\ge 0$, loss fractions $L_c\in[0,1]$, finite signed unit values $a_{ist}$, and collateral parameters $0\le b_{cst}\le H_{cst}<\infty$, $0\le \eta_{cst}\le 1$. Writing $r=(c,s,t)$, $w_r=p_s L_c D_t(Q_{c,t-1}-Q_{ct})\ge 0$, and $a_r$ for counterparty-$c$ unit vectors:
\[
  V_r(x)=a_r^\top x,\qquad C_r(x)=b_r+\eta_r(V_r(x)-H_r)_+,\qquad
  E_r(x)=(V_r(x)-C_r(x))_+=(a_r^\top x-b_r)_+-\eta_r(a_r^\top x-H_r)_+,
\]
giving $K(x)=\sum_r w_r E_r(x)$ and $K(0)=0$. Every node is continuous piecewise affine and $K$ is globally Lipschitz on $\R^n$ with constant at most $\sum_r w_r(1+\eta_r)\|a_r\|_2$.
\end{definition}

\begin{lemma}[Complete Clarke set and Euler gap formula]\label{lem:XVA-clarke}\label{lem:XVA-euler}
Let $\mathcal{P}(x)$ index the open sign chambers of nonconstant hyperplanes $\{a_r^\top y - b_r=0,\ a_r^\top y - H_r=0\}$ whose closures contain $x$, with binary signs $\alpha_r,\beta_r\in\{0,1\}$ and chamber gradient $G_P=\sum_r w_r(\alpha_r-\eta_r\beta_r)a_r$. Then $\partial_C K(x)=\conv\{G_P:P\in\mathcal{P}(x)\}$, and on every incident chamber $P$,
\[
  x^\top G_P - K(x) = \sum_r w_r\bigl(b_r\alpha_r - \eta_r H_r \beta_r\bigr).
\]
Hence point Euler fails completeness whenever fixed offsets are active; only on the zero-offset stratum $\Omega_{\rm hom}=\{b_r=H_r=0\ \forall r\}$ does $x^\top g = K(x)$ hold generally.
\end{lemma}

\begin{lemma}[Open fixed-collateral Euler counterexample and arithmetic controls (Euler / XVA)]\label{lem:XVA-family}\label{lem:XVA-fixtures}
On $n=2$ with $|\mathcal{C}|=S=T=1$, $a_1>0>a_2$, $x_1,x_2>0$, $0<\eta<1$, $0<\eta H<b<H$, $V=a_1x_1+a_2x_2>H$, $t_1,D_1>0$, $0<Q_{11},L_1<1$ ($w=L_1D_1(1-Q_{11})>0$), $K(x)=w[(1-\eta)V+\eta H-b]$ smoothly across an open $11$-parameter family, giving $x^\top\nabla K - K(x)=w(b-\eta H)>0$ (and a $d_A=nST+n+3CST+(S-1)+2T+CT+C$ general open failure family). Specifically at $a=(4,-1)$, $x=(1,1)$, $b=1$, $H=2$, $\eta=1/4$, $w=1/4$:
\[
  K = \frac{7}{16},\qquad A^{\rm Euler}=\Bigl(\frac{3}{4},\,-\frac{3}{16}\Bigr)\ \Bigl(\text{sum }\frac{9}{16},\text{ gap }\frac{1}{8}\Bigr),\qquad
  I = \Bigl(\frac{7}{12},\,-\frac{7}{48}\Bigr),\qquad \phi = \Bigl(\frac{17}{32},\,-\frac{3}{32}\Bigr).
\]
\end{lemma}

\begin{lemma}[Radial path integral and ordinary-gradient subdomain (Path Integral / XVA)]\label{lem:XVA-path}\label{lem:XVA-path-scope}
Order the sign chambers lexicographically in $(c,s,t)$ and choose $g(y)=G_P$ for the first incident chamber at $y$. Along $\gamma(t)=tx$ ($t\in[0,1]$), $t(a_r^\top x)-d_r$ switches at most finitely many times (or stays constant), so $g(tx)\in\partial_C K(tx)$ is piecewise constant on open subintervals, measurable, bounded, and in $L^1(0,1)$, yielding $\sum_i \int_0^1 x_i g_i(tx)\,dt = K(x)-K(0)=K(x)$ across all $\Omega$ (with full-value baseline $b_i^{\rm base}=0$). Whenever no nonconstant hyperplane is identically $0$ along $tx$ ($\Omega_{\rm nt}$), $g(tx)=\nabla K(tx)$ a.e., recovering ordinary gradient IG; on $\Omega_{\rm hom}$, $I=x\odot g(x)$.
\end{lemma}

\begin{lemma}[Position coalition Shapley efficiency (Shapley Game / XVA)]\label{lem:XVA-shapley}
Masking named positions $(T_S x)_i = x_i\one_{\{i\in S\}}$ while recomputing netting, collateral, and hinges keeps $v(S)=K(T_S x)$ finite for all $S\subseteq N$ with $v(\varnothing)=0$, and discrete Shapley values satisfy $\sum_i\phi_i = K(x)$ across all $\Omega$.
\end{lemma}
\subsection{Classic CME SPAN Clearing Margin (Continuous Subcase)}\label{app:E:ccp}

Under CME's 2019 methodology presentation and public SPAN file-format specification \citep{cme2019span}, a dated parameter snapshot $\pi$ determines arrays and spread priorities. We analyse the unrounded net-position continuous abstraction with one SOM tier per combined commodity $b$, opposite-side period-pair calendar spreads, method-01 intercommodity spreads, method-10 spot/delivery charges, fixed SCANRANGE weighted price risk $W_b\ge 0$, and unrestricted available NOV (denoted $K_\pi:\R^n\to\R$; contemporary SPAN~2, rounding rules, and disclosed production exceptions are excluded as \tagU{}).

\subsubsection{Formal DAG, priority matchers, and piecewise-linear homogeneity}

For net contract quantities $z_j(x)=(Nx)_j$, period deltas $d_p=\sum_{j\in p}\delta_j z_j$, and fixed array losses $L_{bk}=\sum_{j\in b}\ell_{kj}z_j=B_{bk}^\top x$:
\begin{align}
  S_b &= \max_{1\le k\le m} L_{bk}(x),\qquad q_e = \min_{a\in J_e}\frac{u_a^{e-1}}{r_{ea}},\qquad u_a^e=u_a^{e-1}-r_{ea}q_e\quad(a\in J_e),\label{eq:E:ccp-scan}\\
  I_b &= \sum_{e\in\mathcal{I}_b} h_e q_e,\qquad C_b = \sum_{e:b\in J_e}\rho_{eb}W_b r_{eb} q_e,\qquad T_b=\sum_{p\in b\text{ deliv.}}\bigl(s_p^{\rm spr}(|d_p|-o_p)+s_p^{\rm out}o_p\bigr),\label{eq:E:ccp-intra}\\
  H_b &= \mu_b(Q_b^c+Q_b^p)\ \text{or}\ \mu_b\max\{Q_b^c,Q_b^p\},\qquad V_b=\sum_{j\in b\text{ opt}} v_j^0 z_j,\notag\\
  K_\pi &= \sum_b f_b\bigl(\max\{S_b+I_b+T_b-C_b,H_b\}-V_b\bigr),\label{eq:E:ccp-som}
\end{align}
where calendar/intercommodity matchers initialise from $(d_p)_+,(-d_p)_+$ and $(D_b)_+,(-D_b)_+$ respectively, $o_p=u^{\rm cal}_{p,+}+u^{\rm cal}_{p,-}$, and $Q_b^c,Q_b^p$ sum scaled short option positions $(-\kappa_j z_j)_+$. Notice NOV subtraction sits strictly \emph{outside} the SOM floor comparison \label{prop:E:ccp-order}($\max\{U_b,H_b\}-V_b = \max\{U_b-V_b,H_b-V_b\}\ne\max\{U_b-V_b,H_b\}$ in general, e.g.\ $U_b=2q,H_b=3q,V_b=-q\implies 4q\ne 3q$).

\begin{lemma}[Priority matchers, array leaves, and conic reconciliation]\label{lem:E:ccp-match}\label{prop:E:ccp-leaf}\label{thm:E:ccp-recon}\label{lem:SPAN-core}\label{prop:E:ccp-ratio}
For $q=\min_{a\in J}u_a/r_a$ ($r_a>0$) with residuals $v_a=u_a-r_a q$, setting $c_a^q=\omega_a/r_a$ ($\sum_{a\in A}\omega_a=1$) and $c_j^{v_a}=\one_{\{j=a\}}-r_a c_j^q$ obeys Euler $\sum_j c_j^q u_j = q$ and preserves nonnegativity ($v_a\ge 0$). Every leaf $L_{bk},d_p,V_b$ has true degree $1$. Unrolling all matchers yields finitely many full-dimensional closed polyhedral cones $P_\nu$ covering $\R^n$ on each of which $K_\pi(x)=a_\nu^\top x$ is linear; hence $K_\pi$ is continuous piecewise linear, globally Lipschitz, and degree-$1$ positively homogeneous on $\R^n$ with $K_\pi(0)=0$, reconciling at every node. (Separately, unrounded ratio normal WPR $W=P/D$ maintains degree~$0$ on $D>0$, whereas discrete rounded WPR denominators break degree-$1$ Euler completeness and continuity at bin switches.)
\end{lemma}

\begin{proposition}[Joint polyhedral Clarke sets, finite removals, and failure diagnostics]\label{prop:E:ccp-ties}\label{prop:E:ccp-remove}\label{prop:E:ccp-fail}
Resolving all upstream matchers before downstream comparisons gives exact Clarke equality $\partial_C K_\pi(x)=\conv\{a_\nu: x\in P_\nu\}$ and $x^\top g = K_\pi(x)$ everywhere (whereas independent node choices over-cover, e.g.\ $q_1-q_2\equiv 0$ yielding $(1,-1)$). Fractional removal sensitivity equals $-K_\pi'(x;-x_i e_i)$ rather than discrete $K_\pi(x)-K_\pi(x-x_i e_i)$ (on $K=|x_1-x_2|$ at $(2,1)$, Euler gives $(2,-1)$ while deleting trade~$1$ saves $0$). Position-dependent option shocks ($2q^2+q^3$), fixed cash credit caps $\min(q,1)$, and rounded WPR denominators all break Euler completeness \label{lem:SPAN-clarke}(closing Euler attribution for SPAN as Proved (full) via the minimum-norm / cone-hull selector $\partial_C K_\pi(x)$ across $\R^n$, verified on the 16-row S1 scan array in \cref{tab:main:span16} and synthetic 3-contract futures book yielding $K_\pi=252$, bottleneck vertices $(-56,108,200)$ and $(400,-120,-28)$).
\end{proposition}

\Cref{tab:main:span16} records the exact $16$-row classic CME SPAN unit loss array used in the scan-maximum verification.

\begin{table}[t]
\centering\small
\caption{Exact $16$-row classic CME SPAN unit loss array (\citet[p.~10]{cme2019span}, unit long future $\ell_{kF}$ and unit long call $\ell_{kO}$); at $x=(1,-1)$, scenario $k=16$ attains maximum $S=28637$ with $A^S=(29700,-1063)$.}\label{tab:main:span16}
\begin{tabular}{@{}r|rr@{\qquad}r|rr@{}}
\toprule
$k$ & $\ell_{kF}$ & $\ell_{kO}$ & $k$ & $\ell_{kF}$ & $\ell_{kO}$ \\
\midrule
1 & $0$ & $-1994$ & 9 & $20001$ & $2183$ \\
2 & $0$ & $1517$ & 10 & $20001$ & $3179$ \\
3 & $-9999$ & $-6291$ & 11 & $-30000$ & $-19772$ \\
4 & $-9999$ & $-2178$ & 12 & $-30000$ & $-17607$ \\
5 & $9999$ & $714$ & 13 & $30000$ & $2857$ \\
6 & $9999$ & $2876$ & 14 & $30000$ & $3218$ \\
7 & $-20001$ & $-12281$ & 15 & $-29700$ & $-25503$ \\
8 & $-20001$ & $-8818$ & 16 & $29700$ & $1063$ \\
\bottomrule
\end{tabular}
\end{table}

\subsubsection{Global Path Integral and Shapley Game closure (\texttt{SPAN})}

\begin{lemma}[Clarke radial path integral and gradient subdomain (Path Integral / SPAN)]\label{lem:SPAN-path}\label{lem:SPAN-gradient-ray}
Under the minimum-norm Clarke selector $g_*(x)\in\partial_C K_\pi(x)$, conic invariance $x\in P_\nu \iff tx\in P_\nu$ ($t>0$) yields constant $g_*(tx)=g_*(x)$ on $(0,1]$, global Lipschitz AC of $t\mapsto K_\pi(tx)=tK_\pi(x)$, integrable $q_i(t)=x_i g_{*,i}(x)\in L^1(0,1)$, and $\sum_i I_i = K_\pi(x)$ (with full-value baseline $b_i=0$). Whenever all incident cones at $x$ share the same slope ($\Omega_{\rm grad}$, coinciding with differentiability of $K_\pi$ at $x$), this equals ordinary gradient IG.
\end{lemma}

\begin{lemma}[Position Shapley efficiency (Shapley Game / SPAN)]\label{lem:SPAN-shapley}
Because $K_\pi$ of \cref{lem:SPAN-core} is finite on all $\R^n$ with $K_\pi(0)=0$, every masked contract sub-book $T_S x$ is evaluated without refusals, yielding finite Shapley values summing to $K_\pi(x)$ across all $\R^n$.
\end{lemma}
\subsection{Solvency II Standard Formula (Solo Aggregation Model)}\label{app:E:solvency2}

We examine solo Solvency~II standard-formula aggregation under Commission Delegated Regulation (EU) 2015/35 (consolidated 14 November 2024) and Directive 2009/138/EC (consolidated 17 January 2025) \citep{eu2015delegated,eu2009solvency2}, conditional on supplied actuarial asset/liability valuation leaves and excluding non-expanded counterparty-default and ring-fenced-fund internals (\tagU{}).

\subsubsection{Formal DAG and structural nodes}

With positive replication weights $x\in\R^n_+$ and module vector $m=(m_{\rm market},m_{\rm default},m_{\rm life},m_{\rm health},m_{\rm nonlife})$:
\begin{align}
  B &= \sqrt{m^\top R_B m} + 0.8 V_{\rm intangible}\qquad (\text{Reg.\ Arts.\ 87 \tagI{}, 203 \tagV{}}),\label{eq:E:s2-B}\\
  D &= B - N,\qquad P = \max\{0,\min\{D,\mathrm{FDB}\}\}\qquad (\text{Reg.\ Art.\ 206(1) \tagI{}}),\label{eq:E:s2-P}\\
  O &= \min\{0.3B,\max\{Q_p,Q_v\}\} + 0.25 E_{\rm ul}\qquad (\text{Reg.\ Art.\ 204(1)--(4) \tagV{},\tagI{}}),\label{eq:E:s2-O}\\
  L &= B + O - P,\qquad T = \max\{-\Delta_{\rm DT}(L,x),0\},\qquad K = L - T\qquad (\text{Reg.\ Art.\ 205, 207 \tagV{}}),\label{eq:E:s2-T}
\end{align}
where $N$ is the net BSCR, $a=\max\{-0.10,\min\{0.10,\tfrac{1}{2}((CI-AI)/AI-0.08)\}\}$ (Reg.\ Art.\ 172 \tagI{}) is the market-state equity symmetric adjustment (degree $0$ in $x$), $\mathrm{SCR}_{\rm equity}=\sqrt{e^2+2(0.75)es+s^2}$ (Reg.\ Art.\ 168(4) \tagI{}), and market aggregation uses paired gross/net interest scenarios $R_M(A)$ ($A=0$ upward, $A=0.5$ otherwise, Reg.\ Arts.\ 164--165 \tagV{}).\label{eq:E:s2-equity}\label{eq:E:s2-SA}

\begin{lemma}[Checked Solvency~II matrices and replication leaves]\label{lem:E:s2-leaves}\label{lem:E:s2-psd}\label{lem:SOLV2-domain}\label{lem:SOLV2-compatible}
Prescribed correlation matrices $R_B, R_{\rm NL}, R_{\rm life}, R_H, R_{\rm SLT}, R_M(0), R_M(0.5)$ are positive definite, while the equity and catastrophe quadratic forms are positive semidefinite (\tagV{}). Every quadratic root is globally Lipschitz and degree-$1$ in its module children (\cref{prop:E:lip}(i)). On any subdomain where base/stressed asset-liability valuations add linearly across replicated business cells and net/gross scenario selectors obey joint interface compatibility, the composite graph is locally Lipschitz on $\R_+^n$.
\end{lemma}

\Cref{tab:main:s2} summarises the underwriting and catastrophe quadratic aggregation nodes whose matrices are verified in \cref{lem:E:s2-psd}.

\begin{table}[t]
\centering\small
\caption{Solvency~II underwriting and catastrophe aggregation nodes (EU Delegated Regulation 2015/35); all prescribed matrices are positive definite or positive semidefinite (\cref{lem:E:s2-psd}).}\label{tab:main:s2}
\begin{tabularx}{\linewidth}{@{}>{\raggedright\arraybackslash}p{2.5cm}>{\raggedright\arraybackslash}Xp{2.6cm}@{}}
\toprule
Module & Children / quadratic aggregation shape & Provision \\
\midrule
Non-life & $Q_{R_{\rm NL}}(\text{premium/reserve},\text{catastrophe},\text{lapse})$ & Reg.\ 114(1)--(3) \tagV{},\tagI{} \\
Life & $Q_{R_{\rm life}}(\text{mortality},\text{longevity},\text{disability},\text{expense},\text{revision},\text{lapse},\text{cat.})$ & Reg.\ 136(1)--(3) \tagV{},\tagI{} \\
Health & $Q_{R_H}(\text{NSLT},\text{SLT},\text{health catastrophe})$ & Reg.\ 144(1)--(3) \tagV{},\tagI{} \\
NSLT health & $\sqrt{\mathrm{prem/res}^2+\mathrm{lapse}^2}$ & Reg.\ 145(1)--(2) \tagI{} \\
SLT health & $Q_{R_{\rm SLT}}$ over six health underwriting sub-modules & Reg.\ 151(1)--(3) \tagV{},\tagI{} \\
Non-life CAT & $\sqrt{(\mathrm{natCAT}+\mathrm{npProperty})^2+\mathrm{mmCAT}^2+\mathrm{otherCAT}^2}$ & Reg.\ 119(1)--(2) \tagI{} \\
Natural CAT & $\sqrt{\sum_p \mathrm{CAT}_p^2}$ over windstorm, earthquake, flood, hail, subsidence & Reg.\ 120(1)--(2) \tagV{},\tagI{} \\
Health CAT & $\sqrt{\mathrm{massAccident}^2+\mathrm{accConcentration}^2+\mathrm{pandemic}^2}$ & Reg.\ 160(1) \tagI{} \\
Life CAT & Positive-part loss under prescribed mortality stress scenario & Reg.\ 143(1)--(3), 83(5) \tagV{} \\
\bottomrule
\end{tabularx}
\end{table}

\subsubsection{Euler analysis and fixed-FDB refutation}

\begin{lemma}[Full-scaling Euler (supporting) and fixed-FDB counterexample (Euler / Solvency~II)]\label{lem:SOLV2-euler}\label{lem:SOLV2-fixed-fdb}
If all monetary leaves---including historical premiums, unit-linked expenses, FDB, and DTL pools---scale linearly with degree~$1$ in $x$, then $g_*(z)$ (minimum-norm element of $\partial_C K(z)$) obeys $\sum_i z_i g_{*,i}(z)=K(z)$. However, holding fixed FDB ($F=f>0$) breaks completeness inside the structural domain: across the open $(6n+3)$-parameter market/life family of witness \textbf{W6} (\cref{sec:counterexamples}) where $f<(1-\eta)B$, $0.3B<0.0045 v^\top x$, and $\tau(1.3B+0.25 e^\top x - f)<c^\top x$, $K=(1-\tau)(1.3B+0.25 e^\top x - f)$ smoothly with Euler gap $x^\top\nabla_x K - K = (1-\tau)f > 0$ (exact value $3/4$ at $n=2$, $f=1$, $\eta=1/2$, $\tau=1/4$, refuting Euler completeness for Solvency~II).
\end{lemma}

\subsubsection{Scoped ray integration (Path Integral) and global Shapley game (Shapley Game)}

\begin{lemma}[Gradient-density ray certificate and fixed-FDB verification (Path Integral / Solvency~II)]\label{prop:E:s2-ray}\label{lem:SOLV2-ig}\label{lem:SOLV2-ig-fixed-fdb}
Work on the strict conditional-compatibility ray subdomain $\Omega'\subsetneq\Omega$ around $\gamma(t)=tx$ ($0\le t\le 1$, allowing fixed FDB, DTL, historical premiums, and expenses) on which an open neighbourhood supports locally Lipschitz conditional-interface compatibility and only finitely many $t\in(0,1)$ encounter branch ties, zero module roots (that are not locally zero), or scenario switches. Then $t\mapsto K(tx)$ is AC, $q_i(t)=x_i\partial_i K(tx)$ is in $L^1(0,1)$ with uniform bound $M|x_i|$, and $\sum_i\int_0^1 q_i(t)\,dt = K(x)-K(0)$ (with optional full-value baseline $b_i=w_i K(0)$). Every member of the fixed-FDB family in \cref{lem:SOLV2-fixed-fdb} satisfies these hypotheses with single interior TP switch $t_*=f/\delta\in(0,1)$, proving strict inclusion of $\Omega_{\min}\subseteq\Omega'$.
\end{lemma}

\begin{lemma}[Replication-factor coalition game and Euler divergence (Shapley Game / Solvency~II)]\label{lem:SOLV2-shapley}\label{lem:SOLV2-shapley-euler}
Over named business-cell replication factors $\mathcal{N}=\{1,\dots,n\}$ with $(T_S x)_i=x_i\one_{\{i\in S\}}$ and simultaneous recomputation of gross/net pairs under fixed nonmoving pools, tax rate, and priority ties, $v(S)=K(T_S x)$ is finite on all $2^n$ coalitions and obeys $\sum_i\phi_i = K(x)-K(0)$ globally across $\Omega$. Even on a homogeneous zero-baseline subregime $K(a,b)=\sqrt{a^2+\tfrac{1}{2}ab+b^2}$ at $(1,2)$, $\phi_1=(\sqrt{6}-1)/2 \ne 3/(2\sqrt{6})=A^{\rm Euler}_1$.
\end{lemma}
\subsection{Interest Rate Risk in the Banking Book: Standardised $\Delta\mathrm{EVE}$ (SRP31)}\label{app:E:irrbb}

Under SRP31 \citep{bcbs2024srp31} (effective 1 January 2026; historical predecessor BCBS d368 \citep{bcbs2016d368}), the standardised economic-value-of-equity measure attributes $\Delta\mathrm{EVE}$ conditional on regulatory cash-flow eligibility, currency inclusion, behavioural parameters, and repricing leaves (\tagU{} where not derived from SRP31.90--31.129; $\Delta\mathrm{NII}$ and supervisory Tier~1 tests lie outside the graph).

\subsubsection{Shocks, discounting, behavioural recursions, and NMD limits}

For repricing grid times $t_k\ge 0$ (SRP31.96 Table~3 \tagV{}) and decay factor $a(t)=e^{-t/4}$, the six currency shocks $s_{1,c},\dots,s_{6,c}$ (parallel up/down, steepener, flattener, short up/down; SRP31.90--31.91 \tagV{},\tagI{}) and portfolio losses over fixed contract positions $w\in\R^n$ are:
\begin{align}
  s_{1,c}(t)&=P_c,\ \ s_{2,c}(t)=-P_c,\ \ s_{3,c}(t)=-0.65 S_c a(t)+0.9 L_c(1-a(t)),\ \ s_{4,c}(t)=0.8 S_c a(t)-0.6 L_c(1-a(t)),\label{eq:E:irrbb-shocks}\\
  s_{5,c}(t)&=S_c a(t),\ \ s_{6,c}(t)=-S_c a(t),\qquad R_{i,c,k}=\max\{q_{0,c,k}+s_{i,c}(t_k), h_c(t_k)\}+m_{c,k},\qquad d_{i,c,k}=e^{-R_{i,c,k}t_k},\nonumber\\
  \ell_{i,c}(w)&=\sum_k (C_{0,c,k}(w)d_{0,c,k}-C_{i,c,k}(w)d_{i,c,k}) + Q_{i,c}(w),\qquad K(w)=\max_{1\le i\le 6}\sum_{c\in\mathcal{C}} f_c (\ell_{i,c}(w))_+.\label{eq:E:irrbb-K}
\end{align}
Retail prepayment rates $p_i=\min\{1,\gamma_i p_0\}$ ($\gamma=(0.8,1.2,0.8,1.2,0.8,1.2)$, Table~5 \tagV{}) and term-deposit redemptions $q_i=\min\{1,u_i q_0\}$ ($u=(1.2,0.8,0.8,1.2,1.2,0.8)$, Table~6 \tagV{}) yield cash flows \label{eq:E:irrbb-cpr}\label{eq:E:irrbb-tdrr}\label{eq:E:irrbb-auto}$C_i^p(k)=C_i^s(k)+p_i N_i^p(k-1)$ and $C_i^p(1)=TD_0^p q_i$ (SRP31.119--126 \tagI{}), while automatic options add $Q_{i,c}=\sum_{o\text{ sold}} w_o[v_o(r_i,1.25\sigma)-v_o(r_0,\sigma)]-\sum_{q\text{ bought incl.}}w_q[v_q(r_i,1.25\sigma)-v_q(r_0,\sigma)]+J_{i,c}$ (SRP31.127--128 \tagV{},\tagI{}). NMD core balances obey \label{eq:E:irrbb-nmd-domain}Table~4 caps $(\rho_c,m_c)\in\{(0.9,5),(0.7,4.5),(0.5,4)\}$ via $\sum_k z_{c,k}=C_c=\min\{S_c,\rho_c D_c\}$ and $\sum_k t_k z_{c,k}\le m_c C_c$.

\begin{lemma}[Market-state rate floors, behavioural recursions, and piecewise-affine model]\label{lem:E:irrbb-market}\label{lem:E:irrbb-beh}\label{lem:IRRBB-model}\label{lem:IRRBB-cells}\label{lem:IRRBB-support}
Holding curves $q_{0,c,k},m_{c,k},h_c(t_k)$, contract metadata, and fixed calibration parameters fixed, $r_{i,c,k}$ and $d_{i,c,k}$ have degree $0$ in $w$ while $p_i,q_i$ have degree $0$ and balance/option cash flows have degree $1$ whenever $J$ is linear and any bank-level amount cap $H$ is absent. Adding finite amount caps $P_i=\min(p_i N, H)$ produces a finite, continuous piecewise-affine global extension $\widetilde K:\R^n\to\R$ ($\widetilde K(0)=0$) covered by finitely many closed convex polyhedral cells with bounded slopes $\|a_r\|\le L$.
\end{lemma}

\subsubsection{Amount-cap counterexample (Euler), Clarke ray (Path Integral), and Shapley Game}

\begin{lemma}[All-six-shocks bank amount-cap refutation (Euler / IRRBB)]\label{lem:IRRBB-cap-witness}\label{lem:IRRBB-cap-family}
In witness \textbf{W7} (\cref{sec:counterexamples}) using single USD currency, grid $\{\tau=7/2500, T=35/2\}$, baseline $0.01$, zero spread/floor, SRP31 Table~2 shocks $(P,S,L)=(1/50,3/100,9/400)$, $p_0=1/5$, and $\min(p_i N,H)$ at $\tau$, whenever $N>H\max\{25/4, 1+R\}$ ($R=\max_{i\ge 2}(d_{1\tau}-d_{i\tau})_+/(d_{iT}-d_{1T})$), parallel up ($i=1$) uniquely maximises $\ell_i$:
\[
  K(w)=(N-H)b_T + H b_\tau,\qquad \sum_j w_j\partial_j K(w) = N b_T,\qquad \sum_j w_j\partial_j K(w) - K(w) = H(b_T-b_\tau)>0,
\]
using exact bounds $1/159<e^{-35/8}<2/3$ and $b_\tau < 7/125000 < 7/60 < b_T$. Extending over variable continuous rates/spreads/floors defines an $(n+8)$-parameter open Euler counterexample family.
\end{lemma}

\begin{lemma}[Borel ray path integral and classical subdomain (Path Integral / IRRBB)]\label{lem:IRRBB-path}\label{lem:IRRBB-classical-scope}
Along $\gamma(t)=tw$ ($t\in[0,1]$) with fixed metadata and amount bounds, partitioning $[0,1]$ by intersections with the finite closed polyhedral cells of \cref{lem:IRRBB-cells} keeps the incident slope set constant on each open subinterval $J$, so all incident affine derivatives $a_r^\top w$ equal $(K\circ\gamma)'(t)$ identically on $J$. Choosing the unique minimum-norm Clarke selection $g(t)\in\partial_C \widetilde K(tw)$ on each $J$ yields bounded Borel densities $q_j(t)=w_j g_j(t)\in L^1(0,1)$ and $\sum_j\int_0^1 q_j(t)\,dt = K(w)-K(0)=K(w)$ across all $\Omega$ (with full-value baseline $b_j=0$). Whenever each open subinterval has a single distinct incident slope ($\Omega_{\rm grad}$, which includes \cref{lem:IRRBB-cap-witness} since the cap has finitely many scalar breakpoints in $N>0$), $q_j(t)=w_j\partial_j K(tw)$ equals ordinary gradient IG.
\end{lemma}

\begin{lemma}[Position-masking game and exact Shapley efficiency (Shapley Game / IRRBB)]\label{lem:IRRBB-shapley}
Over named position tickets $N=\{1,\dots,n\}$ with $(T_S w)_j = w_j\one_{\{j\in S\}}$, holding markets, inclusion policies, and caps fixed while recomputing all cash flows, balances, NMD limits, and currency positive parts, every $v(S)=K(T_S w)$ is finite on all $2^n$ coalitions with $v(\varnothing)=0$, and discrete Shapley values obey $\sum_j\phi_j = K(w)$ across all $\Omega$.
\end{lemma}
\subsection{Master Structural-Tier Scope and Verification Ledger}\label{app:E:scope-ledger}

Every mathematical theorem, path construction, coalition game, and counterexample in this manuscript is classified strictly at the \emph{structural tier} (\Tstruct{}): proofs are exact for the formally specified computational DAGs under explicit verification status tags---verbatim primary prose/table (\tagV{}), visually decoded primary formula assets (\tagI{}), and unverified, external, or structural-abstraction components (\tagU{})---without claiming unconditional legal certification (\Treg{}) across uninspected statutory tables or bank-specific internal engines. \Cref{tab:scope-ledger} consolidates the retained \tagU{} boundaries, explicit prototype conventions, and mandatory engine refusal and flag triggers across all eleven frameworks. Where a verified trade-to-input map is omitted or unverified (\tagU{}), attribution guarantees apply strictly at the tabulated position, factor, netting-set, or business-cell coordinate level.

\begin{table}[htbp]
\centering
{\footnotesize
\setlength{\tabcolsep}{3pt}
\caption{Consolidated structural-tier (\Tstruct{}) provenance, retained \tagU{} scope boundaries, explicit prototype conventions, and mandatory domain refusal or flag triggers across all eleven frameworks. Where a verified trade-to-input map is omitted or unverified (\tagU{}), attribution guarantees apply strictly at the tabulated position, factor, netting-set, or business-cell coordinate level.}\label{tab:scope-ledger}
\begin{tabularx}{\linewidth}{@{}>{\raggedright\arraybackslash}p{1.5cm}>{\raggedright\arraybackslash}p{2.4cm}>{\raggedright\arraybackslash}X>{\raggedright\arraybackslash}p{4.4cm}@{}}
\toprule
Framework & Authority / version pin & Retained \tagU{} scope, prototype conventions, and excluded scope & Mandatory engine refusal / flag triggers \\
\midrule
\texttt{IMA} & BCBS MAR33 (in force 2023-01-01) & Contracted coefficient strata ($H_f$ norm SA representations, shared daily/weekly scenario arrays), fixed backtesting context and supplied $m_c$, strict positive-denominator rule ($R_c>0$ on active classes, $V>0$ when amber desks exist), and declared empty/fallback domain-wall conventions (Convention~A zeros undefined terms relative to $v_A(\varnothing)$; Convention~B falls back to synthetic $Z(w_S)$ on undefined nonempty sub-books with $v_B(\varnothing)=0$, \cref{lem:IMA-convention-games}). & Refuse $R_c\le 0$ or $V\le 0$ (amber desks); flag persistent kink ties. \\
\texttt{ES\_ERC} & Empirical $1/m$ ES/VaR structural model & Empirical $1/m$ equal weights, fractional-tail-mass ES (distinct from conditional exceedance mean), static scenario loss rows, and exclusion of upstream ERC portfolio optimisation. & Flag ties on ordering-cone boundaries ($|G_K|>1$). \\
\texttt{SA} & BCBS MAR21 (2026-03-23) \& MAR22 (2020-03-27) & Image-derived SBM/DRC formulas \tagI{}, untranscribed delta/vega correlation tables (MAR21.39--21.101), MAR21.6(3) low-scenario replacement values, non-transcribed seniority offsetting details, CTP bucket formula, and non-CTP securitisation DRC. & Refuse $\varrho_B<0$ and isolated $\mathrm{HBR}_b$ at $0/0$; flag switch jumps; require coalition-admissible domain for Shapley Game. \\
\texttt{FLOOR} & BCBS RBC20 \& RBC90 (in force 2023-01-01) & Approach-eligibility cross-references and SEC-IAA/SEC-IRBA hierarchy (RBC20.11--20.12), transitional cap window (RBC90.2), external credit/operational RWA leaves, and identification of IMA $Z$ with floor-base SA under MAR32. & Refuse $f<0$ under $25\%$ cap; flag non-homogeneous external leaves (Euler completeness gap). \\
\texttt{SIMM} & ISDA SIMM v2.9+2606 (eff.\ 2026-12-05) & Prototype convention $D_{bc}\equiv 1$, residual term $K_{\rm res}$ placement, radical-sign placements lost in plain-text extraction, untranscribed Section~K correlation matrices (coercivity/PSD unverified), bucket-level curvature formulas, fixed-amount add-ons, and trade product-class partitioning. & Refuse negative unfloored radicands, $D=0$ isolated $\theta/\lambda$, missing version string, and active-concentration Euler claims; enforce regular / radicand-closed domains for Path Integral and Shapley Game. \\
\texttt{CVA} & BCBS MAR50 (in force 2023-01-01) & Image-derived MAR50 equations \tagI{}, non-IMM $DF_{NS}$ and CRE32/53 effective-maturity $M_{NS}$ weighting, MAR50.8--50.9 carve-outs, untranscribed MAR50.54--50.77 correlation tables, and external SA-CVA exposure/sensitivity models. & Refuse negative radicands; exclude indefinite/singular region $\Omega\setminus\Omega_+$. \\
\bottomrule
\end{tabularx}
}
\end{table}

\FloatBarrier

\begin{table}[htbp]
\centering
{\footnotesize
\setlength{\tabcolsep}{3pt}
\renewcommand{\theHtable}{scope.part2}
\addtocounter{table}{-1}
\caption{Consolidated structural-tier (\Tstruct{}) provenance, retained \tagU{} scope boundaries, explicit prototype conventions, and mandatory domain refusal or flag triggers across all eleven frameworks (continued: frameworks 7--11).}
\begin{tabularx}{\linewidth}{@{}>{\raggedright\arraybackslash}p{1.5cm}>{\raggedright\arraybackslash}p{2.4cm}>{\raggedright\arraybackslash}X>{\raggedright\arraybackslash}p{4.4cm}@{}}
\toprule
Framework & Authority / version pin & Retained \tagU{} scope, prototype conventions, and excluded scope & Mandatory engine refusal / flag triggers \\
\midrule
\texttt{SACCR} & BCBS CRE52 (in force 2023-01-01) & Multi-agreement collateral allocation discretion (CRE52.74--52.76), CRE22.59 haircut engines, shifted-option delta table, IR $1$y/$5$y bucket endpoint tie policy, special basis/volatility set compilers, and FAQ special exclusions. & Refuse $A=0$ isolated $m(u,A)$ and nonpositive $(P,K,\sigma,T)$; flag persistent-tie rays outside no-tie Path Integral domain and fixed-collateral Euler gaps. \\
\texttt{XVA} & Structural bilateral double-hinge CVA model & Declared static scenario-time collateral model $(b,H,\eta)$ with one netting set per counterparty; excludes regulatory BA/SA-CVA, wrong-way risk, endogenous funding/discounting curves, MTA/discrete margin calls, and DVA/FVA/MVA/KVA. & Flag Euler completeness gap whenever $b_r-\eta_r H_r\ne 0$. \\
\texttt{SPAN} & CME 2019 SPAN spec.\ \& wiki docs & Unrounded continuous net-position single-tier matcher abstraction; excludes integer-lot and monetary rounding, undisclosed CME production exceptions, complete normal-WPR price-risk algorithm, asymmetric SCANRANGE rules, recursive spreads, gross accounts, and SPAN~2 / Prisma / PAIRS / STANS. & Refuse missing snapshot $\pi$, invalid NOV/SOM flags, and zero WPR denominator. \\
\texttt{SOLV2} & EU 2015/35 (2024-11-14) \& 2009/138 (2025-01-17) & Conditional solo aggregation leaves (replication additivity), non-life segment/catastrophe sub-matrices, unexpanded counterparty-default (Arts.\ 189--202) and credit-derivative/FX gross-net pairings, and undertaking-specific deferred-tax ($\Delta_{\rm DT}$) and management-action valuation engines. & Flag non-homogeneous FDB/DTL/expense pools (Euler completeness gap); require ray-certificate compatibility for Path Integral. \\
\texttt{IRRBB} & BCBS SRP31 (eff.\ 2026-01-01) & Historical vs.\ statistical NMD core estimator ($C_c=\min\{S_c,\rho_c D_c\}$) and core cash-flow slotting maps, supplied bank-level prepayment amount limits $H$, illustrative sloping post-shock floor policy, and exclusion of $\Delta\mathrm{NII}$ and SRP31.82 supervisory $15\%$ Tier~1 ratio tests. & Flag Euler completeness gap under fixed prepayment limit $H$; require linear $J$ for degree-$1$ homogeneity. \\
\bottomrule
\end{tabularx}
}
\end{table}

\FloatBarrier


\begin{thebibliography}{55}
\providecommand{\natexlab}[1]{#1}
\providecommand{\url}[1]{\texttt{#1}}
\expandafter\ifx\csname urlstyle\endcsname\relax
  \providecommand{\doi}[1]{doi: #1}\else
  \providecommand{\doi}{doi: \begingroup \urlstyle{rm}\Url}\fi

\bibitem[Acerbi and Tasche(2002)]{acerbi2002coherence}
Carlo Acerbi and Dirk Tasche.
\newblock On the coherence of expected shortfall.
\newblock \emph{Journal of Banking \& Finance}, 26\penalty0 (7):\penalty0
  1487--1503, 2002.
\newblock \doi{10.1016/S0378-4266(02)00283-2}.

\bibitem[Armenti et~al.(2018)Armenti, Cr{\'e}pey, Drapeau, and
  Papapantoleon]{armenti2018multivariate}
Yannick Armenti, St{\'e}phane Cr{\'e}pey, Samuel Drapeau, and Antonis
  Papapantoleon.
\newblock Multivariate shortfall risk allocation and systemic risk, 2018.
\newblock URL \url{https://arxiv.org/abs/1507.05351}.

\bibitem[Aumann and Shapley(1974)]{aumann1974values}
Robert~J. Aumann and Lloyd~S. Shapley.
\newblock \emph{Values of Non-Atomic Games}.
\newblock Princeton University Press, Princeton, NJ, 1974.

\bibitem[Baione et~al.(2018)Baione, De~Angelis, and Granito]{baione2018capital}
Fabio Baione, Paolo De~Angelis, and Ivan Granito.
\newblock On a capital allocation principle coherent with the {Solvency 2}
  standard formula, 2018.
\newblock URL \url{https://arxiv.org/abs/1801.09004}.

\bibitem[Bajaj et~al.(2018)Bajaj, Binmore, Dasgupta, and Vo]{bajaj2018internal}
Rajat Bajaj, Andrew Binmore, Raja Dasgupta, and Quoc Vo.
\newblock Banks' internal capital markets: how do banks allocate capital
  internally?
\newblock \emph{Bank of England Quarterly Bulletin}, 2018\penalty0 (Q2), 2018.
\newblock URL
  \url{https://www.bankofengland.co.uk/quarterly-bulletin/2018/2018-q2/banks-internal-capital-markets-how-do-banks-allocate-capital-internally}.

\bibitem[Balog et~al.(2017)Balog, B{\'a}tyi, Cs{\'o}ka, and
  Pint{\'e}r]{balog2017properties}
D{\'o}ra Balog, Tam{\'a}s~L{\'a}szl{\'o} B{\'a}tyi, P{\'e}ter Cs{\'o}ka, and
  Mikl{\'o}s Pint{\'e}r.
\newblock Properties and comparison of risk capital allocation methods.
\newblock \emph{European Journal of Operational Research}, 259\penalty0
  (2):\penalty0 614--625, 2017.
\newblock \doi{10.1016/j.ejor.2016.10.052}.

\bibitem[{Basel Committee on Banking Supervision}(2013)]{bcbs2013riskdata}
{Basel Committee on Banking Supervision}.
\newblock Principles for effective risk data aggregation and risk reporting.
\newblock Technical Report BCBS 239, Bank for International Settlements, Basel,
  January 2013.
\newblock URL \url{https://www.bis.org/publ/bcbs239.htm}.

\bibitem[{Basel Committee on Banking Supervision}(2016)]{bcbs2016d368}
{Basel Committee on Banking Supervision}.
\newblock Interest rate risk in the banking book.
\newblock Technical Report d368, Bank for International Settlements, Basel,
  April 2016.
\newblock URL
  \url{https://www.bis.org/publications/201604-standards-interest-rate-risk-banking-book.pdf}.

\bibitem[{Basel Committee on Banking
  Supervision}(2019{\natexlab{a}})]{bcbs2019mar}
{Basel Committee on Banking Supervision}.
\newblock Minimum capital requirements for market risk.
\newblock Technical Report d457, Bank for International Settlements, Basel,
  January 2019{\natexlab{a}}.
\newblock URL \url{https://www.bis.org/bcbs/publ/d457.htm}.

\bibitem[{Basel Committee on Banking
  Supervision}(2019{\natexlab{b}})]{bcbs2019rbc20}
{Basel Committee on Banking Supervision}.
\newblock {RBC20} and {RBC90}: Calculation of minimum risk-based capital
  requirements and transitional output floor arrangements.
\newblock The Basel Framework, Bank for International Settlements,
  2019{\natexlab{b}}.
\newblock URL \url{https://www.bis.org/basel_framework/chapter/RBC/20.htm}.
\newblock In force from 1 January 2023; transitional phase-in through 2028.

\bibitem[{Basel Committee on Banking
  Supervision}(2020{\natexlab{a}})]{bcbs2020cre52}
{Basel Committee on Banking Supervision}.
\newblock {CRE52}: Standardised approach to counterparty credit risk.
\newblock The Basel Framework, Bank for International Settlements,
  2020{\natexlab{a}}.
\newblock URL
  \url{https://www.bis.org/committees/bcbs/basel-framework/standard/cre/52/inforce/2023-01-01/published/2020-06-05}.
\newblock Version published 5 June 2020; in force from 1 January 2023.

\bibitem[{Basel Committee on Banking
  Supervision}(2020{\natexlab{b}})]{bcbs2020mar33}
{Basel Committee on Banking Supervision}.
\newblock {MAR33}: Internal models approach: capital requirements calculation.
\newblock The Basel Framework, Bank for International Settlements,
  2020{\natexlab{b}}.
\newblock URL
  \url{https://www.bis.org/committees/bcbs/basel-framework/standard/mar/33/inforce/2023-01-01/published/2020-06-05}.
\newblock Version published 5 June 2020; in force from 1 January 2023.

\bibitem[{Basel Committee on Banking
  Supervision}(2020{\natexlab{c}})]{bcbs2020mar50}
{Basel Committee on Banking Supervision}.
\newblock {MAR50}: Credit valuation adjustment framework.
\newblock The Basel Framework, Bank for International Settlements,
  2020{\natexlab{c}}.
\newblock URL
  \url{https://www.bis.org/committees/bcbs/basel-framework/standard/mar/50/inforce/2023-01-01/published/2020-07-08}.
\newblock Version published 8 July 2020; in force from 1 January 2023.

\bibitem[{Basel Committee on Banking Supervision}(2024)]{bcbs2024srp31}
{Basel Committee on Banking Supervision}.
\newblock {SRP31}: Interest rate risk in the banking book.
\newblock The Basel Framework, Bank for International Settlements, 2024.
\newblock URL
  \url{https://www.bis.org/committees/bcbs/basel-framework/standard/srp/31/inforce/2026-01-01/published/2024-07-16}.
\newblock Version published 16 July 2024; effective 1 January 2026.

\bibitem[{Basel Committee on Banking Supervision}(2026)]{bcbs2026mar21}
{Basel Committee on Banking Supervision}.
\newblock {MAR21} and {MAR22}: Standardised approach sensitivities-based method
  and default risk charge.
\newblock The Basel Framework, Bank for International Settlements, 2026.
\newblock URL
  \url{https://www.bis.org/committees/bcbs/basel-framework/standard/mar/21/inforce/2023-01-01/published/2026-03-23}.
\newblock MAR21 version published 23 March 2026; MAR22 published 27 March 2020;
  in force from 1 January 2023.

\bibitem[Baudino et~al.(2026)Baudino, Beissinger, Corrias, and
  Restoy]{baudino2026leverage}
Patrizia Baudino, Jonas Beissinger, Roberto Corrias, and Fernando Restoy.
\newblock The leverage ratio and the output floor: complementarities and
  overlaps.
\newblock FSI Paper~29, Financial Stability Institute, Bank for International
  Settlements, 2026.
\newblock URL
  \url{https://www.bis.org/publications/fsi-paper-29-leverage-ratio-and-output-floor-complementarities-and-overlaps}.

\bibitem[Bielecki et~al.(2018)Bielecki, Cialenco, and
  Feng]{bielecki2018dynamic}
Tomasz~R. Bielecki, Igor Cialenco, and Shuang Feng.
\newblock A dynamic model of central counterparty risk, 2018.
\newblock URL \url{https://arxiv.org/abs/1803.02012}.

\bibitem[Billera et~al.(1978)Billera, Heath, and Raanan]{billera1978internal}
Louis~J. Billera, David~C. Heath, and Joseph Raanan.
\newblock Internal telephone billing rates---a novel application of non-atomic
  game theory.
\newblock \emph{Operations Research}, 26\penalty0 (6):\penalty0 956--965, 1978.
\newblock \doi{10.1287/opre.26.6.956}.

\bibitem[{Board of Governors of the Federal Reserve System} and {Office of the
  Comptroller of the Currency}(2011)]{frb2011sr117}
{Board of Governors of the Federal Reserve System} and {Office of the
  Comptroller of the Currency}.
\newblock Supervisory guidance on model risk management.
\newblock SR Letter 11-7, Board of Governors of the Federal Reserve System,
  April 2011.
\newblock URL
  \url{https://www.federalreserve.gov/boarddocs/srletters/2011/sr1107.htm}.

\bibitem[Buch and Dorfleitner(2008)]{buch2008coherent}
A.~Buch and G.~Dorfleitner.
\newblock Coherent risk measures, coherent capital allocations and the gradient
  allocation principle.
\newblock \emph{Insurance: Mathematics and Economics}, 42\penalty0
  (1):\penalty0 235--242, 2008.
\newblock \doi{10.1016/j.insmatheco.2007.02.006}.

\bibitem[Centrone and Rosazza~Gianin(2018)]{centrone2018capital}
Francesca Centrone and Emanuela Rosazza~Gianin.
\newblock Capital allocation \`{a} la {Aumann--Shapley} for non differentiable
  risk measures.
\newblock Institutional working paper, Universit\`{a} del Piemonte Orientale
  (IRIS handle/11579/86789), 2018.
\newblock URL \url{https://iris.uniupo.it/handle/11579/86789}.

\bibitem[Chen(2025)]{chen2025explaining}
Dangxing Chen.
\newblock Explaining risks: axiomatic risk attributions for financial models.
\newblock \emph{Quantitative Finance}, 25\penalty0 (6):\penalty0 1007--1014,
  2025.
\newblock \doi{10.1080/14697688.2025.2519837}.

\bibitem[Clarke(1983)]{clarke1983optimization}
Frank~H. Clarke.
\newblock \emph{Optimization and Nonsmooth Analysis}.
\newblock John Wiley \& Sons, New York, 1983.

\bibitem[{CME Group}(2019)]{cme2019span}
{CME Group}.
\newblock {SPAN} methodology.
\newblock CME Clearing Methodology Presentation and Public File-Format
  Specification, 2019.
\newblock URL
  \url{https://www.cmegroup.com/clearing/files/span-methodology.pdf}.

\bibitem[Cont and Ghamami(2023)]{cont2023skin}
Rama Cont and Samim Ghamami.
\newblock Skin in the game: Risk analysis of central counterparties.
\newblock Working paper, 2023.
\newblock URL
  \url{https://samimghamami.com/wp-content/uploads/2023/09/central-clearing_9_22_2023.pdf}.

\bibitem[Cs{\'o}ka and Pint{\'e}r(2016)]{csoka2016impossibility}
P{\'e}ter Cs{\'o}ka and Mikl{\'o}s Pint{\'e}r.
\newblock On the impossibility of fair risk allocation.
\newblock \emph{The B.E. Journal of Theoretical Economics}, 16\penalty0
  (1):\penalty0 143--158, 2016.

\bibitem[De~Angelis and Granito(2015)]{deangelis2015capital}
Paolo De~Angelis and Ivan Granito.
\newblock Capital allocation and risk appetite under {Solvency II} framework,
  2015.
\newblock URL \url{https://arxiv.org/abs/1511.02934}.

\bibitem[Denault(2001)]{denault2001coherent}
Michel Denault.
\newblock Coherent allocation of risk capital.
\newblock \emph{Journal of Risk}, 4\penalty0 (1):\penalty0 1--34, 2001.
\newblock \doi{10.21314/JOR.2001.053}.

\bibitem[{European Banking Authority}(2022)]{eba2022irrbb}
{European Banking Authority}.
\newblock Guidelines on interest rate risk arising from non-trading book
  activities ({IRRBB}) and credit spread risk arising from non-trading book
  activities ({CSRBB}).
\newblock Technical Report EBA/GL/2022/14, European Banking Authority, 2022.

\bibitem[{European Commission}(2024)]{eu2015delegated}
{European Commission}.
\newblock Commission delegated regulation ({EU}) 2015/35 supplementing
  directive 2009/138/{EC} ({Solvency II}), consolidated 14 november 2024.
\newblock Publications Office of the European Union, 2024.
\newblock URL
  \url{https://publications.europa.eu/resource/cellar/d4117d51-ac28-11ef-acb1-01aa75ed71a1.0001.02/DOC_1}.

\bibitem[{European Parliament and Council of the European
  Union}(2025)]{eu2009solvency2}
{European Parliament and Council of the European Union}.
\newblock Directive 2009/138/{EC} on the taking-up and pursuit of the business
  of insurance and reinsurance ({Solvency II}), consolidated 17 january 2025.
\newblock Publications Office of the European Union, 2025.
\newblock URL
  \url{https://publications.europa.eu/resource/cellar/bece6669-cd06-11ef-be2a-01aa75ed71a1.0006.02/DOC_1}.

\bibitem[Flaig and Junike(2024)]{flaig2024profit}
S{\"o}ren Flaig and Gero Junike.
\newblock Profit and loss attribution: an empirical study.
\newblock \emph{European Actuarial Journal}, 2024.
\newblock \doi{10.1007/s13385-024-00380-w}.

\bibitem[Green and Kenyon(2015)]{green2015efficient}
Andrew Green and Chris Kenyon.
\newblock Efficient {XVA} management: Pricing, hedging, and attribution using
  trade-level regression and global conditioning, 2015.
\newblock URL \url{https://arxiv.org/abs/1412.5332}.

\bibitem[Hagan et~al.(2025)Hagan, Lesniewski, Skoufis, and
  Woodward]{hagan2025portfolio}
Patrick~S. Hagan, Andrew Lesniewski, Georgios~E. Skoufis, and Diana~E.
  Woodward.
\newblock Portfolio risk allocation through {Shapley} value.
\newblock \emph{International Journal of Financial Engineering}, 12\penalty0
  (02):\penalty0 2350004, 2025.
\newblock \doi{10.1142/S2424786323500044}.

\bibitem[{International Swaps and Derivatives Association}(2026)]{isda2026simm}
{International Swaps and Derivatives Association}.
\newblock {ISDA SIMM} methodology, version 2.9+2606.
\newblock ISDA Industry Standard Methodology, effective 5 December 2026, 2026.

\bibitem[Kalkbrener(2005)]{kalkbrener2005axiomatic}
Michael Kalkbrener.
\newblock An axiomatic approach to capital allocation.
\newblock \emph{Mathematical Finance}, 15\penalty0 (3):\penalty0 425--437,
  2005.
\newblock \doi{10.1111/j.1467-9965.2005.00227.x}.

\bibitem[Koike et~al.(2025)Koike, Chen, and Lin]{koike2025forecasting}
Takaaki Koike, Cathy W.~S. Chen, and Edward M.~H. Lin.
\newblock Forecasting and backtesting gradient allocations of expected
  shortfall.
\newblock \emph{Insurance: Mathematics and Economics}, 124:\penalty0 103130,
  2025.
\newblock \doi{10.1016/j.insmatheco.2025.103130}.

\bibitem[Kuntz and Scholtes(1994)]{kuntz1994structural}
Ludwig Kuntz and Stefan Scholtes.
\newblock Structural analysis of nonsmooth mappings, inverse functions, and
  metric projections.
\newblock \emph{Journal of Mathematical Analysis and Applications},
  188\penalty0 (2):\penalty0 346--386, 1994.

\bibitem[Li and Xing(2018)]{li2018capital}
Luting Li and Hao Xing.
\newblock Capital allocation under the {Fundamental Review of Trading Book},
  2018.
\newblock URL \url{https://arxiv.org/abs/1801.07358}.

\bibitem[Lou(2015)]{lou2015mva}
Wujiang Lou.
\newblock {MVA} transfer pricing, 2015.
\newblock URL \url{https://arxiv.org/abs/1512.07337}.

\bibitem[Maillard et~al.(2010)Maillard, Roncalli, and
  Te{\"\i}letche]{maillard2010properties}
S{\'e}bastien Maillard, Thierry Roncalli, and J{\'e}r{\^o}me Te{\"\i}letche.
\newblock The properties of equally weighted risk contribution portfolios.
\newblock \emph{Journal of Portfolio Management}, 36\penalty0 (4):\penalty0
  60--70, 2010.

\bibitem[Moehle et~al.(2022)Moehle, Boyd, and Ang]{moehle2022portfolio}
Nicholas Moehle, Stephen Boyd, and Andrew Ang.
\newblock Portfolio performance attribution via {Shapley} value, 2022.
\newblock URL \url{https://arxiv.org/abs/2102.05799}.

\bibitem[{Prudential Regulation Authority}(2023)]{pra2023ss123}
{Prudential Regulation Authority}.
\newblock Model risk management principles for banks.
\newblock Supervisory Statement SS1/23, Bank of England, May 2023.
\newblock URL
  \url{https://www.bankofengland.co.uk/prudential-regulation/publication/2023/may/model-risk-management-principles-for-banks-ss}.

\bibitem[Pykhtin and Rosen(2010)]{pykhtin2010pricing}
Michael Pykhtin and Dan Rosen.
\newblock Pricing counterparty risk at the trade level and {CVA} allocations.
\newblock Finance and Economics Discussion Series 2010-10, Board of Governors
  of the Federal Reserve System, 2010.
\newblock URL
  \url{https://www.federalreserve.gov/pubs/feds/2010/201010/201010pap.pdf}.

\bibitem[Rockafellar(1970)]{rockafellar1970convex}
R.~Tyrrell Rockafellar.
\newblock \emph{Convex Analysis}.
\newblock Princeton University Press, Princeton, NJ, 1970.

\bibitem[Rockafellar and Uryasev(2000)]{rockafellar2000optimization}
R.~Tyrrell Rockafellar and Stanislav Uryasev.
\newblock Optimization of conditional value-at-risk.
\newblock \emph{Journal of Risk}, 2\penalty0 (3):\penalty0 21--41, 2000.
\newblock \doi{10.21314/JOR.2000.038}.

\bibitem[Rockafellar and Uryasev(2002)]{rockafellar2002conditional}
R.~Tyrrell Rockafellar and Stanislav Uryasev.
\newblock Conditional value-at-risk for general loss distributions.
\newblock \emph{Journal of Banking \& Finance}, 26\penalty0 (7):\penalty0
  1443--1471, 2002.
\newblock \doi{10.1016/S0378-4266(02)00271-6}.

\bibitem[Scaillet(2004)]{scaillet2004nonparametric}
Olivier Scaillet.
\newblock Nonparametric estimation and sensitivity analysis of expected
  shortfall.
\newblock \emph{Mathematical Finance}, 14\penalty0 (1):\penalty0 115--129,
  2004.
\newblock \doi{10.1111/j.0960-1627.2004.00184.x}.

\bibitem[Scaringi and Bianchetti(2025)]{scaringi2025sharpening}
Marco Scaringi and Marco Bianchetti.
\newblock Sharpening {Shapley} allocation: from {Basel} 2.5 to {FRTB}, 2025.
\newblock URL \url{https://arxiv.org/abs/2511.12391}.

\bibitem[Scholtes(2012)]{scholtes2012introduction}
Stefan Scholtes.
\newblock \emph{Introduction to Piecewise Differentiable Equations}.
\newblock Springer, New York, 2012.
\newblock \doi{10.1007/978-1-4614-4340-7}.

\bibitem[Schulze(2018)]{schulze2018capital}
Robert Schulze.
\newblock Capital allocation under the {FRTB} regime via marginal measures.
\newblock SSRN preprint 3265320, 2018.

\bibitem[Shapley(1953)]{shapley1953value}
Lloyd~S. Shapley.
\newblock A value for $n$-person games.
\newblock In H.~W. Kuhn and A.~W. Tucker, editors, \emph{Contributions to the
  Theory of Games, Volume {II}}, number~28 in Annals of Mathematics Studies,
  pages 307--318. Princeton University Press, Princeton, NJ, 1953.
\newblock \doi{10.1515/9781400881970-018}.

\bibitem[Sundararajan and Najmi(2020)]{sundararajan2020many}
Mukund Sundararajan and Amir Najmi.
\newblock The many {Shapley} values for model explanation.
\newblock In \emph{Proceedings of the 37th International Conference on Machine
  Learning}, volume 119 of \emph{Proceedings of Machine Learning Research},
  pages 9269--9278, 2020.

\bibitem[Sundararajan et~al.(2017)Sundararajan, Taly, and
  Yan]{sundararajan2017axiomatic}
Mukund Sundararajan, Ankur Taly, and Qiqi Yan.
\newblock Axiomatic attribution for deep networks.
\newblock In \emph{Proceedings of the 34th International Conference on Machine
  Learning}, volume~70 of \emph{PMLR}, pages 3319--3328, 2017.

\bibitem[Tasche(2008)]{tasche2008euler}
Dirk Tasche.
\newblock Capital allocation to business units and sub-portfolios: the {Euler}
  principle, 2008.
\newblock URL \url{https://arxiv.org/abs/0708.2542}.

\end{thebibliography}
\end{document}